%% file: article.tex
\documentclass[preprint]{elsarticle}

\usepackage{graphicx}
\usepackage{amsfonts}
\usepackage{amsmath}
\usepackage{amssymb}
\usepackage{euscript}
\usepackage{float}
\usepackage{color}
\usepackage{verbatim}
\usepackage{caption}
\usepackage{subcaption}
\usepackage{mathtools}
\usepackage[11pt]{moresize}
\usepackage{xspace}
\usepackage{mathrsfs}
\usepackage{epstopdf}
\usepackage{slashbox}  
\usepackage{hyperref}  

\usepackage{xcolor}
\hypersetup{
    colorlinks,
    linkcolor={blue!50!black},
    citecolor={blue!50!black},
    urlcolor={blue!80!black}
}

\usepackage{ifthen}
\usepackage{enumerate}
\usepackage{cases}
\usepackage{yfonts}

\input {macros.tex}

\usepackage{proof}              
\usepackage{tikz}
\usepackage{pst-node}
\usepackage{tikz-cd}

\begin{document}


\title{Unification and  combination of a class of  traversal strategies made with pattern matching and fixed-points\tnoteref{funding}}
\tnotetext[funding]{This work was supported by LABEX ACTION ANR-11-LABX-0001-01 and
by the European Territorial Cooperation Programme INTERREG IV project OSCAR.}

\author{Walid Belkhir}
\ead{walid.belkhir54@gmail.com}
\author{Nicolas Ratier\corref{cor1}}
\ead{nicolas.ratier@femto-st.fr}
\author{Duy Duc Nguyen}
\ead{nduyduc1989@gmail.com}
\author{Michel Lenczner}
\ead{michel.lenczner@femto-st.fr}
\address{FEMTO-ST Institute, Time and Frequency Department,
Univ. Bourgogne Franche-Comt\'{e} (UBFC), ENSMM, CNRS,
15B avenue des Montboucons, 25030 Besan\c{c}on cedex, France}
\cortext[cor1]{Corresponding author}

\begin{abstract}
Motivated by an ongoing project on computer aided derivation of asymptotic models governed by partial differential equations,
we introduce a class of term transformations that consists of traversal strategies and insertion of contexts.
We define unification and combination operations on this class which  amount to merging transformations in order to obtain more complex ones.
We show that the unification and combination operations enjoy nice algebraic properties like associativity, congruence and the existence of neutral elements.
The main part of this paper is devoted to proving that the unification and combination operations are correct. 
\end{abstract}


\maketitle

\input{introduction_corrected.tex}
\newpage
  \tableofcontents
\newpage

\input{preliminaries_corrected.tex} 
\input{elementary-ces_corrected.tex} 
\input{strategy-ces_corrected.tex}
\input{combination-ces_corrected.tex}
\input{main-results_corrected.tex}
\input{outline-proof_corrected.tex}
\input{properties-ces_corrected.tex}
\input{correction-combination-free_corrected.tex}
\input{correction-combination-definitions.tex}

\input{correction-combination-unfoldings_corrected.tex}
\input{correction-combination_corrected.tex}
\input{algebraic_properties_unif_corrected.tex}
\input{conclusion_corrected.tex}

\bibliographystyle{elsarticle-num}
\bibliography{article}

\newpage
\input{appendix.tex}
\end{document}

%% file: macros.tex
\newtheorem{theorem}{\textbf{Theorem}}

\newtheorem{claim}[theorem]{Claim}

\newtheorem{corollary}[theorem]{Corollary}
\newtheorem{definition}[theorem]{\textbf{Definition}}
\newtheorem{example}[theorem]{Example}
\newtheorem{myexample}[theorem]{Example}

\newtheorem{lemma}[theorem]{\textbf{Lemma}}
\newtheorem{notation}[theorem]{Notation}
\newtheorem{notations}[theorem]{Notations}

\newtheorem{remark}[theorem]{Remark}

\newtheorem{assumptions}[theorem]{Assumptions}
\newtheorem{simplifications}[theorem]{Simplifications}

\newenvironment{proof}[1][Proof]{\textbf{#1.} }{\ \hfill \rule{0.5em}{0.5em}}

\newtheorem{fact}[theorem]{Fact}

 \newtheorem{proposition}[theorem]{\textbf{Proposition}}

\newcommand{\N}{\mathbb{N}}
\newcommand{\dom}{\mathrm{dom}}
 \newcommand{\mycal}[1]{\mathcal{#1}}
 \newcommand{\PosSet}{\mathbb{N}^{\omega}_{\varepsilon}}
 \newcommand \set[1]{\{{#1}\}}
 \newcommand \bset[1]{\big\{{#1}\big\}}
 \newcommand {\E}   {\varepsilon}
 \newcommand \tor {\text{ or }}
  \newcommand \mbf[1] {{#1}}

 \newcommand \tif {\text{ if }}
 
 \newcommand \mmax {\mathtt{max}}
 \newcommand \tifthen[2] {\mathtt{\mathbf{If}}\,{#1} \, \mathtt{\mathbf{Then}}\, {#2}}
 \newcommand \tst {\text{ s.t. }}
 \newcommand{\Dom}     {{\mathsf {Dom}}}
 \newcommand{\gvert}{\;\;|\;\;}
 
 \newcommand \match[2]{#1 \ll #2}
 \newcommand \funset[1]{#1 \rightarrow #1}

\newcommand \oces {\textsf{CE}-strategies\xspace}
\newcommand \oce  {\textsf{CE}-strategy\xspace}

\newcommand \ces {\textsf{HCE}-strategies\xspace}
\newcommand \ce  {\textsf{HCE}-strategy\xspace}

\renewcommand \ces {$\mathtt{T}_\mu$-strategies\xspace}
\renewcommand \ce  {$\mathtt{T}_\mu$-strategy\xspace}

\newcommand \emptylist {\mathfrak{f}}

 \newcommand{\sembrackk}[1]{[\![#1]\!]}
 \newcommand \uberEq[1]{\stackrel{#1}{=}}

 \newcommand \fail {\mathbb{F}}
 \newcommand \totherwise {\textrm{ otherwise}}
 \newcommand{\PPos}{{\cal P}os}
 
 \newcommand{\comb}{\curlyvee}
 \newcommand{\combb}{\curlywedge}
 \newcommand{\nfcombb}{\ {\curlywedge }\ } 
 \newcommand \tand {\textrm{ and }}
 \newcommand{\fixset}{\mycal{Z}}
 
 \newcommand \most {\mathtt{Most}}

 \newcommand{\eceSet}{\mathcal{E}}
 \newcommand{\ceSet}{\mathcal{C}}
 \newcommand{\ceSetFree}{\mathcal{C}_0}
 \newcommand{\preceSet}{\mathcal{P}}

\newcommand{\D}{\mathbf{D^{\star}}}
\newcommand \bb[1] {\hat{#1}}
\newcommand \mf[1] {\mathbf{#1}}

\newcommand{\tuple}   [1] {\langle{#1}\rangle}

\newcommand{\h}{\mathbf{h}}

 \newcommand {\TDDD}  {\mathtt{OuterMost}}
 \newcommand {\IDDD}  {\mathtt{InnerMost}}

 \newcommand \twhere {\text{ where }}
 
 \newcommand \tiff {\text{ iff }}
 \newcommand{\ceSetCan}{\mathcal{C}}
 \newcommand \tthen  {\text{ then }}
 \newcommand{\memsalab}{\texttt{MEMSALab}\xspace}

\newcommand{\arity}{\mathit{ar}}

\newcommand \bigand {\bigwedge}

\newcommand \uand {\wedge}

\newcommand \morphy {\,\mathfrak{R}\,}
\newcommand \morphyy {\,\mathfrak{S}\,}

\newcommand \subcer {\sqsupset}

\newcommand{\reduce}{\;\; \rightarrow \;\;}  
\newcommand{\reduces} {\rightarrow}         
\newcommand{\xreduces}[1] {\xrightarrow[]{#1}}

\newcommand{\qsim}  {$(\ceSet,\ceSetFree)$-quasi-simulation\xspace}

\newcommand \Eu[1]{\EuScript{#1}}

\newcommand{\fresh}[1]{{\mathtt{fresh}({#1})}}

\newcommand{\boundv}[1]{{\mycal{B}ound({#1})}}

\newcommand{\NF}{\mathtt{\mycal{NF}}}
\newcommand{\Unif}{\mathfrak{U}}

\newcommand{\hideProofs}[1]
{\ifthenelse{\equal{#1}{true}}{TRUE}{\NewEnviron{killcontents}{}
                                     \let\proof\killcontents
                                     \let\endproof\endkillcontents}
  \ifthenelse{\equal{#1}{false}}{FALSE}{}
}

\newcommand \ufold[2]{\rho({#1},#2)}
\newcommand \bufold[2]{\rho\big({#1},#2\big)}
\newcommand \bbufold[2]{\rho\Big({#1},#2\Big)}

\renewcommand \ufold[2]{\rho_{#2}({#1})}
\renewcommand \bufold[2]{\rho_{#2}\big({#1}\big)}
\renewcommand \bbufold[2]{\rho_{#2}\Big({#1}\Big)}

\newcommand \unravel [1]{\mathcal{U}({#1})}
\newcommand \bunravel[1]{\mathcal{U}\big({#1}\big)}

\newcommand{\ceSetEquiv}{\ceSet\slash\hspace{-0.18cm}\equiv}

\newcommand\sbullet[1][.5]{\mathbin{\vcenter{\hbox{\scalebox{#1}{$\bullet$}}}}}

\newcommand\nbr{\#}
\newcommand \nbrr[1] {\Omega_{\mu}{#1}}

\newcommand{\m}{\textup{\texttt{-}}} 

\newif\ifhideproofs
\ifhideproofs
\usepackage{environ}
\NewEnviron{hide}{}
\let\proof\hide
\let\endproof\endhide
\fi

%% file: introduction_corrected.tex
\section{Introduction}
The general context of this article is the incremental design of complex models using a notion of extension.
The models we are considering are described by abstract terms and subjected to symbolic transformations.
The latter are assumed to rely on two fundamental operations: the operation of extension that transforms a reference object to a more complex one by enriching it, and the operation of combination that merges several extensions to produce a new one that incorporates all the characteristics and effects of those used for its generation. This process is guided by the semantics of the objects in question, namely the way the extensions operate.

We briefly recall the background of this work. Our motivation originates in an undergoing project for the modeling and simulation of complex systems in micro or nano-technologies, e.g.~\cite{YanBel2013,belkhir2014symbolic,belkhir:SYNASC:15}. The systems under consideration are governed by Partial Differential Equations (PDEs) and are too complex to be simulated by straightforward numerical methods, unless at the time-scale of design engineering. In addition, asymptotic methods for PDEs have been an active domain of mathematics for more than seventy years whose main goal is to derive \textquotedblleft simpler\textquotedblright\ PDEs from those which have small parameters in their geometry or their equations. These methods are called singular perturbation methods in physics. 

They are developed in all fields where PDEs are used for modeling ranging from physics,
biology, finance etc, see for instance the review paper~\cite{charalambakis2010homogenization}.

The use of asymptotic methods for modeling and simulation leads to reduced computation times while retaining the essence of the models. Nevertheless, they suffer from a major drawback which limits their diffusion in the community of engineers which is that their derivation is done on a case by case basis. In other words, for each new problem, the entire process of deriving the model must be redone from scratch even if the new problem has many functionalities in common with one or more problems already modeled. It follows that despite the immense number of existing models, relatively few of them are used in general simulation software. 

Our group has adopted an alternative approach by developing a software package called \memsalab (for MEMS Array Lab)~\cite{EuroSim11} whose function is to build asymptotic models by successive extensions which intend to take into account different characteristics such as scalar or vector forms of solutions, various a priori estimates on the solutions and the sources, thinness or periodicity of geometries, several nested substructures, etc.

Our approach takes advantage of the modularity and the algebraic nature of asymptotic methods by following the approach presented in~\cite{LenSmi07}. It is also based on the so-called \emph{combination of extensions}~\cite{belkhir:SYNASC:15} method that we are now sketching.

- First, remember that the construction of an asymptotic model operates on a PDE comprising small parameters. The construction of an asymptotic model consists in passing to the limit to zero on the small parameters which requires several proof steps leading to a new PDE. The latter can be implemented in generic simulation software. Among all pairs of input PDEs and proofs of asymptotic models, the simplest one is chosen, which is the pair of \textit{reference PDE} and \textit{reference proof}, from which the others can be constructed by successive complexifications. Technically, a proof is implemented by a rewriting strategy, that is to say by a series of transformations made up of rewriting rules accompanied by strategies which specify the way in which the rewriting rules  are applied to PDEs.

- Second, the reference proof is complexified, we say that it is extended, in several ways to take into account new functionalities giving new proofs. This is done by applying an extension to the reference proof in so far as an extension is another kind of rewriting strategies. This results in an \textit{extended proof}. Then, applying an extended proof to a complexified PDE yields a new asymptotic model. 

- Finally, to cover several new elementary functionalities, we merge several elementary extensions by an operation of \textit {combination}. The extension resulting from a combination can itself be applied to the reference proof. It follows a new proof which applied to a complex PDE produces an asymptotic model benefiting from all the characteristics.

To illustrate the concept of combination of extensions, consider the input PDE of reference and the proof of reference both containing the term $\partial_{x} v(x)$, an extension that adds an index $j$ on the variable $x$ of derivation and a second extension which adds an index $i$ on the derived function $v$. The application of each of these two extensions to the reference term yields the terms $\partial _{x}v_{i}(x)$ and $\partial_{x_{j}}v(x)$, respectively. The combination of these two extensions would be another extension that, when applied to the reference term, yields $\partial _{x_{j}}v_{i}(x)$. 

In summary, there are three levels for PDEs, proofs and extensions. A proof can be applied to a PDE, an extension to a proof, and extensions can be combined to produce new extensions. Therefore, combining extensions related to several elementary features allows for building, in an incremental way, new proofs and therefore new asymptotic models.

Although the concepts of extension and combination were introduced for the first time in~\cite{YanBel2013}, in that earlier work the combination of extensions was done by composition, not allowing for conflicts between extensions.  A conflict between two extensions arises when they modify the same part of a proof and when the application of one of them creates new possibilities for the application of the other one. In that restricted framework, the combination of non conflicting extensions simply amounts to their composition. The complete principle of the extension-combination method was introduced in~\cite{belkhir:SYNASC:15} where a user language was defined for specifying proofs and extensions as rewriting strategies. We also have defined the combination on a small class of extensions. However, the question of the correctness of the combination formulas was left open.

When defining a new class of strategies with an operation of combination, there are many difficulties to overcome. A careful attention must be payed to the choice of the constructors out of which these strategies are built up. There are two extreme ways to proceed. One way is to build the strategies by means of the most rudimentary constructors, as in~\cite{combination:2019:rudi-strategies-faillure}. This makes the strategies hard to use in practice due to their huge size. But the advantage of such rudimentary constructors is to allow one to understand the mechanisms behind the combination operation and to define it correctly. Even more, to proceed in this way was inevitable and justifies our work~\cite{combination:2019:rudi-strategies-faillure}.

 The other extreme way is to rather design high level strategy constructors which are easy to use in practice. But this makes it hard to understand the mechanisms behind the combination operation since a high level constructor hides several rudimentary constructors. For instance, given a rule $r$, the translation of the high level strategy $\mathtt{OuterMost}(r)$ into rudimentary constructors requires  three rudimentary constructors since $\mathtt{OuterMost}(r)$ can be written as $\mathtt{\mu X. (r \oplus Most(X))}$, where ``$\mu$" stands for the fixed-point or the recursion constructor, ``$\oplus$" stands for the left-choice constructor, and  ``$\mathtt{Most}$" is the one-step constructor that accesses to all the children of a term if viewed as a tree.  In this case, one has to define the combination of two high level constructors in just one step which is usually difficult or even impossible. Furthermore, this raises the question of the closure of such class of strategies under combination since, for instance, the combination of two $\mathtt{OuterMost}$s is not an $\mathtt{OuterMost}$.  In earlier attempts, we figured out that the combination of extensions based on high-level strategies such as \texttt{BottomUp} or \texttt{TopDown} or \texttt{OuterMost}  can not be expressed with high-level strategies, making such a class not closed under combination.  We thus understood that more rudimentary strategy constructors were needed.

In~\cite{combination:2019:rudi-strategies-faillure} we followed the first way and introduced the large class of \emph{context embedding strategies}, or \oces for short, that involves elementary and, more importantly, an explicit handling of failures which are produced when an application of an extension fails. 

We proved the correctness of the combination operation for a fragment of the class of CE-strategies. The drawback of working with this class is that the definition of the traversal navigation strategies such as \texttt{OuterMost} yields a \oce whose size depends on the signature. Even worse, the size of the resulting combined CE-strategy can be exponential with respect to the size of the two input CE-strategies. 
In this paper we overcome these difficulties by pursing a third way which is in between the two extreme ways exposed above. We introduce another class of strategies, called the \ces, which is built up using both high level and rudimentary constructors inspired by the propositional modal $\mu$-calculus~\cite{rudimemt:mu-calculus:book,BraWal15} rather than strategy languages as in~\cite{RewriteStrat_CHK2003}.
The $\mu $-calculus-like approach involves natural and rudimentary strategy constructors, especially the jumping to a position and the recursion with the fixed-point operator. This makes tractable the question of language closure for combinations.
Moreover, the procedure of combination of \ces together with their verification is also much simplified. Although this new class of \ces is less expressive that the class of \oces of~\cite{combination:2019:rudi-strategies-faillure}, the new class remains powerful enough to be used in practice and its closure is harder to achieve since it incorporates high level and rudimentary constructors, this makes this paper more complete than~\cite{combination:2019:rudi-strategies-faillure}, since the  \oces of \cite{combination:2019:rudi-strategies-faillure} can not be used in practice because of their huge size. Besides, we define a unification and combination operations for the class of \ces.
Roughly speaking, the unification of two \ces   amounts to construct a \ce that  captures the effect of both insofar as they are compatible, where  the compatibility of two \ces depends on each input term and is related to their successful application. The incompatible effects are covered by the combination.
This class enjoys similar algebraic properties as \oces with respect to unification and combination, like associativity, congruence and the existence of a neutral and an absorbing element. The main result of this paper shows that the unification and combination operations of \ces are correct with respect to a correctness criterion that we shall devise and that is guided by the semantics of \ces, Subsection~\ref{correctness:criterion:unif:comb:subsec}.

We notice that the size of the resulting combined \ce is polynomial with respect of the size of the two input \ces. The \ces are reasonably easy to use in practice and they have been implemented and used in \memsalab software in the previous years within a user specification language of mathematical expressions, proofs and extensions and their combination for asymptotic models where the first applications targeted micro and nanotechnology~\cite{YanBel2013,belkhir2014symbolic,belkhir:SYNASC:15}.

The concept of extension, sometimes called refinement in the literature, is developed in different contexts such as the parallel and concurrent systems, for example in~\cite{Luca-Action-Refinement:book,Gorrieri20011047-action-refinement,action-refinement:2020} the refinement is done by replacement of components with more complex components. Combination principles are present in different areas of application, they involve different techniques but follow the same key idea that consists of the merging of structures or algorithms motivated mainly by the incremental design of complex systems by integration of simple and heterogeneous subsystems. For instance, the works in combination of logics~\cite{Ghilardi03algebraicand,Combining:Logics:13}, algorithms and verification methods~\cite{TAP:2015}, decision procedures~\cite{Combining:decision:procedures:MannaZ02}, the  composition and synthesis for service-oriented and agent-oriented computing~\cite{SHENG2014:service:composition,JENSEN201928} in which the synchronous and asynchronous product of automata and transition systems are a form of combination, and the unification of grammars in linguistics~\cite{kahane:grammar:unification:polar:04,kahane-grammar:unification:05,francez_wintner_2011:book-unification-grammars,Richard2017UnOP:outils, Unification-Based-Tree-Adjoining-Grammars-91}. However, the integration of the two concepts of extension and combination seems to have not been addressed in the literature.

\subsection*{Organization of the paper. }
The paper is organized as follows. \\
Section~\ref{Preliminaries} is devoted to a review of the useful concepts of rewriting theory, and to definitions and notations.\\
In Section~\ref{Implement:by:position:sec} we introduce the position-based \ces and their combination. \\
In Section~\ref{Implement:by:strategies:sec} we introduce the larger class of \ces together with their semantics. \\
In Section~\ref{unification:combination:section} we give the unification procedure. \\ 
In Section~\ref{main:results:sec} we state the  results of this paper without proof, namely the correctness of the unification and combination and their algebraic properties. \\
In Section~\ref{structure:proof:main:results:sec} we expose a detailed outline of the proof of the main result, that is, the correctness of the unification of \ces. \\
In Section~\ref{Psi:construction:section}  we construct a mapping that is needed in the formulation of the correctness criterion  of the unification and combination of \ces. \\ 
In Section~\ref{proof:correction:fixed-point-free:sec} we prove the correctness of the unification of the fixed-point free fragment of \ces, that is, the \ces without the fixed-point constructor.\\
In Sections~\ref{correction-combination-definitions:section},~\ref{correction:unif:general:setting:unfold:sec} and ~\ref{equiv:unif:with:unif:unfolding:sec}  we develop the notions and tools
as well as the intermediary results required in the proof of the main result.  This is the technical core of the paper. \\ 
In Section ~\ref{proof:main:results:section}  we sum up the results of the previous three sections and prove the correctness of the unification and combination for the full class of \ces,
from  which we prove  the important algebraic properties of the unification and combination. \\  
In Section~\ref{conclusion:sec} we give a summary, few concluding remarks and we announce future work.  \\
To improve the readability of the paper, some proofs are given  in the Appendix.


%% file: preliminaries_corrected.tex

\section{Preliminaries: terms, substitution, notations, rewriting}
\label{preliminaries:section}
\label{Preliminaries} We introduce preliminary definitions and
notations.

\paragraph*{Terms, contexts}
Let ${\mathcal{F}}=\cup _{n\geq 0}{\mathcal{F}}_{n}$ be a set of symbols
called \emph{function symbols}. The \emph{arity} of a symbol $f$ in ${\mathcal{F}}_{n}$ is $n$ and is denoted $\arity(f)$.
Elements of arity zero are called \emph{constants} and often denoted by the letters $a,b,c,$ etc.
The set ${\mathcal{F}}_{0}$ of constants is always
assumed to be not empty.  Given a denumerable set ${\mathcal{X}}$ of \emph{variable} symbols, the set of \emph{terms} $\mathcal{T}\left( \mathcal{F},\mathcal{X}\right)$ is the smallest set containing ${\mathcal{X}}$ and such that $f(t_{1},\ldots ,t_{n})$ is in $\mathcal{T}\left( \mathcal{F},\mathcal{X}\right) $ whenever
$ar(f)=n$ and $t_{i}\in \mathcal{T}\left(\mathcal{F},\mathcal{X}\right) $ for $i\in \lbrack 1..n]$.
Let   $\square \not\in {\mathcal{X}}$ be an extra variable, the set $\mathcal{T}_{\square }(\mathcal{F},\mathcal{X})$ of
contexts, denoted simply by $\mycal{T}_{\square}$, is made with terms with symbols in $\mathcal{F}\cup \mathcal{X}\cup \{\square \}$ which always includes exactly one occurence of $\square$.
Evidently, $\mathcal{T}_{\square }(\mathcal{F},\mathcal{X})$ and $\mathcal{T}(\mathcal{F},\mathcal{X})$ are two disjoint sets.
For a term $t$ and a context $\tau$, we shall write $\tau[t]$ for the term that results from the replacement
of $\square$ by $t$ in $\tau$.
We shall write simply $\mycal{T}$ (resp. $\mycal{T}_{\square}$) instead of  $\mathcal{T}\left( \mathcal{F},\mathcal{X}\right)$
(resp. $\mathcal{T}_{\square }(\mathcal{F},\mathcal{X})$). We denote by $\mathcal{V}ar\left({t}\right) $ the set of variables occurring in $t$.
We shall write $\arity(t)$ to mean the arity of of the symbol at the root of $t$.
\paragraph*{Positions, prefix-order, substitution}
Let $t$ be a term in $\mathcal{T}\left(\mathcal{F},\mathcal{X}\right)$.
The position $\varepsilon $ is called the root position of  $t$, and the function or variable symbol at this position is called the root symbol of $t$. A position in a tree is a sequence of integers, i.e., an element in $\PosSet=\{{\varepsilon }\}\cup \mathbb{N}\cup (\mathbb{N}\times \mathbb{N}) \cup \cdots$.
In particular we shall write $\mathbb{N}_{\E}$ for $\set{\E} \cup \mathbb{N}$, such positions are called unitary positions.
Given two positions $p=p_{1}p_{2}\ldots p_{n}$ and $q=q_{1}q_{2}\ldots q_{m}$,
the \emph{concatenation} of $p$ and $q$, denoted by  
  $pq$, is the position $p_{1}p_{2}\ldots p_{n}q_{1}q_{2}\ldots q_{m}$.
We notice that in  the examples, when we write, for instance, the position $12$, we mean the concatenation of $1$ and $2$, and not the twelfth position.

The set of positions of the term $t$, denoted by
$\mathcal{P}os\left( t\right)$, is a set of positions of positive
integers such that, if $t\in \mathcal{X}$ is a variable or $t\in
\mycal{F}_{0}$ is a constant, then $\mathcal{P}os\left( t\right)
=\left\{ \varepsilon \right\} $. If $t=f\left( t_{1},...,t_{n}\right) $
then $\mathcal{P}os\left(t\right)=\left\{ \varepsilon \right\} \cup
\bigcup_{i=1,n}\left\{ip\mid p\in \mathcal{P}os\left( t_{i}\right)
\right\}$.

The prefix order defined as $p\leq q$  iff there exists $p^{\prime }$  such that $pp^{\prime }=q$, is a partial order on positions.
If $p^{\prime }\neq \varepsilon$ then we obtain the strict order $p<q$.
We write $\left( p\parallel q\right) $ iff $p$ and $q$ are incomparable with
respect to $\leq$. The binary relations $\sqsubset$ and $\sqsubseteq$ defined by $p \sqsubset q \quad \text{ iff } \quad \big(p < q \tor p\parallel q \big)$ and $p \sqsubseteq q \quad
\text{ iff } \quad \big(p\le q \tor p\parallel q \big) $, are total relations on positions.

For any $p\in \mathcal{P}os(t)$ we denote by $t_{|p}$ the subterm
of $t$ at position $p$, that is, $t_{|{\varepsilon }} =t$, and $f(
t_{1},...,t_{n})_{|iq} =(t_{i})_{|q}$.  For a term $t$, we shall
denote by $\delta(t)$ the depth of $t$, defined by $\delta(t_0)=1$, if $t_0 \in \mycal{X} \cup \mycal{F}_0$ is a variable or a constant,
and $\delta(f(t_1,\ldots,t_n)) = 1+ max(\delta(t_i))$, for $i=1,\ldots,n$.
For any position $p\in \mathcal{P}os\left( t\right)$ we denote by $t\left[ s \right]_{p}$ the term obtained by replacing the subterm of $t$ at
position $p $ by $s$: $t[s]_{\varepsilon } =s$ and $f(t_{1},...,t_{n})[s]_{iq} =f(t_{1},...,t_{i}[s]_{q},...,t_{n})$.

A substitution is a mapping $\sigma :\mathcal{X}\rightarrow \mathcal{T}(\mathcal{F},\mathcal{X})$ such that $\sigma (x)\neq x$ for only finitely many $x$'s. The finite set of variables that $\sigma $ does
not map to themselves is called the domain of $\sigma $: $\Dom(\sigma )\overset{def}{=}\left\{ x\in \mathcal{X}\gvert\sigma (x)\neq x\right\} $. If
$\Dom(\sigma )=\left\{ x_{1},...,x_{n}\right\} $ then we write
$\sigma $ as: $\sigma =\left\{ x_{1}\mapsto \sigma \left(x_{1}\right) ,...,x_{n}\mapsto \sigma\left( x_{n}\right) \right\} $.

A substitution $\sigma :\mathcal{X}\rightarrow {\mathcal{T}(\mathcal{F},\mathcal{X})}$ uniquely extends
to an endomorphism $\widehat{\sigma }:\mathcal{T}(\mathcal{F},\mathcal{X})\rightarrow \mathcal{T}(\mathcal{F},\mathcal{X})$ defined by: $\widehat{\sigma }(x)=\sigma (x)$ for all $x\in \Dom(\sigma )$,
$\widehat{\sigma }(x)=x$ for all $x\not\in {\mathsf{Dom}}(\sigma )$, and $\widehat{\sigma }(f(t_{1},\ldots ,t_{n}))=f(\widehat{\sigma }(t_{1}),\ldots ,\widehat{\sigma }(t_{n}))$ for $f\in \mathcal{F}$.
In what follows we do not distinguish between a substitution and its extension.

For two terms $t,t^{\prime }\in \mycal{T}$, we say that $t$ matches $t^{\prime }$, written $\match{t}{t'}$, iff there exists a substitution $\sigma$, such that $\sigma(t)=t^{\prime }$. It turns out that if
such a substitution exists, then it is unique. A substitution $\sigma'$ is subsumed by a substitution $\sigma$ iff
$\sigma'(t)$ matches  $\sigma(t)$ for any term $t$.

A most general unifier of the two terms $t$ and $t^{\prime }$ is a
substitution $\gamma $ such that $\gamma (t)=\gamma (t^{\prime })$
and, for any other substitution $\gamma ^{\prime }$ satisfying
$\gamma ^{\prime }(t)=\gamma ^{\prime }(t^{\prime })$, we have that
$\gamma ^{\prime }$ is subsumed by $\gamma $. The most general unifier is unique up to a variable renaming.

The composition of functions will be denoted by ``$\circ$''.
The set of all subsets of a set $S$ will be denoted by $\wp(S)$.
For a finite set $S$, we write $|S|$ for the number of elements of $S$.
For a finite  set $S$ of integers, the maximum (resp. minimum) of $S$ will be denoted by $\max(S)$ (resp. $\min(S)$).
\paragraph{Lexicographic ordering}
A lexicographic ordering, denoted by "$<$", on the Cartesian product $\mathbb{N}^{n}=\mathbb{N}\times \ldots \times \mathbb{N}$ ($n$-times), where $n\ge 1$, is inductively defined for any $(a_1,\ldots,a_n)$ and $(b_1,\ldots,b_n)$ in $\mathbb{N}^{n}$ such that $(a_1,\ldots,a_n)<(b_1,\ldots,b_n)$ iff either $n=1$, and in this case $a_1<b_1$. Or $n\ge 2$, and in this case  either \emph{i.)} $a_1<b_1$ or \emph{ii.)} $a_1=b_1$ and  $(a_2,\ldots,a_n)<(b_2,\ldots,b_n)$.


%% file: elementary-ces_corrected.tex
\section{Position-based \ces  and their combination}
\label{Implement:by:position:sec}
The operation of combining \ces requires an abstract operation of merging contexts, a concrete example of which will be provided. The algebraic properties of the combination will be presented in the general case.
\begin{definition}[Merging of contexts]
  \label{Def:Merge:contexts}
Any associative binary operation
  \begin{align*}
    \sbullet: \mycal{T}_{\square} \times \mycal{T}_{\square} \rightarrow \mycal{T}_{\square} 
  \end{align*}
  is called \emph{merging of contexts}.
 \end{definition}

\begin{example}[Merging of contexts by composition]
  \label{ExamCombinationContexts}
     We give an example  of the operation of  merging  of contexts, denoted by  "$\sbullet$", as follows:
        \begin{align*}
          {\tau} \sbullet {\tau}'   =\tau [\tau ^{\prime }]_{\mathcal{P}os\left( \tau,\square \right)}
        \end{align*} 
        where ${\mathcal{P}os\left( t,\square \right)}$ is the position of $\square$ in $t$. 
         This kind of merging   has been introduced in \cite{yang2014contribution} and implemented in  \texttt{MEMSALab} software. 
For instance, the merging of the two contexts $\tau_1=\mathtt{Index}(\square,{i})$ and $\tau_2=\mathtt{Index}(\square,{j})$, used for inserting indices to mathematical variables or functions, is given by 
    \begin{equation*}
        \tau_1 \sbullet \tau_2  = \tau_1[\tau_2]_{1} = \mathtt{Index}(\mathtt{Index}(\square,{j}),{i}),
    \end{equation*}
    where ${i}$ and ${j}$ are terms.
\end{example}

To define the position-based \ces, we introduce two particular  position-based strategies as follows. Firstly, for a position $p$ and a  context $\mbf{\tau}$, we define the jump strategy $@p.\mbf{\tau}$ that, when applied to a term $t$, it  inserts $\mbf{\tau}$ at the position $p$ of  $t$.
Secondly, we define the failing strategy $\emptylist$ that fails when applied to any term.
Their precise semantics are given in Definition \ref{semantics-posi-based:def} of the semantics of position-based \ces.

\begin{definition}[Position-based \ces]
\label{posi-based:def}
Let $p_1,\ldots,p_n$ be  positions in $\PPos$ and  $\tau_1,\ldots,\tau_n$   be  contexts in $\mycal{T}_{_{\square }}$ with $n\ge 1$.
A \emph{position-based \ce} is either the failing strategy $\emptylist$ or  the ordered conjunction 
\begin{align*}
\bigand_{i=1,n} @p_{i}.\tau_i. 
\end{align*} 
The set of position-based \ces is denoted by $\eceSet$. 
\end{definition}
Notice that the order of positions in $\bigand_{i=1,n} @p_{i}.\tau_i$ matters.
We impose that the position-based \ces follow  some constraints regarding the  positions of insertions to avoid conflicts: the order of context insertions must  go  up from the leaves to  the root.
Formally, 

\begin{definition}[Well-founded position-based \ce]
\label{Well-founded:simple:ext:def} 
Let $p_1,\ldots,p_n$   positions in $\PPos$ and  $\tau_1,\ldots,\tau_n$   be contexts in $\mycal{T}_{_{\square }}$ with $n\ge 1$.
A position-based \ce $E$
\begin{align*}
E = \bigand_{i=1,n} @p_{i}.\tau_i
\end{align*} 
is \emph{well-founded} iff
\begin{enumerate}[i.)]
\item every position  occurs at most once  in $E$, i.e. $p_i \neq p_j$ for all $i\neq j$, and
\item lower positions  appear earlier in $E$, i.e.  $i < j$ if $p_i \sqsubset p_j$, for all $i,j \in [1,n]$.
\end{enumerate}
 Moreover, the  position-based \ce $\emptylist$ is well-founded. 
\end{definition}

In all what follows we work only with the set of well-founded
position-based \ces, still denoted by  $\eceSet$. For two position-based
\ces $E$ and $E'$, we shall abuse of notation and write $E=E'$ to
mean that they are equal up to a permutation of their parallel
positions. 
For a position $p$, we let
\begin{align*}
  @p.\bigand_{i=1,n} @p_{i}.\tau_i=\bigand_{i=1,n} @pp_{i}.\tau_i.
\end{align*}

We next define the semantics
of a position-based \ce as a function in $\funset{\mycal{T} \cup
\set{\fail}}$, with the idea that if the application of a
position-based \ce to a term fails, the result is $\fail$. Besides,
we adopt  a stronger  version of  failure, that is,
$\bigand_{i=1,n} @p_{i}.\tau_i$ fails  when each  of
$@p_i.{\tau}_i$ fails. To formalize this notion of failure  we need
to introduce  an intermediary function
\begin{align*}
\eta: (\funset{\mycal{T} \cup \set{\fail}})  \rightarrow  \mycal{T}\cup \set{\fail}
\rightarrow  \mycal{T}\cup \set{\fail},
\end{align*}
that stands for the \emph{fail as identity}. It is defined for any function $f$ in
$\funset{\mycal{T}\cup \set{\fail}}$  and any term $t \in
\mycal{T}\cup \set{\fail}$ by
\begin{align*}
(\eta(f))(t) =  \begin{cases}
f(t) & \tif f(t)  \neq \fail, \\
t    & \totherwise.
\end{cases}
\end{align*}
The semantics of position-based \ces follows.
\begin{definition}[Semantics of position-based \ces]
\label{semantics-posi-based:def} The \emph{semantics} of a position-based
\ce $E$ is a function $\sembrackk{E}$ in $\funset{\mycal{T}\cup
\set{\fail}}$ inductively  defined by:
\begin{align*}
\sembrackk{\emptylist}(t) & \uberEq{def} \fail,   \\
\sembrackk{E}(\fail) & \uberEq{def} \fail, \\
\sembrackk{@p. {\tau}}(t)  & \uberEq{def}
\begin{cases}
    t{[{\tau}[t_{|p}]]}_{p} & \tif p \in \PPos(t) \\
    \fail & \totherwise,
\end{cases}  \\
\sembrackk{\bigand_{i=1,n} @p_{i}.\tau_i}(t)  &
\uberEq{def}
\begin{cases}
  \Big(\big(\eta(\sembrackk{@p_n.\mbf{\tau}_n})\big) \circ \cdots\circ  \big(\eta(\sembrackk{@p_1. \mbf{\tau}_1})\big)\Big) (t) & \textrm{if } \exists p_i \in \set{p_1,\ldots,p_n} 
   \textrm{ s.t. }  p_i \in \PPos(t) \\
  \fail  &  \textrm{ otherwise}.
\end{cases}
\end{align*}
Two \ces $E$ and $E'$ are said to be \emph{semantically equivalent}, if and only if  $\sembrackk{E}(t)=\sembrackk{E'}(t)$, for any term $t$.
\end{definition}
Notice that two position-based \ces are semantically equivalent iff they are equal up to a renaming of parallel positions.

\begin{myexample}\label{Example:SemanticOfCEStrategy}
  We illustrate with an example of  position-based \ces with their application to a term in \texttt{MEMSALab}.
Consider the  two contexts $\tau_1= \mathtt{Index}(\square,i)$ and $\tau_2= \mathtt{Index}( \square,j)$.
Applying the position-based \ce  $@\varepsilon.\tau_1$  to the term  $t=\mathtt{Var}(x,\Omega)$ gives the
transformation of a space variable $x$ defined on a domain $\Omega$ to its coordinate $x_{i}$. 
The procedure is given by
\begin{align*}
    \sembrackk{@\varepsilon.\tau_1}(t) = t[\tau_1[t_{|\varepsilon}]]_\varepsilon = t[\tau_1[t]]_\varepsilon = \tau_1[t]_{\mathcal{P}os(\tau_1,\square)}
         =\mathtt{Index}\left(\mathtt{Var} ( x,\Omega), {i}\right).
\end{align*}
Let  $\tau = \tau_1 \sbullet \tau_2$, where "$\sbullet$" stands for the operation of merging of contexts by composition  as defined in Example \ref{Def:Merge:contexts}.
The application of the position-based \ce $@\varepsilon.\mbf{\tau}$ to the term $t$  gives
\begin{align*}
  \sembrackk{@\varepsilon.\mathbf{\tau}}(t) = \tau[t]_{\mathcal{P}os(\tau_1[\tau_2],\square)}   
    = \mathtt{Index}\left(\mathtt{Index}\left(\mathtt{Var}\left( x, \Omega \right),{j} \right),{i} \right).
\end{align*}
\end{myexample}

\begin{myexample}\label{Example:SemanticOfCEStrategy1}
We  illustrate an application of a position-based \ce on the derivative of a function represented by the term $t'={\partial}_{{x}}{u}$, where
${u}$ is  the derived function,
${x}$ is the  the mathematical variable and
${\partial}_{{x}}$ the derivation operator with respect to ${x}$. 
Let $\tau_1= \mathtt{Index}(\square,i)$ and $\tau_2= \mathtt{Index}( \square,j)$ be contexts.
The application of the position-based \ce $@1.\tau_1 \wedge @2.\tau_2 $  to $t'$ yields the term ${\partial}_{{x}_{{j}}}{u}_{{i}}$.
Since the positions $1$ and $2$ are parallel, this  \ce  is well-founded and its application to $t'$ yields
\begin{align*}
  \sembrackk{@1.\tau_1 \wedge  @2.\tau_2}(t') = (\sembrackk{@1.\tau_1} \circ \sembrackk{@2.\tau_2})(t')  = \sembrackk{@1.\tau_1}(\sembrackk{@2.\tau_2}(t')) 
                                                                                                         = \partial_{x_j}{u_i}.
\end{align*}
The tree structures of $\tau_1$, $\tau_2$, $\partial_x{u}$ and ${\partial}_{{x}_{{j}}}{u}_{{i}}$ are depicted in Figure \ref{fig:JumpStrategy}.
\end{myexample}

\begin{figure}[]
    \centering
    \begin{tikzpicture}[level distance=1cm,
    level 1/.style={sibling distance=1.2cm},
    level 2/.style={sibling distance=1.2cm},
    scale=0.95]
    \node {$\mathtt{Index}$}
    child {node {$\square$}}
    child {node{${i}$}
    };
    \begin{scope}[xshift=2.5cm,level 1/.style={sibling distance=1.2cm}]
    \node {$\mathtt{Index}$}
    child {node {$\square$}}
    child {node{${j}$}
    };
    \end{scope}
    \begin{scope}[xshift=6cm,level 1/.style={sibling distance=1.2cm}]
    \node {${\partial}$}
    child {node {${u}$} edge from parent node[left,draw=none] {}}
    child {node {${x}$} edge from parent node[right,draw=none] {}};
    \end{scope}
    \begin{scope}[xshift=10cm,level 1/.style={sibling distance=1.5cm}, level 2/.style={sibling distance=0.8cm}]
    \node {${\partial}$}
    child {node {$\mathtt{Index}$}
        child{ node {${u}$}} child{ node {${i}$}}
        edge from parent node[left,draw=none] {}}
    child {node {$\mathtt{Index}$}
        child {node {${x}$}} child{ node {${j}$}}
        edge from parent node[right,draw=none] {}};
    \end{scope}
    \end{tikzpicture}
    \caption{The tree structure of the contexts $\tau_1= \mathtt{Index}(\square,i)$ and $\tau_2= \mathtt{Index}( \square,j)$, and the term $\partial_{x}u$.
        The term $\partial_{{x}_{{j}}}{u}_{{i}}$ results from   the application of the position-based \ce  $@1.\tau_1 \wedge  @2.\tau_2$ to $\partial_{x}u$,
              as discussed  in Example \ref{Example:SemanticOfCEStrategy1}.} 
    \label{fig:JumpStrategy}
\end{figure}

The unification of two position-based \ces amounts to sort and merge their positions, and to merge  their contexts if they are inserted at the same position.
To simplify the following Definition \ref{unif:posi:non:empty:def},
when  unifying  position-based \ces $E$ and $E'$ in the general case (\ref{item:2:posi:non:empty:def}), we can assume without loss of generality that each  of them contains an insertion at the root position $\varepsilon$,
because otherwise one can add to each of them  the identity  insertion  $@\varepsilon.\square$ that leaves unchanged any term to which it is applied.

\begin{definition}[Unification of position-based \ces]
\label{unif:posi:non:empty:def} The \emph{unification}  of two position-based \ces is the binary operation $\combb: \eceSet\times
\eceSet\longrightarrow \eceSet$
defined as
\begin{enumerate}
\item
  \begin{enumerate}[(a)]
  \item \label{final:E:1} $\emptylist \combb E  = \emptylist.$
  \item  \label{final:E:2}  $E \combb \emptylist =   \emptylist.$ 
  \end{enumerate}
\item \label{item:2:posi:non:empty:def} If $E=\bigand_{p_i \in I} @p_i.\tau_i \uand @\varepsilon.\tau$  and $E'=\bigand_{q_j \in J} @q_j.\tau'_j \uand @\varepsilon.\tau'$, for two partially ordered sets $I$ and $J$ of positions, then 
  \begin{align*}
  E \combb E'  =  \bigand_{p_i \in I \cap J} @p_i.(\tau_i\sbullet \tau'_i) \uand  R \uand  R'\uand @\varepsilon.(\tau\sbullet \tau'),
  \end{align*}
where \begin{align*} R &= \bigand_{p_i \in I\setminus J} @p_i.\tau_i &\tand &&
                     R'&= \bigand_{q_j \in J\setminus I} @q_j.\tau'_j. && 
   \end{align*}
\end{enumerate}
\end{definition}
Notice that since one can reorder the positions of  $R \uand  R'$, then  the unification of two well-founded position-bases \ces can be turned into an  equivalent  well-founded one, i.e. into a unique  (up to a permutation of parallel positions)  well-founded position-based \ce.

\begin{myexample}
  Consider position-based \ces
  \begin{align*}
  E=@p_1.\tau_1 \wedge @p_2.\tau_2 \wedge @p_3.\tau_3  &&\tand  && E'=@p_1.\tau'_1 \wedge @q_1.\tau'_2  \wedge @q_2.\tau'_3,
  \end{align*}
  and the sets of their positions are $P=\{p_1,p_2,p_3\}$ and $P'=\{p_1,q_1,q_2\}$, respectively. Hence $P \cup P'=\{p_1,p_2,p_3,q_1,q_2\}$ and $P \cap P' = \{p_1\}$. The unification of $E$ and $E'$ is 
    \begin{equation*}
        E \combb E'=@p_1.(\tau'_1 \sbullet  \tau_1)\wedge @p_2.\tau_2\wedge @p_3.\tau_3\wedge @q_1.\tau'_2\wedge @q_2.\tau'_3.
    \end{equation*}
\end{myexample}

For practical reasons, we need to introduce  the combination of two position-based \ces in  the same way as their unification apart that
the combination of a position-based \ce with the failure is the identity.

\begin{definition}[Combination  of two position-based \ces]
\label{comb:posi:def}
The \emph{combination} of two  position-based \ces is a binary operation
$\comb:  \eceSet   \,\times\,  \eceSet   \longrightarrow \eceSet$
defined for any     $E$ and $E'$ in $\eceSet$ by
\begin{align*}
 E \comb  E'=
\begin{cases}
E \combb E'  & \tif E\neq \emptylist \tand E' \neq \emptylist \\
E  & \tif E\neq \emptylist \tand   E'=\emptylist \\
E'  & \tif E = \emptylist.
\end{cases}
\end{align*}
\end{definition}

The algebraic properties of the unification and the combination of position-based \ces are stated in the following
Propositions \ref{main:prop:elemntary:og:prop:1} and \ref{main:prop:elemntary:og:prop:2}, respectively.  
\begin{proposition}
\label{main:prop:elemntary:og:prop:1} 
 The set $\eceSet$ of position-based \ces  together with the unification  operation enjoy the following properties.
    \begin{enumerate}
    \item The neutral element of the  unification   is $@\varepsilon.\square$, 
    \item The absorbing element of the unification is $\emptylist$,
    \item The unification is  associative, i.e. $(E \combb E') \combb E'' = E \combb (E' \combb E'')$. 
    \item The unification  of position-based \ces  is (non-)commutative  if and only if   the operation  "$\sbullet$" of merging of  contexts  is  (non-)commutative. 
    \item The unification    is idempotent if and only if  the operation of merging of the contexts is idempotent, 
          that is, $E \combb  E = E$  for any $E\in \eceSet$  iff  $\tau \sbullet \tau = \tau$  for any contexts  $\tau$ in $\mycal{T}_{\square}$.
\end{enumerate}
\end{proposition}
\begin{proposition}
\label{main:prop:elemntary:og:prop:2}
 The set $\eceSet$ of position-based \ces  together with the unification and combination operations enjoy the following properties.
    \begin{enumerate}
    \item The neutral element of the  combination   is $\emptylist$.
    \item The combination is associative, i.e.  $(E\comb E')\comb E'' = E\comb (E'\comb E'')$.
    \item The combination  of position-based \ces  is (non-)commutative if and only if  the operation of merging of the contexts  "$\sbullet$" is (non-)commutative.
    \item The combination is idempotent  iff the operation  "$\sbullet$" of merging of  contexts is idempotent.
    \end{enumerate}
\end{proposition}

The proof of these propositions does not provide any difficulties since the properties of  associativity, (non)-commutativity,  and idempotence of the unification and combination  are  inherited from
their counterpart properties of the merging  of contexts. 


%% file: strategy-ces_corrected.tex
\section{The class of \ces}
\label{Implement:by:strategies:sec}

As far as the unification is concerned, designing a class of strategies  faces  the following challenging issues:
\emph{1.)} finding the right class of extensions that is closed by combination: a less expressive class would not be closed under combination nor useful in practice, while very expressive extensions are impossible to combine,
\emph{2.)} finding the right basic constructors of the extensions:  very rudimentary constructors would make the size of the extensions very huge and non-practical, while   more general constructors are very hard to combine, 
\emph{3.)} combining  the  "\emph{while}" loops, or iterations,  is the most difficult part and requires a special care,
\emph{4.)} proving the correctness of the combination by taking into account the semantics of the extensions.

We introduced the position-based \ces
to clarify the ideas behind contexts, their insertion as well as their combination.
However,  position-based \ces are not satisfactory  for practical
applications, since the positions are generally not flexible, not accessible and cannot be
used on a regular basis in applications.
So, we enrich this  framework  by supplementing position-based \ces with  navigation  strategies  to form a class of \emph{\ces}  which  is closed under  combination.

\paragraph{Syntax and  semantics of \ces}
A \ce is composed of two parts: a navigation of the input term
without changing it, and an insertion of contexts at certain
positions. The navigation part is built up  using
the left-choice  constructor ($\oplus$),
a conditional constructor ``$\mathtt{if}\textrm{-}\mathtt{then}$'',
a pattern-matching ``$u;S$'' with a pattern $u$,
the $\most(S)$ constructor  that applies $S$ to all the children of the input term,
the jump constructor $@i.S$ to a position $i$ as well as a conjunction of such jumps,
and the fixed-point constructor (``$\mu$'') allowing the recursion in the definition of strategies.
The resulting class is called  the class of  \ces, which stands for \emph{traversal strategies with fixed-points}.

In what follows we assume that there is an enumerable set of \emph{fixed-point variables} denoted by $\fixset$. Fixed-point variables in  $\fixset$
will be denoted by $X,Y,Z,\ldots$

\begin{definition}[Grammar of \ces]
 \label{Def:HCE-strategies}
The class of \ces  is  defined by the following grammar:
\begin{align*}
 S  \; ::= \;  &\emptylist \gvert X \gvert @\varepsilon.\tau \gvert u;{S} \gvert  S \oplus S   \gvert   \mu X.S  \gvert  @i.S \gvert   @i_1.S \uand @i_2.S   \gvert   \most({S})  \gvert  \tifthen{S}{S} \\
\end{align*}
where $X$ is a fixed-point variable in $\fixset$, and $\mbf{\tau}$
is  a  context   in $\mycal{T}_{\square}$, and $u$  is a  term   in
$\mycal{T}$, and $i,i_1,i_2$ are unitary  positions in
$\mathbb{N}_{\varepsilon}$. The set of \ces will be
denoted by $\ceSet$.
The subset of fixed-point free \ces will be denoted by $\ceSetFree$. 
\end{definition}
\paragraph{Notations}
We shall  write "$\tifthen{S_1 \&S_2}{S}$" instead   of   $\tifthen{S_1}{\big(\tifthen{S_2}{S}\big)}$.
If a set of positions $I$ is empty, then the \ce $\bigwedge_{i\in I} @i.S_i$ is just the failure $\emptylist$.

We notice that extending  the class of \ces by allowing the  position $i$ of the jump constructor $@i.S$ to range over arbitrary   positions  in  $\PosSet$ instead of
unitary positions in $\mathbb{N}_{\varepsilon}$ does not increase the expressiveness of the strategy language. This can be achieved by  turning each \ces
$@p.S$, where $p$ is a position in $\PosSet$ into
$@q_1.\cdots.@q_n.S$, with $p=q_1.\cdots.q_n$ and each $q_j$ is a unitary  position in $\mathbb{N}_{\varepsilon}$.

The design of the class of \ces  is inspired by the $\mu $-calculus formalism~\cite{rudimemt:mu-calculus:book}
since  we  need   very rudimentary strategy constructors.
In particular the jumping into the  immediate positions of the term tree
is morally similar to  the diamond and box modalities ($\langle \cdot \rangle$ and $[ \cdot ]$) of the propositional modal $\mu$-calculus.
And the fixed-point constructor is much finer than the iterate operator of e.g.~\cite{RewriteStrat_CHK2003}.
Besides, we incorporate the left-choice strategy constructor and a pattern matching operation.

An occurence of a  fixed-point variable  $X$ is \emph{bound}  in a \ce $S$ if it is under the scope of a "$\mu X$".  Otherwise, it is said \emph{free}.
The set of bound variables of $S$ will be denoted by $\boundv{S}$. A \ce is \emph{closed} if all its fixed-point variables are bound.
We shall sometimes write  $S(X)$  to emphasize that the fixed-point variable $X$ is free in $S(X)$.


\begin{myexample}\label{Example:MuStrategy}
We informally illustrate the semantics of \ces through an example.
Consider the  \ce  defined by $S(X)=(u;\mbf{\tau}) \oplus (@1.X)$ and its iteration $\mu X.S(X)$, where  $u$ is a term and  $\tau$ is a context.
 When applied to a term $t$, the  \ce $\mu X.S(X)$  checks first whether $u$ matches with $t$.
If it is the case, then the context $\tau$ is inserted at the root of $t$ and stops, yielding the term $\tau[t]$. Otherwise,
the \ce jumps to the position $1$ of $t$, i.e. the  left-most  child of $t$,  and  reiterates  the procedure by applying  $\mu X.S(X)$ to  this child. If it reaches
the left-most  leaf of $t$  with which  $u$ does not match, then the \ce $\mu X.S(X)$ fails on $t$.
For instance, the application of $\mu X.S(X)$ to the term  $f(v,f(u,f(u,a)))$ gives $f(v,f(\tau[u],f(u,a)))$, while
it fails on $f(v,f(f(a,u),u))$.
\end{myexample}

\begin{remark}\label{navigation:part:remark}
  Notice that a \ce is composed of two parts: \emph{i.)} a navigation part that consists of the navigation strategies that browse the input term without changing it.
  These strategies are the pattern matching, the left-choice, the iteration, the jump, the conjunction,  the $\most$, and the ``$\mathtt{if}\textrm{-}\mathtt{then}$''.  
  And, \emph{ii.)} an insertion part that modifies the input term and consists of an insertion of contexts.
\end{remark}

\begin{definition}[Unfolding]
  \label{ufold:iteration:def}
For any  \ce  $S(X) \in \ceSet$, and $n \ge 1$,  we define the \emph{unfolding} of  $\mu X.S(X)$  which replaces the fixed-point operator on $X$ by $n$-iterations as follows
\begin{align*}
\mu^{0}X.S(X) \uberEq{def} \emptylist &&\tand && \mu^{n}X.S(X) \uberEq{def} S(\mu^{n-1}X.S(X)).
\end{align*}
\end{definition}

\begin{example}[Unfolding]
  \label{ufold:iteration:example}
   For a  pattern $u \in \mycal{T}$ and a context  $\tau$, let
    \begin{align*}
      {S}(X)  &=(u;@\varepsilon.\tau) \oplus @1.X
    \end{align*}
    be  a \ce.
   We give examples of the replacement of the fixed-point operator of  $\mu X.S(X)$  by $n$-iterations, for $n=0,1,2$,  as follow:
    \begin{align*}
     \mu^{0} X.S(X)  &=  \emptylist.  \\
     \mu^{1} X.S(X) & = S(\mu^{0} X.S(X))  \\
                   & = S( \emptylist)  \\
                   & = (u;@\varepsilon.\tau) \oplus @1.\emptylist.\\
     \mu^{2} X.S(X) & = S(\mu^{1} X.S(X))  \\
                    &= S\big((u;@\varepsilon.\tau) \oplus @1.\emptylist \big) \\
                   & =  (u;@\varepsilon.\tau) \oplus @1.\big( (u;@\varepsilon.\tau) \oplus @1.\emptylist \big). 
    \end{align*}
\end{example}

The formal definition of the semantics of \ces follows.
\begin{definition}[Semantics of \ces]
\label{SemanticsOfCEStrategies}
 The \emph{semantics} of a closed \ce $S$ is the function
$\sembrackk{S} : \funset{\mathcal{T} \cup \mathbb{F}}$,
which is defined inductively  as follows.
\begin{align*}
&  \sembrackk{\emptylist}(t)  \uberEq{def} \fail. \\
&  \sembrackk{{S}}(\fail)  \uberEq{def} \fail. \\
& \sembrackk{u;S}(t)      {\overset{def}{=}}%
                                    \begin{cases}
                                      \sembrackk{S}(t)  & \textrm{if } \match{u}{t}, \\
                                             {\mathbb{F}} & \text{otherwise}.
                                    \end{cases} \\
& \sembrackk{@\varepsilon.\mbf{\tau}}(t) \uberEq{def}  \tau[t]. \\
& \sembrackk{{S}_{1}\oplus {S}_{2}}(t) {\overset{def}{=}}%
\begin{cases}
\sembrackk{{S}_{1}}(t) & {\text{if }}\sembrackk{{S}_{1}}(t)\neq {\mathbb{F}}, \\
\sembrackk{{S}_{2}}(t) & \text{otherwise.}%
\end{cases}    \\
& \sembrackk{\mu X. {S}(X)}(t)  \uberEq{def} \sembrackk{\mu^{\delta(t)} X. S(X)}(t). \\
&\sembrackk{\tifthen{S_1}{{S}}}(t)  \uberEq{def}
\begin{cases}
  \sembrackk{{S}}(t)   & \textrm{if }   \sembrackk{S_1}(t) \neq \fail ,\\
  \fail & \textrm{otherwise}.
\end{cases} \\
&\sembrackk{@p.{S}}(t)  \uberEq{def}
\begin{cases}
  t[\sembrackk{{S}}(t_{|p})]_{p}   & \textrm{if } \sembrackk{{S}}(t_{|p}) \neq \fail  \tand p \in \PPos(t), \\
  \fail & \textrm{otherwise}.
\end{cases} \\
&\sembrackk{\bigand_{i=1,n} @p_i.S_i}(t)  \uberEq{def}   
  \begin{cases}
   \big(\eta(\sembrackk{@p_n.{S}_n})  \circ \cdots \circ \eta(\sembrackk{@p_1.{S}_1})\big)(t) \, &\textrm{ if } \exists i\in[1,n] \tst \sembrackk{@p_i.{S}_i} (t)\neq \fail,  \\
\fail & \totherwise.
 \end{cases}\\
&\sembrackk{\most({S})}(t)  \uberEq{def} \sembrackk{\bigand_{i=1,ar(t)} @i.S}(t).
\end{align*}
We shall refer to  $\sembrackk{S}(t)$ as the application of  $S$  to $t$.
\end{definition}

Notice that when  the application of $\bigand_{i=1,n} @p_i.S_i$ or of $\mu X.S(X)$ fails, it does not return the input term untouched (i.e. it does not behave as the identity), but fails as well.
The reason is that  we want a fine  semantics that distinguishes  between the identity that operates successfully  and returns the input term (e.g. $@\varepsilon.\square$), and the failure that indicates that the \ce was not applied, which may launch other \ces.       
Notice also that $\bigand_{i=1,n} @p_i.S_i$ fails if and only if  each $@p_i.S_i$ fails, and not just one of them fails. 
This is important  because we want to make the semantics of $\bigand_{i}$ compatible with that of $\most$ in terms of failure, that is why  we expressed the latter in terms of the former,
and the only reason for that is to be able to unify $\bigand_{i}$ with  $\most$, see Section~\ref{unification:combination:section}, and  remaining in the same framework of \ces. Otherwise, a richer semantics in terms of
handling the failure requires the framework~\cite{combination:2019:rudi-strategies-faillure} in which the failures are handled explicitly in the formalism, making it impractical.

The general definition of the  semantics of the fixed-point constructor requires an unnecessary  machinery involving Knaster-Tarski fixed-point theorem \cite{Tarski55} and complete lattices.
However, due to  the particular nature of \ces, we gave an adhoc definition of the fixed-point \ce  by
$\sembrackk{\mu X. {S}(X)}(t)  \uberEq{def} \sembrackk{\mu^{\delta(t)}X.S(X)}(t)$, which is the same as that given by the least-fixed point.
The justification  of the    iteration of  ${S}(\emptylist)$ at most $\delta(t)$ times, the depth of $t$, is that the navigation part of a \ce does not change the input term $t$, see Remark~\ref{navigation:part:remark} and Example~\ref{semantics:ces:example}.
Therefore,  either the \ce ${S}$  progresses on the term $t$ and will reach the leaves of $t$ after at most $\delta(t)$ iterations,
or ${S}$ does not progress  and in this case it fails after any iteration.
Examples of \ces  that do not progress are $\mu X. X$ and $\mu X. (u,X)$ for a term $u$.
Technically, we show in Corollary~\ref{general-fixed-point-corollary} that, for every term $t$, the \ce  $R=\mu^{\delta(t)} X.S(X)$ is a fixed-point of $S(X)$ in the sense that  $\sembrackk{S(R)}(t) =\sembrackk{R}(t)$.
Notice that any \ce of the form  $\mu X. S(X,X)$, in which $X$ occures twice,  can be turned into  the  equivalent \ce $\mu X. \mu Y. S(X,Y)$ in which $X$ occurs once.
This equivalence can be proved by induction on $S(X,X)$, and more generaly it holds for any $\mu$-calculus~\cite{rudimemt:mu-calculus:book}.

\begin{example}[Semantics of \ces]
  \label{semantics:ces:example}
We give two examples of \ces and their semantics. 
Let $\tau,\tau'$ be  contexts in $\mycal{T}_{\square}$, and let $f(f(b))$ and $g(b,b',x)$ be  terms in $\mycal{T}$, where $b,b'$ are constants, and $x$ is a rewriting variable. 
\begin{enumerate}
\item
  Consider the \ce
\begin{align*}
  S(X)=(b;@\varepsilon.\tau) \oplus @1.X.
\end{align*}
The \ce $\mu X.S(X)$  checks whether the constant $b$ matches with an input term $t$, if it does  then $\tau$ is inserted at the root of $t$ (i.e. $b$) yielding $\tau[b]$, otherwise it jumps to the position
$1$ of $t$ and iterates the same operation. We next illustrate  the application of $\mu X.S(X)$  to $f(f(b))$.
  Since the depth of $f(f(b))$ is $\delta(f(f(b)))=3$, we need to compute $\mu^3 X.S(X)$, as we did for a similar \ce in  Example~\ref{ufold:iteration:example}, thus we get:
\begin{align*}
\mu^3 X.S(X) = (b;@\varepsilon.\tau) \oplus  @1.\Big((b;@\varepsilon.\tau) \oplus @1.\big( (b;@\varepsilon.\tau) \oplus @1.\emptylist \big)\Big).
\end{align*}
Hence the application of $\mu X.S(X)$ to $f(f(b))$ yields
\begin{align*}
  \sembrackk{\mu X.S(X)}\big(f(f(b))\big) &= \sembrackk{\mu^3 X.S(X)}\big(f(f(b))\big)   \\
  &= \sembrackk{ @1.\Big((b;@\varepsilon.\tau) \oplus @1.\big( (b;@\varepsilon.\tau) \oplus @1.\emptylist \big)\Big)}\big(f(b))\big) \\
  &= \sembrackk{(b;@\varepsilon.\tau) \oplus @1.\big( (b;@\varepsilon.\tau) \oplus @1.\emptylist \big) } (b)  \\
  &= \sembrackk{b;@\varepsilon.\tau} (b)  \\
  &= \sembrackk{@\varepsilon.\tau} (b)  \\
  &= \tau[b]. 
\end{align*}
\item  Consider the \ce
  \begin{align*}
    R(Y) = g(b,b',x); \big(@1.\tau \wedge @2.\tau' \wedge @3. Y\big) .
  \end{align*}
  The \ce $\mu Y.R(Y)$ expects a term of the form  $g(b,b',t')$, then it inserts $\tau$ on its  first child (i.e. $b$), and inserts $\tau'$ on its  second child (i.e. $b'$), then jumps to the
  third child (i.e. $t'$) and iterates the same operation. Hence the application of $\mu Y.R(Y)$ to the term $g(b,b',g(b,b',b))$, which has depth $3$, yields
  \begin{align*}
    \sembrackk{\mu Y.R(Y)} \big(g(b,b',g(b,b',b))\big)
     &= \sembrackk{\mu^3 Y.R(Y)}\big(g(b,b',g(b,b',b))\big)   \\
     &= \sembrackk{R\big(\mu^2 Y.R(Y)\big)}\big(g(b,b',g(b,b',b))\big)   \\
    &= \sembrackk{  g(b,b',x); \Big(@1.\tau \wedge @2.\tau' \wedge @3.\big(\mu^2 Y.R(Y)\big) \Big)}  \big(g(b,b',g(b,b',b))\big) \\
    &= \sembrackk{@1.\tau \wedge @2.\tau' \wedge @3.\big(\mu^2 Y.R(Y)\big)}  \big(g(b,b',g(b,b',b))\big) \\
    &= g\Big(\tau[b],\tau'[b'],\underbrace{\eta\big(\sembrackk{\mu^2 Y.R(Y)}(g(b,b',b))\big)}_{t''}\Big),
  \end{align*}
  hence 
\begin{align*}
  t''&=\eta\big(\sembrackk{R(\mu^1 Y.R(Y))}(g(b,b',b)) \\
  &= \eta\Big( \sembrackk{  g(b,b',x); \big(@1.\tau \wedge @2.\tau' \wedge @3.\big(\mu^1 Y.R(Y)\big) \big)}  \big(g(b,b',b)\big) \Big) \\
  &= g\Big(\tau[b],\tau'[b'],\eta\big(\sembrackk{\mu^1 Y.R(Y)}(b)\big)\Big) \\
  &= g\Big(\tau[b],\tau'[b'],\eta\big(\sembrackk{R(\mu^0 Y.R(Y))}(b)\big)\Big) \\
  &= g\Big(\tau[b],\tau'[b'],\eta\big(\sembrackk{R(\emptylist)}(b)\big)\Big) \\
  &= g\Big(\tau[b],\tau'[b'],\eta\big(\sembrackk{g(b,b',x); \big(@1.\tau \wedge @2.\tau' \wedge @3. \emptylist\big) }(b)\big)\Big) \\
  &= g(\tau[b],\tau'[b'],b).
\end{align*}
Summing up, we get
  \begin{align*}
    \sembrackk{\mu Y.R(Y)} \big(g(b,b',g(b,b',b))\big) &=  g\big(\tau[b],\tau'[b'],g(\tau[b],\tau'[b'],b)\big).
\end{align*}
\end{enumerate}
\end{example}
  
In the following example we show how to encode the  two  standard traversal strategies  $\TDDD$ and  $\IDDD$ in our formalism using  the fixed-point  constructor.
\begin{myexample}
    \label{ex:TD:strategy}
    In  what follows we assume that ${S}$ is a \ce.
    We recall that, when applied to a term $t$, the \ce  $\TDDD({S})$ tries to apply ${S}$ to the maximum  of the sub-terms of $t$ starting from the root of $t$,
    it stops when it is successfully applied.
    And  when applied to a term $t$, the \ce  $\IDDD({S})$ tries to apply ${S}$ to the maximum  of the sub-terms of $t$ starting from the leaves of $t$,
    it stops when it is successfully applied.
    Hence,
    \begin{align*}
    \TDDD(S)             & := \mu X. \big(S \oplus  \most(X) \big)  &\tand&&     \IDDD(S)             & := \mu X. \big(\most(X)  \oplus  S \big).
    \end{align*}
\end{myexample}

\begin{definition}
\label{equivalence:ces:def}
Let  $S,S'$  be \ces  and $n\ge 1$ an integer.  We shall write
\begin{enumerate}[i)]
\item ${S} \equiv {S}'$  iff $\sembrackk{{S}} = \sembrackk{{S}'}$. In this case, $S$ and $S'$ are called \emph{equivalent}.
\item ${S} \equiv_{n} {S}'$  iff $\sembrackk{{S}}(t) = \sembrackk{{S}'}(t)$ for any term $t$ of depth $\delta(t) = n$. In this case,  $S$ and $S'$ are called \emph{$n$-equivalent}.
\end{enumerate}
\end{definition}

Notice that  "$\equiv$" is an equivalence relation  and that $S$ and $S'$ are  equivalent iff they are $n$-equivalent  for any $n \ge 1$.





%% file: combination-ces_corrected.tex
\section{Unification and combination of \ces}
\label{unification:combination:section}

The problem now is to extend the operations of unification and combination of position-based \ces (i.e. Definitions~\ref{unif:posi:non:empty:def} and~\ref{comb:posi:def}) to the larger class of  \ces.
These two   extensions  must fulfill a correctness criterion that  will be devised in Subsection~\ref{correctness:criterion:unif:comb:subsec}.
The subsequent subsections are devoted to the definition of  the extension of unification  (Definition~\ref{unification:def})  and the extension of the combination (Definition~\ref{combination:def}).
Then we  give   an  example of unification of \ces and comment it.

\subsection{A correctness criterion for the extension of the unification and combination to  \ces}
\label{correctness:criterion:unif:comb:subsec}
Since there are many ways to  define  an extension of the  unification operation  from  position-based \ces (i.e. the class $\mycal{E}$) to \ces (i.e. the class $\ceSetCan$), one needs a criterion that both guides the elaboration of a definition  and ensures its correctness. Such a criterion should impose a  compatibility between  the unification operation upon  $\mycal{E}$  and its extension to the larger class~$\ceSetCan$, in the sense that the former operation should stand as the  basis for the latter.

For this purpose, out of a term in $\mycal{T}$ and a \ce in $\ceSetCan$, we shall  construct a unique (up to a permutation of parallel positions) position-based \ce in $\mycal{E}$.
That is, we shall define a  mapping
\begin{align*}
  \Psi :  \mycal{T} \longrightarrow  \ceSet  \longrightarrow \eceSet
\end{align*}
 that associates to any  term   $t$  in $\mycal{T}$ and any closed \ce $S$ in $\ceSet$,  a position-based \ce $(\Psi(t))(S)$ in $\eceSet$, denoted simply by $\Psi_t(S)$,  such that the semantic equivalence is preserved  in the following sense: 
 \begin{align}
   \label{psi:sem:equiv:eq}
 \sembrackk{\Psi_t(S)}(t) = \sembrackk{S}(t).
\end{align}

Since the mapping $\Psi$ takes into account the semantics, then the correctness criterion is nothing but the compatibility between the unification upon  $\eceSet$ and its extension to  $\ceSet$, i.e.
for any  term $t$, the following two operations yield the same result: 
\begin{enumerate}[{i.)}]
\item the unification of  two \ces in $\ceSetCan$,  followed by the  mapping of the result to $\mycal{E}$ by $\Psi_t$, and 
\item the mapping of each of these two \ces to  $\mycal{E}$ by $\Psi_t$, followed by the  unification of the resulting position-based \ces.
\end{enumerate}  
This natural correctness criterion  will be  formalized   in Definition~\ref{correctness:criterion:def} for both the unification and combination.
 However, to simplify the exposition we shall not define the mapping $\Psi$ here but in Definition~\ref{psi:def} of Section~\ref{Psi:construction:section}, since the statement of the main results does not require
 this definition. Furthermore, we shall show in Lemma \ref{psi:sem:lemma} of Section~\ref{Psi:construction:section} that the thus defined  $\Psi$ preserves the semantic equivalence in the sense of Eq.(\ref{psi:sem:equiv:eq}). 
 

\begin{definition}[Correctness criterion for the extension of $\combb$ and $\comb$]
  \label{correctness:criterion:def}
 An extension $\combb: \ceSetCan  \,\times\, \ceSetCan   \longrightarrow \ceSetCan$ of the unification  $\combb: \mycal{E} \,\times\,  \mycal{E}   \longrightarrow  \mycal{E}$  is \emph{correct}, if and only if,
 for every term $t \in \mycal{T}$ and for every \ces  $S$ and $R$  in $\ceSet$,  we have that 
\begin{align*}
  \Psi_t(S \combb R) & = \Psi_t(S)\combb \Psi_t(R).
\end{align*}

Similarly,  an extension $\comb: \ceSetCan  \,\times\, \ceSetCan   \longrightarrow \ceSetCan$ of the combination  $\comb: \mycal{E} \,\times\,  \mycal{E}   \longrightarrow  \mycal{E}$  is correct, if and only if,
 for every term $t \in \mycal{T}$ and for every \ces  $S$ and $R$  in $\ceSet$,  we have that 
\begin{align*}
  \Psi_t(S \comb R) & = \Psi_t(S)\comb \Psi_t(R).
\end{align*}
That is, the following diagrams commute.
    \[\begin{tikzcd}
\mycal{C} \times \mycal{C} \arrow{r}{\combb} \arrow[swap]{d}{\Psi_t \times \Psi_t} &  \mycal{C} \arrow{d}{\Psi_t} \\
\mycal{E} \times \mycal{E}  \arrow{r}{\combb} & \mycal{E}
\end{tikzcd}
\;\;\;\;\;\;\;
\begin{tikzcd}
\mycal{C} \times \mycal{C} \arrow{r}{\comb} \arrow[swap]{d}{\Psi_t \times \Psi_t} &  \mycal{C} \arrow{d}{\Psi_t} \\
\mycal{E} \times \mycal{E}  \arrow{r}{\comb} & \mycal{E}
\end{tikzcd}
\]
\end{definition}

\subsection{Sub-\ces,  memory and  pre-\ces}
Since  the unification of \ces will be defined by induction, we need to define  the notion of the set of sub-\ces of a given \ce. With a slight modification allowing $S(\mu X.S(X))$ to be in the set of such sub-\ces of $\mu X.S(X)$, we also define the set of augmented sub-\ces.

\begin{definition}[Sub-\ces of a \ce]
  \label{def:sub-ces-augmented}
Given a \ce $S$, we inductively define the finite set of \emph{sub-\ces} of $S$, denoted by   $\Phi(S)$, as well as the finite set of  \emph{augmented sub-\ces} of $S$, denoted by   $\widetilde\Phi(S)$, which are similar apart for the fixed-point \ces:
\begin{align*}
\Phi(\emptylist)          &= \set{\emptylist},                          &\widetilde\Phi(\emptylist) &= \set{\emptylist},\\
\Phi(X)                  &= \set{X},                                    &\widetilde\Phi(X)   &= \set{X}, \\
\Phi(@\varepsilon.\mbf{\tau})  &= \set{@\varepsilon.\mbf{\tau}},        &\widetilde\Phi(@\varepsilon.\mbf{\tau})  &= \set{@\varepsilon.\mbf{\tau}},\\
\Phi(u;S)              &= \set{u;S} \cup  \Phi(S),                      &\widetilde\Phi(u;S)  &= \set{u;S} \cup  \widetilde\Phi(S),  \\  
\Phi(@p.S)               &= \set{@p.S} \cup \Phi(S),                    &\widetilde\Phi(@p.S) &= \set{@p.S} \cup \widetilde\Phi(S), \\ 
\Phi(S_1 \oplus  S_2)     &=  \set{S_1 \oplus S_2} \cup  \Phi(S_1) \cup \Phi(S_2),              &\widetilde\Phi(S_1 \oplus  S_2)     &=  \set{S_1 \oplus S_2} \cup  \widetilde\Phi(S_1) \cup \widetilde\Phi(S_2), \\
\Phi(\bigand_{i=1,n} S_i)   & =  \set{\bigand_{i=1,n} S_i} \cup \bigcup_{i=1,n} \Phi(S_i),        &\widetilde\Phi(\bigand_{i=1,n} S_i)   & =  \set{\bigand_{i=1,n} S_i} \cup \bigcup_{i=1,n}  \widetilde\Phi(S_i), \\
\Phi\big(\tifthen{S_1}{S}\big)   & = \set{\tifthen{S_1}{S}} \cup  \Phi(S_1) \cup \Phi(S),       &\widetilde\Phi\big(\tifthen{S_1}{S}\big)   & = \set{\tifthen{S_1}{S}} \cup \widetilde\Phi(S_1) \cup \widetilde\Phi(S), \\
\Phi(\mu X.S(X))          & = \set{\mu X.S(X)} \cup   \Phi(S(X)).                               &\widetilde\Phi(\mu X.S(X))& = \set{\mu X.S(X)}  \cup \widetilde\Phi(S(X))  \cup  \Phi\big(S \big(\mu X.S(X)\big)\big).
\end{align*}
A \ce $R$ is said to be a sub-\ce of $S$ if $R$ is in $\Phi(S)$.

Similarly, the set of all \emph{fixed-point sub-\ces} of $S$, denoted by $\Phi_{\mu}(S)$,  as well as the set of all \emph{augmented fixed-point sub-\ces} of $S$, denoted by $\widetilde\Phi_{\mu}(S)$, are defined  similarly  apart for the fixed-point \ces: 
\begin{align*}
\Phi_{\mu}(\emptylist)          &= \emptyset,                                     &\widetilde\Phi_{\mu}(\emptylist)               &= \emptyset,  \\
\Phi_{\mu}(X)                  &= \emptyset,                                      &\widetilde\Phi_{\mu}(X)                        &= \emptyset, \\
\Phi_{\mu}(@\varepsilon.\mbf{\tau})  &= \emptyset,                                &\widetilde\Phi_{\mu}(@\varepsilon.\mbf{\tau})  &= \emptyset,\\
\Phi_{\mu}(u;S)              &=   \Phi_{\mu}(S),                                  &\widetilde\Phi_{\mu}(u;S)                       &= \widetilde\Phi_{\mu}(S),\\ 
\Phi_{\mu}(@p.S)               &=  \Phi_{\mu}(S),                                 &\widetilde\Phi_{\mu}(@p.S)             &= \widetilde\Phi_{\mu}(S), \\ 
\Phi_{\mu}(S_1 \oplus  S_2)     &=   \Phi_{\mu}(S_1) \cup \Phi_{\mu}(S_2),         &\widetilde\Phi_{\mu}(S_1 \oplus  S_2)     &= \widetilde\Phi_{\mu}(S_1) \cup \widetilde\Phi_{\mu}(S_2),\\
\Phi_{\mu}(\bigand_{i=1,n} S_i)   & = \bigcup_{i=1,n}  \Phi_{\mu}(S_i),             &\widetilde\Phi_{\mu}(\bigand_{i=1,n} S_i)   & =\bigcup_{i=1,n}  \widetilde\Phi_{\mu}(S_i),\\
\Phi_{\mu}\big(\tifthen{S_1}{S}\big)   & =  \Phi_{\mu}(S_1) \cup \Phi_{\mu}(S),    &\widetilde\Phi_{\mu}\big(\tifthen{S_1}{S}\big)   & =  \widetilde\Phi_{\mu}(S_1) \cup \widetilde\Phi_{\mu}(S),\\
\Phi_{\mu}(\mu X.S(X))          & = \set{\mu X.S(X)} \cup \Phi_{\mu}(S(X)).       &\widetilde\Phi_{\mu}(\mu X.S(X))          & = \set{\mu X.S(X)}  \cup \widetilde\Phi_{\mu}(S(X))    \cup \Phi_{\mu}(S(\mu X.S(X))). 
\end{align*}
\end{definition}

Clearly, $\Phi_{\mu}(S) \subset \Phi(S) \subset \widetilde\Phi(S)$  and $\Phi_{\mu}(S) \subseteq \widetilde\Phi_{\mu}(S)$  and  $\tilde\Phi_\mu(S) \subseteq \tilde\Phi(S)$.
Notice that if $S$ is fixed-point free, then $\Phi(S)=\widetilde\Phi(S)$ and  $\Phi_{\mu}(S) = \widetilde\Phi_{\mu}(S) = \emptyset$.
Indeed,   the set of augmented sub-\ces    $\widetilde\Phi(S)$ is finite and this can be easily shown by induction on $S$.
We illustrate the Definition~\ref{def:sub-ces-augmented} with the following example.
\begin{example}[Of $\Phi$, $\widetilde{\Phi}$, $\Phi_{\mu}$ and $\widetilde\Phi_{\mu}$]
  \label{sub-ces:example}
   For a given pattern $u \in \mycal{T}$ and a context $\tau$, let
    \begin{align*}
      {S}(X)  &=(u;@\varepsilon.\tau) \oplus @1.X
    \end{align*}
    be a \ce. Hence the sets  $\Phi(\mu X.S(X))$, $\widetilde{\Phi}(\mu X.S(X))$, $\Phi_{\mu}(\mu X.S(X))$ and $\widetilde\Phi_{\mu}(\mu X.S(X))$  are easily computed as follows:  
    \begin{align*}
      \Phi(\mu X.S(X)) &= \{ \mu X.S(X) \} \cup \Phi(S(X)) \\
      &= \{ \mu X.S(X) \} \cup \{S(X)\}  \cup  \{(u;@\varepsilon.\tau), (@\varepsilon.\tau), @1.X, X \}   \\
      &= \{ \mu X.S(X), S(X),  (u;@\varepsilon.\tau), (@\varepsilon.\tau), @1.X, X \},
    \end{align*}
   and 
    \begin{align*}
      \widetilde{\Phi}(\mu X.S(X))
      &=  \set{\mu X.S(X)}  \cup \widetilde\Phi(S(X)) \cup  \Phi\big(S \big(\mu X.S(X)\big)\big) \\
      &=  \set{\mu X.S(X)}   \cup   \Phi(S(X)) \cup  \Phi\big(S \big(\mu X.S(X)\big)\big)  \tag{Since $S(X)$ is fixed-point free}\\
      &=  \set{\mu X.S(X)} \cup \Phi(S(X)) \cup \{S(\mu X.S(X))\} \cup \Phi \big( (u;@\varepsilon.\tau) \oplus @1.\mu X.S(X)\big)   \tag{Definition of $S(\mu X.S(X))$} \\
      &=  \set{\mu X.S(X)} \cup \Phi(S(X)) \cup \{S(\mu X.S(X))\} \cup \Phi(u;@\varepsilon.\tau) \cup \Phi(@1.\mu X.S(X)\big)    \\
      &=  \set{\mu X.S(X)} \cup \Phi(S(X)) \cup \{S(\mu X.S(X))\} \cup  \Phi(u;@\varepsilon.\tau) \cup \{@1.\mu X.S(X)\} \cup \Phi(\mu X.S(X)\big)  \\
      &=  \set{\mu X.S(X)} \cup \{S(\mu X.S(X))\} \cup \{@1.\mu X.S(X)\} \cup \Phi(\mu X.S(X)\big)  \tag{Since $\Phi(S(X))$ and $\Phi(u;@\varepsilon.\tau)$ are a subset of $\Phi(\mu X.S(X))$} \\
      &= \{ \mu X.S(X)\} \cup   \{S(\mu X.S(X))\} \cup \{@1.\mu X.S(X)\} \cup \{S(X),  (u;@\varepsilon.\tau), (@\varepsilon.\tau), @1.X, X\} \tag{The expression of $\Phi((\mu X.S(X))$ was computed above} \\
            &= \{ \mu X.S(X), S(\mu X.S(X)), @1.\mu X.S(X), S(X),  (u;@\varepsilon.\tau), (@\varepsilon.\tau), @1.X, X\}, \tand  \\
      \Phi_{\mu}(\mu X.S(X))  &= \set{\mu X.S(X)}, \tand \\
    \widetilde\Phi_{\mu}(\mu X.S(X)) & = \set{\mu X.S(X)}.
    \end{align*}
\end{example}

In the Example~\ref{sub-ces:example} above  we have that   $ \Phi_{\mu}(\mu X.S(X)) = \widetilde\Phi_{\mu}(\mu X.S(X))$, but this is not true in  general  as  shown in the following remark. 
\begin{remark}
The inclusion $\Phi_{\mu}(R) \subseteq \widetilde\Phi_{\mu}(R)$ is strict in general, that is,  
there is  a \ce $R$ such that $\Phi_{\mu}(R) \subsetneq \widetilde\Phi_{\mu}(R)$. This is achieved  by taking a \ce $R$ of the form:
  \begin{align*}
    R = \mu X. \mu Y. S(X,Y),
  \end{align*}
and noticing that the \ce $\mu Y. S(R,Y)$ is neither  in $\Phi_{\mu}(R)$ nor in  $\Phi(R)$, but in  $\widetilde\Phi_{\mu}(R)$ and in  $\widetilde\Phi(R)$.
\end{remark}
The unification of two \ces will be given by means of a  reduction system that requires  storing a piece of information, called \emph{memory}, related to the input  \ces. 
Roughly speaking, a memory is a set of triples where the  first and the second element of each triple  is a fixed-point sub-\ce   or an augmented \ce, and
the third element is a fixed-point variable. The idea behind the memory is that the unification of a fixed-point \ce  $\mu X.S(X)$ with a \ce $R$ amounts to the unification of
$S(\mu X.S(X))$ with $R$, or more precisely, the unification of $\mu X.S(X)$ with $R$ produces a \ce $\mu Z.T(Z)$, where $Z$ is a fresh-fixed point variable and $T(Z)$
is the unification of $S(\mu X.S(X))$ with $R$. To ensure that this process terminates we need to store the triple  $(\mu X.S(X),R,Z)$ in the memory so that $Z$ is produced whenever  $\mu X.S(X)$ is unified again with $R$.

The formal definition of the memory follows.

\begin{definition}[Memory]
  Given  an enumerable set $\fixset$ of fixed-point variables, as well as two \ces $S$ and $R$,
  we define   the set  of all \emph{memories} related to $S$ and $R$ with respect to $\fixset$, denoted by $\mathfrak{M}_{\fixset}(S,R)$ or simply  by $\mathfrak{M}(S,R)$, as the following set of sets of triples:
\begin{align*}
\mathfrak{M}(S,R)    &= \wp\Big(\big(\widetilde\Phi_{\mu}(S) \times (\widetilde\Phi(R)\setminus \fixset) \times \fixset \big) \;\cup\;   \big((\widetilde\Phi(S)\setminus \fixset) \times \widetilde\Phi_{\mu}(R) \times \fixset\big)\Big). 
\end{align*}
More generally, the set of all memories, denoted by $\mathfrak{M}$, is defined by 
\begin{align*}
\mathfrak{M}    &= \bigcup_{S,R \in \mycal{C}} \mathfrak{M}(S,R).
\end{align*}
An element in $\mathbf{\mathfrak{M}}(S,R)$ or in $\mathfrak{M}$ is called a memory. 
\end{definition}

An example of a memory related to two \ces follows.
\begin{example}[Memory]
    For given patterns $u,u' \in \mycal{T}$ and contexts $\tau,\tau'$, let
    \begin{align*}
      {S}(X)  &=(u;@\varepsilon.\tau) \oplus @1.X && \tand &      S'(X') &=(u';@\varepsilon.\tau') \oplus @1.X' 
    \end{align*}
    be  \ces.
    From  Example~\ref{sub-ces:example} above  we have that
    \begin{align*}
      \widetilde{\Phi}(\mu X.S(X))  &= \{ \mu X.S(X), S(\mu X.S(X)), @1.\mu X.S(X), S(X),  (u;@\varepsilon.\tau), (@\varepsilon.\tau), @1.X, X\}, \tand \\
      \widetilde{\Phi}(\mu X'.S'(X')) &=  \{ \mu X'.S(X'), S(\mu X'.S(X')), @1.\mu X'.S(X'), S(X'),  (u';@\varepsilon.\tau'), (@\varepsilon.\tau'), @1.X', X'\}.
    \end{align*}

Given $\fixset=\set{Z_1,Z_2,\ldots,Z_{6}}$, we  give an example of a  memory $\Eu{M}$ related to $\mu X.S(X)$ and $\mu X'.S'(X'))$ and $\fixset$, i.e. $\Eu{M} \in \mathfrak{M}_{\fixset}\big(\mu X.S(X),\mu X'.S'(X')\big)$.
\begin{align*}
      \Eu{M}
      =\big\{&\big(\mu X.S(X),  S'(X'), Z_1\big),               &&\big(S(X),\mu X'.S'(X'), Z_2\big), \\
      &\big(\mu X.S(X),  @\varepsilon.\tau', Z_3\big),         &&\big(@\varepsilon.\tau,\mu X'.S'(X'), Z_4\big), \\
      &\big(\mu X.S(X),  @1.X', Z_5\big),                      &&\big(@1.X,\mu X'.S'(X'), Z_6\big)  \big\}.
    \end{align*}
\end{example}
From now on we let  $\fixset$ to be an enumerable set of fixed-point variables.
Since the unification reduction system  will handle two \ces together with a memory, this new object is called a \emph{Pre-\ce} and defined next. 

\begin{definition}[Pre-\ces]
\label{Def:pre-HCE-strategies}
The class of  \emph{pre-\ces}  is  defined by the following grammar:
\begin{align*}
P  \; ::= \;  & S \gvert \tuple{S,S',\Eu{M}} \gvert u;{P} \gvert  P \oplus P   \gvert   \mu X.P  \gvert    @i.P \uand  @i'.P  \gvert  \most(P)  \gvert    \tifthen{S}{P} \\
\end{align*}
where 
$S,S'$ are \ces in $\mycal{C}$,
$\Eu{M}$ is a memory   in $\mathfrak{M}$,
$X$ is a fixed-point variable in $\fixset$,
$u$  is a  term   in $\mycal{T}$ 
and $i,i'$ are unitary  positions in $\mathbb{N}_{\varepsilon}$.
 The set of pre-\ces will be denoted by $\preceSet$. 
\end{definition}

Like in the modal $\mu$-calculus, it is easier and convenient to work with \ces
that make progress when applied to a term. Making progress is guaranteed by a
syntactic requirement, called monotonicity, that imposes that in each fixed-point sub-\ce $\mu X.S(X)$ there is at least a  position jump or a $\most$ from the root of $S(X)$ to $X$. 
\begin{definition}[Monotonicity of \ces]
\label{Monotony}A \ce  $T$ is \emph{monotonic}   if for any $\mu X.S(X) \in \Phi_{\mu}(T)$, there exist  \ces $S'(X)$ and $S''(X)$ each of  which is a sub-\ce of $S(X)$ such that 
 $S'(X)$ is either of the form  $@i.S''(X)$ where $i\in \mathbb{N}_{\varepsilon} \setminus\set{\varepsilon}$,  or of the form  $\most(S''(X))$.
\end{definition}
For instance, the \ce $\mu X.\big((u;@\varepsilon.\tau) \oplus @1.X\big)$  (resp. $\mu X.\big((u;@\varepsilon.\tau) \oplus \most(X)\big)$) is monotonic since there the jump "$@1$"  (resp. "$\most$") between $\mu X$ and $X$.
While  $\mu X.\big((u;@\varepsilon.\tau) \oplus X\big)$ is not monotonic.

We generalize next the condition of well-foundedness from position-based \ces to \ces.  
\begin{definition}[Well-founded \ces.]
\label{Well-founded:strategy:ext:def}
 A \ce  $S$  is \emph{well-founded} iff  every position-based \ce  that  is a sub-\ce  of $S$  is well-founded in the sense of Definition~\ref{Well-founded:simple:ext:def}.
 \end{definition}


\subsection{The procedure of unification of \ces }
From now on we shall  abuse of language and  refer to the extension of the unification operation from position-based \ces to \ces, as   simply the unification of \ces.
Before giving the procedure of unification of \ces, we need the following assumptions on the  structure of \ces. 

\begin{assumptions}
\label{global:assumptions:ces}
Throughout this paper,   each  \ce  is well-founded, monotonic, closed, and in which each fixed-point variable appears once,
and each of their sub-\ces which is of the form $@i.S'$ or $\bigwedge_i@i.S_i$, is preceded by a pattern-matching, i.e. $u ; @i.S'$  and   $u; \bigwedge_i@i.S_i$.
\end{assumptions}
These  assumptions do not exclude interesting cases, since either they  exclude cases which do not make sense (e.g. a \ce with free fixed-point variables, or not well-founded), or
they make the \ces easier to handle in the proofs without missing interesting cases, for instance, imposing that each fixed-point variable appears once is not a restriction since each \ce can be turned into
a  \ce with such a property by applying the following simplification operations which preserve the semantic equivalence (i.e. a \ce is semantically equivalent to its simplification).

\begin{simplifications}
\label{simplification:algorithms}
The simplification operations of \ces consists of:
\begin{enumerate}[(i)]
\item renaming identical bound variables, for instance $\mu X. S(X) \oplus \mu X. R(X)$ can be turned into  $\mu X. S(X) \oplus \mu Y. R(Y)$, this is known in the literature as the $\alpha$-conversion, and
\item  renaming identical occurrences of variables if they are bound to the same fixed-point operator, for instance   if $ S(X,X)$ is a \ce in which $X$ appears twice, then
  we can turn  $\mu X. S(X,X)$ into the  equivalent \ce $\mu X. \mu Y. S(X,Y)$, and 
\item  removing useless $\mu$ contractors, i.e. turning $\mu X.S$ into  $S$ when $X$ does not appear in $S$.
\end{enumerate}
\end{simplifications}


We define next the procedure of unification of \ces by means of a  reduction  system that operates on pre-\ces,
in which the pattern $u$ related to the position $i$ in \ce $S$ will be denoted by $\mathit{Patt}(S,i)$, or simply by  $\mathit{Patt}(i)$ when  $S$ is known. 
\begin{definition}
\label{reduction:unif:def}
We define the reduction  system $\Unif$ operating on pre-\ces and consisting of the following reduction rules with a decreasing order of priority.
   \begin{enumerate}

    \item   \begin{enumerate}[(a)]
         \item \label{final:1} $\tuple{\emptylist, S, \Eu{M}}     \reduce   \emptylist.$
         \item  \label{final:2}  $\tuple{S,\emptylist,\Eu{M}}   \reduce   \emptylist.$
                \end{enumerate}
         \item  \label{final:3}  $\tuple{@\varepsilon.\mbf{\tau}, @\varepsilon.\mbf{\tau}',\Eu{M}}  \reduce   @\varepsilon.(\mbf{\tau}\sbullet \mbf{\tau}').$
          \item \label{pattern:unif:item} \begin{enumerate}[(a)]
                 \item  \label{pattern:ext:1} $\tuple{(u;S), S',\Eu{M}}  \reduce  u; \tuple{S,S',\Eu{M}}.$
                 \item  \label{pattern:ext:2} $\tuple{S',(u;S),\Eu{M}}  \reduce  u; \tuple{S', S,\Eu{M}}.$
                 \end{enumerate}
\item \begin{enumerate}
       \item \label{list:ext'} $\tuple{@i.S,@i.S',\Eu{M}} \reduce @i.\tuple{S,S',\Eu{M}}$.
           
         \item \label{list:ext} If $S=\bigand_{i \in I} @i.S_i \uand @\varepsilon.\tau$  and $S'=\bigand_{j \in J} @j.S'_j \uand @\varepsilon.\tau'$ then \\
         \begin{align*}
           \tuple{S,S', \Eu{M}}  \reduce  \tifthen{S \& S'} {\bigand_{i \in I \cap J} @i.\big(\tuple{S_i,S'_i,\Eu{M}} \oplus S_i \oplus S'_i\big)\uand  R \uand  R'\uand @\varepsilon.(\tau \sbullet \tau')},
         \end{align*}
         where \begin{align*} R &= \bigand_{i \in I\setminus J} @i.S_i  &\tand &&
           R'&= \bigand_{j \in J\setminus I} @j.S'_j.
         \end{align*}
\end{enumerate}
          \item \begin{enumerate}[(a)]
        \item \label{choice:ext:1} $\tuple{(S_1 \oplus S_2), S,\Eu{M}}  \reduce \tuple{S_1,S,\Eu{M}}  \oplus \tuple{S_2, S,\Eu{M}}.$
        \item  \label{choice:ext:2} $\tuple{S,(S_1 \oplus S_2) ,\Eu{M}}  \reduce \tuple{S,S_1,\Eu{M}}  \oplus \tuple{S,S_2,\Eu{M}}.$
               \end{enumerate}
          \item \begin{enumerate}[(a)]
                \item  \label{if:ext:1} $\tuple{(\tifthen{S_1}{S_2}),{S},\Eu{M}}  \reduce \tifthen{S_1}{\tuple{S_2,{S},\Eu{M}}}$.
                \item  \label{if:ext:2} $\tuple{S, (\tifthen{S_1}{S_2}) ,\Eu{M}} \reduce \tifthen{S_1}{\tuple{S,S_2,\Eu{M}}} $.
                \end{enumerate}
           \item \begin{enumerate}[(a)]
             \item \label{most:ext:1} $\tuple{\most(S) ,  \most(S'),\Eu{M}} \reduce
                                                            \mathtt{\mathbf{If}}{\big(\most(S) \& \most(S')\big)}  \mathtt{\mathbf{ Then }} \, \most\big(\tuple{S, S',\Eu{M}} \oplus S \oplus S'\big)$.
                    \item \label{most:ext:2} $\tuple{\most(S), \bigand_{i \in I} @i.S_i ,\Eu{M}}  \reduce  \tuple{\bigand_{i \in [1,\arity(u)]} @i.S, \bigand_{i \in I}  @i.S_i,\Eu{M}}  \twhere u=\mathit{Patt}(i) $
              \item  \label{most:ext:3} $\tuple{\bigand_{i \in I} @i.S_i, \most(S),\Eu{M}}  \reduce  \tuple{\bigand_{i \in I}  @i.S_i,\bigand_{i \in [1,\arity(u)]} @i.S,\Eu{M}}  \twhere u=\mathit{Patt}(i) $
           \end{enumerate}

\item \begin{enumerate}[(a)]
    \item  \label{fixed:ext:1} \begin{align*}
             \tuple{\underbrace{\mu X. S(X)}_{\xi}, S',\Eu{M}}   \reduce
                \begin{cases} \mu Z. \tuple{S(\xi), S',\Eu{M}'},  & \tif (\xi,S',\cdot) \notin \Eu{M}, \\
                                           &\;\;\;\; \twhere \begin{cases} Z  &=  \fresh{\xi,S'}, \\
                                                       \Eu{M}'         &=   \Eu{M} \cup \set{(\xi,S',Z)}.\end{cases}\\
                                           &\\
                              Z  & \tif (\xi,S',Z) \in \Eu{M}.
                \end{cases}
           \end{align*}

          \item \label{fixed:ext:2}

            \begin{align*}
              \tuple{S',\underbrace{\mu X. S(X)}_{\xi},\Eu{M}}   \reduce
              \begin{cases} \mu Z.\tuple{S',S(\xi) ,\Eu{M}'},  & \tif (S',\xi,\cdot) \notin \Eu{M}, \\
                & \;\;\;\; \twhere \begin{cases}  Z     &=  \fresh{S',\xi}, \\
                                \Eu{M}'        &= \Eu{M} \cup \set{(S',\xi,Z)}. \end{cases} \\
                &\\
                Z  & \tif (S',\xi_1,Z) \in \Eu{M}.
              \end{cases}
           \end{align*}
            \end{enumerate}
\end{enumerate}
\end{definition}

\paragraph{Explanation of the rules}
Notice that, by construction,  for any tuple $\tuple{S,S',\Eu{M}}$  produced by the unification reduction system $\Unif$, the  memory $\Eu{M}$ is 
 redundancy free, that is, if   $(S_1,R_1,Z_1)$ and $(S_1,R_1,Z_2)$ are in $\Eu{M}$, then $Z_1=Z_2$.
We comment on the key points in Definition~\ref{reduction:unif:def}.
\begin{description}
\item[\texttt{Pattern matching}:] If we omit the memory $\Eu{M}$ in the rule~\ref{pattern:ext:1} for sake of simplicity, then the unification of $u;S$ with $S'$  is naturally $ u; S''$, where $S''$ is the unification
   of $S$ with $S'$,  since we want that the pattern $u$ proceeds the merging of $S$ and $S'$. 

\item[$\most$:] For the unification of two $\most$s (i.e. rule~\ref{most:ext:1}), we first recall the semantics of this constructor.
When  a \ce  $\most(S)$ is applied    to term $t$, the \ce  $S$ is  applied  to each of its children.
In particular, $\most(S)$ fails on $t$ if and only if $S$ fails on each of $t$'s children. Otherwise, when  $\most(S)$ succeeds on $t$, then  $S$  behaves as the identity
on the children  of $t$ on which it fails, see Definition~\ref{SemanticsOfCEStrategies}.
While unifying two \ces  $\most(S)$ and  $\most(S')$ we need to distinguish two  cases.
\emph{(i)} If one of these \ces fails, then the result should fail. This is achieved by the condition
$\mathtt{\mathbf{If}}{\big(\most(S) \& \most(S')\big)}$ in the resulting \ce.
\emph{(ii)} If both of them do not fail when applied to a term $t$ the we need to consider whether  each of them fails or not on each child
of $t$. In order to explain the $\most\big(\tuple{S, S',\Eu{M}} \oplus S \oplus S'\big)$ part in the resulting unified  \ce, let $t_i$ be a child of $t$,  and   consider  the four  cases: 
\begin{enumerate}[{(ii}.1)]
\item if both of $S$ and $S'$ succeeds on $t_i$, then we need to consider their unification. This explains the $\tuple{S, S',\Eu{M}}$ part.
\item If both of them fails on $t_i$, then their unification should fail on $t_i$ as well. But this holds since $\tuple{S, S',\Eu{M}} \oplus S \oplus S'$ fails as well.
\item If $S$ succeeds on $t_i$ while $S'$ fails on it, then the resulting unified \ce should apply  $S$ to $t_i$.   But this is achieved  by  $\tuple{S, S',\Eu{M}} \oplus S \oplus S'$ which
       is equal to $S$ since $\tuple{S, S',\Eu{M}}$ fails on $t_i$.
\item If $S$ fails  on $t_i$ while $S'$ succeeds on it, then this case is symmetric to the previous one.
\end{enumerate}

The unification of $\most(S)$ with a conjunction of position  jumps $\bigand_{i \in I} @i.S_i$ requires that we encode  $\most(S)$ into a  conjunction of position  jumps as well.

\item[\texttt{Fixed-points}:] 
The idea behind the unification of $\mu X.S(X)$ with $R$ (i.e. rule~\ref{fixed:ext:1}) is to unfold $\mu X.S(X)$  to $S(\mu X.S(X))$  and then  to  unify $S(\mu X.S(X))$ with $R$.
Indeed this process is terminating thanks to the use of memory since we memorized that we passed through the unification of $\mu X.S(X)$ with $R$ and we
generated a fresh fixed-point variable $Z$, this is done by adding the tuple $(\mu X.S(X),R,Z)$ to the memory.  Thanks to the memory,
the next time we face the unification of  $\mu X.S(X)$ with $R$, we shall produce $Z$.  
\end{description}

We shall show in Subsection~\ref{termination:confluence:reduction:unif:sec} that the unification system $\Unif$ is terminating and confluent.
This allows us to define the unification operation in terms of the normal  form with respect to  $\Unif$.
The normal form of  $\tuple{S,S',\emptyset}$ will be denoted by  $\NF \tuple{S,S',\emptyset}$.
The definition of the unification and combination of \ces follow.
We emphasize  that throughout this paper, as far as we are dealing with  the unification and combination,
we assume that the two  sets of the  fixed-point variables of the two input \ces are disjoint.  

\begin{definition}[Unification of  \ces]
    \label{unification:def}
    The \emph{unification}  of   \ces is the  binary operation \\$\combb: \ceSetCan  \,\times\, \ceSetCan   \longrightarrow \ceSetCan$,
    defined  for any  $S$ and $S'$ in $\mycal{C}$ by
        \begin{align*}
          S \combb S' \uberEq{def}  \NF \tuple{S,S',\emptyset}. 
        \end{align*}
\end{definition}
Notice that the unification of two \ces yields a \ce that  captures the effect of both insofar as they are compatible,
where  the compatibility of two \ces depends on each input term and is related to their successful application.
That is, if $S$ and $S'$ can be applied successfully to a term $t$, then the application of their unification $S \combb S'$ on  $t$  succeeds as well and reproduces the effect that $S$ and $S'$ being applied simultaneously.
However, the incompatible effects are covered by the combination in the sense that if $S \combb S'$ fails on a term $t$, then  $S$ or $S'$  fails, and the combination returns the non-failing one, if any.
This justifies the following definition of the combination.
\begin{definition}[Combination of  \ces]
    \label{combination:def}
    The \emph{combination} of   \ces is the  binary operation \\$\comb: \ceSetCan  \,\times\, \ceSetCan   \longrightarrow \ceSetCan$,
    defined  for any  ${{S}}$ and ${{S}}'$ in ${C}$ by
        \begin{align*}
    {{S}} \comb  {{S}}' \uberEq{def}  (S \nfcombb  S') \oplus {{S}} \oplus  {{S}}'.
        \end{align*}
\end{definition}

\begin{myexample}[Unification of \ces]
  \label{unif:example}
  We give an example of the unification of two fixed-point \ces.
    For given patterns $u,u' \in \mycal{T}$ and contexts $\tau,\tau'$, let
    \begin{align*}
      {S}(X)  &=(u;@\varepsilon.\tau) \oplus @1.X && \tand &      S'(X') &=(u';@\varepsilon.\tau') \oplus @1.X' \\
      \xi     &=\mu X.S(X)                 && \tand &      \xi' &=\mu X'.S'(X')
    \end{align*}
    be  \ces.
    We  compute the  unification $\mu X. {S}(X) \combb \mu X'.{S}'(X')$ which is the normal form of the tuple \\
      $\tuple{\mu X. {S}(X),\mu X'. {S}'(X'),\emptyset}$  by applying the  reduction rules of $\Unif$ 
    given  in Definition~\ref{reduction:unif:def}.
    Let 
\end{myexample}
    \begin{align}
    (*) &=\tuple{\mu X. {S}(X),\mu X'. {S}'(X'),\emptyset}  \notag \\
        & \reduces \mu Z. \tuple{S(\xi),\xi',\set{(\xi,\xi',Z)}}                                                             \tag{Rule~\ref{fixed:ext:1}} \\
        & \reduces \mu Z.\mu Z'. \tuple{S(\xi),S'(\xi'),\underbrace{\set{(\xi,\xi',Z),(S(\xi),\xi',Z')}}_{\Eu{M}}}             \tag{Rule ~\ref{fixed:ext:2}} \\
        & = \mu Z.\mu Z'. \tuple{(u; @\varepsilon.\tau) \oplus @1.\xi,S'(\xi'),\Eu{M}}                                              \tag{Def. of $S(X)$} \\
        & \reduces  \mu Z.\mu Z'.\big( \underbrace{\tuple{u;@\varepsilon.\tau,S'(\xi'),\Eu{M}}}_{(\textrm{I})} \oplus   \underbrace{\tuple{@1.\xi,S'(\xi'),\Eu{M}}}_{(\textrm{II})}  \big).  \tag{Rule ~\ref{choice:ext:1}} \\
    (\textrm{I}) & \reduces u;\tuple{@\varepsilon.\tau,S'(\xi'),\Eu{M}}                                                                        \tag{Rule ~\ref{pattern:ext:1}}\\ 
       &   =       u;\tuple{@\varepsilon.\tau,(u';@\varepsilon.\tau') \oplus @1.\xi',\Eu{M}}                                                      \tag{Def. of $S'(X')$}\\ 
       &   \reduces   u;\big(\tuple{@\varepsilon.\tau, u';@\varepsilon.\tau' ,\Eu{M}} \oplus \tuple{@\varepsilon.\tau, @1.\xi',\Eu{M}} \big)       \tag{Rule ~\ref{choice:ext:2}}\\ 
       &   \reduces   u;\big((u';\tuple{@\varepsilon.\tau, @\varepsilon.\tau' ,\Eu{M}}) \oplus \tuple{@\varepsilon.\tau, @1.\xi',\Eu{M}} \big)       \tag{Rule~\ref{pattern:ext:2}}\\ 
       &   \reduces   u;\big((u'; @\varepsilon.(\tau\sbullet\tau')) \oplus \tuple{@\varepsilon.\tau, @1.\xi',\Eu{M}} \big)                             \tag{Rule~\ref{final:3}}\\ 
       &   \reduces   u;\big((u'; @\varepsilon.(\tau\sbullet\tau')) \oplus  (\tifthen{@1.\xi'}{@1.\xi' \uand @\varepsilon.\tau})   \big).     \tag{Rule~\ref{list:ext}}\\ 
 &  \notag \\ 
    (\textrm{II}) &  =  \tuple{@1.\xi,\,(u';@\varepsilon.\tau') \oplus @1.\xi',\,\Eu{M}}                                                  \tag{Def. of  $S'(X')$}\\ 
         &  \reduces  \tuple{@1.\xi,u';@\varepsilon.\tau',\Eu{M}} \oplus \tuple{@1.\xi,@1.\xi',\Eu{M}}                                \tag{Rule~\ref{choice:ext:2}} \\ 
         &  \reduces  \big( u';\tuple{@1.\xi,@\varepsilon.\tau',\Eu{M}}\big) \oplus \tuple{@1.\xi,@1.\xi',\Eu{M}}                     \tag{Rule~\ref{pattern:ext:2}} \\ 
         &  \reduces  \big( u'; (\tifthen{@1.\xi}{@1.\xi \uand\varepsilon.\tau'})\big) \oplus \tuple{@1.\xi,@1.\xi',\Eu{M}}   \tag{Rule~\ref{list:ext}} \\ 
         &  =   \big( u'; \tifthen{@1.\xi}{@1.\xi \uand\varepsilon.\tau'}\big) \oplus @1.\tuple{\xi,\xi',\Eu{M}}                       \tag{Rule~\ref{list:ext'}} \\ 
         &  =   \big( u'; \tifthen{@1.\xi}{@1.\xi \uand\varepsilon.\tau'}\big) \oplus @1.Z.                                       \tag{Rule ~\ref{fixed:ext:1}  since $(\xi,\xi',Z) \in \Eu{M}$ }
    \end{align}

    Summing up, the unification  $(**)$ of $\mu X. {S}(X)$ and  $\mu X'. {S}'(X')$ is:
    \begin{align*}
           (**) &= \mu X. {S}(X) \,\combb\, \mu X'. {S'}(X')\\
                &= \mu Z. \mu Z'. \bigg( u;\big((u'; @\varepsilon.(\tau\sbullet\tau')) \oplus  (\tifthen{@1.\xi'}{@1.\xi' \uand @\varepsilon.\tau})   \big)   \\
                & \hspace{5cm} \oplus   \big( u'; \tifthen{@1.\xi}{@1.\xi \uand\varepsilon.\tau'}\big)  \\
                & \hspace{5cm} \oplus @1.Z \bigg).
    \end{align*}
    Notice that the fixed-point variable $Z'$ does not appear in the resulting \ce and therefore "$\mu Z'$" can be removed.
    The application of the resulting \ce $(**)$ to a term $t$ features four cases.
        \begin{enumerate}[i.)]
        \item  Either  both $u$ and $u'$  match with  $t$, and in this case the context $\mbf{\tau}'\sbullet \mbf{\tau}$ is  inserted at the root of $t$.
        \item  Or  only  $u$   matches  with $t$, and in this case   $\tau$  is inserted at the position $1$ of $t$ provided the \ce $\mu X'. {S}'(X')$ is applied successfully at the position $1$ of $t$.
        \item Or  only  $u'$  matches   with $t$, and in this case   $\tau'$ is inserted at the position $1$ of $t$ provided the \ce $\mu X. {S}(X)$    is applied successfully at the position $1$ of $t$.
        \item  Or both $\mu X. {S}(X)$ and $\mu X'. {S'}(X')$ are applied at the position $1$ of $t$.
        \end{enumerate}

   The unification of two \ces in which each fixed-point variable appears  once may yields a \ce in which a variable appears many times or does not appear at all, e.g. the $Z'$ in the Example~\ref{unif:example}.
   An attention will be payed to this issue since this assumption on the occurrences of fixed-point variables   is not preserved by unification.
   However,   other assumptions listed in Assumptions~\ref{global:assumptions:ces} are preserved.
   Namely, it is easy to show that the unification of two well-founded  \ces is a  well-founded one. And we shall show later that the unification of two monotonic \ces is a monotonic one as well.


%% file: main-results_corrected.tex
\section{Statement of the  results}
\label{main:results:sec}
In this section we state the main results of this paper, that is, the correctness of the procedure of  unification and combination  stated in Subsection~\ref{correction:unif:sub-section:statement},
and the algebraic properties of the unification and combination of \ces stated in Subsection~\ref{alg:prop:subsection}. The proofs of these results can be found in Section \ref{proof:main:results:section}.

\subsection{Correctness of  the unification and combination procedures}
\label{correction:unif:sub-section:statement}

\newcounter{myvar-unif-theorem}
\setcounter{myvar-unif-theorem}{\value{theorem}} 

\begin{theorem}[Correctness of the unification]
 \label{main:theorem:1'}
 The unification of \ces is correct. That is,  
for every term $t \in \mycal{T}$ and for every \ces  $S$ and $R$  in $\ceSet$, 
we have that 
\begin{align*}
  \Psi_t(S \combb R) & = \Psi_t(S)\combb \Psi_t(R).
\end{align*}
\end{theorem}

\begin{theorem}[Correctness of the combination]
  \label{main:theorem:2'}
  The combination of \ces is correct.
That is, for every term $t \in \mycal{T}$ and  for every  \ces  $S$ and $R$  in $\ceSetCan$,
we have that
\begin{align*}
  \Psi_t(S \comb R)  =  \Psi_t(S) \comb \Psi_t(R).
\end{align*}
\end{theorem}


\subsection{Algebraic properties of the unification and  combination}
\label{alg:prop:subsection}
The existence of the neutral  elements  and the associativity property of the unification and combination are obvious for the sub-class of position-based \ces
but they are far from being  so for the larger class of \ces, and it is crucial and useful to have them. Namely, a user of  \ces needs know the algebraic properties of the structure he handles. 
For instance,  he needs combine  many \ces, and thus needs to know if this combination is associative and/or commutative.
Besides, the properties of the  neutral and absorbing elements allow one
to simplify \ces.

We notice that the neutral  elements  and the
associativity property of the unification and combination must be  understood at the semantic level and not at the syntactic level
since there are \ces which are syntactically different but
semantically equivalent. For instance, the \ces  $@\E.\square$ and
$(x;@\E.\square)$ and
 $(x;@\E.\square) \oplus (y;@\E.\square)$, where $x,y$ are variables, are all equivalent.
More generally, the algebraic properties of the unification and combination will be formulated in terms of equivalence classes  of \ces (with respect to the semantic equivalence relation) rather than syntactic  \ces.  

Technically speaking, since the semantic equiavalence  "$\equiv$" (Definition~\ref{equivalence:ces:def}) is an equivalence relation, we shall use the standard notation $[S]$ for the equivalence class of the \ce $S$, i.e. $[S]=\set{S' \in \ceSet \gvert S' \equiv S}$, 
and the notation $\ceSetEquiv$  for the quotient  set of $\ceSet$ by "$\equiv$", i.e. $\ceSetEquiv=\set{[S] \gvert S \in \ceSet}$.
 Moreover, the unification and combination of the equivalence classes of \ces in $\ceSetEquiv$  can be 
defined in a natural way as: 
\begin{align*}
[S_1] \combb [S_2] := [S_1 \combb S_2] &&  [S_1] \comb [S_2] := [S_1 \comb S_2].
\end{align*}
We notice that these two operations are well defined since they are a congruence  by Theorems \ref{main:alg:theorem:3'} and \ref{main:alg:theorem:4'}.  
The algebraic properties  of the unification of \ces follow.
In fact,  the unification of \ces  inherits the properties of associativity, (non-)commutativity and idempotence from
the position-based \ces and the merging of contexts.
\begin{theorem}
\label{main:alg:theorem:1'}
 The quotient set $\ceSetEquiv$ of  \ces  together with the unification   operation  enjoy the  following properties.
    \begin{enumerate}
    \item The neutral element of the   unification upon $\ceSetEquiv$  is $[@\E.\square]$. 
    \item The absorbing element of the unification is $[\emptylist]$.
    \item The unification  of \ces is  associative, i.e. $([S_1] \combb [S_2]) \combb [S_3] =  [S_1] \combb ([S_2] \combb [S_3])$, 
          for any $S_1,S_2,S_3 \in \ceSet$.
        \item  The unification of \ces is (non-)commutative if and only if  the operation of merging of contexts "$\sbullet$"  is (non-) commutative. 
    \item The unification of \ces is idempotent if and only if  the operation of merging of contexts  is idempotent,
          that is,  $[S] \combb  [S]= [S]$  for any $S \in \ceSet$  iff  $\tau \sbullet \tau = \tau$  for any context  $\tau$ in $\mycal{T}_{\square}$.
    \end{enumerate}
\end{theorem}

The algebraic properties  of the combination  of \ces follow.
In fact,  the combination of \ces  inherits the properties of associativity, (non-)commutativity and idempotence from
the position-based \ces and the merging of contexts.
\begin{theorem}
\label{main:alg:theorem:2'}
 The quotient set $\ceSetEquiv$ of  \ces  together with the  combination  operation enjoy the following properties.
    \begin{enumerate}
    \item The neutral element of the   combination upon $\ceSetEquiv$   is $[\emptylist]$.
    \item The  combination of \ces is associative, i.e. $([S_1] \comb [S_2]) \comb [S_3] =  [S_1] \combb ([S_2] \combb [S_3])$, 
          for any $S_1,S_2,S_3 \in \ceSet$.
    \item The combination of \ces is (non-)commutative if and only if  the operation of merging of contexts $\sbullet$  is (non-) commutative. 
    \item The combination of \ces is idempotent if and only if  the operation of merging of contexts  is idempotent.
    \end{enumerate}
\end{theorem}

Since the  mapping  $\Psi: \mycal{T} \longrightarrow  \ceSet  \longrightarrow \eceSet$ preserves the semantic equivalence in the sense of Eq.(\ref{psi:sem:equiv:eq}),
then $\Psi$ induces a mapping  $\dot{\Psi}: \mycal{T} \longrightarrow  \ceSetEquiv  \longrightarrow \eceSet$ in a natural way by $\dot{\Psi}_t([S]):= \Psi_t(S)$, for any term $t$ in  $\mycal{T}$ and \ce $S$ in $\ceSet$.
It turned out that, for any term $t$, the mapping  $\dot{\Psi}_t: \ceSetEquiv  \longrightarrow \eceSet$ is a morphism from the structure
$(\ceSetEquiv,  \combb, \comb, [@\E.\square],[\emptylist])$ to  $(\eceSet,  \combb, \comb, @\E.\square,\emptylist)$ since for any \ces $S_1$ and $S_2$ in $\ceSet$, on the one hand,
\begin{align*}
  \dot{\Psi}_t([@\E.\square]) = @\E.\square   && \tand &&    \dot{\Psi}_t([\emptylist]) = \emptylist,
\end{align*}
and on the other hand, it follows from Theorems~\ref{main:theorem:1'} and~\ref{main:theorem:2'} that
\begin{align*}
  \dot{\Psi}_t([S_1] \combb [S_2]) = \dot{\Psi}_t([S_1])  \combb \dot{\Psi}_t([S_2]) && \tand &&   \dot{\Psi}_t([S_1] \comb [S_2]) = \dot{\Psi}_t([S_1])  \comb \dot{\Psi}_t([S_2]).
\end{align*}

The congruence and non-degeneracy of the unification  and  combination are stated in the two following theorems, respectively.
\begin{theorem}[Congruence and non-degeneracy of the unification]
  \label{main:alg:theorem:3'}
The following holds.
\begin{enumerate}
\item  The unification  of \ces is a congruence, that is,    for any \ces   ${S}_1,{S}_2, {S}$ in $\ceSet$, we have that:
  \begin{align*}
\textrm{If } {S}_1 \equiv {S}_2 &&\tthen&&  {S}_1 \nfcombb {S}  \equiv {S}_2 \nfcombb {S}  \;\tand\; {S} \nfcombb {S}_1  \equiv {S} \nfcombb {S}_2.
\end{align*}

\item The unification  is non-degenerate, that is, for any  \ces $[S]$ and $[S']$ in $\ceSetEquiv$,   we have that
\begin{align*}
    [S] \nfcombb [S']   =  [\emptylist]  &&\tiff &&  [S] =  [\emptylist] \;\tor\;  [S'] =  [\emptylist]. 
\end{align*}

\end{enumerate}
\end{theorem}

\begin{theorem}[Congruence and non-degeneracy of the combination]
 \label{main:alg:theorem:4'}
The following hold.
\begin{enumerate}
\item  The combination  of \ces is a congruence, that is,    for any \ces   ${S}_1,{S}_2, {S}$ in $\ceSet$, we have that:
\begin{align*} 
\textrm{If } {S}_1 \equiv {S}_2  &&\tthen && {S}_1 \comb {S}  \equiv {S}_2 \comb {S} \;\tand\; {S} \comb {S}_1  \equiv {S} \comb {S}_2.
\end{align*}

\item The combination  is non-degenerate, that is, for any  \ces $[S]$ and $[S']$ in $\ceSetEquiv$,   we have that
\begin{align*}
    [S] \comb [S']  =  [\emptylist]  &&\tiff &&   [S] =  [\emptylist] \;\textrm{and }\;   [S'] = [\emptylist].
\end{align*}

\end{enumerate}
\end{theorem}


%% file: outline-proof_corrected.tex
\section{Outline of the proof of the main result}
\label{structure:proof:main:results:sec}
The most lengthy and difficult result  to prove is Theorem~\ref{main:theorem:1'} on the correctness of the unification of \ces. The remaining theorems are more or less a consequence of this theorem.
In this section we give a relatively detailed outline of the proof  of Theorem~\ref{main:theorem:1'} without the technical machinery  which will be developed  in the next sections.
We shall  proceed in four steps:
\begin{description}
\item[Step 1.] We first show that the unification of \ces is correct in the particular setting, where the \ces are \emph{fixed-point free}. 
              More precisely, we shall show that the mapping $\Psi$  permutes with the unification (in the sense of Theorem~\ref{main:theorem:1'}) within  this particular setting.
              The proof is relatively easy and will be exposed in Section~\ref{proof:correction:fixed-point-free:sec}.

\item[Step 2.]  Then we reduce the general setting to the fixed-point free setting by replacing the fixed-point operators  by iterations whose number depends on the input term, Sections~\ref{correction-combination-definitions:section},~\ref{correction:unif:general:setting:unfold:sec} and ~\ref{equiv:unif:with:unif:unfolding:sec}.
  That is,   we  replace in the input \ces each fixed-point \ce  $\mu X.S(X)$  with the unfolding  $S(S(\ldots(S(\emptylist))):=\mu^n X.S(X)$ whose length is an arbitrary fixed integer~$n$.    
  Clearly, the unfolding of a  \ce is a fixed-point free one.
  The key idea is to show that the unification of two \ces is $n$-equivalent to the unification of their unfoldings.
  To accomplish this, we compare the structure of the resulting two \ces and show that they have  a similar  structure (Lemma~\ref{main:lemma:mophism:quasi}).
  We illustrate this idea  of similarity of  structures in a particular case through a simple example,  then we discuss the  more general case.
  For the simple example, let $M(Y), S(X)$ and $R$ be three fixed-point free \ces where $R$ does not contain neither a left choice $\oplus$ nor an $\mathtt{\mathbf{If\m Then}}$. 
  Consider, on the one hand, the unification of $M(\mu X.S(X))$   with   $R$, and on the other hand, the unification of the unfolding of $M(\mu X.S(X))$ with the unfolding of $R$.
  Notice that the unfolding of $R$ is equal to $R$ since $R$ is fixed-point free.
  The structure of the \ce   $M(\mu X.S(X)) \nfcombb {R}$  is depicted on the left of  Figure~\ref{S:T:structure}, while that of the \ce $M(\mu^n X.S(X)) \combb {R}$ is on the right. That is,
  the unification  $M(\mu X.S(X)) \nfcombb {R}$  yields a \ce of the form $T_0(\mu Z_1.T_1(\ldots))$, whereas  the unification $M(\mu^n X.S(X)) \combb {R}$ yields a (fixed-point) free \ce of the form $T_0(T_1(\ldots))$.

  The general case in which  we unify $S$ and $R$ where  both of them  contain many fixed-point operators can be obtained by generalizing the simple example. 
  The general structure of the \ce $S \combb R$ is depicted on the left of Figure~\ref{S:T:structure:complex}, while that of the \ce that results from the unification of an unfolding of $S$ with an unfolding of $R$ is on the right.
  The general structure of the \ce on the left is of the form  $T_0(\mu Z_1.T_1(\mu Z_2.T_2(\ldots \mu Z_m.T_m(Z_m))))$, while that  on the right, is $T_0(T_1(\ldots T_m))$.
  Besides, each \ce $\mathbf{T}_i^j$ on the left is either a fixed-point variable or a fixed-point \ce that results from the unification of two \ces where one of them is a fixed-point.
  Assume that $\mathbf{T}_i^j$    is the normal form of $\tuple{\xi_i^j,R_i^j,\cdot}$, where $\xi_i^j$ is a fixed-point \ce that is a sub-\ce of $S$, while $R_i^j$ is  a sub-\ce of $R$. The main point is that  each $T_i^j$ is the result of the unification of an unfolding of
  $\xi_i^j$ with  an unfolding of $R_i^j$. Furthermore,  the more we go deeper into the right tree, i.e. $j$ increases, the more the size of iterations in the unfoldings decreases. The formalization of the notion of similarity between the unification of two \ces and that   of their unfoldings will be done in Subsection~\ref{quasi:sim:sub:sub:sec}. Proving the existence of such similarity  between the unification and the unification of unfoldings,
  as well as  developing  the properties of this  similarity, namely the decrease of the size of iterations in the unfoldings, will be done in Section~\ref{relating:unif:with:unfolds:sec}.

\item[Step 3.] The third step of  the proof consists of proving that the unification of two \ces is equivalent to that of their unfoldings by using the notion of similarity discussed before. 
  More precisely, we shall show that, for any $n\ge 1$,  the unification of two \ces is $n$-equivalent to  the $\mathbf{s}$-unfolding of them, 
  where  $\mathbf{s}$-unfolding amounts to  replace each fixed-point operator  with an iteration of size $n$.
  This will be proved in  Section~\ref{equiv:unif:with:unif:unfolding:sec},  Proposition~\ref{main:corollary:unif}. We outline next the general idea of this proof in a simple setting in which   the unification $S \combb R$
  yields a \ce $\mu Z_1.T_1(Z_1)$.  Thanks to the notion of similarity, we know that the unification of the $\mathbf{s}$-unfolding of $S$ with the $\mathbf{r}$-unfolding of $R$
  is of the form $T_1(T^1_1)$, where $T^1_1$ is the result of the unification of an $\mathbf{s}_1$-unfolding of $S$ with a $\mathbf{s}_2$-unfolding of $R$, where  the  $\mathbf{s}_1$-unfolding (resp. $\mathbf{s}_2$-unfolding)
  replaces each fixed-point operator in $S$ (resp. in $R$) with certain number of iterations that  can be computed.
  To show that $\mu Z_1.T_1(Z_1)$ is $n$-equivalent to  $T_1(T^1_1)$ it suffices to show that $T_1(T^1_1)$ is a fixed-point of $T_1(Z_1)$, i.e. that $T_1(T^1_1)$ is $n$-equivalent to $T_1(T_1(T^1_1))$.
  To achieve this it is enough to show that $T^1_1$ is $n'$-equivalent to $T_1(T^1_1)$ provided that there is at least $n-n'$ jumps between the root of  $T_1(Z_1)$ and $Z_1$, where  $n'$ is a constant that depends on
  $\mathbf{s}_1$ and $\mathbf{s}_2$. This  raises two technical problems.
  \emph{(i)} Since  $T_1(T^1_1)$ is the  unification of the $\mathbf{s}$-unfolding of $S$ with the $\mathbf{s}$-unfolding of $R$, and since $T^1_1$ is the unification of the $\mathbf{s}_1$-unfolding of $S$ with the
  $\mathbf{s}_2$-unfolding of $R$,  how to relate in general the unification of two unfoldings  with  the unification of two other unfoldings of the same \ces? We shall address this problem in  Section~\ref{equiv:unif:unfoldings:sec} and show that the two
  resulting \ces are equivalent up to a constant that depends on the four unfoldings. And \emph{(ii)} how to compute a lower bound on the number of such jumps? This question will be  addressed in Subsections
 ~\ref{measures:sub:sub:sec} and~\ref{derived:sub:sub:sec}. These results will be summed up in Subsection~\ref{equivalent:sub:sub:sec} to show  the main result of this third step (i.e. Proposition~\ref{main:corollary:unif}).

\item[Step 4.] Since an  unfolding of a \ce is a fixed-point free one, we shall  rely on  Proposition~\ref{main:corollary:unif} together with the correctness of the unification and combination for the fixed-point free setting outlined in Step 1, to prove   the correctness of the unification and combination in the general setting.  The proof turns to be  relatively straightforward and will be exposed  in Subsection~\ref{correctness:unification:combination:proof}.
\end{description}

\begin{figure}[H]
  \centering
  \input{tree-ext-unif-simple.tex} 
\caption{The structure of the  the \ce   $M(\mu X.S(X)) \combb  {R}$  (left)  and that of  $M(\mu^n X.S(X)) \combb {R}$ (right),  where $M(Y)$, $S(X)$ and $R$ are fixed-point free, and $n\ge 1$.}
\label{S:T:structure}
\end{figure}

\begin{figure}[H]
  \input{tree-ext-unif-complex.tex} 
  \caption{The general structure of a \ce   $S \combb {R}$  (left)  and  that of the \ce that results from the unification of a full unfolding of  $S$ with a full unfolding of  $R$ (right), where each \ce $\mathbf{T}_j^i$ is either
    a fixed-point \ce, or a fixed-point variable. Each $T_j^i$ is a  unification of two unfoldings of the same \ces involved in  $\mathbf{T}_j^i$. Inductively, the structure of  each  $\mathbf{T}_j^i$ is again similar to the one of $T_j^i$. 
 }
\label{S:T:structure:complex}
\end{figure}


%% file: tree-ext-unif-simple.tex
 
\begin{tikzpicture}
\begin{scope}[xshift = 0cm] 
    \draw[line width=1pt,color=black] (0,0) -- (-1,-1) -- (1,-1) -- cycle;
    \draw[line width=1pt,color=black] (2,-2) -- (1,-3) -- (3,-3) -- cycle;

    \draw[line width=1pt,->] (0.89,-1)  -- (2,-2);     
    \draw[line width=1pt,->] (2.89,-3)  -- (4,-4);     

    \node [] at (2,-1.3) {$\mu Z_1$};  
    \draw[line width=0.7pt,->] (1.2,-3)  -- (0.7,-3.5);     
    \node [] at (3.99,-3.3) {$\mu Z_{2}$};  

    \node [] at (0,-0.6) {${T}_0$};

     \node [] at (2.1,-2.6) {${T}_1$};
     \node [] at (0.7,-3.7) {\begin{small}$Z_1$\end{small}};
     \node [] at (4.5,-4.3) {\begin{small}$\tuple{\mu X.S(X),R_2,\cdot}$\end{small}};
\end{scope}

  \begin{scope}[transform canvas={xshift = 6cm}, yshift=0cm]
    \draw[line width=1pt,color=black] (0,0) -- (-1,-1) -- (1,-1) -- cycle;
    \draw[line width=1pt,color=black] (2,-2) -- (1,-3) -- (3,-3) -- cycle;

    \draw[line width=1pt,->] (0.89,-1)  -- (2,-2);     

    \draw[line width=0.7pt,->] (1.2,-3)  -- (0.7,-3.5);     
    \draw[line width=1pt,->] (2.9,-3)  -- (3.9,-4);     

    \node [] at (0,-0.6) {$T_0$};

     \node [] at (2.1,-2.6) {$T_{1}$};
     \node [] at (0.7,-3.8) {\begin{small}$T_1^1$\end{small}};
     \node [] at (4.1,-4.3) {\begin{small}$\tuple{\mu^{n-1} X.S(X), R_2,\emptyset}$\end{small}};

\end{scope}

\end{tikzpicture}


%% file: tree-ext-unif-complex.tex
 
\begin{tikzpicture}
  \begin{scope}
    \draw[line width=1pt,color=black] (0,0) -- (-1,-1) -- (1,-1) -- cycle;
    \draw[line width=1pt,color=black] (2,-2) -- (1,-3) -- (3,-3) -- cycle;
    \draw[line width=1pt,color=black] (4,-4) -- (3,-5) -- (5,-5) -- cycle;
    \draw[line width=1pt,color=black] (6,-7) -- (5,-8) -- (7,-8) -- cycle;

    \draw[line width=1pt,->] (0.89,-1)  -- (2,-2);     
    \draw[line width=1pt,->] (2.89,-3)  -- (4,-4);     

    \draw[line width=0.7pt,->] (-0.9,-1)  -- (-1.2,-1.5);     
    \draw[line width=0.7pt,->] (-0,-1)  -- (-0.1,-1.5);     
    \node [] at (2,-1.3) {$\mu Z_1$};  
    \draw[line width=0.7pt,->] (1.2,-3)  -- (0.7,-3.5);     
    \draw[line width=0.7pt,->] (1.9,-3)  -- (1.85,-3.5);     
    \node [] at (3.99,-3.3) {$\mu Z_{2}$};  
    \draw[line width=0.7pt,->] (3.2,-5) -- (2.6,-5.5);
    \draw[line width=0.7pt,->] (3.9,-5)  -- (3.85,-5.5);     
    \draw[line width=1pt,->] (4.7,-5)  -- (5.5,-5.9);
    \node [] at (5.7,-5.3) {$\mu Z_{3}$};  
    
    \draw[line width=0.7pt,->] (5.1,-8) -- (4.5,-8.5);
    \draw[line width=0.7pt,->] (6.4,-8) -- (6.6,-8.5);
    \node [] at (6.5,-6.8) {$\mu Z_{m}$};  

    \node [] at (0,-0.6) {${T}_0$};
    \node [] at (-1.3,-1.85) {$\mathbf{T}_1^0$};
    \node [] at (-0.1,-1.89) {$\mathbf{T}_{k_0}^{0}$};
     \node [] at (-0.5,-1.3) {$\mathbf{\ldots}$};

     \node [] at (2.1,-2.6) {${T}_1$};
     \node [] at (0.7,-3.8) {$\mathbf{T}_1^1$};
     \node [] at (2.1,-3.8) {$\mathbf{T}_{k_1}^1$};
    \node [] at (1.35,-3.3) {$\mathbf{\ldots}$};

     \node [] at (4.1,-4.6) {$T_2$};
     \node [] at (2.6,-5.8) {$\mathbf{T}_1^2$};
     \node [] at (3.9,-5.8) {$\mathbf{T}_{k_2}^2$};
     \node [] at (3.3,-5.3) {$\mathbf{\ldots}$};

     \node [] at (6.1,-7.6) {${T}_{m}$};
     \node [] at (4.4,-8.8) {$\mathbf{T}_{1}^m$};
     \node [] at (6.5,-8.8) {$\mathbf{T}_{k_m}^m$};
     \node [] at (5.7,-8.3) {$\mathbf{\ldots}$};

     \node [] at (5.7,-6.39) {\rotatebox{-65}{$.........$}}; 
  \end{scope}
  \begin{scope}[transform canvas={xshift = 9cm}]
    \draw[line width=1pt,color=black] (0,0) -- (-1,-1) -- (1,-1) -- cycle;
    \draw[line width=1pt,color=black] (2,-2) -- (1,-3) -- (3,-3) -- cycle;
    \draw[line width=1pt,color=black] (4,-4) -- (3,-5) -- (5,-5) -- cycle;
    \draw[line width=1pt,color=black] (6,-7) -- (5,-8) -- (7,-8) -- cycle;

    \draw[line width=1pt,->] (0.89,-1)  -- (2,-2);     
    \draw[line width=1pt,->] (2.89,-3)  -- (4,-4);     

    \draw[line width=0.7pt,->] (-0.9,-1)  -- (-1.2,-1.5);     
    \draw[line width=0.7pt,->] (-0,-1)  -- (-0.1,-1.5);     
    \draw[line width=0.7pt,->] (1.2,-3)  -- (0.7,-3.5);     
    \draw[line width=0.7pt,->] (1.9,-3)  -- (1.85,-3.5);     
    \draw[line width=0.7pt,->] (3.2,-5) -- (2.6,-5.5);
    \draw[line width=0.7pt,->] (3.9,-5)  -- (3.85,-5.5);     
    \draw[line width=1pt,->] (4.7,-5)  -- (5.5,-5.9); 
    
    \draw[line width=0.7pt,->] (5.1,-8) -- (4.5,-8.5);
    \draw[line width=0.7pt,->] (6.4,-8) -- (6.6,-8.5);

    \node [] at (0,-0.6) {$T_0$};
    \node [] at (-1.3,-1.75) {${T}^0_1$ };
    \node [] at (-0.1,-1.75) {${T}^0_{k_0}$ };
     \node [] at (-0.5,-1.3) {$\mathbf{\ldots}$};

     \node [] at (2.1,-2.6) {$T_{1}$};
     \node [] at (0.7,-3.8) {${T}^1_{1}$};
     \node [] at (1.9,-3.8) {$T^1_{k_1}$};
    \node [] at (1.35,-3.3) {$\mathbf{\ldots}$};

     \node [] at (4.1,-4.6) {$T_2$};
     \node [] at (2.5,-5.85) {${T}^2_{1}$};
     \node [] at (3.9,-5.85) {${T}^2_{k_2}$};
     \node [] at (3.3,-5.3) {$\mathbf{\ldots}$};

     \node [] at (6.1,-7.6) {$T_{m}$};
     \node [] at (4.5,-8.9) {${T}^m_{1}$};
     \node [] at (6.5,-8.9) {${T}^m_{k_m}$};
     \node [] at (5.7,-8.3) {$\mathbf{\ldots}$};

     \node [] at (5.7,-6.39) {\rotatebox{-65}{$.........$}}; 
  \end{scope}
\end{tikzpicture}



%% file: properties-ces_corrected.tex
\section{From \ces to position-based \ces: the definition of the mapping $\Psi$}
\label{Psi:construction:section}

In this section we define the mapping $\Psi$ announced in Section~\ref{unification:combination:section}, then   state and prove its properties.
Before doing this, we need to define the \emph{tree depth}  of a  \ce that corresponds to the usual notion of depth of such a  \ce after removing 
all the back-edges. We warn the reader that we shall use the same notation $\delta$ used for the depth of terms introduced in the preliminaries section~\ref{preliminaries:section}.

\begin{definition}[Tree depth  of a  \ce]
\label{def:Delta:strategy'}
The \emph{tree depth}  of a \ce  is the depth of its underlying tree  after we have ignored the fixed-point constructors  of this \ce.
That is, it is the function $\delta: \ceSet  \longrightarrow    \mathbb{N}$ defined inductively as follows:
\begin{align*}
\mathbf{\delta}(\emptylist)               &= 0 \\
\mathbf{\delta}(X)                        &= 0   \\
\mathbf{\delta}(@\varepsilon.\mbf{\tau})  &= 1 \\
\mathbf{\delta}(u;S)                &=  1 + \mathbf{\delta}(S)  \\ 
\mathbf{\delta}(@p.S)               &= 1 + \mathbf{\delta}(S)  \\ 
\mathbf{\delta}(S_1 \oplus  S_2)     &=  1 + \mmax\{\mathbf{\delta}(S_1), \mathbf{\delta}(S_n)\}  \\
\mathbf{\delta}(\bigand_{i=1,n} S_i)   & =  1 + \mmax\{\mathbf{\delta}(S_1), \ldots, \mathbf{\delta}(S_n)\} \\
\mathbf{\delta}\big(\tifthen{S_1}{S}\big)   & =  1 + \mmax\{\mathbf{\delta}(S_1),\mathbf{\delta}(S)\} \\
\mathbf{\delta}(\mu X.S(X))          & =  \mathbf{\delta}(S(X)).
\end{align*}
\end{definition}

It is useful  to normalize \ces which  are almost position-based \ces, i.e. they involve position  jumps and conjunctions, by concatenation of their nested  positions and by removing the failures.
For instance, turning $@i.@j.S$ into $@ij.S$, and turning $@i.S \wedge @j.\emptylist$ into $@i.S$. The definition of the normalization follows.

\begin{definition}[Normalization]
 \label{normalisation:position:def}
 The \emph{normalization} is the function $\theta$ that turns any \ce built up with just position  jumps and conjunctions to a position-based \ce as follows for any set of positions $J$:
\begin{align*}
\theta(@i.\tau)   &=  @i.\tau \\
\theta(@i.@j.S) &=  \theta(@ij.S) \\
\theta\big( \bigand_{j\in J} @j.S_j \big)      & =   \theta\big(\bigand_{j \in J \setminus \set{i}}  @j.S_j\big) \;\; \tif S_i=\emptylist \\
  \theta\big( \bigand_{j \in J} @j.S_j \big)      & =   \bigand_{j \in J}  \theta(@j.S_j) \\
  \theta\big(@i.\big(\bigand_{j \in J} @j.S_j  \big) \big)      & =   \bigand_{j \in J}  \theta(@ij.S_j). 
\end{align*}
\end{definition}

\begin{example}[Normalization]
  Let $\tau$, $\tau'$ and $\tau''$ be contexts in  $\mycal{T}_{\square}$.
  Let $S$ be the following \ce:
  \begin{align*}
    S= @1.\big( @2.\tau \wedge @3.(@4.\tau' \wedge @5.\tau'') \big).
  \end{align*}
  Then its normalization yields:
  \begin{align*}
    \theta(S)= @12.\tau \wedge @134.\tau' \wedge @135.\tau''.
  \end{align*}
\end{example}

Guided by the semantics of \ces, we next define the mapping  $\Psi$.
\begin{definition}[The mapping $\Psi$]
\label{psi:def} 
 We define the mapping 
\begin{align*}
  \Psi : \mycal{T} \longrightarrow  \ceSet   \longrightarrow \eceSet
\end{align*}
 that associates to any term  $t$  in $\mycal{T}$ and any closed \ce $S$ in $\ceSet$  a position-based  \ce $\Psi_t(S)$ in $\eceSet$ by
\begin{enumerate}
\item  \label{psi:def:item:empty} $\Psi_t(\emptylist) = \emptylist$.
\item \label{psi:def:item:insert} $\Psi_t(@\varepsilon.\mbf{\tau})= @\varepsilon.\mbf{\tau}$.
\item \label{psi:def:item:choice} $ \Psi_t(S \oplus S') =   \begin{cases}
                                                      \Psi_t(S) & \textrm{if } \Psi_t(S) \neq \emptylist, \\
                                                      \Psi_t(S') & \textrm{otherwise}.
                                                      \end{cases}$

\item \label{psi:def:item:mu} $\Psi_t(\mu X.S(X)) = \Psi_t\big(\mu^{\delta(t)} X.S(X)\big)$.
\item \label{psi:def:item:pattern} $ \Psi_t(u;S)  = \begin{cases}
                 \Psi_t(S)  & \textrm{if }    \match{u}{t}, \\
                 \emptylist & \textrm{otherwise}.
    \end{cases}$

\item \label{psi:def:item:ifthen} $ \Psi_t\big(\tifthen{S'}{S}\big) =
   \begin{cases}
     \Psi_t(S)   & \textrm{if }  \Psi_t(S') \neq \emptylist ,\\
     \emptylist            & \textrm{otherwise}.
   \end{cases}$ 

\item \label{psi:def:item:and} $\Psi_t\big(\bigand_{i=1,n}@p_i.S_i\big) =   \theta\big(\bigand_{i=1,n} @p_i.\Psi_{t_{|p_i}}(S_i) \big)$.
\item \label{psi:def:item:most} $\Psi_t(\most(S)) = \Psi_t\big(\bigand_{i=1,ar(t)} @i.S \big)$.


\end{enumerate}



\end{definition}

\begin{example}
  If we consider the two \ces $S(X)$ and $R(Y)$ defined in  Example~\ref{semantics:ces:example} by
  \begin{align*}
  S(X) &=(b;@\varepsilon.\tau) \oplus @1.X, \\
   R(Y)& = \mu Y.\Big(g(b,b',x); \big(@1.\tau \wedge @2.\tau' \wedge @3. Y\big) \Big),
  \end{align*}
  together with the two terms $t= f(f(b))$ and  $t'=g(b,b',g(b,b',b))$, then
  \begin{align*}
    \Psi_{t}(\mu X.S(X))  &=  @11.\tau, \\
    \Psi_{t'}(\mu Y.R(Y)) &= @1.\tau \wedge @2.\tau' \wedge @31.\tau \wedge @32.\tau'.
  \end{align*}
  
\end{example}

\begin{lemma}
  \label{psi:sem:lemma}
The mapping $\Psi$  preserves the semantic equivalence in the sense that,   
for any  term   $t$  in $\mycal{T}$ and any \ce $S$ in $\ceSet$, we have that 
\begin{align*}
 \sembrackk{\Psi_t(S)}(t) = \sembrackk{S}(t).
\end{align*}
\end{lemma}
The proof of this Lemma  does not provide any difficulties
since the definition of $\Psi$ is close to the definition of the semantics of \ces.
The previous  Lemma  can be restated in terms of explicit  properties as follows.

\begin{lemma}
\label{nice:prop:Psi:lemma}
The mapping $\Psi$ satisfies the following properties  for any terms  $t,u$, and   for any closed \ces $S,S',R,R',E'$, where $E'$ is built using only jumps and failures, and for any position-based \ce $E$: 
\begin{enumerate}

\item    \begin{enumerate}
          \item \label{Properties-of-Psi:Lemma:item:0}  $\Psi_t(E) = E$.     
          \item \label{Properties-of-Psi:Lemma:item:0'} $\Psi_t(\Psi_t(S)) =  \Psi_t(S)$.
         \end{enumerate}

\item \label{Properties-of-Psi:Lemma:item:1}  $\Psi_t(u;S) =  \Psi_t(u;\Psi_t(S))$.
\item \label{Properties-of-Psi:Lemma:item:2}  $\Psi_t(S \oplus S')   = \Psi_t(\Psi_t(S)\oplus \Psi_t(S'))$.

\item    \begin{enumerate}
          \item \label{Properties-of-Psi:Lemma:item:3''}  $\Psi_t(\tifthen{S'}{S}) =  \Psi_t(\tifthen{\Psi_t(S')}{S})$.     
          \item \label{Properties-of-Psi:Lemma:item:3}  $\Psi_t(\tifthen{S'}{S}) =  \Psi_t(\tifthen{S'}{\Psi_t(S)})$.     
          \item \label{Properties-of-Psi:Lemma:item:3'''}  $\Psi_t(\tifthen{S'}{S}) =  \Psi_t(\tifthen{R'}{S})$ if $\Psi_t(S')=\Psi_t(R')$.     
          \item \label{Properties-of-Psi:Lemma:item:3'} $\Psi_t(\tifthen{E'}{S}) =  \Psi_t(\tifthen{\theta(E')}{\Psi_t(S)})$.

         \end{enumerate}

\item    \begin{enumerate}
          \item \label{Properties-of-Psi:Lemma:item:4}  $\Psi_t(S \wedge R) =  \Psi_t\big( S \wedge R' \big)$ if $ \Psi_t(R)=\Psi_t(R')$, whenever $S,R,R'$ are a conjunction of  jumps.    
          \item \label{Properties-of-Psi:Lemma:item:4'} $\Psi_t(S \wedge R) =  \Psi_t(S)$    if $ \Psi_t(R)=\emptylist$, whenever $S,R$ are a conjunction of  jumps.        
         \end{enumerate}

\end{enumerate}
\end{lemma}
It turns out that the  mapping  $\Psi$ (Definition~\ref{psi:def})
preserves the semantics of \ces  in the following sense. 

\begin{lemma}
\label{nice:prop:Psi:lemma:}
The mapping  $\Psi$ enjoys the following properties.
\begin{enumerate}[i.)]
\item \label{item:1:nice:prop:Psi:lemma} For any position-based \ces $E, E'$ in $\eceSet$, we have that
$E =  E'$ iff $\Psi_t(E)=\Psi_t(E')$  for any  term $t$.

\item \label{item:2:nice:prop:Psi:lemma} For any  \ces ${S},{S}'$ in $\ceSet$, we have that
 ${S} \equiv {S}'$ iff  $\Psi_t(S)=\Psi_t(S')$  for any  term $t$.

\item \label{item:3:nice:prop:Psi:lemma} For any  \ces ${S},{S}'$ in $\ceSet$, we have that
  ${S} \equiv_n {S}'$ iff  $\Psi_t(S)=\Psi_t(S')$  for any  term $t$ of depth $\delta(t)=n$.
\end{enumerate}
\end{lemma}

\begin{proof}
We only prove Item \emph{ii.)}, the other items follow immediately
from the definition of $\Psi$. On the one hand, from the definition
of $\equiv$  we have that
\begin{align*}
S \equiv S' &&\tiff && \sembrackk{S}(t) =  \sembrackk{S'}(t), \;\;
\forall t \in \mycal{T}.
\end{align*}
However, it follows from Lemma~\ref{psi:sem:lemma} that 
\begin{align*}
\sembrackk{S}(t) = \sembrackk{\Psi_t(S)}(t) &&\tand &&
\sembrackk{S'}(t) = \sembrackk{\Psi_t(S')}(t).
\end{align*}
Therefore,
\begin{align*}
\sembrackk{\Psi_t(S)}(t)  = \sembrackk{\Psi_t(S')}(t),  \;\;  \forall t \in \mycal{T}.
\end{align*}
Since, both $\Psi_t(S)$ and $\Psi_t(S')$ are position-based \ces, it
follows from 
 Item~\ref{item:1:nice:prop:Psi:lemma}.) of this Lemma that $\Psi_t(S) =\Psi_t(S')$.
\end{proof}

We show in the following lemma that  the mapping  $\Psi$ can be pushed over the \ce constructors.
\begin{lemma}
\label{psi:unif-congruence:Lemma}
The mapping  $\Psi$ satisfies the following properties for any closed \ces $S,S'$ and  any position-based \ce $E$ and any terms  $t,u$: 
\begin{enumerate}
\item \begin{enumerate} \item \label{psi:unif-congruence:Lemma:item:1} $ \Psi_t\big(u ; \big(\Psi_t(S) \combb E  \big) \big) = \Psi_t(u ; S) \combb E$.
                        \item \label{psi:unif-congruence:Lemma:item:1'} $\Psi_t\big(u ; \big(E \combb \Psi_t(S)\big) \big) =  \Psi_t \big(E \combb \Psi_t(u ; S)\big) $.
      \end{enumerate}                    
\item \begin{enumerate} \item \label{psi:unif-congruence:Lemma:item:2} $\Psi_t\big( (\Psi_t(S) \oplus \Psi_t(S')) \combb E \big)   = \Psi_t (S \oplus S') \combb E$.
                       \item  \label{psi:unif-congruence:Lemma:item:2'} $E \combb \Psi_t\big(\Psi_t(S)\oplus \Psi_t(S')\big)    = E \combb \Psi_t(S \oplus S')$.
      \end{enumerate}  
\item \begin{enumerate} \item  \label{psi:unif-congruence:Lemma:item:3} $\Psi_t\big(\tifthen{S'}{(\Psi_t(S) \combb E)}\big)= \Psi_t\big(\tifthen{S'}{\Psi_t(S)}\big)  \combb E$.
                       \item  \label{psi:unif-congruence:Lemma:item:3'}  $\Psi_t\big(\tifthen{S'}{(E \combb \Psi_t(S))}\big) = E  \combb  \Psi_t(\tifthen{S'}{\Psi_t(S)})$.
      \end{enumerate}  
\end{enumerate}
\end{lemma}
\begin{proof}
We only prove the cases~\ref{psi:unif-congruence:Lemma:item:1} and~\ref{psi:unif-congruence:Lemma:item:2} and~\ref{psi:unif-congruence:Lemma:item:3} since 
the proof of the cases~\ref{psi:unif-congruence:Lemma:item:1'} and~\ref{psi:unif-congruence:Lemma:item:2'} and~\ref{psi:unif-congruence:Lemma:item:3'} is similar.
\begin{enumerate}
\item  \begin{enumerate}
         \item We distinguish two cases depending on whether $u$ matches with $t$ or not.  If $u$  matches with $t$ then  the left-hand side of the equation is
           \begin{align} 
             \Psi_t\big(u ; \big(\Psi_t(S) \combb E  \big)\big) &= \Psi_t(\Psi_t(S) \combb E)  \tag{Def. ~\ref{psi:def} of $\Psi$} \\
                                                                &= \Psi_t(S) \combb E \tag{since $\Psi_t(S) \combb E$  is a position-based \ce, Item~\ref{Properties-of-Psi:Lemma:item:0} of  Lemma~\ref{nice:prop:Psi:lemma}},
           \end{align}
and the right-hand side of the equation is $\Psi_t(u ; S) \combb E  =  \Psi_t(S) \combb E$  by the  Definition of $\Psi$, which is equal to the left-hand side.
 If $u$ does not match with $t$ then,  the left-hand side of the equation is $\emptylist$ by the definition of $\Psi$; and the right-hand side is 
 $\Psi_t(u ; S) \combb E= \emptylist \combb E = \emptylist$.
\end{enumerate}
\item \begin{enumerate}
      \item  We distinguish two cases depending on whether $\Psi_t(S)=\emptylist$ or not. If $\Psi_t(S)=\emptylist$ then the left-hand side of the equation is
          \begin{align}
            \Psi_t\big( (\Psi_t(S) \oplus \Psi_t(S')) \combb E \big)  & = \Psi_t\big( (\emptylist \oplus \Psi_t(S')) \combb E \big) \notag \\ 
                                                                     & = \Psi_t\big( \Psi_t(S') \combb E \big)  \notag  \\ 
                                                                     &=\Psi_t(S') \combb E, \tag{since $\Psi_t(S') \combb E$  is position-based,  Item~\ref{Properties-of-Psi:Lemma:item:0} of  Lemma~\ref{nice:prop:Psi:lemma}} 
            \end{align}
         and the right-hand side of the equation is $\Psi_t (S \oplus S') \combb E = \Psi_t(S') \combb E$  by the definition of $\Psi$,  which is equal to the left-hand side.
         If $\Psi_t(S) \neq \emptylist$, then  left-hand side of the equation is $\Psi_t\big( (\Psi_t(S) \oplus \Psi_t(S')) \combb E \big) = \Psi_t\big(\Psi_t(S)  \combb E \big)$ 
         by the definition of $\Psi$ on the left-choice, which is equal to $\Psi_t(S)  \combb E $, since $\Psi_t(S)  \combb E$ is  position-based. 
         For the right-hand side,  we have  $\Psi_t(S \oplus S')=\Psi_t(S)$ by the definition of $\Psi$, thus we  get the desired result.
      \end{enumerate}
\item  We distinguish two cases depending on whether $\Psi_t(S')=\emptylist$ or not. If $\Psi_t(S')=\emptylist$ then the left-hand side of the equation is
       $\Psi_t\big(\tifthen{S'}{(\Psi_t(S) \combb E)}\big)=\emptylist$ by the definition of $\Psi$, 
       and the right-hand side is $\Psi_t\big(\tifthen{S'}{\Psi_t(S)}\big)  \combb E = \emptylist \combb E = \emptylist$.
       If $\Psi_t(S') \neq \emptylist$ then  left-hand side of the equation is  \\ $\Psi_t\big(\tifthen{S'}{(\Psi_t(S) \combb E)}\big)=\Psi_t(\Psi_t(S) \combb E)$ which is equal  to $\Psi_t(S) \combb E$ 
        since $\Psi_t(S) \combb E$  is a position-based \ce, by the Item~\ref{Properties-of-Psi:Lemma:item:0} of  Lemma~\ref{nice:prop:Psi:lemma}. 
        And the right-hand side is  $\Psi_t\big(\tifthen{S'}{\Psi_t(S)}\big)  \combb E=\Psi_t(\Psi_t(S)) \combb E$ which is equal to $\Psi_t(S) \combb E$ by the same item.
\end{enumerate}
\end{proof}

%% file: correction-combination-free_corrected.tex
\section{Proof of the correctness of the unification of \ces: the fixed-point free setting}
\label{proof:correction:fixed-point-free:sec}

In this section we prove the correctness of the unification procedure in the case where the two input  \ces are fixed-point free (Proposition~\ref{main:proposition:unif:fixed-point-free}). 
This is an important step since we shall reduce in the next three sections~\ref{correction-combination-definitions:section},~\ref{correction:unif:general:setting:unfold:sec},~\ref{equiv:unif:with:unif:unfolding:sec} the general setting  into the fixed-point free one.

We notice that, in the fixed-point free setting,  the memory  involved in the unification system $\Unif$  remains empty and does not play any role since
the only rules that modify  the contexts are the fixed-point ones. Obviously, such rules are not  applied since the input \ces are fixed-point free.
Besides, in this setting, the proof of the  termination and the confluence of $\Unif$ is trivial.   
Indeed, $\Unif$ terminates since each rule transforms a left-hand side \ce into its immediate sub-\ces.


\newcounter{myvar-fact-set-theory}
\setcounter{myvar-fact-set-theory}{\value{theorem}} 

We need a simple  set theoretic fact.
\begin{fact}
\label{set:theoretic:fact}
Let $I',J',J''$ be sets.
Then, $(I' \cap J'')\cup (I' \setminus  (J' \cup J''))=I' \setminus  J'$.
\end{fact}

Since the definition of the mapping $\Psi$ involves the normalization of positions (function $\theta$ in  Item~\ref{psi:def:item:and} of  Definition~\ref{psi:def}),
we need to show  that this normalization does not disturb the unification in the following sense.  
\begin{lemma}
\label{normalization:of:unif:Lemma}
 Let $S=\bigand_{i \in I} @i.S_i$  and  $R=\bigand_{j \in J} @j.R_j$ be two \ces where each $S_i$ and $R_i$ is either the failure   $\emptylist$  or the insertion  $@\varepsilon.\tau_i$, for a context  $\tau_i$ in $\mycal{T}_{\square}$. 
 Then, 
\begin{align}
\label{normalization:of:unif:Lemma:eq}
\Psi_t\big(S\combb R\big)=\Psi_t\big(\theta(S)\combb \theta(R)\big).
\end{align}
\end{lemma}

\begin{proof}
Assume that 
\begin{align*}
   S=\bigand_{i \in I'} @i.S_i \wedge \bigand_{i \in I''} @i.\emptylist  &&\tand&&    R=\bigand_{j \in J'} @j.R_j \wedge \bigand_{j \in J''} @j.\emptylist, \\
\end{align*}
where $S_i \in \mycal{T}_{\square} $ for any $i\in I'$, and $R_j \in \mycal{T}_{\square} $ for any $j\in J'$, and $I' \cap I''=\emptyset$ and $J' \cap J''=\emptyset$. Therefore,
\begin{align*}
   \theta(S)=\bigand_{i \in I'} @i.S_i  &&\tand&&    \theta(R)=\bigand_{j \in J'} @j.R_j. 
\end{align*}
Consider the \ce $\tilde{\Lambda}$:
\begin{align}
\tilde{\Lambda} &= \bigand_{i \in I' \cap J'} @i.(S_i\combb R_i \oplus S_i \oplus R_i)     \wedge    \bigand_{i \in I' \setminus J'} @i.S_i   \wedge   \bigand_{i \in J' \setminus I'} @i.R_i.  \notag
\end{align}
By computing  the \ces $\theta(S) \combb \theta(R)$  and  $S \combb R$  involved in the  right-hand side and the left-hand side of Eq.(\ref{normalization:of:unif:Lemma:eq}) respectively, we get: 
\begin{align}
\Psi_t(\theta(S) \combb \theta(R)) &= \Psi_t(\tifthen{\theta(S)\& \theta(R)}{\tilde{\Lambda}}) \tag{Item~\ref{list:ext} of Def.~\ref{reduction:unif:def} of $\combb$ } \\
                                  & =\Psi_t(\tifthen{S\& R}{\tilde{\Lambda}}) \tag{Item~\ref{Properties-of-Psi:Lemma:item:3'} of Lemma~\ref{nice:prop:Psi:lemma}}\\
                                  & =\Psi_t(\tifthen{S\& R}{\Psi_t(\tilde{\Lambda})}), \tag{Item~\ref{Properties-of-Psi:Lemma:item:3} of Lemma~\ref{nice:prop:Psi:lemma}} \\
                                  & \tand \notag \\ 
\Psi_t\big(S \combb R\big) &= \Psi_t\big(\tifthen{S\& R}{\Lambda}\big) \tag{Item~\ref{list:ext} of Def.~\ref{reduction:unif:def} of $\combb$ }  \\
                                   &= \Psi_t(\tifthen{S\& R}{\Psi_t(\Lambda)}), \tag{Item~\ref{Properties-of-Psi:Lemma:item:3} of Lemma~\ref{nice:prop:Psi:lemma}}
\end{align}
where $\Lambda$ is the \ce
\begin{align}
\Lambda  &=   \bigand_{i \in I' \cap J'} @i.(S_i\combb R_i \oplus S_i \oplus R_i)     \wedge   
        \bigand_{i \in I' \cap J''} @i.(S_i\combb \emptylist  \oplus   S_i \oplus \emptylist)   \wedge  
         \bigand_{i \in I''\cap J'} @i.(\emptylist \combb R_i  \oplus \emptylist  \oplus R_i )   \wedge \notag \\
  &\;\;  \bigand_{i \in I'' \cap J''}  @i.\emptylist   \wedge 
        \bigand_{i \in I' \setminus  (J' \cup J'') }  @i.S_i \wedge 
        \bigand_{i \in I'' \setminus  (J' \cup J'') } @i. \emptylist  \wedge 
        \bigand_{i \in J' \setminus  (I' \cup I'') }  @i.R_i  \wedge 
        \bigand_{i \in J'' \setminus  (I' \cup I'') } @i. \emptylist. \tag{Item~\ref{list:ext} Def.~\ref{reduction:unif:def} of $\combb$} 
\end{align}

Hence to prove Eq.(\ref{normalization:of:unif:Lemma:eq}) we  need to  show that $\Psi_t(\tilde{\Lambda})=\Psi_t(\Lambda)$.
It follows that $\Psi_t(\Lambda)$ can be written as
\begin{align}
\Psi_t(\Lambda)
&=   \Psi_t\Big(\bigand_{i \in I' \cap J'} @i.(S_i\combb R_i \oplus S_i \oplus R_i)     \wedge   
        \bigand_{i \in I' \cap J''} @i.S_i   \wedge  
         \bigand_{i \in I''\cap J'} @i.R_i   \wedge \notag 
      \bigand_{i \in I'' \cap J''}  @i.\emptylist   \wedge\\ 
      &\;\; \bigand_{i \in I' \setminus  (J' \cup J'') }  @i.S_i \wedge 
        \bigand_{i \in I'' \setminus  (J' \cup J'') } @i. \emptylist  \wedge 
        \bigand_{i \in J' \setminus  (I' \cup I'') }  @i.R_i  \wedge 
        \bigand_{i \in J'' \setminus  (I' \cup I'') } @i. \emptylist \Big)  \tag{since $\Psi_t(S_i\combb \emptylist  \oplus   S_i \oplus \emptylist)=\Psi_t(S_i)$ and $\Psi_t(\emptylist \combb R_i  \oplus \emptylist  \oplus R_i)=\Psi_t(R_i)$, 
                    by Item~\ref{Properties-of-Psi:Lemma:item:4} of Lemma~\ref{nice:prop:Psi:lemma}}\\
& \notag \\
&=  \Psi_t\Big(\bigand_{i \in I' \cap J'} @i.(S_i\combb R_i \oplus S_i \oplus R_i)     \wedge   
        \bigand_{i \in I' \cap J''} @i.S_i   \wedge  
        \bigand_{i \in I' \setminus  (J' \cup   J'') }  @i.S_i \wedge 
         \bigand_{i \in J' \setminus  (I' \cup I'') }  @i.R_i \wedge    
         \bigand_{i \in I''\cap J'} @i.R_i    \Big) \tag{since $\Psi_t(@i.\emptylist)=\emptylist$, by Item~\ref{Properties-of-Psi:Lemma:item:4'} of Lemma~\ref{nice:prop:Psi:lemma}}\\
& \notag \\
  &=  \Psi_t\Big( \bigand_{i \in I' \cap J'} @i.(S_i\combb R_i \oplus S_i \oplus R_i)     \wedge   
        \bigand_{i \in I' \setminus J'} @i.S_i   \wedge  
      \bigand_{i \in J' \setminus I'} @i.R_i \Big)   \tag{since $(I' \cap J'')\uplus (I' \setminus  (J' \cup J''))=I' \setminus  J'$ and $(J'\cap I'')\uplus (J'\setminus (I'\cup I''))=J' \setminus  I'$, by Fact~\ref{set:theoretic:fact}} \\
  & = \Psi_t(\tilde{\Lambda}). \tag{Def. of $\tilde{\Lambda}$} 
    \end{align}
\end{proof}


\begin{notation}
Throughout this paper the set of fixed-point free \ces will be denoted  by $\ceSetFree$. 
\end{notation}

Now we are ready to show the main result of this section, that is, that the unification of fixed-point free \ces is correct. 
\begin{proposition}
\label{main:proposition:unif:fixed-point-free}
For every term $t \in \mycal{T}$ and for every fixed-point free \ces  $S$ and $R$  in $\ceSetFree$, 
we have that 
\begin{align}
\label{main:lemma:unif:fixed-point-free:eq}
\Psi_t(S \combb R) & = \Psi_t(S) \combb \Psi_t(R).
\end{align}
Or, equivalently, the following diagram commutes.
\[\begin{tikzcd}
\ceSetFree \times \ceSetFree \arrow{r}{\combb} \arrow[swap]{d}{\Psi_t \times \Psi_t} &  \ceSetFree \arrow{d}{\Psi_t} \\
\mycal{E} \times \mycal{E}  \arrow{r}{\combb} & \mycal{E}
\end{tikzcd}
\]
\end{proposition}
\begin{proof}
The proof is by structural induction on $S$ and $R$, which amounts to consider $\delta(S)$  the depth of $S$, and $\delta(R)$  the depth of $R$.
\begin{description} 
\item \textbf{Base case}. If $(\delta(S),\delta(R))=(0,0)$ then $S=\emptylist$ or $S=@\varepsilon.\tau$,  and   $R=\emptylist$ or $R=@\varepsilon.\tau'$.   
In this case  the proof is trivial since $\Psi_t(S) =S$ and   $\Psi_t(R) =R$.
\item \textbf{Induction step}.  We assume that the claim holds for some $S'$ and $R'$ and we shall show it for any $S$ and $R$ such that either 
  \emph{i.)} $S'$ is an immediate  sub-\ce of $S$ and $R'=R$, or
  \emph{ii.)} $R'$ is an immediate  sub-\ce of $R$ and $S'=S$, or
  \emph{iii.)} $S'$ is an immediate  sub-\ce of $S$, and $R'$ is an immediate  sub-\ce of $R$.
\begin{enumerate}
\item If $S=u;S'$ and $R$ is arbitrary then 
\begin{align}
 \Psi_t(S \combb R)   & = \Psi_t((u;S') \combb R)  \notag \\   
                      & =  \Psi_t\big(u;(S' \combb R)\big) \tag{Item~\ref{pattern:ext:1} of Def.~\ref{reduction:unif:def} of $\combb$} \\
                      & =  \Psi_t\big(u;\Psi_t(S' \combb R)\big) \tag{Item ~\ref{Properties-of-Psi:Lemma:item:1} of Lemma~\ref{nice:prop:Psi:lemma}} \\
                      &=  \Psi_t\big(u; \big(\Psi_t(S') \combb \Psi_t(R)\big)\big)  \tag{Ind. hypothesis since $S'$ is an immediate sub \ce of $S$, and $R=R'$} \\
                      &=  \Psi_{t}(u; S')  \combb \Psi_t(R) \tag{Item~\ref{psi:unif-congruence:Lemma:item:1} of Lemma~\ref{psi:unif-congruence:Lemma}}  \\
                      & = \Psi_t(S)  \combb \Psi_t(R). \tag{Def. of $S$}
\end{align}
\item If $S=S' \oplus S''$ and $R$ is arbitrary then 

\begin{align}
 \Psi_t(S \combb R)   &= \Psi_t((S' \oplus S'') \combb R)                              \notag \\   
                      &=   \Psi_t\big((S' \combb R) \oplus (S'' \combb R)\big)         \tag{Item~\ref{choice:ext:1} of Def.~\ref{reduction:unif:def} of $\combb$} \\
                      &= \Psi_t\big(\Psi_t(S' \combb R) \oplus \Psi_t(S'' \combb R)\big)        \tag{Item ~\ref{Properties-of-Psi:Lemma:item:2} of Lemma~\ref{nice:prop:Psi:lemma}} \\
                      &=  \Psi_t\big(\big(\Psi_t(S') \combb \Psi_t(R)\big) \oplus \big(\Psi_t(S'') \combb \Psi_t(R)\big)\big)                    \tag{Ind. hypothesis} \\
 &=  \Psi_t\big(\big(\Psi_t(S')  \oplus \Psi_t(S'') \big) \combb \Psi_t(R)\big)                    \tag{Def. of $\combb$} \\
                      &=  \Psi_t(S' \oplus S'')  \combb \Psi_t(R)                          \tag{Item~\ref{psi:unif-congruence:Lemma:item:2} of Lemma~\ref{psi:unif-congruence:Lemma}} \\
                      &= \Psi_t(S)  \combb \Psi_t(R). \tag{Def. of $S$}
\end{align}

\item If $S=\tifthen{S'}{S''}$ and $R$ is arbitrary  then
  \begin{align}
    \Psi_t(S \combb R) &= \Psi_t((\tifthen{S'}{S''}) \combb R) \notag \\
                       &= \Psi_t(\tifthen{S'}{(S''\combb R)})                                    \tag{Item~\ref{if:ext:1} Def.~\ref{reduction:unif:def} of $\combb$} \\
                       &= \Psi_t\big(\tifthen{S'}{\Psi_t((S''\combb R))}\big)                    \tag{Item~\ref{Properties-of-Psi:Lemma:item:3} of Lemma~\ref{nice:prop:Psi:lemma}} \\
                       &= \Psi_t\big(\tifthen{S'}{\big(\Psi_t(S'')\combb \Psi_t(R)\big)}\big)     \tag{Ind. hypothesis} \\
                       &= \Psi_t\big(\tifthen{S'}{\Psi_t(S'')}\big) \combb \Psi_t(R)             \tag{Item~\ref{psi:unif-congruence:Lemma:item:3} of Lemma~\ref{psi:unif-congruence:Lemma}} \\
                       &= \Psi_t\big(\tifthen{S'}{S''}\big) \combb \Psi_t(R)                     \tag{Item~\ref{Properties-of-Psi:Lemma:item:3} of Lemma~\ref{nice:prop:Psi:lemma}} \\
                       & = \Psi_t(S)  \combb \Psi_t(R). \tag{Def. of $S$}
   \end{align} 

\item \label{proof:lemma:unif:free-fp:item:list:list} If $S=\bigand_{i \in I} @i.S_i$  and  $R=\bigand_{j \in J} @j.R_j$  then   let 
 \begin{align*} 
   M_1        &=\bigand_{i \in I\setminus J} @i.S_i              &\tand&&         M_2        &=\bigand_{j \in J\setminus I} @j.R_j,   \\
   M^{\star}_1 &=\bigand_{i\in I \setminus J}@i.\Psi_{t_{|i}}(S_i)   &\tand&&          M^{\star}_2 &=\bigand_{j\in J \setminus I}@j.\Psi_{t_{|j}}(R_j).
   \end{align*} 
The  left-hand side of Eq.(\ref{main:lemma:unif:fixed-point-free:eq}) can be written as
\begin{align}
\textrm{LH}.\ref{main:lemma:unif:fixed-point-free:eq}
&=\Psi_t(  S \combb R) \notag \\
&=  \Psi_t\big(\tifthen{S\&R}{\bigand_{i \in I \cap J} @i.(S_i \combb R_i \oplus S_i \oplus R_i) \uand  M_1 \uand  M_2}\big) \tag{Item~\ref{list:ext} of Def.~\ref{reduction:unif:def}  of $\combb$}\\
&=  \Psi_t\Big(\tifthen{S\&R}{\Psi_t\big(\bigand_{i \in I \cap J} @i.(S_i \combb R_i \oplus S_i \oplus R_i) \uand  M_1 \uand  M_2\big)}\Big) \tag{Item~\ref{Properties-of-Psi:Lemma:item:3} of Lemma~\ref{nice:prop:Psi:lemma}}\\
&= \Psi_t\bigg(\tifthen{S\&R}{\theta\Big(\bigand_{i \in I \cap J} @i.\Psi_{t_{|i}}(S_i \combb R_i\oplus S_i \oplus R_i) \uand M^{\star}_1  \uand  M^{\star}_2 \Big)}\bigg)  \tag{Item~\ref{psi:def:item:and} of Def.~\ref{psi:def} of $\Psi_t(\bigand(\cdot))$}\\
&= \Psi_t\bigg(\tifthen{S\&R}{\bigand_{i \in I \cap J} @i.\Psi_{t_{|i}}(S_i \combb R_i\oplus S_i \oplus R_i) \uand M^{\star}_1  \uand  M^{\star}_2 }\bigg)  \tag{Item~\ref{Properties-of-Psi:Lemma:item:3'} of Lemma~\ref{nice:prop:Psi:lemma}} \\
&= \Psi_t\bigg(\tifthen{S\&R}{\bigand_{i \in I \cap J} @i.\big(\Psi_{t_{|i}}(S_i) \combb \Psi_{t_{|i}}(R_i) \oplus \Psi_{t_{|i}}(S_i) \oplus \Psi_{t_{|i}}(R_i) \big) \uand M^{\star}_1 \uand  M^{\star}_2 }\bigg) \tag{Ind. hyp.}\\
&=\Psi_t\bigg(\bigand_{i\in I} @i.\Big(\Psi_{t_{|i}}(S_i)\Big) \combb \bigand_{j\in J} @i.\Big(\Psi_{t_{|i}}(R_i) \Big)  \bigg)                                    \tag{Item~\ref{list:ext} of Def. ~\ref{reduction:unif:def} of $\combb$}\\
&=\Psi_t\bigg(\theta\Big(\bigand_{i\in I} @i.\big(\Psi_{t_{|i}}(S_i)\Big) \combb \theta\Big(\bigand_{j\in J} @i.\big(\Psi_{t_{|i}}(R_i) \big)\big) \Big) \bigg)     \tag{Lemma~\ref{normalization:of:unif:Lemma}}\\
&=\Psi_t\bigg(\Psi_{t}\Big(\bigand_{i\in I} @i.S_i \Big) \combb  \Psi_{t}\Big(\bigand_{j\in J} @j.R_j \Big)\bigg)                                              \tag{Item~\ref{psi:def:item:and} of Def.~\ref{psi:def} of $\Psi_t(\bigand\cdot)$}\\
&=\Psi_{t}\Big(\bigand_{i\in I} @i.S_i \Big) \combb  \Psi_{t}\Big(\bigand_{j\in J} @j.R_j \Big)  \tag{Lemma~\ref{nice:prop:Psi:lemma}} \\
&= \Psi_t(S)  \combb \Psi_t(R). \tag{Def. of $S$ and $R$}
\end{align}

\item If $S=\most(S')$ and $R=\most(R')$  then assume that $t$ is neither a constant nor a rewriting variable, i.e. $\delta(t)\ge 2$, the case when  $\delta(t)=1$ being trivial since both sides of the equation are equal to $\emptylist$.
 In this case we rewrite $\most(\cdot)$ as $\bigand_i(\cdot)$ and we apply Item~\ref{proof:lemma:unif:free-fp:item:list:list} of this proof.
Let 
\begin{align*}
S^{\star}=\bigand_{i=1,ar(t)} @i.S' &\tand && R^{\star}=\bigand_{i=1,ar(t)} @i.R', 
\end{align*}
and notice that $\Psi_t(S^{\star})=\Psi_t(S)$ and $\Psi_t(R^{\star})=\Psi_t(R)$.
Hence
\begin{align}
 \Psi_t(S \combb R)   &= \Psi_t \big(\most(S') \combb \most(R')\big)                           \notag \\   
                      &= \Psi_t\big( \tifthen{(S\&R)}{\big(\most\big((S'\combb R') \oplus S' \oplus R'\big)\big)}\big)             \tag{Def.~\ref{reduction:unif:def} of $\combb$} \\                    
                      &= \Psi_t \big(\tifthen{(S\&R}{\Psi_t\big(\most\big((S'\combb R') \oplus S' \oplus R'\big)\big)} \big)                    \tag{Item ~\ref{Properties-of-Psi:Lemma:item:3} of Lemma~\ref{nice:prop:Psi:lemma}} \\
                      &= \Psi_t \big(\tifthen{(S^{\star}\&R^{\star})}{\Psi_t\big(\most\big((S'\combb R') \oplus S' \oplus R'\big)\big)} \big)        \tag{Item ~\ref{Properties-of-Psi:Lemma:item:3'''} of Lemma~\ref{nice:prop:Psi:lemma}} \\
                      &= \Psi_t \bigg(\tifthen{(S^{\star}\&R^{\star})}{\Psi_t \Big(\bigand_{i=1,ar(t)} @i.\big((S'\combb R') \oplus S' \oplus R' \big)\Big)} \bigg)   \tag{Item~\ref{psi:def:item:most} of Def.~\ref{psi:def}  of $\Psi_t(\most(\cdot))$} \\
                      &= \Psi_t \bigg(\tifthen{(S^{\star}\&R^{\star})}{\bigand_{i=1,ar(t)} @i.\big((S'\combb R') \oplus S' \oplus R' \big)} \bigg)   \tag{Item ~\ref{Properties-of-Psi:Lemma:item:3'} of Lemma~\ref{nice:prop:Psi:lemma}} \\
                      &= \Psi_t \Big(\bigand_{i=1,ar(t)} @i.S'  \combb \bigand_{i=1,ar(t)} @i.R'\Big)   \tag{Item~\ref{list:ext} of Def. ~\ref{reduction:unif:def} of $\combb$ in which $I=J=\{1,\ldots,ar(t)\}$} \\
                      &= \Psi_t \Big(\bigand_{i=1,ar(t)} @i.S'\Big)  \combb \Psi_t \Big(\bigand_{i=1,ar(t)} @i.R'\Big)   \tag{Item~\ref{proof:lemma:unif:free-fp:item:list:list} of this proof} \\
                      &=  \Psi_t \big(\most(S')\big)  \combb \Psi_t\big(\most(R')\big)     \tag{Item~\ref{psi:def:item:most} of Def.~\ref{psi:def} of $\Psi_t(\most(\cdot))$}\\
                      &= \Psi_t(S) \combb \Psi_t(R).                                         \tag{Def. of $S$ and $R$}
\end{align}
\end{enumerate}
\end{description}
\end{proof}


%% file: correction-combination-definitions.tex
\section{Properties of the unification reduction system and of \ces}
\label{correction-combination-definitions:section}
This section, together with the next two sections~\ref{correction:unif:general:setting:unfold:sec}  and~\ref{equiv:unif:with:unif:unfolding:sec}, are devoted to developing the ingredients
required  in the proof of the  main result of this paper regarding the correctness of the unification of \ces in the general setting, in which the \ces contain fixed-point operators.
In this  section we introduce definitions and show preliminary results which will be used in   the next two  sections.
In  Subsection~\ref{measures:trees:subsec} we define some measures on the structure of \ces,  namely the number of   nested fixed-point operators of a \ce  and its size.
In Subsection~\ref{termination:confluence:reduction:unif:sec} we show the termination and the confluence of the unification reduction system.
Un Subsection~\ref{iter:unfolding:subsec}  we introduce  the operation of unfolding which turns all fixed-point operators of a \ce into  iterations of arbitrary fixed size.
In Subsection  ~\ref{prop:sem:ces:subsec}  we show some useful properties related to  \ces, namely the semantic equivalence of two \ces when applied to  terms of a certain depth, as well as a condition under which a \ce is equivalent to a fixed-point one.
In Subsection  ~\ref{composition:lemma:section} we show a key  Lemma, called \emph{composition lemma}, that expresses the unification of two \ces in terms of  their  sub-\ces.

\subsection{Measures of \ces: the star height and the depth  of   \ces}
\label{measures:trees:subsec}
Taking into account that the structure of a \ce is no longer a tree but  a tree with back-edges 
that may contain cycles,  we slightly modify the standard  measure  of 
the depth of trees in order to capture both the number of nested 
loops, caused by the nested application  of the fixed-point constructor $\mu$, and the 
distance from the root of the tree to the leaves. 
Many  proofs will be done by  induction   on  this measure.

We adapt the definition of the star height~\cite{eggan1963,Courcelle84} that measures  the depth of Kleen operator $\star$ in regular languages to 
 \ces in order to  capture the  number of the nested fixed-point constructor.

\begin{definition}[Star height  of a \ce]
\label{def:star:height:strategy}
The \emph{star height}  of a \ce   is the function \\ $\h: \ceSet  \longrightarrow    \mathbb{N}$ defined inductively as follows:
\begin{align*}
\h(S) = \begin{cases}
          0 & \tif S \textrm{ is fixed-point free} \\
          \mmax\big\{\h(S'(X_1,\ldots,X_n)),\h(R_1),\ldots, \h(R_n)\big\}  &\tif S = S'(R_1,\ldots,R_n), n \ge 1 \\
          1+ \h(S') &\tif S=\mu X.S'.
       \end{cases}
\end{align*}
\end{definition}

\begin{example}[Star height]
  If $S(X)$ and $R(Y)$  are fixed-point free \ces with distinct free fixed-point variables, then
  \begin{align*}
    \h(S(X))= \h(R(Y))=0.
  \end{align*}
  We compute the star height of the \ces  $\mu X.S(X) \oplus \mu Y.R(Y)$ and  $\mu X.\mu Y. (S(X) \oplus R(Y))$ and $\mu X. \big(S(X) \oplus \mu Y. R(Y)\big)$.
Since  the two fixed-point operators  in $\mu X.S(X) \oplus \mu Y.R(Y)$  are not nested, we have that:
  \begin{align*}
    \h\big(\mu X.S(X) \oplus \mu Y.R(Y)\big) &= \mmax\big\{\h(\mu X.S(X)), \h(\mu Y.R(Y)) \big\} \\
     &= \mmax\big\{1+\h(S(X)), 1+\h(R(Y)) \big\} \\
     &=1.
  \end{align*}

  However, since the two fixed-point operators in  $\mu X.\mu Y. (S(X) \oplus R(Y))$ are  nested, we have that:
\begin{align*}
  \h\big(\mu X.\mu Y. (S(X) \oplus R(Y))\big)  &=  1+ \h\big(\mu Y. (S(X) \oplus R(Y))\big) \\
                                               &=  1+ 1+\h\big(S(X) \oplus R(Y)\big) \\
                                               &=2.
\end{align*}
And similarly, the two fixed-point operators in $\mu X. \big(S(X) \oplus \mu Y. R(Y)\big)$ are nested, thus we  get: 
\begin{align*}
  \h\big(\mu X. \big(S(X) \oplus \mu Y. R(Y)\big)\big)  &=2.
\end{align*}
\end{example}

We combine the star height  and the tree depth, defined in Definition~\ref{def:Delta:strategy'}, to obtain the desired measure that 
takes into account both the number of the nested fixed-point constructors  and the size of a \ce.
\begin{definition}[Depth  of a  \ce]
\label{def:Delta:strategy}
The \emph{depth}    of a \ce $S$ is the function \\ $\Delta: \ceSet  \longrightarrow    \mathbb{N} \times \mathbb{N}$ defined by 
\begin{align*}
\Delta(S)=(\h(S),\delta(S)).
\end{align*}
\end{definition}

Notice that if a \ce $S$ is fixed-point free, i.e. it does not contain the fixed-point 
constructor $\mu$,  then its depth  $\Delta(S)=(0,n)$, for some $n \in \mathbb{N}$.

The following fact shows that the depth of a fixed-point \ce is strictly greater than the depth of its unfolding.
\begin{fact}
\label{Delta:monotonic:fact}
Let $\mu X.S(X)$  be a \ce where $X$ is free in $S(X)$.
Then for any integer $n\ge 0$ we have  
\begin{align*}
\Delta(\mu^{n} X.S(X)) < \Delta(\mu X.S(X)).
\end{align*}
\end{fact}
\begin{proof}
The case    when $n=0$ is trivial since $\Delta(\mu^{0} X.S(X))=\Delta(\emptylist)=(0,0)$. 
We show next that $\h(\mu X.S(X))=1 + \h(\mu^{n} X.S(X))$ for any $ n \ge 1$. 
It follows from the definition of the star height that  $\h(\mu^{n} X.S(X))=\h\big(S(S(\ldots(S(\emptylist))))\big)=\mmax\{\h(S(X)),\h(S(\emptylist))\}=\h(S(\emptylist)) = \h(S(X))$.
On the other hand, by the definition of the star height $\h(\mu X.S(X))=1 + \h(S(X))$. And it follows from the lexicographic order that $\Delta(\mu^{n} X.S(X))< \Delta(\mu X.S(X))$. 
\end{proof}

We next define the number of jumps (i.e. \ces which are position jumps of the form  $@i.S$ or $\most$s)  that lie  between  the root of a  \ce to   a free fixed-point variable.
The idea is that by meeting jumps,  the \ce makes progress. In particular, if at least one jump  lies between any fixed-point constructor $\mu X$ and the occurrence of $X$  in a \ce $S$, then  $S$  is monotonic.
Besides, we can  compare  the semantics  of two \ces $M(S)$ and $M(R)$ thanks to  number of jumps  between the root of $M(X)$ and $X$.

\begin{definition}
\label{number:of:jumps:def}
Let $S(X)$ be a \ce where the fixed-point variable $X$ is free and appears once.
The  \emph{number of jumps}   between  the root of $S(X)$ and $X$, denoted by  $\Pi_X(S(X))$,  is inductively defined  as follows:
\begin{align*}
\Pi_X(X)                              &= 0   \\
\Pi_X(u;S'(X))                        &=     \Pi_X(S'(X)) \\
\Pi_X(S_1(X) \oplus  S_2)             &=  \Pi_X(S_1(X)) \\
\Pi_X(S_1 \oplus  S_2(X))             &=  \Pi_X(S_2(X)) \\
\Pi_X\big(\tifthen{S''}{S'(X)}\big)   & =    \Pi_X(S'(X)) \\
\Pi_X\big(\tifthen{S''(X)}{S'}\big)   & =    \Pi_X(S''(X)) \\
\Pi_X(\mu Y.S'(X,Y))                  & =  \Pi_X(S'(X,Y))\\
\Pi_X\big( (\bigand_{i=1,m} @i.S_i) \wedge @j.S'(X)\big)      & = 1+ \Pi_X(S'(X))  \\
\Pi_X(\most(S'(X))             & = 1+  \Pi_X(S'(X)).
\end{align*}
\end{definition}

\begin{example}
  Let $u,u'$ be patterns in $\mycal{T}$, and let $S'$ be a fixed-point free \ce.
  Let $S(X)$ be the following \ce:
  \begin{align*}
    S(X) = u; @1.\big(\most(u';@2.X) \oplus S'\big).
  \end{align*}
  There are three jumps between the root of $S(X)$ and $X$, which are   $@1.(\cdot)$ and $\most(\cdot)$ and $@2.(\cdot)$. That is,
  \begin{align*}
    \Pi_{X}(S(X)) & = \Pi_{X}\big(u; @1.\big(\most(u';@2.X) \oplus S'\big) \big) \\
    & = \Pi_{X}\big(@1.\big(\most(u';@2.X) \oplus S'\big) \big) \\
    & = 1+\Pi_{X}\big(\most(u';@2.X) \oplus S'  \big)  \\
    & = 1+\Pi_{X}\big(\most(u';@2.X) \big)  \\
    & = 2+\Pi_{X}\big(u';@2.X  \big)  \\
    & = 2+\Pi_{X}\big(@2.X  \big)  \\
    & = 3+\Pi_{X}\big(X  \big)  \\
    & = 3.
  \end{align*}
\end{example}
Notice that if $S$ is monotonic, then for every sub-\ce $\mu X.S'(X)$ of $S$, we have that $\Pi_X(S'(X))\ge 1$.

\subsection{Termination and confluence of the  unification reduction system}
\label{termination:confluence:reduction:unif:sec}
To show the termination of the reduction system $\Unif$
we need to define a measure on the tuples that strictly decreases with each derivation rule. 
Notice that all the  reduction rules strictly decrease the size of one or both of the left-hand side \ces except 
the  fixed-point rules (\ref{fixed:ext:1}) and (\ref{fixed:ext:2}) which can replace  $\mu X.S(X)$  with  $S(\mu X.S(X))$ that is larger than $\mu X.S(X)$. 
However,  these  fixed-point rules increase the size of the memory  because  the right-hand side memory  is augmented with $(\mu X.S(X),R,\cdot)$. 
Since the size of any  memory  related to two fixed \ces is bounded, to ensure the termination of $\Unif$, we need to define a  measure that 
 couples the difference between such bound and the size of the memory    with the size of the \ces.

\begin{definition}
Let $S$ and $R$ be \ces, and let $\Eu{M}$ be a memory in $\mathfrak{M}(S,R)$. We pose
\begin{align*}
\Lambda(S,R,\Eu{M}) := |\Phi_{\mu}(S)|\cdot|\Phi(R)| + |\Phi(S)|\cdot|\Phi_{\mu}(R)| - |\Eu{M}|
\end{align*}
and define the measure $(\Lambda(S,R,{\Eu{M}}),\Delta(S),\Delta(R))$.
\end{definition}

\begin{proposition} 
 The unification reduction system $\Unif$ enjoys the following properties.
\begin{enumerate} 
\item The reduction system $\Unif$  is terminating and confluent. 
\item The  normal form of a pre-\ce with respect to $\Unif$  is a \ce in $\mycal{C}$ (i.e. the normal form does not contain tuples).
\end{enumerate}
\end{proposition}
\begin{proof}
\begin{enumerate} 
\item The termination is guaranteed by the fact that each reduction rule strictly decreases the measure \\ $(\Lambda(S,R,{\Eu{M}}),\Delta(S),\Delta(R))$ with respect to the lexicographic order.
  The confluence     is  guaranteed by  the  priority order imposed  on the reduction rules. 

\item Each rule either advances in the \ce of the tuple of the left-hand side part of this rule, or reduces the left-hand side part into a \ce.
\end{enumerate}
\end{proof}

We show next in Lemma~\ref{position:cross:in:generated:ces:lemma} a useful property of the unification
of monotonic \ces: if the same fixed-point \ce appears twice in a derivation with respect to the unification reduction system $\Unif$,  then this derivation produces a jump.  
Indeed this is a direct   consequence of monotonicity.


\begin{lemma}
\label{position:cross:in:generated:ces:lemma}
Let $\mu X.S(X)$, $R$ and $R'$ be \ces.
Let  $T(Z)$ be a  pre-\ce.
Let $\Eu{M}, \Eu{M}'  \in  \mathfrak{M}$ be memories.
If there is a series of derivations of one of the following forms:
\begin{align*}
  \tuple{\mu X.S(X),R,\Eu{M}}  & \xreduces{\star} T(\tuple{\mu X.S(X),R',\Eu{M}'}) \\
  \textrm{ or } &\\
  \tuple{R,\mu X.S(X),\Eu{M}}  & \xreduces{\star} T(\tuple{R',\mu X.S(X),\Eu{M}'}) \\
  \textrm{ or } &\\
  \tuple{\mu X.S(X),R,\Eu{M}}  & \xreduces{\star} T(\mu X.S(X)) \\
  \textrm{ or } &\\
  \tuple{R,\mu X.S(X),\Eu{M}}  & \xreduces{\star} T(\mu X.S(X))
\end{align*}
in $\Unif$, then there is  at least one jump between the root of  $T(Z)$ and  $Z$. That is,  
\begin{align*}
  \Pi_{Z}(T(Z)) \ge 1.
\end{align*}
\end{lemma}
\begin{proof}
We only consider the first derivation since the other ones can be obtained by the same arguments. 
Recall that $\mu X.S(X)$ is monotonic by the general Assumption~\ref{global:assumptions:ces}, that is,  between  $\mu X.S(X)$ and  $X$ there is a   position  jump  or $\most(\cdot)$.
This implies that, there exist  \ces $\tilde S$ and $\tilde  R$, a memory $\tilde{\Eu{M}}$, a tuple $\tilde{T}(\tilde{Z})$, and  a series of  derivations 
\begin{align*}
\tuple{\mu X.S(X),R,\Eu{M}} \xreduces{\star} \tilde{T}(\tuple{\tilde{S},\tilde{R},\tilde{\Eu{M}}})  \xreduces{\star} T(\tuple{\mu X.S(X),R',\Eu{M}'})
\end{align*}
in $\Unif$ where $\tilde{S}$ is either of the form $\tilde{S}=\bigand_i @i.S'_i$ or  $\tilde{S}=\most(S'')$. This implies that one of the rules (\ref{list:ext'}), (\ref{list:ext}), (\ref{most:ext:1}), (\ref{most:ext:2}), (\ref{most:ext:3})  is applied in the derivation from $\tilde{T}(\tuple{\tilde{S},\tilde{R},\tilde{\Eu{M}}})$ to  $T(\tuple{\mu X.S(X),R',\Eu{M}'})$.
Each of which produces a position jump  or $\most(\cdot)$.  
\end{proof}

An immediate  consequence of the previous Lemma~\ref{position:cross:in:generated:ces:lemma}  is the following Corollary.
\begin{corollary}
  \label{monotonic:unif:corolarry}
The unification of two monotonic \ces is a monotonic \ce. 
\end{corollary}
\subsection{Iteration mapping and (generalized) unfolding of \ces}
\label{iter:unfolding:subsec}
We next generalize the notion of unfolding of \ces  to allow the replacement of  each fixed-point constructor  of a \ce  by an iteration of arbitrary fixed size.
The resulting \ce is obviousely fixed-point free.
\begin{definition}[Iteration mapping, unfolding of a \ce]
\label{ufold:def}
Let $S$ be a \ce with bound fixed-point variables $X_1,\ldots,X_r$ and  let $\mathbf{s}:\{X_1,\ldots,X_r\} \to  \mathbb{N}$ be a mapping, called hereafter \emph{iteration mapping}.  
The \emph{unfolding} of $S$ with respect to $\mathbf{s}$, denoted by  $\ufold{S}{\mathbf{s}}$, consists of replacing each fixed-point constructor by a certain number of iterations given by $\mathbf{s}$.
It is inductively defined  as follows:
\begin{align*}
\ufold{\emptylist}{\mathbf{s}}               &= \emptylist \\
\ufold{X}{\mathbf{s}}                        &= X   \\
\ufold{@\varepsilon.\mbf\tau}{\mathbf{s}}    &=@\varepsilon.\mbf\tau \\
\ufold{u;S}{\mathbf{s}}                     &=   u;\ufold{S}{\mathbf{s}}  \\ 
\ufold{@p.S}{\mathbf{s}}                   &= @p.\ufold{S}{\mathbf{s}}  \\ 
\ufold{S_1 \oplus  S_2}{\mathbf{s}}         &=  \ufold{S_1}{\mathbf{s}} \oplus \ufold{S_2}{\mathbf{s}}  \\
\ufold{\bigand_{i=1,m} S_i}{\mathbf{s}}      & =  \bigand_{i=1,m} \ufold{S_i}{\mathbf{s}} \\
\ufold{\tifthen{S_1}{S_2}}{\mathbf{s}}     & =  \tifthen{\ufold{S_1}{\mathbf{s}}}{\ufold{S_2}{\mathbf{s}}} \\
\ufold{\mu X.S(X)}{\mathbf{s}}             & =  \mu^{\mathbf{s}(X)}X.\ufold{S(X)}{\mathbf{s}}. 
\end{align*}
For two iteration mappings $\mathbf{s}$ and $\mathbf{s}'$ defined on the same domain, we shall write $\mathbf{s} \ge \mathbf{s}'$ to mean that  $\mathbf{s}(X) \ge \mathbf{s}'(X)$ for any $X$ in the domain.
We shall write also $\mathbf{s} > \mathbf{s}'$ to mean that $\mathbf{s} \ge \mathbf{s}'$ and there exists  $X$ in the domain such that  $\mathbf{s}(X) > \mathbf{s}'(X)$ .
\end{definition}
Notice that, for a \ce $S$ and an iteration mapping $\mathbf{s}$, if $S$ is fixed-point free  then $\ufold{S}{\mathbf{s}}=S$.
\begin{example}[Unfolding of a \ce]
  Let $S(X)$ and $R(X)$ be fixed-point free \ces. Let
   \begin{align*}
     & T(X) =S(X) \oplus \mu Y. R(Y)
    \end{align*}
    be  a \ce. Consider the iteration mapping $\mathbf{s}$  defined by :
    \begin{align*}
      \mathbf{s} =\set{X \mapsto 2, Y\mapsto 3}.
    \end{align*}
    Then the unfolding  of the \ce $\mu X.T(X)$ with respect to  $\mathbf{s}$ is defined as follows:
    \begin{align*}
      \ufold{\mu X.T(X)}{\mathbf{s}} &= \mu^2 X.\ufold{T(X)}{\mathbf{s}}    \\
       &= \mu^2 X.\bufold{S(X) \oplus \mu Y. R(Y)}{\mathbf{s}}    \\
      &= \mu^2 X.\big(\ufold{S(X)}{\mathbf{s}} \oplus  \bufold{\mu Y. R(Y)}{\mathbf{s}}\big)    \\
      &= \mu^2 X.\big(S(X) \oplus  \bufold{\mu Y. R(Y)}{\mathbf{s}}\big) \tag{Since $S(X)$ is fixed-point free}   \\
      &= \mu^2 X.\big(S(X) \oplus  \mu^3 Y.\ufold{R(Y)}{\mathbf{s}}\big)   \\
      &= \mu^2 X.\big(S(X) \oplus  \mu^3 Y.R(Y)\big).   \tag{Since $S(X)$ is fixed-point free}
    \end{align*}
    Further computations involve the replacement of each fixed-point operator by an iteration,  given in  Definition~\ref{ufold:iteration:def}, as follow.
    Let $ T'(X)$ be the \ce:
     \begin{align*}
         T'(X) =S(X) \oplus  \mu^3 Y.R(Y),
     \end{align*}
     hence,
        \begin{align*}
          \mu^2 X.\big(S(X) \oplus  \mu^3 Y.R(Y)\big) &= \mu^2 X.T'(X) \\
          &=  T'\big(\mu^1 X.T'(X)\big) \\
          &=  T'\big(T'(\emptylist)\big) \tag{Def. ~\ref{ufold:iteration:def}} \\
          &=  S\big(T'(\emptylist)\big) \oplus  \mu^3 Y.R(Y)  \tag{Def. of $T'(X)$} \\
          &=  S\big( S(\emptylist) \oplus  \mu^3 Y.R(Y) \big) \oplus  \mu^3 Y.R(Y)  \tag{Def. of $T'(X)$}\\
          &=  S\big( S(\emptylist) \oplus  R(R(R(\emptylist))) \big) \oplus  R(R(R(\emptylist))).
        \end{align*}
\end{example}

\subsection{Properties   of \ces and their fixed-points}
\label{prop:sem:ces:subsec}

We give fundamental properties of \ces regarding their semantics and fixed-points.
Namely the properties related, on the one hand, to the composition of \ces
in the sense of a \ce being a sub-\ce of another one (Lemma~\ref{depth:position:composition:lemma}), and on the other hand, a sufficient condition under which a \ce is  equivalent to a fixed-point one (i.e. Corollary~\ref{general-fixed-point-corollary}). Finally, we study the equivalence between a \ce and its unfolding (Lemma~\ref{unfold:equiv:lemma}).

\begin{lemma}
\label{depth:position:composition:lemma}
Let $S(X)$, $R$ and $R'$ be  \ces where  the fixed-point variable $X$ appears  once  in $S(X)$, and let  $n',n''\ge 1$.
\begin{enumerate}
\item \label{depth:position:composition:lemma:item:1} If $R \equiv_{n'} R'$ and   $n''=\Pi_{X}(S(X))$ then  $S(R)\equiv_{n'+n''} S(R')$. 
\item \label{depth:position:composition:lemma:item:2} If $R \equiv_{n'} R'$ and  $n''\le n'$  then  $S(R)\equiv_{n''} S(R')$. 
\item \label{depth:position:composition:lemma:item:3} For any fixed-point free  \ce $\tilde{S}(X^1,\ldots,X^k)$, and \ces $S_1,\ldots,S_k$ with $k\ge 1$, we have that 
  \begin{align}
    \label{depth:position:composition:lemma:eq}
    \tilde{S}(S_1,\ldots,S_k) \equiv_{\mathbf{m}}\tilde{S}(\emptylist,\ldots,\emptylist)
    &&\textrm{ where }&&
    \mathbf{m}=\min \bset{ \Pi_{X^i}(\tilde{S}(X^1,\ldots,X^k)) \gvert  S_i \neq \emptylist, i=1,\ldots,k}.
  \end{align}
  \end{enumerate} 
\end{lemma}
\begin{proof}
  The proof of the two first items can be easily done by a straightforward induction on $S(X)$ and does not provide any difficulties.
  The proof of the third item can be easily done by a straightforward induction on $\tilde{S}(X^1,\ldots,X^k)$ since it is a generalization of the first item. 
\end{proof}

Notice that Eq.(\ref{depth:position:composition:lemma:eq}) holds as well  if we omit the condition $S_i \neq \emptylist$, i.e.
  \begin{align}
    \label{depth:position:composition:lemma:eq'}
    \tilde{S}(S_1,\ldots,S_k) \equiv_{\mathbf{m}}\tilde{S}(\emptylist,\ldots,\emptylist)
    &&\textrm{ where }&&
    \mathbf{m}=\min \bset{ \Pi_{X^i}(\tilde{S}(X^1,\ldots,X^k)) \gvert i=1,\ldots,k}
  \end{align}
  because
  \begin{align*}
    \min \bset{ \Pi_{X^i}(\tilde{S}(X^1,\ldots,X^k)) \gvert  i=1,\ldots,k} \le \min \bset{ \Pi_{X^i}(\tilde{S}(X^1,\ldots,X^k)) \gvert  S_i \neq \emptylist, i=1,\ldots,k}.
    \end{align*}
  Indeed, Eq.(\ref{depth:position:composition:lemma:eq}) is more refined than Eq.(\ref{depth:position:composition:lemma:eq'}) but we shall sometimes use the latter one.
   
From Item 2 of Lemma ~\ref{depth:position:composition:lemma}  it follows that one has to keep in mind that the notion of $n$-equivalence between \ces can be equivalently restated as follows:
two \ces are $n$-equivalent if they give the same result  when applied to any term $t$ of depth $\delta(t) \le n$ and not just
of depth $\delta(t) = n$ as initially  defined   in Definition~\ref{equivalence:ces:def}.
\begin{corollary}
\label{fixed:point:semantics:corollary}
Let $T(X)$ and  $R$ be  \ces.
For any term $t$ of depth $\delta(t)=n$ and any  positive integer  $n' \ge 1$, if we denote by  $T^{(n)}(R)$ the $n$-times iteration \ce $T(T(\ldots(R)))$, 
then we have that 
\begin{align}
\label{fixed:point:semantics:eq}
\sembrackk{\mu^{n+n'}X.T(X))}(t) = \sembrackk{\mu^{n} X.T(X)}(t) = \sembrackk{T^{(n)}(R)}(t) . 
\end{align}
\end{corollary}

The following Corollary is a crucial one. It guarantees  that to  show that two \ces $\mu X.T(X)$ and  $R$ are $n$-equivalent,
it is enough to show that $R$ is a fixed-point of $T(X)$ in the sense that $T(R)$ and $R$ are $n$-equivalent.
\begin{corollary}
\label{general-fixed-point-corollary}
Let   $T(X)$ and $R$ be \ces. 
For  any $n\ge 1$, we have that 
\begin{align*}
 \tif &&  T(R) \equiv_{n} R   && \tthen &&  \mu X.T(X) \equiv_{n} R.
\end{align*}
\end{corollary}
\begin{proof}
  Let $t$ be a term of depth $n$. If $\sembrackk{T(R)}(t) =  \sembrackk{R}(t)$ then clearly $\sembrackk{T^{(n)}(R)}(t) =  \sembrackk{R}(t)$.
  On the other hand, it  follows  from the second equality of Eq.(\ref{fixed:point:semantics:eq}) of   Corollary~\ref{fixed:point:semantics:corollary}  that  $ \sembrackk{\mu^{n}X.T(X)}(t) = \sembrackk{T^{(n)}(R)}(t)$.
  Hence $\sembrackk{\mu^{n}X.T(X)}(t) = \sembrackk{R}(t)$. But $\sembrackk{\mu^{n}X.T(X)}(t)=\sembrackk{\mu X.T(X)}(t)$  holds by  Definition~\ref{SemanticsOfCEStrategies} of the semantics of \ces.
\end{proof}

We show in the following Lemma~\ref{unfold:equiv:lemma} that a \ce is $\mathbf{m}$-equivalent to its unfolding, where $\mathbf{m}$ is the minimal number of iterations in the unfolding.
To achieve this, we need a technical property  (i.e. Eq.(\ref{unfold:equiv:lemma:eq:2})) that will be used later on in other proofs. 

\newcounter{counter:unfold:equiv:lemma}
\setcounter{counter:unfold:equiv:lemma}{\value{theorem}}

\begin{lemma} 
\label{unfold:equiv:lemma}
Let $S$ be a \ce with (bound) fixed-point variables $X_1,\ldots,X_s$ and  let $\mathbf{s}:\{X_1,\ldots,X_s\} \to  \mathbb{N}$  be an iteration  mapping.
\begin{enumerate}[{(i)}]
\item \label{unfold:equiv:lemma:item:4} If $S$ is  a fixed-point \ce, say $\mu X.S'(X)$ with $X \in \{X_1,\ldots,X_s\}$, then there exists a fixed-point free \ce $\tilde{S}(X^1,\ldots,X^m)$ with $m\ge 1$, and \ces $S_1,\ldots,S_{m-1},S_m(X)$ such that for any $n\ge 1$,
  \begin{align}
    \mu^n X.S'(X)                     &=\tilde{S}\Big(S_1,\ldots,S_{m-1}, S_m\big(\mu^{n-1} X.S'(X)\big)\Big)  \label{unfold:equiv:lemma:eq:1}\\
      \ufold{\mu X.S'(X)}{\mathbf{s}} &=\tilde{S}\Big(\ufold{S_1}{\mathbf{s}},\ldots,\ufold{S_{m-1}}{\mathbf{s}},\bufold{S_m(\mu X.S'(X))}{\mathbf{s}'}\Big) \label{unfold:equiv:lemma:eq:2}
  \end{align}
where $\mathbf{s}'$ is the iteration mapping defined on $\set{X_1,\ldots,X_s}$ by  $\mathbf{s}'(X)=\mathbf{s}(X)-1$ and $\mathbf{s}'(X')=\mathbf{s}(X')$ for  $X'\neq X$.
\item \label{unfold:equiv:lemma:item:3} If   $\mathbf{m}=\min\set{\mathbf{s}(X_1),\ldots,\mathbf{s}(X_s)}$, then $S \equiv_{\mathbf{m}} \ufold{S}{\mathbf{s}}$. 
\end{enumerate}
\end{lemma}  

\begin{remark}
  \label{double:ufold:rq}
  Let $S$ be a \ce with (bound) fixed-point variables $X_1,\ldots,X_s$ and  let $\mathbf{s}_1,\mathbf{s}_2:\{X_1,\ldots,X_s\} \to  \mathbb{N}$  be iteration  mappings where $\mathbf{s}_1 \ge  \mathbf{s}_2$.
  Let $\mathbf{m}_1=\min\set{\mathbf{s}_1(X_i)  \gvert i =1,\ldots,s}$, $\mathbf{m}_2=\min\set{\mathbf{s}_2(X_i)  \gvert i =1,\ldots,s}$ and $\mathbf{m}=\min(\mathbf{m_1},\mathbf{m}_2)$.
  Then it follows from Item (\ref{unfold:equiv:lemma:item:3})  of Lemma~\ref{unfold:equiv:lemma} that $\ufold{S}{\mathbf{s}_1} \equiv_{\mathbf{m}} \ufold{S}{\mathbf{s}_2}$ since
  $S  \equiv_{\mathbf{m}_1} \ufold{S}{\mathbf{s}_1}$ and  $S  \equiv_{\mathbf{m}_2} \ufold{S}{\mathbf{s}_2}$. Besides, it follows from the proof of this  item that 
  $\ufold{S}{\mathbf{s}_1}$  and $\ufold{S}{\mathbf{s}_2}$ can be written as
  \begin{align*}
    \ufold{S}{\mathbf{s}_1}& = T(T_1,\ldots,T_m) \\
    \ufold{S}{\mathbf{s}_2} &= T(\emptylist,\ldots,\emptylist).
  \end{align*}
 \end{remark}


\newcounter{counter:composition:unif:lemma}
\setcounter{counter:composition:unif:lemma}{\value{theorem}} 

\subsection{The composition Lemma}
\label{composition:lemma:section}
In the following  key Lemma~\ref{composition:unif:lemma} we shall formulate  how the unification of two given  \ces behaves with respect to  their sub-\ces.
This Lemma  is very useful and will be heavily used throughout this paper, namely  when it  comes to make a structural induction on the given \ces.
More precisely, we  shall show, under some assumptions,  that  the unification of a \ce $S'(\xi_1,\ldots,\xi_k)$
with  a \ce  $R'(\zeta_1,\ldots,\zeta_l)$  yields a \ce $T(T_1,\ldots,T_m)$ where each $T_i$
is either the unification of some $\xi_j$ with a sub-\ce of $R'(\zeta_1,\ldots,\zeta_l)$, or
the unification of some sub-\ce of $S'(\xi_1,\ldots,\xi_k)$ with a $\zeta_j$. 

counter:composition:unif:lemma

\begin{lemma}[Composition Lemma]
\label{composition:unif:lemma}
Let $S$ and $R$ be \ces.
Assume that  there are  fixed-point free \ces  $S'(X_1,\ldots, X_k)$  and  $R'(Y_1,\ldots, Y_l)$, where $k\ge 1$ and $l\ge 1$, and
 \ces  $\xi_1,\ldots,\xi_k$  where $\xi_i \in \Phi(S)$, and
 \ces  $\zeta_1,\ldots,\zeta_l$ where $\zeta_i \in \Phi(R)$, such that  $S$ and $R$ can be written as:
  \begin{align*}
    S= S'(\xi_1,\ldots,\xi_k) &&&
    R= R'(\zeta_1,\ldots,\zeta_l). 
  \end{align*}

Then, there is a fixed-point free \ce $T(Z_1,\ldots,Z_m)$  and \ces $T_1,\ldots,T_m$, where $m \ge 1$, such that 
  \begin{align*}
    S \combb R &= T(T_1,\ldots,T_m)
  \end{align*}
  where for any $i=1,\ldots,m$, there is an alternative between the two following choices.
  
  \begin{enumerate}[(a)]
  \item \label{composition:unif:lemma:item:2:a} There are $j \in \set{1,\ldots,k}$, a \ce  $R^i(Y^1,\ldots, Y^{l'})$ that is  a sub-\ce of $R'(Y_1,\ldots, Y_l)$ with $l'\le l$, and a set of \ces
    $\set{\zeta^1,\ldots,\zeta^{l'}} \subseteq \set{\zeta_1,\ldots,\zeta_{l}}$   such that 
    \begin{align}
      \label{composition:lemma:prop:1}
    T_i = \xi_j \combb R^i(\zeta^1,\ldots, \zeta^{l'})   &&\tor &&      T_i = \xi_j.
  \end{align}
\item \label{composition:unif:lemma:item:2:b} There are $j \in \set{1,\ldots,l}$, a \ce  $S^i(X^1,\ldots, X^{k'})$ that is  a sub-\ce of $S'(X_1,\ldots, X_k)$ with $k'\le l$, and a set of \ces
    $\set{\xi^1,\ldots,\xi^{k'}} \subseteq \set{\xi_1,\ldots,\xi_k}$   such that 

  \begin{align}
    \label{composition:lemma:prop:2}
    T_i = S^i(\xi^1,\ldots,\xi^{k'}) \combb \zeta_j    &&\tor &&      T_i = \zeta_j.
      \end{align}
  \end{enumerate}
\end{lemma}

%% file: correction-combination-unfoldings_corrected.tex
\section{Unification and unfolding}
\label{correction:unif:general:setting:unfold:sec}

In this section we show two independent results which will a crucial ingredient  for the next section~\ref{equiv:unif:with:unif:unfolding:sec} in which the main theorems  will be proved.
The first result, shown in Subsection~\ref{equiv:unif:unfoldings:sec}, establishes a semantic  equivalence between the unification of  unfoldings of two \ces, say $\ufold{S}{\mathbf{s}} \combb \ufold{R}{\mathbf{r}}$,  and  the unification of other unfoldings of the same \ces, say $\ufold{S}{\mathbf{s}'} \combb \ufold{R}{\mathbf{r}'}$.
The second  result, shown in Subsection~\ref{relating:unif:with:unfolds:sec},  relates  the structure of the unification of two \ces, say $S \combb R$,  with  that of their unfoldings, say $\ufold{S}{\mathbf{s}} \combb \ufold{R}{\mathbf{r}}$, according to a relation of similarity  that will be formalized in Subsection~\ref{quasi:sim:sub:sub:sec}.
In Subsection~\ref{tree:simulation:subsection} we introduce some notions, namely the underlying structure of the set of fixed-point  sub-\ces of a given \ce.
In Subsection ~\ref{quasi:sim:sub:sub:sec} we introduce  two relations  of similarity, a strong and a  weak one, between a \ce and a fixed-point free one.


\subsection{The  equivalence between the unification of several unfoldings of two \ces}
\label{equiv:unif:unfoldings:sec}
The purpose of this section is to relate two kinds of fixed-point free  \ces:
the \ce that results from the unification of an unfolding of two \ces, and  
the \ce that results from the unification of a different  unfolding of the same two \ces. The purpose is to show that these two resulting \ces are equivalent for any term of a certain depth that depends on the  unfoldings.
Given four  iteration mappings $\mathbf{s},\mathbf{s}',\mathbf{r},\mathbf{r}'$ where $\mathbf{s}$ and $\mathbf{s}$ (resp. $\mathbf{r}$ and $\mathbf{r}'$) are defined on the same domain,
we shall devise a  measure  between two the pairs   $(\mathbf{s},\mathbf{r})$ and $(\mathbf{s}',\mathbf{r}')$, called \emph{codistance} and denoted by  $D^{\star}\big((\mathbf{s},\mathbf{r}),(\mathbf{s}',\mathbf{r}')\big)$, and show that the \ces  $\ufold{S}{\mathbf{s}} \combb \ufold{R}{\mathbf{r}}$ and  $\ufold{S}{\mathbf{s}'} \combb \ufold{R}{\mathbf{r}'}$ are equivalent for any term of depth at most
$D^{\star}\big((\mathbf{s},\mathbf{r}),(\mathbf{s}',\mathbf{r}')\big)$.
The definition of this measure will be given in Definition~\ref{distance:iteration:mapping:def}.
We shall compare in  Lemma~\ref{comparing:unif:unfolding:lemma}   the fixed-point free \ce $\ufold{S}{\mathbf{s}} \combb \ufold{R}{\mathbf{r}}$ with the  (fixed-point free) \ce  $\ufold{S}{\mathbf{s}'} \combb \ufold{R}{\mathbf{r}'}$, namely  when $\mathbf{s} \ge \mathbf{s}'$ and  $\mathbf{r} \ge \mathbf{r}'$, 
by showing that  \ce $\ufold{S}{\mathbf{s}} \combb \ufold{R}{\mathbf{r}}$  and  \ce  $\ufold{S}{\mathbf{s}'} \combb \ufold{R}{\mathbf{r}'}$  have the same structure except that the former \ce is deeper than  the latter.
The equivalence is the main result of this section and  will be proved in  Corollary~\ref{comparing:unif:unfolding:corollary}.


To illustrate the idea and justify the name of the codistance between two pairs of iteration mappings, we  first  consider the 
codistance between two iteration mapping with the same domain. Let $\mathbf{s}$, $\mathbf{s}'$ and $\mathbf{s}''$ be iteration mappings and  assume
\begin{align}
\mathbf{s}&=\set{X_1\mapsto 100, X_2\mapsto 100, X_3\mapsto 5 }  \label{iter-mapping:example:1} \\
\mathbf{s}'&=\set{X_1\mapsto 100, X_2\mapsto 60, X_3\mapsto  5 } \label{iter-mapping:example:2}\\
\mathbf{s}''&=\set{X_1\mapsto 100, X_2\mapsto 60, X_3\mapsto  4 }. \label{iter-mapping:example:3}
\end{align}
It is clear that for any \ce $S$ with bound variables $X_1, X_2, X_3$ that  the (fixed-point free) \ces  $\ufold{S}{\mathbf{s}}$ and  $\ufold{S}{\mathbf{s}'}$ are equivalent for any $t$ of depth
at most $60$.  This number  corresponds to the minimal $\mathbf{s}'(X_i)$ such   that $\mathbf{s}'(X_i) \neq \mathbf{s}(X_i)$, for $i=1,\ldots,3$.
For the same reason,   $\ufold{S}{\mathbf{s}}$ and  $\ufold{S}{\mathbf{s}''}$ are equivalent for any term  of depth at most $4$. Obviously, $\ufold{S}{\mathbf{s}}$ is equivalent with itself  for any term, and this will  be taken into account in the definition of codistance~\ref{distance:iteration:mapping:def}  by saying that the codistance between an iteration mapping and itself is infinity. Besides, the more two iterations mappings are far from each others,
the less is their codistance, which justifies the name of codistance.
This idea of codistance between two iteration mappings  can be adapted as well to measure the codistance between two pairs of iteration mappings  as follows.
\begin{definition}[Codistance between pairs of  iteration mappings]
\label{distance:iteration:mapping:def}
  Let  $\mathbf{s},\mathbf{s}': \{X_1,\ldots,X_s\} \to  \mathbb{N}$ and $\mathbf{r},\mathbf{r}':\{Y_1,\ldots,Y_r\} \to  \mathbb{N}$ be iteration mappings such that
$\mathbf{s}\ge \mathbf{s}'$ and $\mathbf{r}\ge \mathbf{r}'$.
We define the \emph{codistance} between $\mathbf{s}$ and $\mathbf{s}'$ by:
\begin{align*}
  d^{\star}(\mathbf{s},\mathbf{s}') &=  \begin{cases} \min\set{\mathbf{s}'(X_i) \gvert \mathbf{s}'(X_i) \neq \mathbf{s}(X_i) \textrm{ for } i=1,\ldots,s}  &\tif \mathbf{s} > \mathbf{s}' \\
    \infty & \tif \mathbf{s}=\mathbf{s}'.
    \end{cases}
\end{align*}
We define the \emph{codistance} between the pairs $(\mathbf{s},\mathbf{r})$ and $(\mathbf{s}',\mathbf{r}')$ by:
\begin{align*}
 D^{\star}\big((\mathbf{s},\mathbf{r}), (\mathbf{s}',\mathbf{r}')\big)&= \min\set{d^{\star}(\mathbf{s},\mathbf{s}'),   d^{\star}(\mathbf{r},\mathbf{r}')}.
\end{align*}
\end{definition}
\begin{example}
  We only give an example of the codistance $d^{\star}$ since the computation of $D^{\star}$ is straightforward.
  If we consider the iteration mappings $\mathbf{s}, \mathbf{s}',\mathbf{s}''$ defined above  by Eqs.(\ref{iter-mapping:example:1}),(\ref{iter-mapping:example:2}),(\ref{iter-mapping:example:3}) respectively, then 
  \begin{align*}
    d^{\star}(\mathbf{s},\mathbf{s})  &=  d^{\star}(\mathbf{s}',\mathbf{s}') =  d^{\star}(\mathbf{s}'',\mathbf{s}'') = \infty, \\
    d^{\star}(\mathbf{s},\mathbf{s}') &= 60, \\
    d^{\star}(\mathbf{s},\mathbf{s}'') &= 4. \\
  \end{align*}
\end{example}

In the following Lemma~\ref{comparing:unif:unfolding:lemma} (which makes use of Lemma~\ref{unfold:equiv:lemma}) and Corollary ~\ref{comparing:unif:unfolding:corollary} we use the following definitions:
let $S$  and $R$ be \ces with bound fixed-point variables $X_1,\ldots,X_s$ and $Y_1,\ldots, Y_r$, respectively.
Let $\mathbf{s}_1,\mathbf{s}_2: \set{X_1,\ldots,X_s} \to  \mathbb{N}$  and  $\mathbf{r}_1,\mathbf{r}_2: \set{Y_1,\ldots,Y_r} \to  \mathbb{N}$ be  iteration  mappings  where $\mathbf{s}_1 \ge \mathbf{s}_2$ and $\mathbf{r}_1 \ge \mathbf{r}_2$. 

\newcounter{counter:comparing:unif:unfolding:lemma}
\setcounter{counter:comparing:unif:unfolding:lemma}{\value{theorem}}

\begin{lemma}
\label{comparing:unif:unfolding:lemma}
There exist fixed-point free \ces $T_1,\ldots,T_m, T(Z_1,\ldots,Z_m)$, where each $Z_i$ is a free fixed-point variable and $m\ge 1$, such that
$\ufold{S}{\mathbf{s}_1}\combb \ufold{R}{\mathbf{r}_1}$
and $ \ufold{S}{\mathbf{s}_2}\combb \ufold{R}{\mathbf{r}_2}$ can be written as  
\begin{align*}
  \ufold{S}{\mathbf{s}_1}\combb \ufold{R}{\mathbf{r}_1} &= T(T_1,\ldots,T_m)              \\
 \ufold{S}{\mathbf{s}_2}\combb \ufold{R}{\mathbf{r}_2} &=  T(\emptylist,\ldots,\emptylist).
\end{align*}
\end{lemma}

The following Corollary~\ref{comparing:unif:unfolding:corollary} follows from Lemma ~\ref{comparing:unif:unfolding:lemma}, it confirms that the
definition of codistance between two pairs of iteration mappings is the right one since
it provides an upper bound for the depth of terms on which the \ces $\ufold{S}{\mathbf{s}}\combb \ufold{R}{\mathbf{r}}$ and $\ufold{S}{\mathbf{s}'}\combb \ufold{R}{\mathbf{r}'}$ are equivalent.
It is easy to construct examples where this bound is reached.

\begin{corollary} 
\label{comparing:unif:unfolding:corollary} 
We have that 
\begin{align}
  \label{comparing:unif:unfolding:corollary:eq}
\ufold{S}{\mathbf{s}_1}\combb \ufold{R}{\mathbf{r}_1} \equiv_{ D^{\star}((\mathbf{s}_1,\mathbf{r}_1), (\mathbf{s}_2,\mathbf{r}_2))} \ufold{S}{\mathbf{s}_2}\combb \ufold{R}{\mathbf{r}_2}.
\end{align}
\end{corollary}
\begin{proof} 
It follows from Lemma~\ref{comparing:unif:unfolding:lemma} that there exist fixed-point free \ces $T_1,\ldots,T_m, T(Z_1,\ldots,Z_m)$, where each $Z_i$ is a fixed-point variable and $m\ge 1$, such that $\ufold{S}{\mathbf{s}_1}\combb \ufold{R}{\mathbf{r}_1}$ and $ \ufold{S}{\mathbf{s}_2}\combb \ufold{R}{\mathbf{r}_2}$ can be written as  
\begin{align*}
  \ufold{S}{\mathbf{s}_1}\combb \ufold{R}{\mathbf{r}_1} &= T(T_1,\ldots,T_m)              \\
 \ufold{S}{\mathbf{s}_2}\combb \ufold{R}{\mathbf{r}_2} &=  T(\emptylist,\ldots,\emptylist).
\end{align*}
From Item (\ref{depth:position:composition:lemma:item:3}) of Lemma~\ref{depth:position:composition:lemma}, if follows that to prove Eq.(\ref{comparing:unif:unfolding:corollary:eq}) it suffices to show that
\begin{align}
\min \bset{\Pi_{Z_i}\big(T(Z_1,\ldots,Z_m)\big) \gvert T_i \neq \emptylist, i=1, \ldots,m}  \ge D^{\star}((\mathbf{s}_1,\mathbf{r}_1),(\mathbf{s}_2,\mathbf{r}_2)).
\end{align}
Assume that $D^{\star}((\mathbf{s}_1,\mathbf{r}_1),(\mathbf{s}_2,\mathbf{r}_2)) = d^{\star}(\mathbf{s}_1,\mathbf{s}_2)$ with $\mathbf{s}_1>\mathbf{s}_2$.
Assume that there exists $v\in \set{1,\ldots,s}$ such that \\
$\min\set{\mathbf{s}_2(X_j) \gvert j=1,\ldots,s } = \mathbf{s}_2(X_v)$. Let $\mu X_v.S_v(X_v)$ be the fixed-point \ce related to $X_v$.
From the monotonicity property it follows that  the shortest  path  in terms of number of jumps  from the root of $T(Z_1,\ldots,Z_m)$ to some $Z_i$, say $Z_w$ with $w \in  \set{1,\ldots,m}$,
$\mu X_v.S_v(X_v)$ unfolded $\mathbf{s}_2(X_v)$ times giving arise to at least  $\mathbf{s}_2(X_v)$ positions. If $\mathbf{s}_1(X_v)>\mathbf{s}_2(X_v)$ then in this first  case we have by the definition of $d^{\star}$ that
$ d^{\star}(\mathbf{s}_1,\mathbf{s}_2)=\mathbf{s}_2(X_v)$ and we are done.  If  $\mathbf{s}_1(X_v)=\mathbf{s}_2(X_v)$ then in this case $T_w=\emptylist$ and we pick another $v' \in \set{1,\ldots,s}\setminus\set{v}$ such that 
$\min\bset{\mathbf{s}_2(X_j) \gvert j \in \set{1,\ldots,s} \setminus \set{v} } = \mathbf{s}_2(X_{v'})$. If such $v'$ does not exist then this means that $T_i=\emptylist$ for any $i\in \set{1,\ldots,m}$ thus
$\ufold{S}{\mathbf{s}_1} \combb  \ufold{R}{\mathbf{r}_1}$ and  $\ufold{S}{\mathbf{s}_2}\combb \ufold{R}{\mathbf{r}_2}$ are equivalent in the strong sense. Otherwise,  we reiterate the same reasoning of the first case with $v'$ instead of $v$.
\end{proof}

\subsection{Fixed-point tree and  fixed-point sequence}
\label{tree:simulation:subsection}
This Subsection is first devoted to the definitions of two notions related to the tree-like structure  underlying the set of all  fixed-point sub-\ces of a given \ce, Definition~\ref{tree:fixed-point:def}.
Roughly speaking, if we look at all fixed-point sub-\ces of a given \ce, they form a tree in the sense  that there is an arrow from a fixed-point \ce $S_1$ to a fixed-point \ce $S_2$ if $S_2$ is a sub-\ce of $S_1$ together
with further conditions.

\begin{definition}[Fixed-point tree and fixed-point sequence of a \ce]
\label{tree:fixed-point:def}
Let $S$ be a strategy in which each fixed-point variable appears once.
\begin{enumerate}[i)]
\item The \emph{fixed-point tree} of $S$, denoted by $\Eu{T}(S)$ or simply $\Eu{T}$,  is the pair  $(\Phi_{\mu}(S),\sqsupset)$, where $\sqsupset$ is a binary relation over  $\Phi_{\mu}(S)$ defined as follows:
  $S_1 \sqsupset S_2$ iff $S_2$ is a   sub-\ce of $S_1$ with $S_1 \neq S_2$, and there is no \ce $S'$ in $\Phi_{\mu}(S)$ such that $S' \neq S_1$, $S' \neq S_2$, $S_2$ is a sub-\ce of $S'$, and $S'$ is a sub-\ce of $S_1$.

\item A  \emph{sequence} $S_1,\ldots,S_m$ in the tree $\Eu{T}(S)$, where $m\ge 1$, is a  set of \ces  where each $S_i$ is  in $\Phi_{\mu}(S)$ such that either $m=1$ or $m \ge 2$ and in this case
  $S_i \sqsupset S_{i+1}$, for $i=1,\ldots,m-1$. Such a sequence will be denoted by   $S_1 \sqsupset \ldots \sqsupset S_m$.

\item A sequence $S_1 \sqsupset \ldots \sqsupset S_m$ in  $\Eu{T}(S)$  is \emph{left-maximal} (resp. \emph{right-maximal})    if there is no  $S' \in \Phi_{\mu}(S)$ such that $S' \neq S_1$ (resp. $S' \neq S_m$)
  and $S' \sqsupset S_1$ (resp. $S_m \sqsupset S'$). A sequence is maximal if it is both left-maximal and right-maximal. In such case $S_1$ is called a root, while $S_m$ is called a leaf.

\item A fixed-point tree  $\Eu{T}$ is \emph{connected} if it has just one root.
\end{enumerate}
\end{definition}
Notice that if a fixed-point tree is not connected then it is composed of many  fixed-point sub-trees each of which is connected.
 
\begin{example}
  Let $M_1(Y), S_1(X_1), M_2(Z,Z')$ and $S_3(X_3)$ be fixed-point free \ces.
  Consider the \ce
  \begin{align*}
    H= \mu Z.\big( M_1\big(\mu X_1.S_1(X_1)\big) \oplus  \mu X_2.M_2\big(Z,X_2,\mu X_3.S_3(X_3)\big)\big).
  \end{align*}
  The set of fixed-point sub-\ces is 
  \begin{align*}
    \Phi_{\mu}(H)=\set{H, \mu X_1.S_1(X_1),  \mu X_2.M_2(Z,X_2, \mu X_3.S_3(X_3)),  \mu X_3.S_3(X_3))}
  \end{align*}
  and the fixed-point tree  $\Eu{T}(H)=(\Phi_{\mu}(H),\sqsupset)$ comes with the two maximal sequences
  \begin{align*}
    H  &\sqsupset   \mu X_1.S_1(X_1)  \\
     H  &\sqsupset   \mu X_2.M_2(Z,X_2, \mu X_3.S_3(X_3))  \sqsupset \mu X_3.S_3(X_3) 
  \end{align*}
  Indeed, $H$ is the root of $\Eu{T}(H)$ while  $\mu X_1.S_1(X_1)$  and $\mu X_3.S_3(X_3)$  are leaves. However if we take $H'=\mu X_4.S_4(X_4)$ then the tree $\Eu{T}(H  \oplus H')$ is no longer connected
  since it has two roots: $H$ and $H'$.
\end{example}

\subsection{The relations of $(\ceSet,\ceSetFree)$-simulation and $(\ceSet,\ceSetFree)$-quasi-simulation}
\label{quasi:sim:sub:sub:sec}
In Section~\ref{structure:proof:main:results:sec} we informally outlined the proof of the main result of this paper.
Namely, we described how to relate  the structure of the \ce that results from  the unification of  two \ces to  the structure of the (fixed-point free)  \ce  that results from the unification of 
their related unfolding, and we illustrated the idea in  Figure~\ref{S:T:structure} for a simple case, and in Figure ~\ref{S:T:structure:complex} for the general case.
Now we formalize this idea that   relates   a \ce in $\mycal{C}$ to  a fixed-point free one in $\mycal{C}_0$ in terms of $(\ceSet,\ceSetFree)$-simulation defined   next.

\begin{definition}[$(\ceSet,\ceSetFree)$-simulation]
  \label{C-C0-morphism:def}
  For any  \ce $S$ in  $\ceSet$ and any fixed-point free \ce $R$ in  $\ceSetFree$,  a \emph{$(\ceSet,\ceSetFree)$-simulation} is a binary relation $\morphy$ between the sets of  augmented sub-\ces of $S$ and of sub-\ces of $R$, i.e.
  \begin{align*}
    \morphy \subseteq  \widetilde{\Phi}(S) \times \Phi(R),
  \end{align*}
   inductively defined from   $S \morphy R$ that  fulfills    the inference rules of Table~\ref{Inference Rules}.
  \begin{table}[H]
    \centering  
    \fbox{
      \parbox{8cm}{
        \begin{align*}
&\infer[S',R' \in \ceSetFree]{S' \morphy R'}{%
       S'=R'} 
\;\;\;\;&&&
&\infer[S'(X_1,\ldots,X_m) \in \ceSetFree]{S'(S_1,\ldots,S_m) \morphy S'(R_1,\ldots,R_m)}{%
       S_i \morphy R_i} \\
&&&&& \\
&\infer{\mu X.S'(X) \morphy \emptylist}{%
      } 
&&&
&\infer{\mu X.S'(X) \morphy R'}{%
       S(\mu X.S'(X)) \morphy R'}
        \end{align*}
        }}
   \caption{Inference rules for $(\ceSet,\ceSetFree)$-simulations.}
        \label{Inference Rules}
    \end{table} 
\end{definition}

 
Notice  that if there is a $(\ceSet,\ceSetFree)$-simulation between two \ces, then it is unique, i.e. if $S {\morphy}_1 R $ and $S {\morphy}_2 R $, for two $(\ceSet,\ceSetFree)$-simulations $ {\morphy}_1$ and ${\morphy}_2$, then $ {\morphy}_1= {\morphy}_2$.

\begin{example}
Let $S(X)$ be a fixed-point free \ce. Then, for any $n\ge 0$, there is a $(\ceSet,\ceSetFree)$-simulation between $\mu X.S(X)$  and $\mu^n X.S(X)$,
since $\mu^n X.S(X)$ is nothing but the $n$-times iteration $S(S(\ldots S(\emptylist)))$.
\end{example}

The following claims are not hard to prove. 
\begin{remark}\label{ufold:mophism:rq}
For any \ce $S$ with bound fixed-variables $X_1,\ldots, X_s$ with $s\ge 0$, any iteration mapping $\mathbf{s}: \set{X_1,\ldots, X_s}\to \mathbb{N}$,  and any \ce $M(Y)$, the following holds.
\begin{enumerate}
\item \label{ufold:mophism:rq:item:1}  There is  a $(\ceSet,\ceSetFree)$-simulation between  $S$ and  $\ufold{S}{\mathbf{s}}$.
\item \label{ufold:mophism:rq:item:1'}  If there is a $(\ceSet,\ceSetFree)$-simulation $\morphy$ between $S$ and  $S'$, then there is a $(\ceSet,\ceSetFree)$-simulation $\morphy'$ between $M(S)$ and $M(S')$  for any fixed-point free \ce $M(Y)$.
   That is, the following diagram commutes.
  \[\begin{tikzcd}
   \mycal{C} \arrow{r}{M(\cdot)} \arrow[swap,dash]{d}{\morphy} &  \ceSet\arrow[dash]{d}{\morphy'} \\
   \ceSetFree  \arrow{r}{M(\cdot)} & \ceSetFree
  \end{tikzcd}
  \] 
\item \label{ufold:mophism:rq:item:2} If there is a $(\ceSet,\ceSetFree)$-simulation  $\morphy$ between  $S$ and  $S'$ and if $\tilde{S}$ results from $S$ by Simplifications  (\ref{simplification:algorithms}), denoted hereby $\Eu{O}$, that transform  a \ce  into an equivalent \ce in which each fixed-point variable occurs once,  then there is a $(\ceSet,\ceSetFree)$-simulation  $\tilde{\morphy}$ between $\tilde{S}$ and $S'$ as well. That is, the following diagram commutes.
  \[\begin{tikzcd}
   \mycal{C} \arrow{r}{\Eu{O}} \arrow[dash,swap]{dr}{\morphy } &  \mycal{C} \arrow[dash]{d}{\tilde{\morphy}} \\
    & \ceSetFree
  \end{tikzcd}
  \]
\end{enumerate}
\end{remark}

We next define a weaker relation  of $(\ceSet,\ceSetFree)$-simulation by relaxing the constraint imposed by the fixed-point rule that unravels $\mu X.S(X)$ into $S(\mu X.S(X))$.
The motivation is that in the upcoming  proofs, rather than proving the existence of a $(\ceSet,\ceSetFree)$-simulation,  it is much  easier and less cumbersome  to proceed in two steps by firstly constructing the weaker relation and then strengthening it by  deducing its properties.
\begin{definition}[$(\ceSet,\ceSetFree)$-quasi-simulation]
  \label{quasi:simulation:def}
For any  \ce $S$ in  $\ceSet$ and any fixed-point free \ce $R$ in  $\ceSetFree$,  a \emph{$(\ceSet,\ceSetFree)$-quasi-simulation} is a  binary relation  $\morphyy$ between sub-\ces of $S$ and of $R$, i.e.
  \begin{align*}
    \morphyy  \subseteq  \Phi(S) \times \Phi(R),
  \end{align*}
  inductively defined from  $S \morphyy R$  that   fulfills    the inference rules of Table~\ref{Inference Rules:nu} which are the same as the inference rules of  Table~\ref{Inference Rules} apart for the fixed-point rule which
  is replaced by new two rules.
   \begin{table}[H]
    \centering  
    \fbox{
      \parbox{8cm}{
        \begin{align*}
&\infer[S',R' \in \ceSetFree]{S' \morphyy R'}{%
       S'=R'} 
\;\;\;\;&&&
&\infer[S'(X_1,\ldots,X_m) \in \ceSetFree]{S'(S_1,\ldots,S_m) \morphyy S'(R_1,\ldots,R_m)}{%
       S_i \morphyy R_i} \\
&&&&& \\
&\infer{\mu X.S'(X) \morphyy \emptylist}{%
      } 
&&&
&\infer{\mu X.S'(X) \morphyy R'}{%
  S'(X) \morphyy R'}
\;\;\;\;\;\;\;\;\;\;\;\; \infer[X \in \fixset]{X \morphyy R'}{%
       }
        \end{align*}
      }
    }
   \caption{Inference rules for  $(\ceSet,\ceSetFree)$-quasi-simulation}
        \label{Inference Rules:nu}
    \end{table} 
\end{definition}
Notice that $(\ceSet,\ceSetFree)$-quasi-simulation is strictly  weaker than the $(\ceSet,\ceSetFree)$-simulation in the sense that  if there is  a $(\ceSet,\ceSetFree)$-simulation between $S$ and $R$  then there is a $(\ceSet,\ceSetFree)$-quasi-simulation between $S$ and $R$ as well, while the opposite does not hold in general. This is due to the fact that the $(\ceSet,\ceSetFree)$-simulation imposes that $\mu X.S(X)$ and $X$ must correspond to the "same" \ce in $R$, while the $(\ceSet,\ceSetFree)$-quasi-simulation does not impose that. For instance, there is a $(\ceSet,\ceSetFree)$-quasi-simulation between $\mu X.S(X)$ and $S(S')$ whatever maybe $S'$ since $X \morphyy S'$, while it is not the case  that there is in general a $(\ceSet,\ceSetFree)$-simulation between $\mu X.S(X)$ and $S(S')$ unless further constraints are imposed on $S'$. Hence Remark~\ref{ufold:mophism:rq} holds for  $(\ceSet,\ceSetFree)$-quasi-simulations as well.

A $(\ceSet,\ceSetFree)$-quasi-simulation $\morphyy$ between $S$ and $R$ induces two mappings, the first one that  maps each fixed-point \ce of $S$ to a sub-\ce of $R$,
and the second one that maps each bound variable of $S$ to a sub-\ce of $R$ as well.
\begin{definition}
  \label{induced:morphisms:def}
 For any  $(\ceSet,\ceSetFree)$-quasi-simulation relation $\morphyy$ between $S$ and $R$, with $S\in \ceSet$ and $R \in \ceSetFree$,  define  the mappings  
  \begin{align*}
    \phi_{\mu}: \Phi_{\mu}(S) \cup  \boundv{S}  \rightarrow \Phi(R)  &&\tand &&      \phi_{\nu}^\mu: \boundv{S} \rightarrow \Phi(R)
  \end{align*}
  as:
  \begin{itemize}
    \item $\phi_{\mu}(S')$,  for any $S'$, being  the unique $R'\in \Phi(R)$ such that $S' \morphyy R'$, and
    \item $\phi_{\nu}^\mu(X)$, for any $X$,  being  $\phi_{\mu}(\mu X.T(X))$ where $\mu X.T(X)$ is the (unique) fixed-point \ce related to $X$.
  \end{itemize}
  Besides, the mapping   $\phi_{\nu}^\mu$  extends uniquely to an endomorphism  $\widehat{\phi}_{\nu}^\mu: \ceSet \rightarrow \ceSet$ defined for any \ce by:
  \begin{align*}
    \widehat{\phi}_{\nu}^\mu\big(T(T_1,\ldots,T_m)\big)=T\big(\widehat{\phi}_{\nu}^\mu(T_1),\ldots,\widehat{\phi}_{\nu}^\mu(T_m)\big) &&\tand && \widehat{\phi}_{\nu}^{\mu}(Z)=\phi_{\nu}^\mu(Z).
  \end{align*}
\end{definition}
Notice that it follows from  Definition~\ref{C-C0-morphism:def} that any $(\ceSet,\ceSetFree)$-simulation from $S$ to $R$ is not defined  for the bound fixed-point variables of $S$ since any fixed-point sub-\ce $\mu X.S'(X)$ of $S$
is unravelled  to $S'(\mu X.S'(X))$ and hence $X$ is never reached.

\begin{example}[The mappings $\phi_{\mu}$ and  $\phi_{\nu}^\mu$]
Let $S(X)$ and $M(Y)$ be fixed-point free \ces.  Then there is a $(\ceSet,\ceSetFree)$-simulation $\morphy$ and a $(\ceSet,\ceSetFree)$-quasi-simulation $\morphyy$ between $M(\mu X.S(X))$ and $M(\mu^3 X.S(X))$ and  we have that:
\begin{align*}
   \morphy &= \big\{\big(\mu X.S(X), \mu^3 X.S(X)\big), \big(\mu X.S(X), \mu^2 X.S(X)\big),\big(\mu X.S(X), \mu^1 X.S(X)\big),\big(\mu X.S(X), \emptylist\big)\big\} \\
   \morphyy & = \big\{\big(\mu X.S(X), \mu^3 X.S(X)\big),   \big(X, \mu^2 X.S(X)\big) \big\}
\end{align*}
And therefore, 
\begin{align*}
  \phi_{\mu}(\mu X.S(X)) &= \mu^3 X.S(X) = S(S(S(\emptylist))) \\
  \phi_{\mu}(X) &= \mu^2 X.S(X) = S(S(\emptylist)) \\
    \phi_{\nu}^{\mu}(X) &=   \phi_{\mu}(\mu X.S(X)) =\mu^3 X.S(X)
\end{align*}
\end{example}

\begin{remark}
  \label{mu:morphi:prop}
  Notice that
  if $T(Z_1,\ldots,Z_m)$ is a fixed-point free \ce in which $Z_1,\ldots,Z_m$ are free fixed-point variables,
  and if $T_1\ldots T_m$ are \ces where each of $T_i$ is either fixed-point \ce or a fixed-point variable,
  then it follows from Definition~\ref{induced:morphisms:def}
  together with Definition~\ref{quasi:simulation:def}  of $(\ceSet,\ceSetFree)$-quasi-simulations  that
  \begin{align*}
    \phi_{\mu}\big(\mu Z.T(T_1,\ldots,T_m)\big) = T\big(\phi_{\mu}(T_1),\ldots,\phi_{\mu}(T_m) \big).
  \end{align*}
\end{remark}


\subsection{Relating the structure of the unification of two \ces with  that of their unfolding}
\label{relating:unif:with:unfolds:sec}
The purpose of this section is to relate the structure of the unification of two \ces with that of their unfolding
as  illustrated  in  Section~\ref{structure:proof:main:results:sec} in Figure~\ref{S:T:structure} for a simple case, and in Figure ~\ref{S:T:structure:complex} for the general case.
More precisely, we show in the  following  Lemma~\ref{main:lemma:mophism:quasi} that the unification commutes with the unfolding  in the following sense:
there is a $(\ceSet,\ceSetFree)$-quasi-simulation  between the \ce that results from the unification of two \ces  
and the fixed-point free \ce that results from the unification  of   their related unfolding.

We illustrate with  a simple unification example how this $(\ceSet,\ceSetFree)$-quasi-simulation  is constructed, and thus how the two structures in  Figure~\ref{S:T:structure} are obtained.
Let $M(Y),S(X)$ and $R$ be fixed-point \ces, and  
consider, on one hand, the unification of $M(\mu X.S(X))$   with   $R$, and on the other hand, the unification of the unfolding of $M(\mu X.S(X))$ with the unfolding of $R$, which is $R$, since $R$ is fixed-point free. 
We explain  how  these two unifications  are related. During the unification process that starts from  $\tuple{M(\mu X.S(X)),{R},\emptyset}$ on one side, and
  from $\tuple{M(\mu^n X.S(X)),R,\emptyset}$, where  $n\ge 1$, on the other, we distinguish many cases.
  
  \begin{enumerate}[(I)]
  \item \label{step:1:structure} As far as we have   $\tuple{M'(\mu^n X.S(X)),{R}',\emptyset}$ on one side, and  $\tuple{M'(\mu X.S(X)), R', \emptyset}$ on the other, where $M'$ (resp. $R'$) is a sub-\ce of $M$ (resp. $R$),
    the constructed \ce is the same on both sides, it is $T_0$ in Figure~\ref{S:T:structure}.

  \item \label{step:2:structure} If the derivation reaches  a fixed-point \ce, that is, it reaches $\langle\mu X.S(X),R'',\emptyset\rangle$  on the left side, and   $\langle\mu^n X.S(X),\allowbreak  R'',\emptyset\rangle$
    on the right one, where $R''$ is a sub-\ce of  $R$,  then the left derivation  produces  $\allowbreak \mu Z_1.\langle S(\mu X.S(X)),R'',\cdot\rangle$ and continues from $\allowbreak \tuple{S(\mu X.S(X)),R'',\cdot}$, while the right one continues from   $\allowbreak  \tuple{S(\mu^{n-1} X.S(X)),R'',\emptyset}$. 
    This goes back to  case (\ref{step:1:structure}), in which we take $M'(\cdot)=S(\cdot)$, and in which the left derivation will produce $\mu Z_1.T_1(\ldots)$, while the right one will produce  $T_1(\ldots)$ .
    During the generation of $T_1(\ldots)$  two cases can happen:
    \begin{enumerate}
    \item  If the left derivation reaches  $\tuple{\mu X.S(X),R'',\cdot}$, then the right derivation reaches   $\tuple{\mu^{n-1} X.S(X),R'',\emptyset}$.
      The left derivation continues and produces the fixed-point $Z_1$ generated at the end of  case (\ref{step:2:structure}), while the right derivation  produces  the \ce $T_1^1$   depicted on the right of Figure~\ref{S:T:structure}.
    \item If the  left derivation reaches  $\tuple{\mu X.S(X),R_2,\cdot}$  with $R_2 \neq R''$, then the right derivation  reaches  $\tuple{\mu^{n-1} X.S(X),R_2,\emptyset}$.
      Thus the left derivation   produces the fixed-point   \ce \\  $\mu Z_2.\tuple{S(\mu X.S(X)),R_2,\cdot}$ and continues from $\tuple{S(\mu X.S(X)),R_2,\cdot}$   (see left of Figure~\ref{S:T:structure}), while the right derivation
      continues from $\tuple{\mu^{n-1} X.S(X),R_2,\emptyset}$, see right of Figure~\ref{S:T:structure}.
      This goes back to  case (\ref{step:1:structure}).
    \end{enumerate}
  \end{enumerate}

In the following Lemma~\ref{main:lemma:mophism:quasi} we construct  a  $(\ceSet,\ceSetFree)$-quasi-simulation between the unification
of any two \ces and the unification of their any unfolding. 

\newcounter{counter:main:lemma:mophism:quasi}
\setcounter{counter:main:lemma:mophism:quasi}{\value{theorem}}

\begin{lemma}
  \label{main:lemma:mophism:quasi}
  Let $S$ and $R$ be  \ces with  bound  fixed-point variables $\boundv{S}=\set{ X_1,\ldots,X_s}$ and $\boundv{R}=\set{Y_1,\ldots,Y_r}$.
  Let $\Eu{M} \in \mathfrak{M}(S,R)$ be a memory  with respect  to $S$ and $R$.
  Let  $\mathbf{s}:\{X_1,\ldots,X_s\} \to  \mathbb{N}$ and $\mathbf{r}:\{X_1,\ldots,X_r\} \to  \mathbb{N}$ be iteration mappings.
There is a $(\ceSet,\ceSetFree)$-quasi-simulation $\morphyy$  between   $\NF(\tuple{S,R,\Eu{M}})$ and  $\NF(\tuple{\ufold{S}{\mathbf{s}},\ufold{R}{\mathbf{r}},\emptyset})$.
In particular, the following diagram commutes.
\[\begin{tikzcd}
\mycal{C} \times \mycal{C} \arrow{r}{\combb} \arrow[swap]{d}{\ufold{\cdot}{\mathbf{s}}  \times \ufold{\cdot}{\mathbf{r}}} &  \mycal{C} \arrow[dash]{d}{\morphyy} \\
\ceSetFree \times \ceSetFree  \arrow{r}{\combb} & \ceSetFree
\end{tikzcd}
\]
\end{lemma}

It turned out that the  $(\ceSet,\ceSetFree)$-quasi-simulation  of~\ref{main:lemma:mophism:quasi} is actually a $(\ceSet,\ceSetFree)$-simulation,  but we can not prove it now  in this section, because it requires
the further developments of the next section~\ref{equiv:unif:with:unif:unfolding:sec}, and there we shall be ready to prove it in  Corollary~\ref{quasi:sim:is:sim:cor}.  


Now we can state and prove in Lemma~\ref{properties:quasi:simulation:lemma} useful properties of the $(\ceSet,\ceSetFree)$-quasi-simulation constructed in the proof of the previous  Lemma~\ref{main:lemma:mophism:quasi}.
Roughly speaking, we need to distinguish in the resulting \ce $S\combb R$ between two kinds of  fixed-point \ces:
\begin{enumerate}[(i.)]
  \item a fixed-point \ce $\mu Z.T(Z)$, where $Z$ is fresh, that is generated by the fixed-point rules (\ref{fixed:ext:1}) and (\ref{fixed:ext:2}) of the unification reduction system given  in Definition~\ref{reduction:unif:def}, and
  \item a fixed-point \ce $\mu X'.S'(X')$ that is a sub-\ce of $S$ or $R$.
\end{enumerate}
In the first case, the fixed-point \ce $\mu Z.T(Z)$ is related by the  $(\ceSet,\ceSetFree)$-quasi-simulation to the unification of an iteration (over a fixed-point \ce) with  a \ce, or symmetrically,
to the unification of a \ce with  an iteration (over a fixed-point \ce).
However, in the second  case, the fixed-point \ce $\mu X.S(X)$ is related by the  $(\ceSet,\ceSetFree)$-quasi-simulation to its unfolding.
Formally,

\begin{lemma}
\label{properties:quasi:simulation:lemma}
Let $S$ and $R$ be  \ces with  bound  fixed-point variables $\boundv{S}=\set{ X_1,\ldots,X_s}$ and $\boundv{R}=\set{Y_1,\ldots,Y_r}$.
Let  $\mathbf{s}:\{X_1,\ldots,X_s\} \to  \mathbb{N}$ and $\mathbf{r}:\{Y_1,\ldots,Y_r\} \to  \mathbb{N}$ be iteration mappings.
The $(\ceSet,\ceSetFree)$-quasi-simulation $\morphyy$ between $S\combb R$ and   $\ufold{S}{\mathbf{s}} \combb \ufold{R}{\mathbf{r}}$ constructed in the proof of Lemma~\ref{main:lemma:mophism:quasi} has the following property.

For any   sub-\ce $\mathbf{T}$  in   $S \combb R$ that is either a fixed-point or a bound variable, i.e.
\begin{align*}
  \mathbf{T} \in \Phi_{\mu}(S \combb R) \cup \boundv{S \combb R},
  \end{align*}
there exist  \ces $\mu X'.S'(X')$ and $R'$, iteration mappings $\mathbf{s}':\{X_1,\ldots,X_s\} \to  \mathbb{N}$ and $\mathbf{r}':\{Y_1,\ldots,Y_r\} \to  \mathbb{N}$ and a memory $\Eu{M}'$ such that
one of the four following case holds:
\begin{enumerate}[(a)]
\item\label{quasi:prof:item:1}
  \begin{align*}
  \mathbf{T} =\NF\big(\tuple{\mu X'.S'(X'),R',\Eu{M}'}\big)  && \tand &&
  \mathbf{T} \morphyy \big(\ufold{\mu X'.S'(X')}{\mathbf{s}'} \combb \ufold{R'}{\mathbf{r}'}\big) 
\end{align*}
and in this case  $\mu X'.S'(X') \in \widetilde\Phi_{\mu}(S)$  and  $R' \in  \widetilde\Phi(R)$.
\item \label{quasi:prof:item:2}
  \begin{align*}
   \mathbf{T}=\NF\big(\tuple{R',\mu X'.S'(X'),\Eu{M}'}\big) & &  and &&
   \mathbf{T} \morphyy  \big(\ufold{R'}{\mathbf{r}'} \combb \ufold{\mu X'.S'(X')}{\mathbf{s}'}\big) 
\end{align*}
  and in this case  $R' \in   \widetilde\Phi_{\mu}(S)$  and  $\mu X'.S'(X') \in  \widetilde\Phi(R)$.

\item \label{quasi:prof:item:3}
  \begin{align*}
   \mathbf{T}=\mu X'.S'(X')   & &  and &&
   \mathbf{T} \morphyy  \ufold{\mu X'.S'(X')}{\mathbf{s}'}
\end{align*}
  with   $\mu X'.S'(X') \in  \widetilde\Phi_{\mu}(S)$ and $X' \in \set{X_1,\ldots,X_s}$.

\item \label{quasi:prof:item:4}
  \begin{align*}
   \mathbf{T}=\mu X'.S'(X')   & &  and &&
   \mathbf{T} \morphyy  \ufold{\mu X'.S'(X')}{\mathbf{r}'}
\end{align*}
  with   $\mu X'.S'(X') \in  \widetilde\Phi_{\mu}(R)$ and $X' \in \set{Y_1,\ldots,Y_r}$.  
\end{enumerate}
\end{lemma}
\begin{proof} 
 Item (\ref{quasi:prof:item:1}) follows immediately from the case (\ref{case:3:lemma:unif:morphism:appendix})    of the proof of Lemma ~\ref{main:lemma:mophism:quasi} since 
 any  fixed-point \ce $\mu Z.T(Z)$  and any fixed-point  variable $Z$ in the resulting \ce $S \combb R$ results from the unification of a fixed-point \ce with an arbitrary \ce such that 
  $\mu Z.T(Z)$ and  $Z$ is related by the  $(\ceSet,\ceSetFree)$-quasi-simulation $\morphyy$  to a unification of  two \ces where the \emph{left} one is an iteration over a fixed-point \ce (i.e. $\rho(\mu X.S'(X),\mathbf{s})$).

Item (\ref{quasi:prof:item:2}) follows from the symmetric case of case (\ref{case:3:lemma:unif:morphism:appendix})   of the proof of Lemma ~\ref{main:lemma:mophism:quasi} which we omitted and in which
any  fixed-point \ce  $\mu Z.T(Z)$  and any fixed-point  variable $Z$ in the resulting \ce $S \combb R$ results from the unification of an arbitrary \ce with  a fixed-point \ce such that
  $\mu Z.T(Z)$ and  $Z$ would be related by the  $(\ceSet,\ceSetFree)$-quasi-simulation $\morphyy$ to a unification of  two \ces where the \emph{right}  one is an iteration over a fixed-point \ce.

Items (\ref{quasi:prof:item:3}) and (\ref{quasi:prof:item:4}) follow from the case (\ref{case:3:lemma:unif:morphism:appendix})  of the proof of Lemma  (\ref{main:lemma:mophism:quasi}) together with the explicit computations made in 
the composition Lemma ~\ref{composition:unif:lemma},  with properties (\ref{composition:lemma:prop:1}) and (\ref{composition:lemma:prop:2}),
in which we   take one of the  $\xi_1,\ldots,\xi_k$ or one of the  $\zeta_1,\ldots,\zeta_l$ as fixed-point \ce, and by taking one of the $T_1,\ldots, T_m$ as a fixed-point \ce that is either $\xi_i$ or $\zeta_j$, for some $i \in \set{1,\ldots,k}$ and some $j\in \set{1,\ldots,l}$.
\end{proof} 

In the following example we illustrate the cases (\ref{quasi:prof:item:1}) and (\ref{quasi:prof:item:3}) of Lemma~\ref{properties:quasi:simulation:lemma},
we omit the cases  (\ref{quasi:prof:item:2}) and (\ref{quasi:prof:item:4}) since they are respectively symmetrical to the two former ones. 
\begin{example}
\label{properties:quasi:simulation:example}
Let $\xi_1,\xi_2,\xi_3,R_1,R_2,R_3$ be \ces such that $R_1,R_2,R_3$ are fixed-point free, and $\xi_1,\xi_2,\xi_3$ are fixed-point \ces of the form
\begin{align*}
\xi_1= \mu X_1.S_1(X_1), && \xi_2= \mu X_2.S_2(X_2), && \xi_3= \mu X_3.S_3(X_3), 
\end{align*}
where $\xi_3$ is a sub-\ce of $\xi_2$ which is a sub-\ce of $\xi_1$.
Consider the following iteration mappings in which  $n\ge 1$:
\begin{align*}
  \mathbf{s}_1=\set{X_1 \mapsto n, X_2\mapsto n, X_3\mapsto n }, &&  \mathbf{s}_2=\set{X_1 \mapsto n-1, X_2\mapsto n, X_3\mapsto n}, && \mathbf{s}_3=\set{X_1 \mapsto n-1, X_2\mapsto n-1, X_3\mapsto n }.
\end{align*}
Since $R_1,R_2,R_3$ are fixed-point free, then they are equal to their unfoldings.  We next consider the two unifications $\xi_1 \combb R_1$ and $\ufold{\xi_1}{\mathbf{s}_1} \combb R_1$ that  result respectively  from the
following two derivations, in which we omit the  explicit expression of the \ces $T_1$ and of $T_2$, where $\Eu{M}_2$ and $\Eu{M}_3$ are memories: 
\begin{alignat}{2}
  \tuple{\xi_1,R_1,\emptyset}                        &\xreduces{\star}  \mu Z_1. T_1(\xi_1,\tuple{\xi_2,R_2,\Eu{M}_2}) &&  \xreduces{\star}   \mu Z_1. T_1\big(\xi_1, \mu Z_2. T_2(\tuple{\xi_3,R_3,\Eu{M}_3}\big), \label{deriv:1}\\
  \tuple{\ufold{\xi_1}{\mathbf{s}_1},R_1,\emptyset}  &\xreduces{\star}   T_1\big(\ufold{\xi_1}{\mathbf{s}_2},\tuple{\ufold{\xi_2}{\mathbf{s}_2},R_2,\emptyset}\big) && \xreduces{\star}   T_1\big(\ufold{\xi_1}{\mathbf{s}_2},
      T_2(\tuple{\ufold{\xi_3}{\mathbf{s}_3},R_3,\emptyset}\big). \label{deriv:2}
\end{alignat}
If we assume, for the sake of simplicity, that the normal form of $\tuple{\xi_3,R_3,\Eu{M}_3}$  in the derivation (\ref{deriv:1}) produces just  one fixed-point \ce, then the set of fixed-point sub-\ces of $\xi_1 \combb R_1$ is:
\begin{align*}
  \Phi_{\mu}(\xi_1 \combb R_1) = \set{\xi_1, \underbrace{\NF(\tuple{\xi_1,R_1,\emptyset})}_{\mathbf{T}_1},
                                  \underbrace{\NF(\tuple{\xi_2,R_2,\Eu{M}_2})}_{\mathbf{T}_2}, \underbrace{\NF(\tuple{\xi_3,R_3,\Eu{M}_3})}_{\mathbf{T}_3}}.
\end{align*}
Recall that, from  Definition~\ref{unification:def} of unification, we have
\begin{align*}
 \NF(\tuple{\ufold{\xi_2}{\mathbf{s}_2},R_2,\emptyset}) = \ufold{\xi_2}{\mathbf{s}_2} \combb R_2 &&\tand&&   \NF(\tuple{\ufold{\xi_3}{\mathbf{s}_3},R_3,\emptyset}) = \ufold{\xi_3}{\mathbf{s}_3} \combb R_3.
\end{align*}
Therefore, the cases (\ref{quasi:prof:item:1}) and (\ref{quasi:prof:item:3}) of Lemma~\ref{properties:quasi:simulation:lemma} are as follows:
\begin{align*}
  \mathbf{T}_1 &\morphyy   (\ufold{\xi_1}{\mathbf{s}_1} \combb  R_1) &&\tand&&   \mathbf{T}_2 \morphyy   (\ufold{\xi_2}{\mathbf{s}_2} \combb  R_2) && \tand&&   \mathbf{T}_3 \morphyy   (\ufold{\xi_3}{\mathbf{s}_3} \combb  R_3) 
  \tag{Case (\ref{quasi:prof:item:1}) of Lemma~\ref{properties:quasi:simulation:lemma}} \\
  \xi_1 &\morphyy \ufold{\xi_1}{\mathbf{s}_2}.    \tag{Case (\ref{quasi:prof:item:3}) of Lemma~\ref{properties:quasi:simulation:lemma}}
\end{align*}
\end{example}

We recall that a \ce that  results from a unification $S \combb R$ may contain useless fixed-point constructor of the form $\mu Z.T$ where $Z$ does not appear in $T$,
or it may contain a fixed-point variable that appears many times. We noticed in Item~\ref{ufold:mophism:rq:item:2} of  Remark~\ref{ufold:mophism:rq} that Simplifications (\ref{simplification:algorithms})
preserve the relation of $(\ceSet,\ceSetFree)$-simulation and therefore the  $(\ceSet,\ceSetFree)$-quasi-simulation.
Hence  we can assume from now on that the \ces that result  from the unification follow Assumptions~\ref{global:assumptions:ces}.   
It is  simpler for later development, to lift   the properties of the relation $(\ceSet,\ceSetFree)$-quasi-simulation of
Lemma~\ref{properties:quasi:simulation:lemma} to  its induced mapping $\phi_{\mu}$.
This will be done in Lemma~\ref{properties:morphism:lemma} together with a simple and useful property on the image by $\phi_{\mu}$ of fixed-points \ces in $S \combb R$.
Roughly speaking, this property states that if $\mathbf{T}$ and $\mathbf{T}'$ are fixed-point \ces in $S\combb R$ where $\mathbf{T}'$ is an immediate sub-\ce of $\mathbf{T}$, then
the number of iterations over a certain fixed-point \ce decreases by one from  $\phi_{\mu}(\mathbf{T})$ to  $\phi_{\mu}(\mathbf{T}')$.


\newcounter{counter:properties:morphism:lemma}
\setcounter{counter:properties:morphism:lemma}{\value{theorem}}

\begin{lemma}
\label{properties:morphism:lemma}
Let $S$ and $R$ be  \ces with  bound  fixed-point variables $\boundv{S}=\set{ X_1,\ldots,X_s}$ and $\boundv{R}=\set{Y_1,\ldots,Y_r}$.
Let  $\mathbf{s}:\{X_1,\ldots,X_s\} \to  \mathbb{N}$ and $\mathbf{r}:\{X_1,\ldots,X_r\} \to  \mathbb{N}$ be iteration mappings.
Let $\phi_\mu$  be the mapping induced by the  $(\ceSet,\ceSetFree)$-quasi-simulation $\morphyy$ between $S \combb R$ and   $\ufold{S}{\mathbf{s}} \combb \ufold{R}{\mathbf{r}}$
constructed in the proof of Lemma~\ref{main:lemma:mophism:quasi}.
The mapping  $\phi_\mu$ enjoys the following properties.
\begin{enumerate}
\item \label{properties:morphism:lemma:item:1} 
  For any  fixed-point \ce $\mathbf{T}$  in   $S \combb R$, there exist  \ces $\mu X'.S'(X')$ and $R'$, mappings $\mathbf{s}':\{X_1,\ldots,X_s\} \to  \mathbb{N}$ and $\mathbf{r}':\{Y_1,\ldots,Y_r\} \to  \mathbb{N}$, and a memory $\Eu{M}'$ such that one of the four following cases holds.
  \begin{enumerate}
  \item \label{properties:morphism:lemma:item:1:1}   $\mathbf{T}=\NF\big(\tuple{\mu X'.S'(X),R',\Eu{M}'}\big)$ and 
                                                    $\phi_{\mu}(\mathbf{T})= \big(\ufold{\mu X'.S'(X')}{\mathbf{s}'} \combb \ufold{R'}{\mathbf{r}'}\big)$. 

  \item \label{properties:morphism:lemma:item:1:2}   $\mathbf{T}=\NF\big(\tuple{R',\mu X'.S'(X'),\Eu{M}'}\big)$ and
    $\phi_{\mu}(\mathbf{T})=\big(\ufold{R'}{\mathbf{r}'} \combb \ufold{\mu X'.S'(X')}{\mathbf{s}'}\big)$. 

   \item \label{properties:morphism:lemma:item:1:3}   $\mathbf{T}=\mu X'.S'(X')$, with $X' \in \set{X_1,\ldots,X_s}$  and $\mu X'.S'(X') \in \Phi_{\mu}(S)$, and 
        $\phi_{\mu}(\mathbf{T})= \ufold{\mu X'.S'(X')}{\mathbf{s}'}$.   

   \item \label{properties:morphism:lemma:item:1:4}   $\mathbf{T}=\mu X'.S'(X')$, with $X' \in \set{Y_1,\ldots,Y_s}$  and $\mu X'.S'(X') \in \Phi_{\mu}(R)$, and 
        $\phi_{\mu}(\mathbf{T})= \ufold{\mu X'.S'(X')}{\mathbf{r}'}$.   
  \end{enumerate}

\item \label{properties:morphism:lemma:item:2} For any   fixed-point sequence  
  \begin{align*}
     \mathbf{T}_1 \subcer \cdots \subcer  \mathbf{T}_m 
   \end{align*}
  in     $\Eu{T}(S \combb R)$ with $m \ge 1$ and   for any $i=1,\ldots,m$, there are  iteration  mappings   $\mathbf{s}_i:\{X_1,\ldots,X_s\} \to  \mathbb{N}$ and $\mathbf{r}_i:\{Y_1,\ldots,Y_r\} \to  \mathbb{N}$, such that
  one of the following two cases holds:
  \begin{enumerate}
  \item \label{case:a} There is a  \ce $S_i(X^i) \in \Phi(S)$ with $X^i \in \set{ X_1,\ldots,X_s}$, and a \ce $R_i \in \Phi(R)$ such that 
  \begin{align*}
    \phi_{\mu}(\mathbf{T}_i)= \ufold{\mu X^i.S_i(X^i)}{\mathbf{s}_i} \combb \ufold{R_i}{\mathbf{r}_i},
  \end{align*}
   and  for $i=1,\ldots,m-1$ and for any  $X \in \set{ X_1,\ldots,X_s}$ and any $Y \in \set{ Y_1,\ldots,Y_r}$,  we have that
 \begin{align}
   \label{iteration:X:eq}
   \mathbf{s}_{i+1}(X) = 
    \begin{cases}
      \mathbf{s}_{i}(X), & \tif X\neq X^i\\
      \mathbf{s}_{i}(X^i)-1, & \tif X = X^i
    \end{cases} && \tand &&     \mathbf{r}_{i+1}(Y) =  \mathbf{r}_{i}(Y) 
 \end{align}
  \item \label{case:b} There is a \ce  $S_i \in \Phi(S)$, and  a  \ce $R_i(Y^i) \in \Phi(R)$ with $Y^i \in \set{ Y_1,\ldots,Y_r}$, such that 
  \begin{align*}
    \phi_{\mu}(\mathbf{T}_i)=  \ufold{S_i}{\mathbf{s}_i} \combb \ufold{\mu Y^i.R_i(Y^i)}{\mathbf{r}_i}, 
      \end{align*}
\end{enumerate}
   and  for $i=1,\ldots,m-1$ and for any  $X \in \set{ X_1,\ldots,X_s}$ and any $Y \in \set{ Y_1,\ldots,Y_r}$,  we have that
\begin{align}
  \label{iteration:Y:eq}
   \mathbf{s}_{i+1}(X) =  \mathbf{s}_{i}(X) 
   && \tand &&  
   \mathbf{r}_{i+1}(Y) = 
    \begin{cases}
      \mathbf{r}_{i}(Y), & \tif Y\neq Y^i\\
      \mathbf{r}_{i}(Y^i)-1, & \tif Y = Y^i
    \end{cases}
  \end{align}
 
\end{enumerate}
\end{lemma} 
\begin{example}
  We consider the two unifications $\xi_1 \combb R_1$ and $\ufold{\xi_1}{\mathbf{s}_1} \combb R_1$ of  Example~\ref{properties:quasi:simulation:example}, as well as 
  their related  derivations (\ref{deriv:1})  and (\ref{deriv:2}).
\begin{enumerate}
\item  The cases (\ref{properties:morphism:lemma:item:1:1}) and (\ref{properties:morphism:lemma:item:1:3}) of Lemma~\ref{properties:morphism:lemma} correspond to the following equalities:
  \begin{align*}
  \phi_{\mu}(\mathbf{T}_1) &=  (\ufold{\xi_1}{\mathbf{s}_1} \combb  R_1) &&\tand&&   \phi_{\mu}(\mathbf{T}_2) = (\ufold{\xi_2}{\mathbf{s}_2} \combb  R_2) && \tand&&   \phi_{\mu}(\mathbf{T}_3)=  (\ufold{\xi_3}{\mathbf{s}_3} \combb  R_3) 
  \tag{Case (\ref{properties:morphism:lemma:item:1:1}) of Lemma~\ref{properties:morphism:lemma}} \\
   \phi_{\mu}(\xi_1) &= \ufold{\xi_1}{\mathbf{s}_2}.     \tag{Case (\ref{properties:morphism:lemma:item:1:3}) of Lemma~\ref{properties:morphism:lemma}} 
\end{align*}
\item    Notice that from the derivation  (\ref{deriv:1}) of Example~\ref{properties:quasi:simulation:example}, we  have that  $\mathbf{T}_3$ is a  (fixed-point) sub-\ce of $\mathbf{T}_2$, which is a (fixed-point) sub-\ce of $\mathbf{T}_1$.   Thus we have the fixed-point sequence
\begin{align*}
  \mathbf{T}_1  \subcer \mathbf{T}_2  \subcer \mathbf{T}_3
\end{align*}
 in $\Eu{T}(\xi_1 \combb R_1)$.
Recall that the iteration mappings $\mathbf{s}_1$, $\mathbf{s}_2$ and $\mathbf{s}_3$ were  defined  in Example~\ref{properties:quasi:simulation:example} as:
\begin{align*}
  \mathbf{s}_1=\set{X_1 \mapsto n, X_2\mapsto n, X_3\mapsto n }, &&  \mathbf{s}_2=\set{X_1 \mapsto n-1, X_2\mapsto n, X_3\mapsto n}, && \mathbf{s}_3=\set{X_1 \mapsto n-1, X_2\mapsto n-1, X_3\mapsto n }.
\end{align*}
  Therefore, Eq.(\ref{iteration:X:eq}) of the case (\ref{case:a})  of  Lemma~\ref{properties:morphism:lemma} expresses $\mathbf{s}_2$ in terms of $\mathbf{s}_1$, as well as $\mathbf{s}_3$ in terms of $\mathbf{s}_2$ as follows:
   \begin{align*}
   \mathbf{s}_{2}(X) = 
    \begin{cases}
      \mathbf{s}_{1}(X), & \tif X\neq X_1\\
      \mathbf{s}_{1}(X_1)-1, & \tif X = X_1
    \end{cases} && \tand &&
   \mathbf{s}_{3}(X) = 
    \begin{cases}
      \mathbf{s}_{2}(X), & \tif X\neq X_2\\
      \mathbf{s}_{2}(X_2)-1, & \tif X = X_2,
    \end{cases}
   \end{align*}
   where $X \in \set{X_1,X_2,X_3}$.
   \end{enumerate}
\end{example}  


%% file: correction-combination_corrected.tex
\section{The equivalence between the  unification of two \ces and that of  their unfoldings}
\label{equiv:unif:with:unif:unfolding:sec}

This is the most technical section in which we develop the last  ingredient required  in the proof of  the main result of this paper. 
The purpose of this section is to show  that  the unification of two \ces  is equivalent to the unification of their unfolding for any term of depth at most a certain  bound  that depends on the two unfoldings, i.e. Proposition ~\ref{main:corollary:unif}.
More precisely,  we shall prove that for  any two iteration mappings $\mathbf{s}$ and $\mathbf{r}$ with $\mathbf{s}(X_i)=\mathbf{r}(Y_j)=n$,
the \ces $S \combb R$ and $\ufold{S}{\mathbf{s}} \combb  \ufold{R}{\mathbf{r}}$ are equivalent  for any term of depth at most $n$, where $X_i$ (resp. $Y_j$) are the bound variables
of $S$ (resp. $R$).

To achieve this we need, on the one hand,  the main result of Subsection~\ref{relating:unif:with:unfolds:sec} that ensures the existence of
a $(\ceSet,\ceSetFree)$-quasi-simulation between $S \combb R$ and $\ufold{S}{\mathbf{s}} \combb  \ufold{R}{\mathbf{r}}$ (Lemma~\ref{main:lemma:mophism:quasi}),
together with the properties of this relation (Lemma~\ref{properties:morphism:lemma}).
Indeed such results guarantee   that $S \combb R$ and $\ufold{S}{\mathbf{s}} \combb  \ufold{R}{\mathbf{r}}$ have almost the same structure and should be equivalent.  
However, on the other hand, the structure of $S \combb R$ differs from that of $\ufold{S}{\mathbf{s}} \combb  \ufold{R}{\mathbf{r}}$ when it comes to certain sub-\ces.
Therefore, to complete the proof we need to show that any such sub-\ce  of $S \combb R$ is equivalent to its related sub-\ce of   $\ufold{S}{\mathbf{s}} \combb  \ufold{R}{\mathbf{r}}$ with respect
to any term of a certain depth that depends on the position of such distinct  sub-\ce  in $S \combb R$, or equivalently in  $\ufold{S}{\mathbf{s}} \combb  \ufold{R}{\mathbf{r}}$.

To  illustrate the idea, we consider the simplest case where $S \combb R=\mu Z.T(Z)$ such that $T(Z)$ is fixed-point free.
Let $\mathbf{E}=\ufold{S}{\mathbf{s}} \combb  \ufold{R}{\mathbf{r}}$.
Therefore, thanks to the  $(\ceSet,\ceSetFree)$-quasi-simulation and its properties, we know that
$\mathbf{E} =T(\mathbf{E}')$, where $\mathbf{E}'= \ufold{S}{\mathbf{s}'} \combb  \ufold{R}{\mathbf{r}'}$, for iteration mappings  $\mathbf{s}'$ and $\mathbf{r}'$.
Hence to show that $\mu Z.T(Z)$ is $n$-equivalent to $\mathbf{E}$, it suffices to show that $\mathbf{E}$ is a fixed-point of $T(Z)$, i.e. that $T(\mathbf{E})\equiv_n \mathbf{E}$.
But since $\mathbf{E}=T(\mathbf{E'})$, we need to show that $T(\mathbf{E})$ is $n$-equivalent to $T(\mathbf{E}')$. To achieve this, it suffices to show that   $\mathbf{E}$ is $n'$-equivalent to $\mathbf{E}'$ for some $n'$ provided that the  number of jumps between the root of  $T(Z)$ and $Z$ is at least $n-n'$.
That is,    $\ufold{S}{\mathbf{s}} \combb  \ufold{R}{\mathbf{r}}$ and $\ufold{S}{\mathbf{s}'} \combb  \ufold{R}{\mathbf{r}'}$ are $n'$-equivalent where $n'$ depends on the iteration mappings $\mathbf{s}$, $\mathbf{r}$, $\mathbf{s}'$ and $\mathbf{r}'$. It turns out that in this simple  case, $n'$ is nothing but the codistance $D^{\star}((\mathbf{s},\mathbf{r}),(\mathbf{s}',\mathbf{r}'))$, and $n-D^{\star}((\mathbf{s},\mathbf{r}),(\mathbf{s}',\mathbf{r}'))$ is a lower bound for the number of jumps   between the root of  $T(Z)$ and $Z$.

However, for the general case where $S\combb R$ contains many nested fixed-point \ces, say
\begin{align*}
  S \combb R=\mu Z_1.T_1(\mu Z_2.T_2(\cdots \mu Z_m.T_m(Z_m)))
\end{align*}
which yields the fixed-point sequence $\Eu{S}=\mu Z_1.T_1(Z_1) \subcer \cdots \subcer \mu Z_m.T_m(Z_m)$, many difficulties arise.
Namely, the codistance  $D^{\star}$ is no longer an exact lower  bound to the number of jumps, say   between the root of $T_i(Z_i)$ and $\mu Z_j.T_{j}(Z_j)$, where $1\le i < j \le m$. 
However the same technique remains:  to prove that $T(\mathbf E) \equiv_n T(\mathbf E')$ it is enough to show that $\mathbf E \equiv_{n'} \mathbf E'$, for some $n'$,
provided that the number of jumps from the root of $T(\mathbf{E})$ to $\mathbf{E}$ is at least $n-n'$.
Besides, the same principle remains: the more we go deeper in the sequence $\Eu{S}$, the more the iterations in the related \ces of $\Eu{S}$
(i.e. that result in  $\ufold{S}{\mathbf{s}} \combb  \ufold{R}{\mathbf{r}}$ and have the form $\ufold{S_i}{\mathbf{s}_i} \combb  \ufold{R_i}{\mathbf{r}_i}$ for $i=1,\ldots,m$) decrease, and the more we get  more jumps from
the root of $T_1(Z_1)$ to $\mu Z_m.T_m(Z_m)$. 

Having said that, we need to supplement the codistance $D^{\star}$ with  further measures that will be introduced in Subsection~\ref{measures:sub:sub:sec} together with their properties.
In  Subsection~\ref{derived:sub:sub:sec} we shall show that these measures  provide enough information to compute  an adequate  lower bound for the number of jumps.
More precisely,  these measures will allow us to extract a subsequence from the fixed-point sequence $\Eu{S}$, called the \emph{derivative sequence},  with the property that there is at least one jump  between any two successive \ces
in this Subsequence. Summing  up these results we shall show in Subsection~\ref{equivalent:sub:sub:sec} that the unification of two \ces, say $S\combb R$,  is $n$-equivalent to the unification of their unfolding, i.e.
say  $\ufold{S}{\mathbf{s}} \combb  \ufold{R}{\mathbf{r}}$, where the iteration mappings $\mathbf{s}$ and $\mathbf{r}$ associate to each fixed-point variable the constant $n$.
We shall make use of  the main result of Subsection~\ref{equiv:unif:unfoldings:sec} that allows to compare the equivalence of the unification of an unfolding of two \ces with the unification of another unfolding of the same two \ces.

\subsection{Measures and codistance on fixed-point tree}
\label{measures:sub:sub:sec}
  
We define next  the  number of  occurences of a \ce in  a sequence of tuples.

\begin{definition}
\label{nbr:occurence:in:sequence:example}
Let $S$ and $R$ be \ces.
Let  $\Eu{S}$ be  a sequence of tuples
\begin{align*}
  \Eu{S} = \tuple{S_1,R_1,\Eu{M}_1}, \ldots,\tuple{S_m,R_m,\Eu{M}_m}
\end{align*}
with $m \ge 1$, $S_i \in \Phi(S)$ and $R_i \in \Phi(R)$ for $i=1,\ldots,m$.
Let $S'$ be a fixed-point sub-\ce of $S$.  We  shall denote by $\nbr_{\Eu{S}}(S')$ the \emph{number of occurrences} of  $S'$  in the sequence $\Eu{S}$,  that is
\begin{align}
  \label{number:occurences:eq}
  \nbr_{\Eu{S}}(S') = |\set{S_i \gvert S_i=S', \;\; i=1,\ldots,m}|.
\end{align}
For a \ce $R'$ that is a fixed-point sub-\ce of $R$, the definition of $\nbr_{\Eu{S}}(R')$ is similar to $\nbr_{\Eu{S}}(R')$   by taking $R_i$ instead of $S_i$ in Eq.(\ref{number:occurences:eq}).
\end{definition}

We shall  use the following notations throughout this Subsection and the next Subsection~\ref{derived:sub:sub:sec} as well. 
Let $S$ (resp. $R$)  be  a \ce  with  fixed-point variables $X_1,\ldots, X_s$ (resp. $Y_1,\ldots, Y_r$). 
Let $n \ge 1$ and let  $\mathbf{s}: \{X_1,\ldots,X_s\} \to  \mathbb{N}$ and $\mathbf{r}:\{Y_1,\ldots,Y_r\} \to  \mathbb{N}$ be iteration  mappings with $\mathbf{s}(X_i)=\mathbf{s}(X_j)=n$.
Let $\phi_{\mu}$ be the mapping induced by the $(\ceSet,\ceSetFree)$-quasi-simulation $\morphyy$ between $S\combb R$ and  $\ufold{S}{\mathbf{s}} \combb \ufold{R}{\mathbf{r}}$. 
Let $\Eu{T}$ be the fixed-point tree of $S \combb R$.  Recall that  $\Eu{T}$  is not necessarily connected, i.e. it may be composed of many connected sub-trees and thus it may have many roots.

In the following Definition~\ref{definition:omega:Omega:def}  we define  three  measures, one to count the maximal number of repetitions of \ces
in a sequence of  tuples, a second one that is $n$ minus the previous measure, and the third one that transfers the codistance $D^{\star}$ from $\ufold{S}{\mathbf{s}} \combb \ufold{R}{\mathbf{r}}$ to the related  \ces that belong to the  fixed-point tree $\Eu{T}$. 

\begin{notations}[For Definition~\ref{definition:omega:Omega:def}]
In the following Definition~\ref{definition:omega:Omega:def} we let  $\Eu{S}$  to be a left-maximal  sequence  \begin{align*}     \mathbf{T}_1 \subcer  \ldots \subcer \mathbf{T}_m \end{align*}
in $\Eu{T}$  (i.e. $\mathbf{T}_1$ being a root of $\Eu{T}$)  with $m\ge 1$.  According to Items (\ref{properties:morphism:lemma:item:1:1}) and (\ref{properties:morphism:lemma:item:1:2}) of Lemma~\ref{properties:morphism:lemma},  we know that
for any $i =1,\ldots, m$,  one of the following two cases holds.
\begin{enumerate}[(i)]
\item $\mathbf{T}_i \in \widetilde{\Phi}(S \combb R) \setminus \big(\widetilde{\Phi}(S) \cup \widetilde{\Phi}(R)\big)$ and hence it   can be written as  $\mathbf{T}_i= \NF(\tuple{S_i,R_i,\Eu{M}_i})$ and in this case $\mathbf{T}_{i'} \in \widetilde{\Phi}(S \combb R) \setminus \big(\widetilde{\Phi}(S) \cup \widetilde{\Phi}(R)\big)$ for $i'<i$ . 
\item $\mathbf{T}_i \in \widetilde{\Phi}(S) \cup \widetilde{\Phi}(R)$  and in this case $i=m$ and  $\mathbf{T}_m= S_m \in \widetilde{\Phi}_{\mu}(S)$    or $\mathbf{T}_m= R_m \in  \widetilde{\Phi}_{\mu}(R)$. 
\end{enumerate}
This yields  the finite sequence of tuples
\begin{align*} \tilde{\Eu{S}} = \tuple{S_1,R_1,\Eu{M}_1}, \ldots \end{align*}
that either ends with a tuple $\tuple{S_m,R_m,\Eu{M}_m}$ or a fixed-point \ce $S_m \in  \widetilde{\Phi}_{\mu}(S)$ or
a fixed-point \ce $R_m \in  \widetilde{\Phi}_{\mu}(R)$.
Besides, the mapping $\phi_{\mu}$ associates to the sequence  $\tilde{\Eu{S}}$  the  sequence \begin{align*}\phi_{\mu}(\mathbf{T}_1),\ldots,\phi_{\mu}(\mathbf{T}_m)\end{align*}  in $\ufold{S}{\mathbf{s}} \combb \ufold{R}{\mathbf{r}}$ which is
\begin{align*}\ufold{S_1}{\mathbf{s}_1} \combb \ufold{R_1}{\mathbf{r}_1},\ldots \end{align*} that ends with the \ce $\ufold{S_m}{\mathbf{s}_m} \combb \ufold{R_m}{\mathbf{r}_m}$ or $\ufold{S_m}{\mathbf{s}_m}$  or  $\ufold{R_m}{\mathbf{r}_m}$, for iteration mappings $\mathbf{s}_i$ and $\mathbf{r}_i$.
\end{notations}

\begin{definition}[Measures on \ces of fixed-point tree]
  \label{definition:omega:Omega:def}
We define three measures.
  \begin{enumerate}
\item  We define  $\Omega^{\#}_{\Eu{S}}(\mathbf{T}_i,\mathbf{T}_j)$, for $1\le i \le j$, as the maximal number of occurrences of \ces that appear in the series of tuples in $\tilde{\Eu{S}}$ starting from the tuple related  to $\mathbf{T}_i$ and ending with the tuple related to $\mathbf{T}_j$.
  That is, if $\mathbf{T}_j \in \widetilde{\Phi}(S \combb R) \setminus \big(\widetilde{\Phi}(S) \cup \widetilde{\Phi}(R)\big)$, then 
  \begin{align}
    \label{definition:omega:Omega:def:Omega:1}
    \Omega^{\#}_{{\Eu{S}}}(\mathbf{T}_i,\mathbf{T}_j)   &=
    \max\set{\nbr_{\Eu{S}}(S_p), \nbr_{\Eu{S}}(R_p) \gvert  S_p \in \widetilde{\Phi}_{\mu}(S), R_p \in \widetilde{\Phi}_{\mu}(R), p=i,\ldots,j}.
  \end{align}

  If $\mathbf{T}_j \in \widetilde{\Phi}(S) \cup \widetilde{\Phi}(R)$, then
  \begin{align}
    \Omega^{\#}_{\Eu{S}}(\mathbf{T}_i,\mathbf{T}_j)   &= \begin{cases}
      0 & \tif i=j  \\
      \Omega^{\#}_{\Eu{S}}(\mathbf{T}_i,\mathbf{T}_{j-1}) & \tif i>j.
    \end{cases}
  \end{align}

\item For $i \in \set{1,m}$, define
  \begin{align}
    \omega_{\Eu{S}}(\mathbf{T}_i)   &= n-\Omega^{\#}_{\Eu{S}}(\mathbf{T}_1,\mathbf{T}_i).
  \end{align}
\item   For $i \in \set{1,m}$, we define the \emph{codistance} between $\mathbf{T}_1$ and $\mathbf{T}_i$ as follows.
  
  If $\mathbf{T}_i \in \widetilde{\Phi}(S \combb R) \setminus \big(\widetilde{\Phi}(S) \cup \widetilde{\Phi}(R)\big)$, then 
  \begin{align}
        \label{definition:omega:Omega:def:DS:1}
    \D_{\Eu{S}}(\mathbf{T}_i)   &= \begin{cases}
      n & \tif i=1  \\
      D^{\star}\big((\mathbf{s}_1,\mathbf{r}_1), (\mathbf{s}_i,\mathbf{r}_i)\big) & \tif i>1.
    \end{cases}
    \end{align}
If $\mathbf{T}_i \in \widetilde{\Phi}(S) \cup \widetilde{\Phi}(R)$, then 
\begin{align}
      \label{definition:omega:Omega:def:DS:2}
         \D_{\Eu{S}}(\mathbf{T}_i)   &= \begin{cases}
           n & \tif i=1  \\
           \min\big\{\D(\mathbf{T}_{i-1}), d^{\star}(\mathbf{s}_1,\mathbf{s}_i)\big\} & \tif i>1.
         \end{cases}
\end{align}
\end{enumerate}
\end{definition}

When the sequence $\Eu{S}$ is known and if there is no ambiguity we shall  simplify the notations by omitting $\Eu{S}$ and simply writing  $\Omega^{\#}(\mathbf{T}_i,\mathbf{T}_j)$, $\omega(\mathbf{T}_i)$ and $\D(\mathbf{T}_i)$ instead of  $\Omega^{\#}_{\Eu{S}}(\mathbf{T}_i,\mathbf{T}_j)$, $\omega_{\Eu{S}}(\mathbf{T}_i)$ and $\mathbf{D}^{\star}_{\Eu{S}}(\mathbf{T}_i)$. These three measures are illustrated with the following example.
\begin{example}
  \label{all:measures:example}
  Let $\xi_1,\xi_2,\xi_3$ be  fixed-point \ces, where $\xi_i=\mu X_i.S_i(X_i)$ for $i=1,2,3$, such that  $\xi_3$ is a sub-\ce of $\xi_2$ which is a sub-\ce of $\xi_1$. 
Let $R_1,\ldots,R_6$ be fixed-point free \ces,  and let  $\Eu{M}_1,\ldots,\Eu{M}_6$ be  memories with $\Eu{M}_1=\emptyset$.

 Firstly, we consider the unification  $\xi_1 \combb R_1$.
 We do not  make explicit the  derivation that starts from  $\tuple{\xi_1,R_1,\Eu{M}_1}$ because  it has  been detailed in the  similar and  simpler Example~\ref{properties:quasi:simulation:example}, see Eq.(\ref{deriv:1}).
 We assume  that the unification $\xi_1 \combb R_1$ gives rise to   the following  sequence of tuples, in which $\xi_1$ occurs $3$ times, $\xi_2$ occurs $2$ times  and  $\xi_3$ occurs once:
  \begin{align*}
    \tuple{\xi_1,R_1,\Eu{M}_1},
    \tuple{\xi_2,R_2,\Eu{M}_2},
    \tuple{\xi_1,R_3,\Eu{M}_3},
    \tuple{\xi_3,R_4,\Eu{M}_4},
    \tuple{\xi_2,R_5,\Eu{M}_5},
    \tuple{\xi_1,R_6,\Eu{M}_6}.
  \end{align*}
This  yields  the following (fixed-point) left-maximal  sequence, denoted by $\Eu{S}$: 
\begin{align}
  \label{all:measures:example:fixed-point:seqience:eq}
\mathbf{T}_1 \subcer \mathbf{T}_2 \subcer \cdots \subcer \mathbf{T}_6
  \end{align}
  in $\Eu{T}(\xi_1 \combb R_1)$, where each $\mathbf{T}_i$ is the normal form of the related triplet (i.e. $\mathbf{T}_1=\NF(\tuple{\xi_1,R_1,\Eu{M}_1})$, etc). 
  
  Secondly, we consider the unification  of an unfolding of $S_1$   with an unfolding of $R_1$.
  Recall that  $R_1$, as well as the other $R_2,\ldots,R_6$, are fixed-fixed point free, and therefore they are equal to their unfolding. 
  Hence,  we define  the following iteration mappings in which  $n\ge 1$:
\begin{align*}
&  \mathbf{s}_1=\set{X_1 \mapsto n, X_2\mapsto n, X_3\mapsto n }, &&  \mathbf{s}_2=\set{X_1 \mapsto n\m 1, X_2\mapsto n, X_3\mapsto n}, && \mathbf{s}_3=\set{X_1 \mapsto n\m 1, X_2\mapsto n\m 1, X_3\mapsto n }, \\
 & \mathbf{s}_4=\set{X_1 \mapsto n\m 2, X_2\mapsto n\m 1, X_3\mapsto n }, &&   \mathbf{s}_5=\set{X_1 \mapsto n\m 2, X_2\mapsto n\m 1, X_3\mapsto n\m 1 },  &&   \mathbf{s}_6=\set{X_1 \mapsto n\m 2, X_2\mapsto n\m 2, X_3\mapsto n\m 1},
\end{align*}
and we consider  the unification $\ufold{\xi_1}{\mathbf{s}_1} \combb R_1$, which is related to the unification $\xi_1 \combb R_1$ via the mapping $\phi_\mu$ as follows:
\begin{align*}
  \phi_{\mu}(\mathbf{T}_1) = \ufold{\xi_1}{\mathbf{s}_1} \combb R_1, &&
  \phi_{\mu}(\mathbf{T}_2) = \ufold{\xi_2}{\mathbf{s}_2} \combb R_2, &&
  \phi_{\mu}(\mathbf{T}_3) = \ufold{\xi_1}{\mathbf{s}_3} \combb R_3, && \\
  \phi_{\mu}(\mathbf{T}_4) = \ufold{\xi_3}{\mathbf{s}_4} \combb R_4, &&
  \phi_{\mu}(\mathbf{T}_5) = \ufold{\xi_2}{\mathbf{s}_5} \combb R_5, && 
  \phi_{\mu}(\mathbf{T}_6) = \ufold{\xi_1}{\mathbf{s}_6} \combb R_6. && 
\end{align*}
The measures  $\Omega^{\#}_{\Eu{S}}$ and $\omega_{\Eu{S}}$ and  $\D_{\Eu{S}}$,  related   to the (fixed-point) sequence $\Eu{S}$, are given in Table ~\ref{table:all:measures}, in which
the second row shows the \ce among $\set{\xi_1,\xi_2,\xi_3}$ that appears in $\mathbf{T}_i$;
and   the third row shows the iteration mapping $\mathbf{s}_i$ involved in $\phi_{\mu}(\mathbf{T}_i)$.

\begin{table}[H]
  \begin{center}
\begin{tabular}{ |c||c|c|c|c|c|c| } 
 \hline  
\backslashbox{}{\\ \\ \\ \\ $\mathbf{T_i}$}                                  & $\mathbf{T}_1$ & $\mathbf{T}_2$ & $\mathbf{T}_3$  &  $\mathbf{T}_4$ & $\mathbf{T}_5$  & $\mathbf{T}_6$ \\ \hline  \hline
 $\xi_j$ in $\mathbf{T}_i$  & $\xi_1$ & $\xi_2$& $\xi_1$  &   $\xi_3$     &  $\xi_2$&  $\xi_1$\\ \hline
 $\mathbf{s}_i$ & $(n, n, n)$ & $(n\m 1, n, n)$ & $(n\m 1, n\m 1, n)$  &  $(n\m 2, n\m 1, n)$       & $(n\m 2, n\m 1, n\m 1)$ & $(n\m 2, n\m 2, n\m 1)$ \\ \hline 
 $\Omega^{\#}_{\Eu{S}}(\mathbf{T}_1,\mathbf{T}_i)$ & $1$ & $1$ & $2$  &   $2$       & $2$  & $3$\\ \hline
 $\omega_{\Eu{S}}(\mathbf{T}_i)$                  & $n\m 1$ & $n \m 1$ & $n\m 2$  & $n\m 2$       &  $n\m 2$ & $n\m 3$\\ \hline
 $\D_{\Eu{S}}(\mathbf{T}_i)$ & $n$ & $n\m 1$ & $n\m 1$  & $n\m 2$      & $n\m 2$ & $n\m 2$ \\ \hline
\end{tabular}
\end{center}
  \caption{An example of the measures  $\Omega^{\#}_{\Eu{S}}$ and $\omega_{\Eu{S}}$, and the codistance $\D_{\Eu{S}}$ related to  the fixed-point left-maximal sequence $\Eu{S}=\mathbf{T}_1 \subcer \cdots \subcer \mathbf{T}_6$ defined in  Eq.(\ref{all:measures:example:fixed-point:seqience:eq}). The second row shows the \ce among $\set{\xi_1,\xi_2,\xi_3}$ that appears in $\mathbf{T}_i$, for $i=1,\ldots,6$.
  The third row shows the iteration mapping $\mathbf{s}_i$ involved in $\phi_{\mu}(\mathbf{T}_i)$, where the triplet  $(n_1,n_2,n_3)$ refers to  the  iteration mapping $\set{X_1\mapsto n_1, X_2\mapsto n_2, X_3\mapsto n_3}$.}
\label{table:all:measures}
\end{table}
\end{example}


In Lemma~\ref{two:bound:of:Omega} we shall establish a useful relation between  $\omega$ and $\D$.

\newcounter{counter:two:bound:of:Omega}
\setcounter{counter:two:bound:of:Omega}{\value{theorem}} 

\begin{lemma} 
  \label{two:bound:of:Omega}
For any left-maximal sequence 
  \begin{align*}
    \mathbf{T}_1  \subcer \cdots \subcer  \mathbf{T}_m  
  \end{align*}
 in $\Eu{T}$ with $m\ge 2$, and  for any $p$ and $q$ where $1 \le p <q \le m$,
\begin{enumerate}
\item If for $i=1,\ldots,q$, there are \ces $S_i \in \widetilde{\Phi}(S)$ and $R_i \in \widetilde{\Phi}(R)$, and iteration mappings \\$\mathbf{s}_i:\set{X_1,\ldots,X_s} \to \N$ and $\mathbf{r}_i:\set{Y_1,\ldots,Y_s} \to \N$ such that 
  \begin{align*}
  \phi_{\mu}(\mathbf{T}_i)& = \ufold{S_i}{\mathbf{s}_i} \combb  \ufold{R_i}{\mathbf{r}_i}
\end{align*}
then

\begin{align}
       \omega(\mathbf{T}_q) &\in \set{D^{\star}\big((\mathbf{s}_1,\mathbf{r}_1),(\mathbf{s}_q,\mathbf{r}_q)\big), D^{\star}\big((\mathbf{s}_1,\mathbf{r}_1),(\mathbf{s}_q,\mathbf{r}_q)\big)-1}.       \label{two:bound:of:Omega:eq:1}
  \end{align}

\item If  there is a \ce $\xi_m\in \widetilde{\Phi}(S) \cup \widetilde{\Phi}(R)$ and an iteration mapping $\mathbf{s}_m$ such that
\begin{align*}
  \phi_{\mu}(\mathbf{T}_m)& = \ufold{\xi_m}{\mathbf{s}_m} 
\end{align*}
then
\begin{align}
    \label{two:bound:of:Omega:eq:3}
    \min\set{\mathbf{s}_m(X) \gvert X \in \dom(\mathbf{s}_m)} \ge \D(\mathbf{T}_m) .
  \end{align}
\end{enumerate}
\end{lemma}    

From Lemma ~\ref{comparing:unif:unfolding:lemma}  we get the following corollary that establishes, in addition to another property,  the  semantic equivalence between
$\phi_{\mu}(Z_i)$ and $\phi_{\nu}^{\mu}(Z_i)$ for a fixed-point variable $Z_i$ of a \ce $\mu Z_i.T_i(Z_i)$ that appears in  the fixed-point tree $\Eu{T}$.
Roughly speaking, this corollary will be useful to prove that $\phi_{\mu}(\mu Z_i.T_i(Z_i))$ is a fixed-point of $T_i(Z_i)$ as explained at the beginning of Subsection~\ref{equivalent:sub:sub:sec},
and  used in the proof of Lemma~\ref{main:lemma:unif}.



\begin{corollary}
\label{properties:morphism:corollary}
Let $\Eu{S}$ be a  sequence
\begin{align*}
 \mu Z_{1}.T_{1}(Z_{1}) \subcer \cdots \subcer \mu Z_m.T_m(Z_m)  \subcer Z_i
\end{align*}
in  $\Eu{T}$  with $m \ge 1$ and  $i \in \set{1,\ldots,m}$.
\begin{enumerate}
\item If $Z_i \in \boundv{S \combb R}\setminus (\boundv{S} \cup \boundv{R})$  then 
\begin{align}
    \label{properties:morphism:corollary:eq:1}
    \phi_{\mu}(Z_i) \equiv_{\mathbf{D}^{\star}_{\Eu{S}}(Z_i)} \phi_{\nu}^{\mu}(Z_i).
\end{align}

\item If  $Z_m \in \boundv{S} \cup \boundv{R}$ (i.e. $\mu Z_m.T_m(Z_m) \in \widetilde{\Phi}_{\mu}(S) \cup \widetilde{\Phi}_{\mu}(R)$) then
\begin{align}
    \label{properties:morphism:corollary:eq:2}
    \phi_{\mu}(\mu Z_m.T_m(Z_m)) \equiv_{\mathbf{D}^{\star}_{\Eu{S}}( \mu Z_m.T_m(Z_m))} \mu Z_m.T_m(Z_m).
  \end{align}
\end{enumerate}
\end{corollary}

\begin{proof}
 \begin{enumerate}
   \item  Since $Z_i\in \boundv{S \combb R}\setminus (\boundv{S} \cup \boundv{R})$ then it follows from Items (\ref{properties:morphism:lemma:item:1:1}) and (\ref{properties:morphism:lemma:item:1:2}) of  Lemma~\ref{properties:morphism:lemma}  that
 there are \ces $S_i \in \Phi(S)$ and $R_i\in \Phi(R)$,   and iteration mappings  $\mathbf{s}_i,\mathbf{s}_j:\{X_1,\ldots,X_s\} \to  \mathbb{N}$ and $\mathbf{s}_i,\mathbf{r}_j:\{X_1,\ldots,X_r\} \to  \mathbb{N}$ such that
  \begin{align*}
    \phi_{\nu}^{\mu}(Z_i) &= \ufold{S_i}{\mathbf{s}_j} \combb \ufold{R_i}{\mathbf{r}_j}  \\
    \phi_{\mu}(Z_i) &= \ufold{S_i}{\mathbf{s}_i} \combb \ufold{R_i}{\mathbf{r}_i}
  \end{align*}
  where $\mathbf{s}_j \ge \mathbf{s}_i$ and  $\mathbf{r}_j \ge \mathbf{r}_i$. Thus, by Eq. (\ref{definition:omega:Omega:def:DS:1}) of Definition~\ref{definition:omega:Omega:def} of $\D$, we have $\mathbf{D}^{\star}_{\Eu{S}}(Z_i)=D^{\star}((\mathbf{s}_j,\mathbf{r}_j), (\mathbf{s}_i,\mathbf{r}_i))$.
  Hence the claim  follows from   Corollary~\ref{comparing:unif:unfolding:corollary}  that states that  
  \begin{align*}
    \ufold{S_i}{\mathbf{s}_j} \combb \ufold{R_i}{\mathbf{r}_j}
    \equiv_{D^{\star}((\mathbf{s}_j,\mathbf{r}_j), (\mathbf{s}_i,\mathbf{r}_i))} 
    \ufold{S_i}{\mathbf{s}_i} \combb \ufold{R_i}{\mathbf{r}_i}. 
  \end{align*}
\item    Assume that $\mu Z_m.T_m(Z_m)$ is a sub-\ce of $S$, the case when it is a sub-\ce of $R$ is similar.
From  Items (\ref{properties:morphism:lemma:item:1:3}) of  Lemma~\ref{properties:morphism:lemma}   if follows that there is an iteration mapping $\mathbf{s}_m:\set{X_1,\ldots,X_s} \to \mathbb{N}$ such that 
  \begin{align*}
    \phi_{\mu}(\mu Z_m.T_m(Z_m)) = \ufold{\mu Z_m.T_m(Z_m)}{\mathbf{s}_m}.
  \end{align*}
 
Let $\mathbf{m}=\min\set{\mathbf{s}_m(X_i) \gvert i=1,\ldots,s}$. We know from Item (\ref{unfold:equiv:lemma:item:3})  Lemma~\ref{unfold:equiv:lemma} that
 \begin{align*}
   \mu Z_m.T_m(Z_m) \equiv_{\mathbf{m}} \ufold{\mu Z_m.T_m(Z_m)}{\mathbf{s}_m}
    \end{align*}
 Hence it follows from   From Item (\ref{depth:position:composition:lemma:item:2}) of Lemma~\ref{depth:position:composition:lemma}  that to show Eq.(\ref{properties:morphism:corollary:eq:2}) it suffices to show
 \begin{align*}
 \D( \mu Z_m.T_m(Z_m)) \le \mathbf{m}
 \end{align*}
But this was proved in  Lemma ~\ref{two:bound:of:Omega}, see Eq. (\ref{two:bound:of:Omega:eq:3}).
 \end{enumerate}
\end{proof}

Although the following corollary will not be used in any further proofs, it is worth mentioning it. 
\begin{corollary}
  \label{quasi:sim:is:sim:cor}
  The $(\ceSet,\ceSetFree)$-quasi-simulation  that results between the unification of
  two \ces, say  $S \combb R$, and that of their unfolding, say $\ufold{S}{\mathbf{s}} \combb \ufold{R}{\mathbf{r}}$,  (i.e. constructed in the proof of  Lemma~\ref{main:lemma:mophism:quasi}) is actually  a $(\ceSet,\ceSetFree)$-simulation.
\end{corollary}
\begin{proof}
  This follows  immediately from Corollary~\ref{properties:morphism:corollary}. That is, on the one hand for any \ce $\mu Z.T(Z)$ in $S \combb R$ that is not a sub-\ce of $S$ nor $R$,
  we have that each of $\phi_{\mu}(\mu Z.T(Z))$ (i.e. $\phi_{\nu}^{\mu}(Z)$)  and $\phi_{\mu}(Z)$ corresponds to the unification of two unfoldings of the same two \ces.
  And on the other hand,  for any \ce $\mu Z.T(Z)$ in $S \combb R$ that is a sub-\ce of $S$ or $R$, there is a $(\ceSet,\ceSetFree)$-simulation between $\mu Z.T(Z)$ and any unfolding of it.
\end{proof}


\subsection{Derived  tree and a lower bound for the number of jumps}
\label{derived:sub:sub:sec}   
The Eq.(\ref{two:bound:of:Omega:eq:1}) of Lemma~\ref{two:bound:of:Omega} allows one to distinguish between
elements of  $\Eu{T}$ whose $\omega$ and $\D$ are equal, and those whose  $\omega$ and $\D$   are different by $1$. The latter  elements form the \emph{derived} tree of $\Eu{T}$.
The name "derived" tree is justified by the fact that  we want to focus on the elements of $\Eu{T}$ on which $\D$ changes and increases by $1$.

\begin{definition}[Derived tree $\partial \Eu{T}$ of $\Eu{T}$]
  \label{derived:tree:def}
Recall that  $\Eu{T}=(\Phi_{\mu}(S \combb R),\subcer)$.
We define the \emph{derived tree} of $\Eu{T}$, denoted by $\partial\Eu{T}$, as the pair  $\partial\Eu{T}=(A,\subcer)$ where  $A \subseteq \Phi_{\mu}(S \combb R)$ is defined by
\begin{align*}
  A =\set{\mathbf{T} \in \Phi_{\mu}(S \combb R)  \gvert \omega_{\Eu{S}}(\mathbf{T}) = \mathbf{D}^\star_{\Eu{S}}(\mathbf{T}) -1, \textrm{ for any maximal sequence }\Eu{S} \textrm{ in } \Eu{T} \textrm{ containing } \mathbf{T}}.
\end{align*}
\end{definition}
\begin{example}[Derived tree $\partial{\Eu{T}}$]
  We consider  Example~\ref{all:measures:example},  and we assume that the fixed-point tree $\Eu{T}$ of $\xi_1 \combb R_1$ contains just the sequence $\Eu{S}=\mathbf{T}_1 \subcer \mathbf{T}_2 \subcer \cdots \subcer \mathbf{T}_6$, 
  defined in Eq. (\ref{all:measures:example:fixed-point:seqience:eq}). By examining the last two rows of  Table~\ref{table:all:measures} that respectively exhibit  $\omega_{\Eu{S}}(\mathbf{T}_i)$ 
  and $\mathbf{D}^\star_{\Eu{S}}(\mathbf{T}_i)$, we notice that  the equality $\omega_{\Eu{S}}(\mathbf{T}) = \mathbf{D}^\star_{\Eu{S}}(\mathbf{T}) -1$ holds for $\mathbf{T}=\mathbf{T}_1,\mathbf{T}_3,\mathbf{T}_6$. 
  Hence it follows that the derived tree  $\partial{\Eu{T}}$ is composed of  $\mathbf{T}_1,\mathbf{T}_3,\mathbf{T}_6$. Besides, in $\partial{\Eu{T}}$, we have $\mathbf{T}_1 \subcer \mathbf{T}_3 \subcer  \mathbf{T}_6$.
\end{example}

The following remark provides useful observations  that can be illustrated  by the Table~\ref{table:all:measures} of the Example~\ref{all:measures:example}.
\begin{remark}
  \label{link:tree:derivative:rq}
Notice that, for any  maximal sequence $\Eu{S}$ in $\Eu{T}$,  the following statements follow from Eq.(\ref{two:bound:of:Omega:eq:1}) of Lemma~\ref{two:bound:of:Omega} and from Definition~\ref{derived:tree:def}.
\begin{enumerate}
  \item Any  (fixed-point) \ce  $\mathbf{T}$  which  is  in $\Eu{T}$ but not in $\partial \Eu{T}$ has the property $\omega_{\Eu{S}}(\mathbf{T}) = \mathbf{D}^{\star}_{\Eu{S}}(\mathbf{T})$. 
  \item Since  by  Items (\ref{properties:morphism:lemma:item:1:1}) and (\ref{properties:morphism:lemma:item:1:2}) of  Lemma~\ref{properties:morphism:lemma}  we know that each of $\D$ and $\omega$ can be incremented by at most $1$ from a \ce to its immediate sub-\ce in $\Eu{T}$, then if $\mathbf{\bb{T}}_1$ is in $\partial\Eu{T}$ and $\mathbf{T}_2$ is in $\Eu{T}$ such that $\mathbf{\bb T}_1 \subcer \mathbf{T}_2$,
    then  $\mathbf{D}^{\star}_{\Eu{S}}(\mathbf{\bb T}_1)=\mathbf{D}^{\star}_{\Eu{S}}(\mathbf{T}_2)+1$ and hence $\omega_{\Eu{S}}(\mathbf{\bb T}_1)=\omega_{\Eu{S}}(\mathbf{T}_2)$. 

  \item \label{link:tree:derivative:rq:item:2} Similarly,  if $\mathbf{T}_1,\ldots,\mathbf{T}_m$ are in $\Eu{T}$, and $\mathbf{\bb T}_2$ is in $\partial\Eu{T}$ such that
    $\mathbf{T}_1 \subcer \ldots \subcer \mathbf{T}_m \subcer \mathbf{\bb T}_2$, then  $\mathbf{D}^{\star}_{\Eu{S}}(\mathbf{T}_i)=\mathbf{D}^{\star}_{\Eu{S}}(\mathbf{\bb T}_2)$ and   $\omega_{\Eu{S}}(\mathbf{T}_i)=\omega_{\Eu{S}}(\mathbf{T}_2)+1$, for any $i \in \set{1,\ldots,m}$.   
  \item \label{link:tree:derivative:rq:item:3} In particular, if $\mathbf{\bb T}_1$ and $\mathbf{\bb T}_2$ are in $\partial \Eu{T}$   such that  $\mathbf{\bb T}_1 \subcer \mathbf{\bb T}_2$,
    then  $\mathbf{D}^{\star}_{\Eu{S}}(\mathbf{\bb T}_1)=\mathbf{D}^{\star}_{\Eu{S}}(\mathbf{\bb T}_2)+1$ and  $\omega_{\Eu{S}}(\mathbf{\bb T}_1)=\omega_{\Eu{S}}(\mathbf{\bb T}_2)+1$.
\end{enumerate}
\end{remark}

\newcommand \T { \mf{\bb{T}}_{2}}
\newcommand \TOne { \mf{\bb{T}}_{1}}

Thanks to Lemma~\ref{two:bound:of:Omega} and Remark~\ref{link:tree:derivative:rq},
we show in the following Lemma~\ref{distance:fixed-point-in-tree} a crucial property of the derived tree $\partial \Eu{T}$ that was behind its
introduction: if two \ces $\TOne$ and $\T$ are in $\partial \Eu{T}$ with $\TOne \subcer   \T$, then  the number of jumps between the root of $\TOne$ and $\T$ is at least one.

\newcounter{counter:distance:fixed-point-in-tree}
\setcounter{counter:distance:fixed-point-in-tree}{\value{theorem}} 

\begin{lemma}
  \label{distance:fixed-point-in-tree}
Let 
  \begin{align*}
    \mu \bb{Z}_1.\bb{T}_1(\bb Z_1)  \subcer   \T
  \end{align*}
be a  sequence in $\partial \Eu{T}$.
 Define $\bb T^{\star}_{1}(Z)$ to  be the (unique) \ce satisfying
\begin{align*}
  \bb T^{\star}_1(\T)= \bb T_{1}(\bb Z_1).
\end{align*}
We have that
\begin{align}
  \label{distance:fixed-point-in-tree:eq}
   1  \le \Pi_{Z}(\bb T^{\star}_1(Z)).   
\end{align}
\end{lemma}


\subsection{The unification of two \ces is equivalent to the unification of their unfolding}
\label{equivalent:sub:sub:sec}

We arrive at the key lemma  that will allows us to show that  unification of two \ces is $n$-equivalent to the unification of their unfolding.
We already explained  at the beginning of this Section~\ref{equiv:unif:with:unif:unfolding:sec} that, in the  particular setting where $S\combb R$ is composed of just one \ce, say $\mu Z.T(Z)$,
the purpose is to show that $\mathbf{E}$ is a fixed-point of $T(Z)$, where $\mathbf{E}$ is the unification of the unfolding of two \ces. That is, we want to show that $\mathbf{E}$ is $n$-equivalent to $T(\mathbf{E})$.

However,  if we consider  the general setting in which the fixed-point \ces in   $S \combb R$ can be nested, namely if we have a sequence  $\Eu{S}=\mu Z_1.T_1(Z_1) \subcer \cdots \subcer \mu Z_m.T_m(Z_m)$ in $S \combb R$,  
then a fixed-point variable $Z_i$ may appear in any \ce $\mu Z_j.T_{j}(Z_{j})$  for $1\le i \le j \le m$. Therefore, we need a general and inductive  way to formulate and then to show that certain fixed-point free 
\ces $\mathbf{E}_1,\ldots,\mathbf{E}_m$ (which are in the \ce that results from the unification of the unfolding of $S$ with the unfolding of $R$) are a fixed-point of
$T_1(Z_1),\ldots, T_m(Z_m)$, respectively, in the sense that $\mathbf{E}_i$ is $n_i$-equivalent to $T_i(\mathbf{E}_i)$, for $i=1,\ldots,m$, where $n_i$ is an appropriate constant. 
This general and inductive  way of formulating such requirements is achieved  thanks to the mappings  $\phi_{\mu}$ and $\widehat{\phi}_{\nu}^{\mu}$  by just  imposing that
$\phi_{\mu}\big(\mu Z_i.T_i(Z_i)\big)$  and  $\widehat{\phi}_{\nu}^{\mu}(T_i(Z_i))$ must be  $\D_{\Eu{S}}(\mu Z_i.T_i(Z_i))$-equivalent.
In particular,  $\phi_{\mu}\big(\mu Z_i.T_i(Z_i)\big)$ corresponds to $\mathbf{E}_i$, while $\widehat{\phi}_{\nu}^{\mu}(T_i(Z_i))$ corresponds to $T(\mathbf{E}_i)$ since, roughly speaking,  
$\widehat{\phi}_{\nu}^{\mu}(T_i(Z_i))$ corresponds to $T_i\big(\widehat{\phi}_{\nu}^{\mu}(Z_i)\big)$ which is $T_i(\mathbf{E}_i)$.

\begin{lemma} 
  \label{main:lemma:unif}
Let $S$ (resp. $R$)  be  a \ce  with bound   fixed-point variables $X_1,\ldots, X_s$ (resp. $Y_1,\ldots, Y_r$), and let  $n \ge 1$. 
Let  $\mathbf{s}:\{X_1,\ldots,X_s\} \to  \mathbb{N}$ and $\mathbf{r}:\{Y_1,\ldots,Y_r\} \to  \mathbb{N}$ be iteration  mappings  with $\mathbf{s}(X_i)=\mathbf{r}(Y_j)=n$, for $i=1,\ldots,s$ and $j=1,\ldots,r$.
Let $\Eu{T}$ be the fixed-point tree of $S \combb R$  rooted at  $\mu Z_1.T_1(Z_1)$. 
Let $\Eu{T}_i$ be a right maximal sub-tree of $\Eu{T}$  rooted at  $\mu Z_i.T_i(Z_i)$ yielding the unique sequence $\Eu{S}^i$:
 \begin{align*}
  \mu Z_1.T_1(Z_1)   \subcer \cdots \subcer     \mu Z_i.T_i(Z_i)
\end{align*}
 in $\Eu{T}$ and let
\begin{align*}
  \omega(i)&=\omega_{\Eu{S}^i}(\mu Z_i.T_i(Z_i)) \\
  \D(i)&=\D_{\Eu{S}^i}(\mu Z_i.T_i(Z_i)).
  \end{align*}
 Then for any $i=1,\ldots,\delta(\Eu{T})$, and any  maximal sequence 
 \begin{align*}
  \mu Z_i.T_i(Z_i)   \subcer \cdots \subcer     \mu Z_m.T_m(Z_m)
\end{align*}
  in $\Eu{T}_i$ where $ i\le m$, 
  either
\begin{enumerate}[(i)]
\item  $Z_i \in \boundv{S\combb R} \setminus \big(\boundv{S}\cup \boundv{R}\big)$ and in this case  we have that
\begin{align}
\label{distance:fixed-point-in-tree:remove:fixed-point:lemma:eq}
\phi_{\mu}\big(\mu Z_i.T_i(Z_i)\big)   & \equiv_{\D(i)}   \widehat{\phi}_{\nu}^{\mu}(T_i(Z_i)),
\end{align}

\item or $i=m$ and $Z_m \in \boundv{S}\cup \boundv{R}$, and in this case we have that 
\begin{align}
  \label{distance:fixed-point-in-tree:remove:fixed-point:lemma:eq'}
  \phi_{\mu}\big(\mu Z_m.T_m(Z_m)\big)   & \equiv_{\D(m)}   \mu Z_m.T_m(Z_m).
\end{align}
\end{enumerate}
\end{lemma}
\begin{proof} 
  The proof is by a double induction. The outer one is a  structural induction on the tree  $\partial\Eu{T}_i$. 
  \begin{description}
  \item[Outer base case $\delta(\partial\Eu{T}_i)=0$.] %
In this case consider a maximal  sequence
\begin{align*}
\mu Z_i.T_i(Z_i)  \subcer \cdots \subcer \mu Z_m.T_m(Z_m)
\end{align*}
in $\Eu{T}_i$ with $1\le i \le m$. Indeed, since $\delta(\partial\Eu{T}_i)=0$ then
\begin{align}
  \label{ind:equality:eq:1}
  \D(i)=\D({i+1})=\ldots=\D({m}).
\end{align}
 
We make an inner  structural induction on $\Eu{T}_i$.
\\\underline{Inner base case: $\delta(\Eu{T}_i)=1$.}

Since $\Eu{T}_i$ is not necessarily connected, it may contain many  maximal sequences, but each one of them  is composed of just one fixed-point \ce. Hence, consider a   maximal sequence 
  \begin{align*}
    \mu Z_m.T_m(Z_m)  
  \end{align*}
  in $\Eu{T}_i$. 
 And we need to show that
\begin{numcases}{\phi_{\mu}\big(\mu Z_m.T_m(Z_m)\big)   \equiv_{\D(m)}}
  \widehat{\phi}_{\nu}^{\mu}(T_m(Z_m)) &  if $Z_m \in  \boundv{S\combb R} \setminus \big(\boundv{S}\cup \boundv{R}\big)$  \label{distance:fixed-point-in-tree:remove:fixed-point:lemma:eq:base-case'} \\ 
  \mu  Z_m.T_m(Z_m)                   &  if $Z_m \in \boundv{S}\cup \boundv{R}.$                                          \label{distance:fixed-point-in-tree:remove:fixed-point:lemma:eq:base-case:X}
  \end{numcases}
If $Z_m \in \boundv{S}\cup \boundv{R}$ then Eq (\ref{distance:fixed-point-in-tree:remove:fixed-point:lemma:eq:base-case:X}) follows from Eq.(\ref{properties:morphism:corollary:eq:2}) of  Corollary~\ref{properties:morphism:corollary}.

If  $Z_m \in  \boundv{S\combb R} \setminus \big(\boundv{S}\cup \boundv{R}\big)$ then notice that $T_m(Z_m)$ is fixed point-free but  may contain free fixed-point variables besides  $Z_m$. Therefore  there exists a  fixed point-free \ce
$T^{\star}(Z^1,\ldots,Z^l,Z_m)$ with $l \ge 0$ and  $\set{Z^1,\ldots,Z^l} \subseteq \set{Z_1,\ldots,Z_{m}} \setminus \set{Z_m}$ such that $T_m(Z_m)=T^{\star}_m(Z^1,\ldots,Z^{l},Z_m)$.
Hence we need to show that 
\begin{align}
\label{distance:fixed-point-in-tree:remove:fixed-point:lemma:eq:base-case''}
\phi_{\mu}\big(\mu Z_m.T^\star_m(Z^1,\ldots,Z^{l},Z_m)\big)   & \equiv_{\D(m)}   \widehat{\phi}_{\nu}^{\mu}(T_m^\star(Z^1,\ldots,Z^{l},Z_m)).
\end{align}
On the one hand, it follows from Remark~\ref{mu:morphi:prop} that  the left-hand side of Eq. (\ref{distance:fixed-point-in-tree:remove:fixed-point:lemma:eq:base-case''})
can be written as
\begin{align*}
  \textrm{LHS}. (\ref{distance:fixed-point-in-tree:remove:fixed-point:lemma:eq:base-case''}) &=\phi_{\mu}\big(\mu Z_m.T_m^\star(Z^1,\ldots,Z^{l},Z_m)\big)\\
  &=T_m^\star\big(\phi_{\mu}(Z^1),\ldots,\phi_{\mu}(Z^{l}),\phi_{\mu}(Z_m)\big).
\end{align*}

On the other hand,  by the Definition~\ref{induced:morphisms:def} of $\widehat{\phi}_{\nu}^{\mu}$,  the right-hand side of Eq. (\ref{distance:fixed-point-in-tree:remove:fixed-point:lemma:eq:base-case''}) can be written as 
\begin{align*}
  \textrm{RHS}. (\ref{distance:fixed-point-in-tree:remove:fixed-point:lemma:eq:base-case''}) &=   \widehat{\phi}_{\nu}^{\mu}(T_m^\star(Z^1,\ldots,Z^l,Z_m)) \\ 
   &= T_m^\star(\widehat{\phi}_{\nu}^{\mu}(Z^1),\ldots,\widehat{\phi}_{\nu}^{\mu}(Z^l),\widehat{\phi}_{\nu}^{\mu}(Z_m))) \\
   &= T_m^\star(\phi_{\nu}^{\mu}(Z^1),\ldots,\phi_{\nu}^{\mu}(Z^l),\phi_{\nu}^{\mu}(Z_m)).
 \end{align*}
 Thus we need to show that
 \begin{align}
   \label{distance:fixed-point-in-tree:remove:fixed-point:lemma:eq:base-case'''}
   T_m^\star(\phi_{\mu}(Z^1),\ldots,\phi_{\mu}(Z^{l}),\phi_{\mu}(Z_m))
   \equiv_{\D(m)}
   T_m^\star(\phi_{\nu}^{\mu}(Z^1),\ldots,\phi_{\nu}^{\mu}(Z^l),\phi_{\nu}^{\mu}(Z_m)).
\end{align}
 From Eq.(\ref{properties:morphism:corollary:eq:1}) of Corollary~\ref{properties:morphism:corollary} we have that
 \begin{align*}
\phi_{\mu}(Z_m) \equiv_{\D(Z_m)} \phi_{\nu}^{\mu}(Z_m) && \tand&&  \phi_{\mu}(Z^j) \equiv_{\D(Z^j)} \phi_{\nu}^{\mu}(Z^j) \textrm{ for } j=1,\ldots,l.
 \end{align*}
But we know from Eq.(\ref{ind:equality:eq:1}) above  that $\D(Z_m)=\D(m)$ as well as $\D(Z_m)=\D(Z^j)$ for $j=1,\ldots,l$.
Thus Eq.(\ref{distance:fixed-point-in-tree:remove:fixed-point:lemma:eq:base-case'''}) holds  by Item (\ref{depth:position:composition:lemma:item:2}) of Lemma~\ref{depth:position:composition:lemma}. 

\underline{Inner induction step.} 
 Assume that Eq.(\ref{distance:fixed-point-in-tree:remove:fixed-point:lemma:eq}) holds for a  fixed-point sub-tree $\Eu{T}_i$  of $\Eu{T}$,
 and  we shall prove it for the (unique)  fixed-point sub-tree $\Eu{T}_{i-1}$ (of $\Eu{T}$) that contains $\Eu{T}_i$  such that  $\Eu{T}_i$  is an immediate sub-tree of  $\Eu{T}_i$.  Assume that
 $\Eu{T}_{i-1}$ is rooted at $\mu Z_{i-1}.T_{i-1}(Z_{i-1})$.

 Consider such  tree   $\Eu{T}_{i-1}$  and a  maximal sequence
\begin{align*}
 \mu Z_{i-1}.T_{i-1}(Z_{i-1})  \subcer   \mu Z_i.T_i(Z_i)  \subcer \cdots \subcer \mu Z_m.T_m(Z_m)
\end{align*}
in  $\Eu{T}_{i-1}$. We recall that we have
\begin{align}
 \label{ind:equality:eq:2}
 \D(i-1)=\ldots=\D(m).
 \end{align}

The \ce  $T_{i-1}(Z_{i-1})$ can be written in terms of its immediate fixed-point sub-\ces and  fixed-point variables in the sense that  there exist  $k \ge 1$ and $l\ge 0$ and 
  \begin{enumerate}[i.)]
  \item  a fixed-point free \ce   $T_{i-1}^{\star}(X^1 ,\ldots,X^{k+l})$  in which  each fixed-point variable $X^j$ is free,  and 
  \item \ces   $\mathbf{T}_1,\ldots,\mathbf{T}_k$  where each $\mathbf{T}_j$ is either a fixed-point \ce in $\Phi_{\mu}(S\combb R)$, and 
  \item fixed-point variables   $Z^1,\ldots,Z^l$ where  $\set{Z^1,\ldots,Z^l} \subseteq \set{Z_1,\ldots,Z_{i-1}}$,
  \end{enumerate}
   such that   $T_{i-1}(Z_{i-1})$ can be written as
 \begin{align*}
    T_{i-1}(Z_{i-1}) = T_{i-1}^{\star}(\mathbf{T}_1 ,\ldots,\mathbf{T}_{k},Z^1,\ldots,Z^l). 
 \end{align*}
 Hence,  we need to show that 
\begin{align}
\label{distance:fixed-point-in-tree:remove:fixed-point:lemma:eq:induction}
\phi_{\mu}\big(\mu Z_{i-1}.T_{i-1}^{\star}(\mathbf{T}_1 ,\ldots,\mathbf{T}_{k},Z^1,\ldots,Z^l)  \big)
& \equiv_{\D{(i-1)}}
\widehat{\phi}_{\nu}^{\mu}( T_{i-1}^{\star}(\mathbf{T}_1 ,\ldots,\mathbf{T}_{k},Z^1,\ldots,Z^l)).
\end{align}

On the one hand, it follows from Remark~\ref{mu:morphi:prop} that  the left-hand side of Eq.(\ref{distance:fixed-point-in-tree:remove:fixed-point:lemma:eq:induction})
can be written as
\begin{align*}
  \textrm{LH}.(\ref{distance:fixed-point-in-tree:remove:fixed-point:lemma:eq:induction}) &=
  \phi_{\mu}\big(\mu Z_{i-1}.T_{i-1}^{\star}(\mathbf{T}_1 ,\ldots,\mathbf{T}_{k},Z^1,\ldots,Z^l)  \big) \\
  &=T_{i-1}^{\star}\big(\phi_{\mu}(\mathbf{T}_1) ,\ldots,\phi_{\mu}(\mathbf{T}_{k}),\phi_{\mu}(Z^1),\ldots,\phi_{\mu}(Z^l)\big).
\end{align*}

On the one hand, by Definition~\ref{induced:morphisms:def} of $\widehat{\phi}_{\nu}^{\mu}$, the right-hand side of Eq.(\ref{distance:fixed-point-in-tree:remove:fixed-point:lemma:eq:induction}) can be written as 
\begin{align*}
  \textrm{RH}.(\ref{distance:fixed-point-in-tree:remove:fixed-point:lemma:eq:induction})& =
  \widehat{\phi}_{\nu}^{\mu}\big( T_{i-1}^{\star}(\mathbf{T}_1 ,\ldots,\mathbf{T}_{k},Z^1,\ldots,Z^l)\big) \\
   &=T_{i-1}^{\star}\big(\widehat{\phi}_{\nu}^{\mu}(\mathbf{T}_1) ,\ldots,\widehat{\phi}_{\nu}^{\mu}(\mathbf{T}_{k}),\phi_{\nu}^{\mu}(Z^1),\ldots,\phi_{\nu}^{\mu}(Z^l)\big).
\end{align*}

Therefore showing Eq.(\ref{distance:fixed-point-in-tree:remove:fixed-point:lemma:eq:induction}) amounts  to show that 
\begin{align}
  \label{distance:fixed-point-in-tree:remove:fixed-point:lemma:eq:induction'}
  T_{i-1}^{\star}\big(\phi_{\mu}(\mathbf{T}_1) ,\ldots,\phi_{\mu}(\mathbf{T}_{k}),\phi_{\mu}(Z^1),\ldots,\phi_{\mu}(Z^l)\big)
 \equiv_{\D(i-1)}
  T_{i-1}^{\star}\big(\widehat{\phi}_{\nu}^{\mu}(\mathbf{T}_1) ,\ldots,\widehat{\phi}_{\nu}^{\mu}(\mathbf{T}_{k}),\phi_{\nu}^{\mu}(Z^1),\ldots,\phi_{\nu}^{\mu}(Z^l)\big).
\end{align}

We recall that from Eq.(\ref{ind:equality:eq:2}), we have $\D(\mathbf{T}_j)=\D(i-1)$ for $j=1,\ldots,k$ as well as $\D(Z^j)=\D(i-1)$ for $j=1,\ldots,l$. 
Therefore from Item~\ref{depth:position:composition:lemma:item:2} of Lemma~\ref{depth:position:composition:lemma} it follows that to show Eq.(\ref{distance:fixed-point-in-tree:remove:fixed-point:lemma:eq:induction'}) it is enough to show that for any $j=1,\ldots,k$ 
\begin{align}
  \phi_{\mu}(\mathbf{T}_j)                         & \equiv_{\D(\mathbf{T}_j)}   \widehat{\phi}_{\nu}^{\mu}(\mathbf{T}_j) \label{distance:fixed-point-in-tree:remove:fixed-point:lemma:eq:induction':1} 
\end{align}
and that for any $j=1,\ldots,l$,
\begin{align}
  \phi_{\mu}(Z^j)                         & \equiv_{\D(Z^j)}   \phi_{\nu}^{\mu}(Z^j)   \label{distance:fixed-point-in-tree:remove:fixed-point:lemma:eq:induction':3} 
\end{align}

To achieve this, consider the two cases. 
\begin{itemize}
\item \underline{For Eq.(\ref{distance:fixed-point-in-tree:remove:fixed-point:lemma:eq:induction':1})}
  assume that $\mathbf{T}_{j}$ is of the form $\mathbf{T}_{j}=\mu\mathbf{Z}^j. \mathbf{T}_{j}^{\star}(\mathbf{Z}^j)$. 
  We distinguish  two cases depending on whether
  $\mathbf{Z}^j \in \boundv{S\combb R} \setminus \big(\boundv{S} \cup \boundv{R}\big)$ or $\mathbf{Z}^j \in (\boundv{S} \cup \boundv{R})$.
  For the first case we have   the sequence 
 \begin{align*}
  \mu Z_{i-1}.T_{i-1}(Z_{i-1})   \subcer \mathbf{T}_{j} \subcer \mu Z_{i+1}.T_{i+1}(Z_{i+1}) \subcer \cdots \subcer     \mu Z_m.T_m(Z_m)
\end{align*}
in $\Eu{T}$.  Let  $\Eu{T}^j$  be the maximal sub-tree of $\Eu{T}_{i-1}$ which is rooted at $\mathbf{T}_{j}$. 
  Since     $\Eu{T}^j$ is an immediate sub-tree of $\Eu{T}_{i-1}$, then  Eq. (\ref{distance:fixed-point-in-tree:remove:fixed-point:lemma:eq:induction':1}) follows from  the inner induction hypothesis.
However, for the second case where $\mathbf{Z}^j \in (\boundv{S} \cup \boundv{R})$  we have   the sequence 
 \begin{align*}
  \mu Z_{i-1}.T_{i-1}(Z_{i-1})   \subcer \mathbf{T}_{j}
\end{align*}
 in $\Eu{T}$. That is, in this case  we remind that the fixed-point \ce  $\mathbf{T}_{j}$ is either a sub-\ce of $S$ or of $R$.
 Hence, we have $\widehat{\phi}_{\nu}^{\mu}(\mathbf{T}_{j})=\mathbf{T}_{j}$.
 It follows from  the base case, i.e. Eq.(\ref{distance:fixed-point-in-tree:remove:fixed-point:lemma:eq'}) that
 \begin{align*}
   \phi_{\mu}(\mathbf{T}_{j})       & \equiv_{\D{(\mathbf{T}_j)}}  \mathbf{T}_{j}.  
  \end{align*}
\item \underline{For Eq.(\ref{distance:fixed-point-in-tree:remove:fixed-point:lemma:eq:induction':3})}, 
    it follows from  Eq.(\ref{properties:morphism:corollary:eq:1}) of  Corollary ~\ref{properties:morphism:corollary}.  
\end{itemize}
\underline{Outer induction step.} 
 Assume that Eq.(\ref{distance:fixed-point-in-tree:remove:fixed-point:lemma:eq}) holds for a  fixed-point sub-tree $\partial\Eu{T}_i$  of $\Eu{T}$,
 we shall prove it for the (unique)  fixed-point sub-tree $\partial(\Eu{T}_{i-1})$ (of $\Eu{T}$) that contains $\partial\Eu{T}_i$  such that  $\partial\Eu{T}_i$  is an immediate sub-tree of  $\partial\Eu{T}_i$.
Let  $\mu \bb Z_{\bb i-1}.\bb T_{\bb i-1}(\bb Z_{\bb i-1})$ be a root of  $\partial(\Eu{T}_{i-1})$, and let  $\mu \bb Z_{\bb i}.\bb T_{\bb i}(\bb Z_{\bb i})$ be a root of $\partial(\Eu{T}_{i})$.
 Consider such  tree   $\partial(\Eu{T}_{i-1})$  and a  maximal sequence
\begin{align*}
  \mu \bb Z_{\bb i-1}.\bb T_{\bb i-1}(\bb Z_{\bb i-1})  \subcer   \mu \bb Z_{\bb i}.\bb T_{\bb i} (\bb Z_{\bb i})  \subcer \cdots \subcer \mu \bb Z_{\bb m}.\bb T_{\bb m}(\bb Z_{\bb m})
\end{align*}
in  $\partial\Eu{T}_{i-1}$.  
Assume that  the maximal sequence in $\Eu{T}_{i-1}$ that lays between the root of $\partial\Eu{T}_{i-1}$ and the root  of $\partial(\Eu{T}_{i})$ is non-empty,
the case where it is empty can be handled similarly. Let the following be such a sequence: 
\begin{align*}
  \mu \bb Z_{\bb i-1}.\bb T_{\bb i-1}(\bb Z_{\bb i-1})  \subcer   \mu Z_p.T_p(Z_p)  \subcer  \mu Z_{p+1}.T_{p+1}(Z_{p+1}) \subcer \cdots \subcer \mu Z_q.T_q(Z_q) \subcer \mu \bb Z_{\bb i}.\bb T_{\bb i} (\bb Z_{\bb i})
\end{align*}
where $1 \le p \le q$.
In this case, by Definition~\ref{derived:tree:def} of the derived tree $\partial\Eu{T}$  we have
\begin{align}
  \D\big(  \mu \bb Z_{\bb i-1}.\bb T_{\bb i-1}(\bb Z_{\bb i-1}) \big)&=\D\big( \mu Z_p.T_p(Z_p)\big) +1   \label{ind:equality:eq:3'}\\
  \D(\mu Z_j.T_j(Z_j))&=\D(\mu Z_{j+1}.T_{j+1}(Z_{j+1})) =  \D\big(\mu \bb Z_{\bb i}.\bb T_{\bb i} (\bb Z_{\bb i})\big)  & \textrm{ for }  j=p,\ldots, q-1.  \label{ind:equality:eq:3''}
\end{align}

On the one hand, from the outer induction hypothesis we have that
\begin{align}
  \phi_{\mu} \big(\mu \bb Z_{\bb i}.\bb T_{\bb i} (\bb Z_{\bb i})\big) \equiv_{\D(\bb T_{\bb i} (\bb Z_{\bb i}))}   \widehat{\phi}_{\mu}^{\nu}\big(\bb T_{\bb i} (\bb Z_{\bb i})\big).
\end{align}

On the other hand, as far as Eq.(\ref{ind:equality:eq:3''}) holds, using the same kind of induction made in the inner  base case, we can easily show that,  for $j=1,\ldots, q-1$, we have 
\begin{align}
  \phi_{\mu} \big( \mu Z_j.T_j(Z_j) \big) \equiv_{\D(\mu Z_j.T_j(Z_j))}   \widehat{\phi}_{\mu}^{\nu}\big(T_j(Z_j)   \big).
\end{align}

Since $ \mu \bb Z_{\bb i}.\bb T_{\bb i} (\bb Z_{\bb i})$ is a sub-\ce of $\mu \bb Z_{\bb i-1}. \bb T_{\bb i-1}(\bb Z_{\bb i-1})$  then there  is a \ce $\bb T^{\star}_{i-1}(Z)$ such that $\bb T_{\bb i-1}(\bb Z_{\bb i-1})$ can be written as
\begin{align*}
\bb T_{\bb i-1}(\bb Z_{\bb i-1}) =   T^{\star}_{i-1}\big(  \mu \bb Z_{\bb i}.\bb T_{\bb i} (\bb Z_{\bb i}) \big).
\end{align*}

Since  by Eq.(\ref{ind:equality:eq:3'}) we know that   $\D\big(  \mu \bb Z_{\bb i-1}.\bb T_{\bb i-1}(\bb Z_{\bb i-1}) \big)=\D\big( \mu Z_p.T_p(Z_p)\big) +1$,
then it follow from  Item~\ref{depth:position:composition:lemma:item:1} of Lemma~\ref{depth:position:composition:lemma}  that  to show Eq.(\ref{distance:fixed-point-in-tree:remove:fixed-point:lemma:eq}),
it suffices to show that
\begin{align*}
1 \le \Pi_{Z}\big( \bb T^{\star}_{i-1}(Z) \big), 
\end{align*}  
but this was proved in Lemma~\ref{distance:fixed-point-in-tree}, see Eq.(\ref{distance:fixed-point-in-tree:eq}).
  \end{description}
\end{proof}


In the following Corollary we show that the unification of two \ces is equivalent to that of their unfolding in the particular
setting in which one of these two \ces is a fixed-point one.

\begin{corollary}
\label{main:corollary:unif-0}
Let $S$ (resp. $R$)  be  a \ce  with bound   fixed-point variables $X_1,\ldots, X_s$ (resp. $Y_1,\ldots, Y_r$) and let $n \ge 1$.
Let  $\mathbf{s}:\{X_1,\ldots,X_s\} \to  \mathbb{N}$ and $\mathbf{r}:\{Y_1,\ldots,Y_r\} \to  \mathbb{N}$ be iteration mappings with $\mathbf{s}(X_i)=\mathbf{r}(Y_j)=n$  for  $i=1,\ldots,s$ and $j=1,\ldots,r$.
If either  $S$ or  $R$  is  a fixed-point \ce then
\begin{align}
\label{main:corollary:unif-0:goal:eq}
S \combb R   \equiv_n  \ufold{S}{\mathbf{s}} \combb \ufold{R}{\mathbf{r}}  
\end{align}
Or, the following two diagrams commute where $\mycal{C}_{\mu}$ stands of the set of   fixed-point \ces.
\[\begin{tikzcd}
\mycal{C}_{\mu} \times \mycal{C} \arrow{r}{\combb} \arrow[swap]{d}{\ufold{\cdot}{\mathbf{s}}  \times \ufold{\cdot}{\mathbf{r}}} &  \mycal{C} \arrow{d}{\equiv_n} \\
\ceSetFree \times \ceSetFree  \arrow{r}{\combb} & \ceSetFree
\end{tikzcd}\quad\quad \quad
\begin{tikzcd}
\mycal{C} \times \mycal{C}_{\mu} \arrow{r}{\combb} \arrow[swap]{d}{\ufold{\cdot}{\mathbf{s}}  \times \ufold{\cdot}{\mathbf{r}}} &  \mycal{C} \arrow{d}{\equiv_n} \\
\ceSetFree \times \ceSetFree  \arrow{r}{\combb} & \ceSetFree
\end{tikzcd}
\]
\end{corollary}
\begin{proof}
  Let
\begin{align*}
  \mathbf{E}=\ufold{S}{\mathbf{s}} \combb \ufold{R}{\mathbf{r}},
\end{align*}
and assume that $S\combb R=\mu Z_1.T_1(Z_1)$ for some \ce $T_1(Z_1)$.
The key idea of the proof is to show that $\mathbf{E}$ is a fixed-point of $T_1(Z_1)$ in the sense that  $T_1(\mathbf{E})\equiv_{n}\mathbf{E}$. 
To achieve this we  take  $i=1$ in  Eq.(\ref{distance:fixed-point-in-tree:remove:fixed-point:lemma:eq}) of Lemma~\ref{main:lemma:unif}, and we get 
  \begin{align*}
\phi_{\mu}\big(\mu Z_1.T_1(Z_1)\big)   & \equiv_{\D(\mu Z_1.T_1(Z_1))}   \widehat{\phi}_{\nu}^{\mu}(T_1(Z_1)) 
  \end{align*}
  But since  $\D\big(\mu Z_1.T_1(Z_1)\big)=n$ by the Eq.(\ref{definition:omega:Omega:def:DS:1}) of  Definition~\ref{definition:omega:Omega:def} of $\D$, we get 

\begin{align}
\label{main:corollary:unif-0:proof:eq}
\phi_{\mu}\big(\mu Z_1.T_1(Z_1)\big)   & \equiv_{n}   \widehat{\phi}_{\nu}^{\mu}(T_1(Z_1)).
 \end{align}

On the one hand, by Definition~\ref{induced:morphisms:def} of $\phi_{\mu}$  together with Lemma~\ref{properties:morphism:lemma} on the properties of $\phi_{\mu}$,
it follows that  the left-hand side of Eq.(\ref{main:corollary:unif-0:proof:eq}) can be written as:
  \begin{align}
    \label{intermezzo:eq}
   \textrm{LH}.(\ref{main:corollary:unif-0:proof:eq}) = \phi_{\mu}\big(\mu Z_1.T_1(Z_1)\big) = \mathbf{E}.
  \end{align}

  On the other hand, the right-hand side of Eq.(\ref{main:corollary:unif-0:proof:eq}) can be written as:
  \begin{align*}
    \textrm{RH}.(\ref{main:corollary:unif-0:proof:eq})
    &= \widehat{\phi}_{\nu}^{\mu}(T_1(Z_1))  \\
    &= T_1( \widehat{\phi}_{\nu}^{\mu}(Z_1)) \tag{Since $Z_1$ is the only free  fixed-point variable of  $T_1(Z_1)$}\\
    &= T_1(\phi_{\nu}^{\mu}(Z_1)) \tag{Definition~\ref{induced:morphisms:def} of $\widehat{\phi}_{\nu}^\mu$ }\\
    &= T_1\big(\phi_{\mu}\big(\mu Z_1.T_1(Z_1)\big) \tag{Definition~\ref{induced:morphisms:def} of $\phi_{\nu}^\mu$} \\ 
    &= T_1(\mathbf{E}). \tag{From Eq.(\ref{intermezzo:eq})} \\ 
  \end{align*}
  Summing up, and relying on  Eq.(\ref{main:corollary:unif-0:proof:eq}), we get 
  \begin{align*}
    T_1(\mathbf{E})\equiv_{n}\mathbf{E}.
    \end{align*}
  It follows from Corollary~\ref{general-fixed-point-corollary}  that
\begin{align*}
 \mu Z_1.T_1(Z_1) \equiv_{n} \mathbf{E}.
\end{align*}
But since, by definition, we have that   $S\combb R=\mu Z_1.T_1(Z_1)$ and  $\ufold{S}{\mathbf{s}} \combb \ufold{R}{\mathbf{r}}=\mathbf{E}$,
then we get the desired result, i.e.  Eq.(\ref{main:corollary:unif-0:goal:eq}).
\end{proof}

We generalize Corollary~\ref{main:corollary:unif-0} by relaxing the assumption on the input \ces and letting them to be arbitrary instead of being  fixed-point ones.
We thus arrive at the main result of this Subsection.
\begin{proposition} 
\label{main:corollary:unif}
Let $S$ (resp. $R$)  be  a \ce  with bound   fixed-point variables $X_1,\ldots, X_s$ (resp. $Y_1,\ldots, Y_r$) and let $n \ge 1$.
Let  $\mathbf{s}:\{X_1,\ldots,X_s\} \to  \mathbb{N}$ and $\mathbf{r}:\{Y_1,\ldots,Y_r\} \to  \mathbb{N}$ be iteration  mappings with $\mathbf{s}(X_i)=\mathbf{r}(Y_j)=n$  for  $i=1,\ldots,s$ and $j=1,\ldots,r$.
Then,
\begin{align*}
S \combb R   \equiv_n  \ufold{S}{\mathbf{s}} \combb \ufold{R}{\mathbf{r}},
\end{align*}
which is illustrated by the commutative diagram below. 
\[\begin{tikzcd}
\mycal{C} \times \mycal{C} \arrow{r}{\combb} \arrow[swap]{d}{\ufold{\cdot}{\mathbf{s}}  \times \ufold{\cdot}{\mathbf{r}}} &  \mycal{C} \arrow{d}{\equiv_n} \\
\ceSetFree \times \ceSetFree  \arrow{r}{\combb} & \ceSetFree
\end{tikzcd}
\]
\end{proposition} 
\begin{proof}     
  There are  fixed-point free \ces  $S'(X^1,\ldots, X^k)$  and  $R'(Y^1,\ldots, Y^l)$, where $k\ge 1$ and $l\ge 1$,  as well as  fixed-point \ces  $\xi_1,\ldots,\xi_k$ and $\zeta_1,\ldots,\zeta_l$
  such that  $S$ and $R$ can be written as:
  \begin{align*}
    S= S'(\xi_1,\ldots,\xi_k) &&&
    R= R'(\zeta_1,\ldots,\zeta_l) 
  \end{align*}
   On the one hand, it follows  from the composition Lemma~\ref{composition:unif:lemma} that there is a fixed-point free  \ce $T(Z_1,\ldots,Z_m)$ and \ces $T_1,\ldots,T_m$, where $m\ge 1$, such that   $S \combb R$ can be written as
  \begin{align*}
    S \combb R &= T(T_1,\ldots,T_m), 
  \end{align*}
  where for any $i=1,\ldots,m$, one of the following cases holds.
\begin{enumerate}
\item  There is $j\in \set{1,\ldots,k}$  and a \ce $R^i$ that is a sub-\ce of $R$  such that
  \begin{align*}
    T_i = \xi_j \combb R^i  && \tor&&    T_i = \xi_j.    
  \end{align*}
\item   There is $j\in \set{1,\ldots,l}$  and a \ce $S^i$ that is a sub-\ce of $S$  such that
  \begin{align*}
    T_i = S^i \combb \zeta_j  && \tor&&    T_i = \zeta_j.    
  \end{align*}
  \end{enumerate}
We only discuss the first case since the second one is similar. On the other hand,
since there is a   $(\ceSet,\ceSetFree)$-quasi-simulation between $S \combb R$  and $\ufold{S}{\mathbf{s}} \combb \ufold{R}{\mathbf{r}}$  (i.e. Lemma~\ref{main:lemma:mophism:quasi})
to together  with the properties of the induced mapping $\phi_{\mu}$ (Item (\ref{properties:morphism:lemma:item:1}) of Lemma~\ref{properties:morphism:lemma}) it follows  that  $\ufold{S}{\mathbf{s}} \combb \ufold{R}{\mathbf{r}}$ can be written as
\begin{align*}
  \ufold{S}{\mathbf{s}} \combb \ufold{R}{\mathbf{r}}   &= T(\tilde{T}_1,\ldots,\tilde{T}_m)
  \end{align*} 
such that for any $i=1,\ldots,m$, we have
  \begin{align*}
 \tilde{T}_i =\phi_{\mu}(T_i) =\ufold{\xi_j}{\mathbf{s}} \combb \ufold{R^i}{\mathbf{r}}   && \tor && \tilde{T}_i= \phi_{\mu}(T_i) =\ufold{\xi_j}{\mathbf{s}}.
  \end{align*}

  If $T_i=\xi_j$  then it follows from Item (\ref{unfold:equiv:lemma:item:3}) of Lemma~\ref{unfold:equiv:lemma} that $\xi_j\equiv_{n} \ufold{\xi_j}{\mathbf{s}}$ since $\mathbf{s}(X)=n$ for any $X$ in $\set{X_1,\ldots,X_s}$.
Otherwise, if     $T_i = \xi_j \combb R^i$ then   it follows from Corollary~\ref{main:corollary:unif-0} that $ \xi_j \combb R^i \equiv_{n} \ufold{\xi_j}{\mathbf{s}} \combb \ufold{R^i}{\mathbf{r}}$.
Therefore,
  \begin{align*}
    T_i \equiv_n \tilde{T}_i, && \textrm{for } i=1,\ldots,m.
  \end{align*}
  Hence,
    \begin{align*}
   T(T_1,\ldots,T_m)  \equiv_n T(\tilde{T}_1,\ldots,\tilde{T}_m). 
    \end{align*}
Thus the desired result follows.
\end{proof}



%% file: algebraic_properties_unif_corrected.tex
\section{Proof of the main results}
\label{proof:main:results:section}

In this section we prove the main results of this paper stated in Section \ref{main:results:sec}. 
The  correctness of the unification and combination operations  for arbitrary  \ces will be proved in Subsection~\ref{correctness:unification:combination:proof}.
The algebraic properties of the unification and combination follow immediately from the correctness result, and   will be proved  in Subsection~\ref{algebraic:prop:section}.

\subsection{The correctness of the unification and combination}
\label{correctness:unification:combination:proof}
Now we are ready to prove the first main theorem of this paper regarding the correctness of the unification of \ces, Theorem~\ref{main:theorem:1}.
Its  proof  relies  mainly  on Proposition~\ref{main:corollary:unif} and on the correctness 
of the unification for the fixed-point free fragment of  \ces stated and proved in Proposition~\ref{main:proposition:unif:fixed-point-free}.

\setcounter{theorem}{\value{myvar-unif-theorem}}

\begin{theorem}[Correctness of the unification]
\label{main:theorem:1}
For every term $t \in \mycal{T}$ and for every \ces  $S$ and $R$  in $\ceSet$, 
we have that 
\begin{align*}
\Psi_t(S \combb R) & = \Psi_t(S) \combb \Psi_t(R).
\end{align*}
\end{theorem}
\begin{proof}

Let $n$ be the  depth  of $t$. Assume that $X_1,\ldots, X_s$ (resp. $Y_1,\ldots, Y_r$) are the (bound) fixed-point variables of $S$ (resp. $R$) 
 and  let $\mathbf{s}$ and $\mathbf{r}$ be iteration  mappings with $\mathbf{s}(X_i)=\mathbf{r}(Y_j)=n$, for  $i=1,\ldots,s$ and $j=1,\ldots,r$.
The proof follows from the commutativity of the following diagram.
\[\begin{tikzcd}
& \mycal{C} \times \mycal{C} \arrow[dddd, swap, bend right=99, "\Psi_t \times \Psi_t"]
                              \arrow{r}{\combb}
                              \arrow[swap]{dd}{\ufold{\cdot}{\mathbf{s}}  \times \ufold{\cdot}{\mathbf{r}}} &
                                                                                                           \mycal{C} \arrow{dd}{\equiv_n}
                                                                                                          \arrow[dddd, bend left=90, "\Psi_t"]
                                                                                                           \\
&                                                                                                           && \\
& \ceSetFree \times \ceSetFree  \arrow{r}{\combb} \arrow[swap]{dd}{\Psi_t \times \Psi_t} & \ceSetFree  \arrow{dd}{\Psi_t}\\
& \\                                                                                                           
& \mycal{E} \times \mycal{E}  \arrow{r}{\combb}  & \mycal{E} \\
\end{tikzcd}
\]
Indeed, it  follows from
Proposition  ~\ref{main:corollary:unif},
Proposition~\ref{main:proposition:unif:fixed-point-free},
Item (\ref{unfold:equiv:lemma:item:3}) of Lemma~\ref{unfold:equiv:lemma} + Item (\ref{item:3:nice:prop:Psi:lemma}) of  Lemma~\ref{nice:prop:Psi:lemma}, and
Item (\ref{item:3:nice:prop:Psi:lemma})  of  Lemma~\ref{nice:prop:Psi:lemma},  
respectively, that the following  diagrams commute.

\[
\begin{tikzcd}[scale=1.6]
\mycal{C} \times \mycal{C} \arrow{r}{\combb} \arrow[swap]{dd}{\ufold{\cdot}{\mathbf{s}}  \times \ufold{\cdot}{\mathbf{r}}} &  \mycal{C} \arrow{dd}{\equiv_n} \\
& \\
\ceSetFree \times \ceSetFree  \arrow{r}{\combb} & \ceSetFree
\end{tikzcd}
\;\;\;
\begin{tikzcd}
\ceSetFree \times \ceSetFree \arrow{r}{\combb} \arrow[swap]{dd}{\Psi_t \times \Psi_t} &  \ceSetFree \arrow{dd}{\Psi_t} \\
& \\
\mycal{E} \times \mycal{E}  \arrow{r}{\combb} & \mycal{E}
\end{tikzcd}
\;\;\;\;
\begin{tikzcd}
  \mycal{C} \times \mycal{C}  \arrow{dd}[swap]{\Psi_t \times \Psi_t}  \arrow{dr}{\ufold{\cdot}{\mathbf{s}}  \times \ufold{\cdot}{\mathbf{r}}} &   \\
  &  \ceSetFree \times \ceSetFree  \arrow{dl} {\Psi_t \times \Psi_t}\\
  \mycal{E} \times \mycal{E}
\end{tikzcd}
\;\;
\begin{tikzcd}
 & \mycal{C}   \arrow{dd}{\Psi_t}  \arrow{dl}[swap]{\equiv_n}    \\
    \ceSetFree  \arrow[swap]{dr} {\Psi_t}  &\\
& \mycal{E}
\end{tikzcd}
\]

We restate  these arguments  in the language of equations rather than the language of diagrams.
Let
 \begin{align*}
   \mathbf{S} = \ufold{S}{\mathbf{s}}  && \tand &&
   \mathbf{R} =\ufold{R}{\mathbf{r}}.
   \end{align*}
 
We have that 
\begin{align}
 S \combb R  &\equiv_n \mathbf{S} \combb \mathbf{R}                   & \tag{Proposition~\ref{main:corollary:unif}} \\
 \Psi_t(S \combb R) &= \Psi_t\big( \mathbf{S} \combb \mathbf{R}\big).   & \tag{Item  (\ref{item:3:nice:prop:Psi:lemma}) of  Lemma~\ref{nice:prop:Psi:lemma}}\\
 \Psi_t\big(S \combb R\big)& =\Psi_t(\mathbf{S}) \combb \Psi_t(\mathbf{R}). &\tag{Proposition~\ref{main:proposition:unif:fixed-point-free}, since  $\mathbf{S}$ and $\mathbf{R}$ are fixed-point free}
\end{align}

On the other hand, 
\begin{align*}
  \mathbf{S} &\equiv_n S &&\tand&   \mathbf{R} &\equiv_n R \tag{Item (\ref{unfold:equiv:lemma:item:3}) of Lemma~\ref{unfold:equiv:lemma}} \\
  \Psi_t(\mathbf{S})&= \Psi_t(S) && \tand&  \Psi_t(R)&= \Psi_t(\mathbf{R}). \tag{Item (\ref{item:3:nice:prop:Psi:lemma}) of  Lemma~\ref{nice:prop:Psi:lemma}}
\end{align*}
Therefore
\begin{align*}
  \Psi_t(S \combb R)  =  \Psi_t(S) \combb \Psi_t(R).
  \end{align*}
\end{proof}

We  can now state and  prove the second  main theorem of this paper on the correctness of the combination  of \ces.
In fact, the correctness of the combination follows from the correctness of the unification that we stated and proved in Theorem~\ref{main:theorem:1} above.

\begin{theorem}[Correctness of the combination]
\label{main:theorem:2}
For every term $t \in \mycal{T}$ and  for every  \ces  $S$ and $R$  in $\ceSetCan$,
we have that
\begin{align*}
\Psi_t(S \comb R)  =  \Psi_t(S) \comb \Psi_t(R).
\end{align*}
\end{theorem}
\begin{proof}
\begin{align*}
\Psi_t(S \comb R) 
&=       \Psi_t\big((S \combb R) \oplus S \oplus R\big) \tag{Def.~\ref{combination:def} of $\comb$} \\
&=  \Psi_t\big(\Psi_t(S \combb R) \oplus \Psi_t(S) \oplus \Psi_t(R)\big) \tag{Item (\ref{Properties-of-Psi:Lemma:item:2}) of Lemma~\ref{nice:prop:Psi:lemma} }  \\
&=    \Psi_t\Big(\big(\Psi_t(S) \combb \Psi_t(R)\big) \oplus \Psi_t(S) \oplus \Psi_t(R)\Big) \tag{Theorem~\ref{main:theorem:1}}  \\
&=     \Psi\big(\Psi_t(S) \comb \Psi_t(R)\big) \tag{Def.~\ref{combination:def} of $\comb$}\\
&=     \Psi_t(S) \comb \Psi_t(R). \tag{Item (\ref{Properties-of-Psi:Lemma:item:0}) of Lemma~\ref{nice:prop:Psi:lemma} since $\Psi_t(S) \comb \Psi_t(R)$ is position-based}
\end{align*}
\end{proof}

\subsection{The algebraic properties of the unification and combination}
\label{algebraic:prop:section}

Thanks to Theorems~\ref{main:theorem:1} and~\ref{main:theorem:2}, and using  mapping $\Psi$ (Definition~\ref{psi:def}),  we   can transfer    all the algebraic properties of the combination and
unification of position-based \ces  (stated in Propositions~\ref{main:prop:elemntary:og:prop:1} and~\ref{main:prop:elemntary:og:prop:2}) to \ces.

\begin{theorem}
\label{main:alg:theorem:1}
 The quotient set $\ceSetEquiv$ of  \ces  together with the unification   operation  enjoy the  following properties.
    \begin{enumerate}
    \item The neutral element of the   unification upon $\ceSetEquiv$  is $[@\E.\square]$. 
    \item The absorbing element of the unification is $[\emptylist]$.
    \item The unification  of \ces is  associative, i.e. $([S_1] \combb [S_2]) \combb [S_3] =  [S_1] \combb ([S_2] \combb [S_3])$, 
          for any $S_1,S_2,S_3 \in \ceSet$.
    \item  The unification of \ces is (non-)commutative if and only if  the operation of merging of contexts "$\sbullet$"  is (non-)commutative. 
    \item The unification of \ces is idempotent if and only if  the operation of merging of contexts  is idempotent,
          that is,  $[S] \combb  [S]= [S]$  for any $S \in \ceSet$  iff  $\tau \sbullet \tau = \tau$  for any contexts  $\tau$ in $\mycal{T}_{\square}$.
    \end{enumerate}
\end{theorem}
\begin{proof}
  We only prove the associativity property.
  To prove the associativity of the   unification   for \ces we rely on the associativity
  of the   unification   of position-based \ces
  (Proposition~\ref{main:prop:elemntary:og:prop:1}) together with the
  property of the function $\Psi_t$ (Theorems~\ref{main:theorem:1}).
  Let $S_1,S_2$ and $S_3$ be \ces in $\ceSet$.
  To prove $([S_1] \combb [S_2]) \combb [S_3] =  [S_1] \combb ([S_2] \combb [S_3])$ we shall prove $[(S_1 \combb S_2) \combb S_3] =  [S_1 \combb (S_2 \combb S_3)]$, i.e.  
  \begin{align*}
    S_1 \combb (S_2 \combb S_3) \equiv (S_1 \combb S_2) \combb S_3.
\end{align*}
  It follows from Item \emph{iii.)} of Lemma~\ref{nice:prop:Psi:lemma} that    it suffices to prove that, for any term $t \in \mycal{T}$, we have that
  \begin{align*}
    \Psi_t\big(S_1 \combb (S_2 \combb S_3)\big) =  \Psi_t\big((S_1 \combb S_2) \combb S_3\big).
  \end{align*}
  But this  follows from an  easy computation:
\begin{align*}
  \Psi_t\big(S_1 \combb (S_2 \combb S_3)\big)
  &= \Psi_t(S_1) \combb \Psi_t(S_2 \combb S_3) \tag{Theorem~\ref{main:theorem:2}}  \\
  &= \Psi_t(S_1) \combb (\Psi_t(S_2) \combb \Psi_t(S_3)) \tag{Theorem~\ref{main:theorem:2}} \\
  &=  (\Psi_t(S_1) \combb \Psi_t(S_2)) \combb \Psi_t(S_3) \tag{Proposition~\ref{main:prop:elemntary:og:prop:1}}\\
  &= \Psi_t(S_1 \combb S_2) \combb \Psi_t(S_3)  \tag{Theorem~\ref{main:theorem:2}} \\
  &= \Psi_t\big((S_1 \combb (S_2 \combb S_3)\big) \tag{Theorem~\ref{main:theorem:2}}.
\end{align*}
\end{proof}

The algebraic properties  of the combination  of \ces follow. They  inherit the properties of associativity, (non-)commutativity and idempotence from
the position-based \ces and the merging of contexts.

  \begin{theorem}
\label{main:alg:theorem:2}
 The quotient set $\ceSetEquiv$ of  \ces  together with the  combination  operation enjoy the following properties.
    \begin{enumerate}
    \item The neutral element of the   combination upon $\ceSetEquiv$   is $[\emptylist]$.
    \item The  combination of \ces is associative, i.e. $([S_1] \comb [S_2]) \comb [S_3] =  [S_1] \comb ([S_2] \comb [S_3])$, 
          for any $S_1,S_2,S_3 \in \ceSet$.
    \item The combination of \ces is (non-)commutative if and only if  the operation of merging of contexts $\sbullet$  is (non-)commutative. 
    \item The combination of \ces is idempotent if and only if  the operation of merging of contexts  is idempotent,
          that is,  $[S] \comb  [S]= [S]$  for any $S \in \ceSet$  iff  $\tau \sbullet \tau = \tau$  for any contexts  $\tau$ in $\mycal{T}_{\square}$.
    \end{enumerate}
\end{theorem}
\begin{proof}
Very similar to the proof of   Theorem~\ref{main:alg:theorem:1}.
\end{proof}

The congruence and non-degeneracy of the unification  and  combination are stated in the two following theorems, respectively.
\begin{theorem}[Congruence and non-degeneracy of the unification]
  \label{congruence:unif}
The following holds.
\begin{enumerate}
\item  The unification  of \ces is a congruence, that is,    for any \ces   ${S}_1,{S}_2, {S}$ in $\ceSet$, we have that:
\begin{align*} 
\textrm{If } {S}_1 \equiv {S}_2 &&\tthen&&  {S}_1 \nfcombb {S}  \equiv {S}_2 \nfcombb {S}  \;\tand\; {S} \nfcombb {S}_1  \equiv {S} \nfcombb {S}_2.
\end{align*}

\item The unification  is non-degenerate, that is, for any  \ces $[S]$ and $[S']$ in $\ceSetEquiv$,   we have that
\begin{align*}
    [S] \nfcombb [S']   =  [\emptylist]  &&\tiff &&  [S] =  [\emptylist] \;\tor\;  [S'] =  [\emptylist]. 
\end{align*}

\end{enumerate}
\end{theorem}
\begin{proof}
  We only prove the first Item. 
  On the one hand, if follows from Theorem~\ref{main:theorem:2} that
\begin{align*}
 \Psi_t(S_1  \nfcombb S) =  \Psi_t(S_1)  \combb  \Psi_t(S).
\end{align*}
On the other hand, since $S_1 \equiv S_2$, it follows from Item
\emph{iii.)} of Lemma~\ref{nice:prop:Psi:lemma} that
\begin{align*}
\Psi_t(S_1) = \Psi_t(S_2).
\end{align*}
Hence we get
\begin{align*}
 \Psi_t(S_1  \nfcombb S) &=  \Psi_t(S_2)  \combb  \Psi_t(S) \\
                       &= \Psi_t(S_2  \nfcombb S). \tag{Theorem~\ref{main:theorem:2}}
 \end{align*}
Again, from Item \emph{iii.)} of Lemma~\ref{nice:prop:Psi:lemma}, we
get
\begin{align*}
S_1  \nfcombb S \equiv S_2  \nfcombb S.
\end{align*}
The proof of the remaining claims is similar.
\end{proof}

\begin{theorem}[Congruence and non-degeneracy of the combination]
The following holds.
\begin{enumerate}
\item  The combination  of \ces is a congruence, that is,    for any \ces   ${S}_1,{S}_2, {S}$ in $\ceSet$, we have that:
\begin{align*} 
\textrm{If } {S}_1 \equiv {S}_2  &&\tthen && {S}_1 \comb {S}  \equiv {S}_2 \comb {S} \;\tand\; {S} \comb {S}_1  \equiv {S} \comb {S}_2.
\end{align*}

\item The combination  is non-degenerate, that is, for any  \ces $[S]$ and $[S']$ in $\ceSetEquiv$,   we have that
\begin{align*}
    [S] \comb [S']  =  [\emptylist]  &&\tiff &&   [S] =  [\emptylist] \;\textrm{and }\;   [S'] = [\emptylist].
\end{align*}

\end{enumerate}
\end{theorem}

\begin{proof}
Similar to the proof of Theorem~\ref{congruence:unif}.
\end{proof}


%% file: conclusion_corrected.tex
\section{Conclusion and future work}
\label{conclusion:sec}
We  addressed the  problem of extension and combination of
proofs encountered in the field of computer aided asymptotic  model derivation.
We introduced  a class of rewriting strategies on which the operations of unification and 
combination were defined and proved correct.  
The design of this class is inspired by the $\mu $-calculus formalism~\cite{rudimemt:mu-calculus:book} together with practical needs emerging from  asymptotic  model derivation.
  
The \ces are indeed modular  in the sense that they   navigate in the tree without modifying it,  then 
they insert contexts. This makes our  formalism flexible since it allows one to modify and enrich the navigation part and/or the insertion
part without disturbing the set-up. Besides, the ideas and techniques behind the unification  and combination of the navigation part, namely the unification of fixed-point \ces or recursion,  are generic and could be used in several applications beyond rewriting strategies as far as they incorporate recursion.
Although the \ces can be viewed as a finite algebraic representation  of infinite trees~\cite{COURCELLE-infinite-trees:83,CY91infinite_terms},   
our technique of unification and combination involving  $\mu$-terms and their unfolding is new.  
We envision consequences of these
results on the study of the syntactic (or modulo a theory)
unification and the pattern-matching of infinite trees once they are expressed as  $\mu$-terms in the same way we expressed the \ces. It follows that a rewriting
language that transforms algebraic infinite trees and incorporates the least and greatest fixed-point operators could  be elaborated.

We implemented the unification  procedure within a user specification language of mathematical expressions, proofs and extensions and their combination for asymptotic models.  We noticed  that  the size of the resulting \ces is big, and the good news is that they  contain many redundant and inaccessible  parts in the same way a graph or a transition system contains equivalent sub-parts, and a  program contains inaccessible code. This raises the question of the  minimization or reduction of \ces which remains open.
We managed recently to design an algorithm that decides whether two \ces are  semantically  equivalent by  looking at their structure. This is known as the word problem in other fields, e.g. in universal algebras~\cite{KNUTH-word-problem:univ:algebra:1970, word:problem:survey}.
Proving the correctness of this algorithm is under way. This semantic equivalence algorithm  will probably be useful for the minimization of \ces since  one can factorize the equivalent sub-parts.
This technique is similar to the techniques of   reduction of Petri nets and transition systems and event structures by the bisimulation equivalence relation~\cite{SchSid-atpn-Biss-reduction-petri-2000,tool:beh:equivalence:91,Branching-Bisi-Minim-Compos-2006,ARMASCERVANTES20161110}, and to  the reduction of graphs by internal isomorphisms, or automorphism ~\cite{hand-book:Graph:theory,core-graphs-SAT-19}.

Since the class of \ces can be viewed as $\mu$-calculus in the sense that it supplements elementary strategies  with the fixed-point operator, one can pose the hierarchy problem for it. The hierarchy problem asks whether,  for any $n \ge 1$, there exists a \ce with $n$ bound fixed-point variables, such that no  \ce with less than $n$ bound fixed-point variables    is equivalent to it. The hierarchy problem was posed  for many $\mu$-calculi~\cite{Gerwanger:mu,BELKHIR20106-mu} and it might  help in reducing  the size of the \ces, namely in minimizing the number of bound fixed-point variables.


%% file: appendix.tex
\section*{Appendix: proofs of Lemmas}

\addcontentsline{toc}{section}{Appendix: proofs of Lemmas}
\renewcommand\thesubsection{\textbf{\Alph{subsection}}}

\setcounter{theorem}{\value{myvar-fact-set-theory}}

\subsection{Proofs for Section~\ref{proof:correction:fixed-point-free:sec}}
\begin{fact} 
\label{set:theoretic:fact:appendix}
Let $I',J',J''$ be sets.
Then, $(I' \cap J'')\cup (I' \setminus  (J' \cup J''))=I' \setminus  J'$.
\end{fact}
\begin{proof}
\begin{align*}
(I' \cap J'')\cup (I' \setminus  (J' \cup J''))
& = \set{x\gvert x \in I' \tand x\in J''}  \cup \set{x \gvert x\in I' \tand x \notin J' \cup J''} \\
& = \set{x\gvert x \in I' \tand x\in J''}  \cup \set{x \gvert x\in I' \tand x \notin J' \tand x \notin   J''} \\
& = \set{x\gvert (x \in I' \tand x\in J'')  \tor  (x\in I' \tand x \notin J' \tand x \notin   J'')} \\
& = \set{x\gvert x \in I' \tand (x\in J''  \tor x \notin J' \tor x \notin  J'')} \\
& = \set{x\gvert x \in I' \tand   x \notin J'  }\\
& = I' \setminus  J'.
\end{align*}
\end{proof}

\subsection{Proofs for Section~\ref{correction-combination-definitions:section}}

\setcounter{theorem}{\value{counter:unfold:equiv:lemma}}

\begin{lemma}
\label{unfold:equiv:lemma:appendix}
Let $S$ be a \ce with (bound) fixed-point variables $X_1,\ldots,X_s$ and  let $\mathbf{s}:\{X_1,\ldots,X_s\} \to  \mathbb{N}$  be an iteration  mapping.
\begin{enumerate}[{(i)}]
\item \label{unfold:equiv:lemma:item:4:appendix} If $S$ if a fixed-point \ce, say $\mu X.S'(X)$ with $X \in \{X_1,\ldots,X_s\}$, then there exists a fixed-point free \ce $\tilde{S}(X^1,\ldots,X^m)$ with $m\ge 1$, and \ces $S_1,\ldots,S_{m-1},S_m(X)$ such that for any $n\ge 1$,
  \begin{align}
    \mu^n X.S'(X)                     &=\tilde{S}\Big(S_1,\ldots,S_{m-1}, S_m\big(\mu^{n-1} X.S'(X)\big)\Big)  \label{unfold:equiv:lemma:eq:1:appendix}\\
      \ufold{\mu X.S'(X)}{\mathbf{s}} &=\tilde{S}\Big(\ufold{S_1}{\mathbf{s}},\ldots,\ufold{S_{m-1}}{\mathbf{s}},\bufold{S_m(\mu X.S'(X))}{\mathbf{s}'}\Big) \label{unfold:equiv:lemma:eq:2:appendix}
  \end{align}
where $\mathbf{s}'$ is the iteration mapping defined on $\set{X_1,\ldots,X_s}$ by  $\mathbf{s}'(X)=\mathbf{s}(X)-1$ and $\mathbf{s}'(X')=\mathbf{s}(X')$ for  $X'\neq X$.
\item \label{unfold:equiv:lemma:item:3:appendix} If   $\mathbf{m}=\min\set{\mathbf{s}(X_1),\ldots,\mathbf{s}(X_s)}$, then $S \equiv_{\mathbf{m}} \ufold{S}{\mathbf{s}}$. 
\end{enumerate}
\end{lemma}  
\begin{proof} 
  For Item (\ref{unfold:equiv:lemma:item:4:appendix}), indeed $S'(X)$ can be written in terms of its immediate fixed-point sub-\ces where $X$ appears free in one of them since $X$ appears once in $S'(X)$.  That is, there exists a fixed-point free \ce  $\tilde{S}(X^1,\ldots,X^m)$ with $m\ge 1$, and (fixed-point) \ces $S_1,\ldots,S_{m-1},S_m(X)$ such that $S'(X)$ can be written as $S'(X)=\tilde{S}(S_1,\ldots,S_{m-1},S_m(X))$.
  To show Eq.(\ref{unfold:equiv:lemma:eq:1:appendix})  we rely on the fact that  $\tilde{S}(X^1,\ldots,X^m)$ is fixed-point free and 
  on  Definition~\ref{ufold:def}   of unfolding together with  a  simple structural induction on $\tilde{S}(X^1,\ldots,X^m)$. The computations are straightforward and we don't make them.
  To show Eq.(\ref{unfold:equiv:lemma:eq:2:appendix}),
  let $\mathbf{\tilde{s}}$ be the iteration mapping defined  on $\set{X_1,\ldots,X_s}$ as  the restriction of $\mathbf{s}$ on  $\set{X_1,\ldots,X_s}\setminus\set{X}$.
  Since $\tilde{S}(X^1,\ldots,X^m)$ is fixed-point free, then
  by the definition~\ref{ufold:def}   of unfolding and making use of Eq.(\ref{unfold:equiv:lemma:eq:1:appendix})  we get 
\begin{align}
    \ufold{\mu X.S'(X)}{\mathbf{s}}  
    & =\bufold{\mu X.\tilde{S}(S_1,\ldots,S_{m-1},S_m(X))}{\mathbf{s}}  \notag   \\
    &=\mu^{\mathbf{s}(X)} X.\bufold{\tilde{S}(S_1,\ldots,S_{m-1},S_m(X))}{\mathbf{s}} \notag   \\
    &=\mu^{\mathbf{s}(X)} X.\tilde{S}\Big(\ufold{S_1}{\mathbf{s}},\ldots,\ufold{S_{m-1}}{\mathbf{s}},\bufold{S_m(X)}{\mathbf{s}}\Big)   \notag \\
    &=\mu^{\mathbf{s}(X)} X.\tilde{S}\Big(\ufold{S_1}{\mathbf{s}},\ldots,\ufold{S_{m-1}}{\mathbf{s}},\bufold{S_m(X)}{\mathbf{\tilde{s}}}\Big)  \notag\\
    &=\tilde{S}\Big(\ufold{S_1}{\mathbf{s}},\ldots,\ufold{S_{m-1}}{\mathbf{s}},\bbufold{S_m\big(\mu^{\mathbf{s}(X)-1} X.S'(X)\big)}{\mathbf{\tilde{s}}}\Big) \tag{By Eq.(\ref{unfold:equiv:lemma:eq:1:appendix})}\\
    &=\tilde{S}\Big(\ufold{S_1}{\mathbf{s}},\ldots,\ufold{S_{m-1}}{\mathbf{s}},\bufold{S_m(\mu X.S'(X))}{\mathbf{s}'}\Big). \notag
  \end{align}
To show the Item (\ref{unfold:equiv:lemma:item:3:appendix}) we next  generalize the idea that a \ce $\mu^n X.S'(X)$ could be written as $S'(S'(\cdots (\emptylist))$ where the number of jumps between its root and $\emptylist$ is at least $n$, as well as the fact that $\mu X.S'(X)$ could be written as  $S'(S'(\cdots (\mu X.S'(X))))$.
Technically, we rely on Eq.(\ref{depth:position:composition:lemma:eq'}) and  we shall show that there exist \ces $\xi_1,\ldots,\xi_m$ and a fixed-point free \ce $T(X^1,\ldots,X^m)$ and an unravelling $\unravel{\cdot}$ of $S$,
such that   $\unravel{S}$ and $\rho(S)$ can be written as
\begin{align}
  \unravel{S}            &= T(\xi_1,\ldots,\xi_m)  \label{unfold:equiv:lemma:item:3:eq:1} \\
  \ufold{S}{\mathbf s}  &= T(\emptylist,\ldots,\emptylist)  \label{unfold:equiv:lemma:item:3:eq:2}
\end{align}
such that
\begin{align}
\min\bset{\Pi_{X^i}\big(T(X^1,\ldots,X^m) \big) \gvert i=1,\ldots,m} \ge \mathbf{m}.  \label{unfold:equiv:lemma:item:3:eq:3}
\end{align}
We make a double induction: the outer one being on  $\Im(\mathbf{s})\uberEq{def}(\mathbf{s}(X_1),\ldots,\mathbf{s}(X_s))$ with the lexicographic order, and the inner one being on the number of nested fixed-point sub-\ces of  $S$, i.e. on $\mathbf{h}(S)$ the star height of $S$.
The outer base case  when $\Im(\mathbf{s})=(0,\ldots,0)$  holds trivially since in this case the set of terms of depth $0$ is empty.
For the outer induction step,  we assume that the claim  holds for $\mathbf{s}'$ and we shall prove it for any $\mathbf{s}$ with $\Im(\mathbf{s})=\Im(\mathbf{s}')+(b_1,\ldots,b_s)$ where  there is $i\in \set{1,\ldots,s}$
such that $b_i=1$ and $b_j=0$ for any $i \neq j$. We make an inner induction on $\mathbf{h}(S)$.
The inner base case $\mathbf{h}(S)=0$  holds trivially since in this case $S$ is fixed-point free because the unfolding of $S$ is  $S$.
For the inner induction step we assume that the claim holds for a \ce $S'$ and we shall prove it for any $S$ with $\mathbf{h}(S)= \mathbf{h}(S')+1$.
We only discuss the case when $S$ if a fixed-point \ce, say $S=\mu X.S'(X)$, since the case when $S$ is of the form $S=\tilde{S}(\xi_1,\ldots,\xi_k)$, for    a fixed-point free \ce $\tilde{S}(X^1,\ldots,X^k)$ and  a fixed-point \ces
$\xi_1,\ldots,\xi_k$ with $k \ge 1$, does not provide difficulties since it is easily reducible to the case under discussion, because
$\bufold{\tilde{S}(\xi_1,\ldots,\xi_k)}{\mathbf{s}}=\tilde{S}\big(\ufold{\xi_1}{\mathbf s},\ldots,\ufold{\xi_k}{\mathbf s}\big)$ .  We rely on the fact that $S'(X)$ can be written as $S'(X)=\tilde{S}(S_1,\ldots,S_{m-1},S_m(X))$, for 
fixed-point \ces $S_1,\ldots,S_{m-1},S_m(X)$   in $\widetilde{\Phi}_{\mu}(S)$,  $\tilde{S}(X^1,\ldots,X^m)$ being a fixed-point free \ce.
From  Eq.(\ref{unfold:equiv:lemma:eq:1:appendix}) above  we have that $\ufold{S}{\mathbf{s}}=\ufold{\mu X.S'(X)}{\mathbf{s}}=\tilde{S}\Big(\ufold{S_1}{\mathbf{s}},\ldots,\ufold{S_{m-1}}{\mathbf{s}},\bufold{S_m(\mu X.S'(X))}{\mathbf{s}'}\Big)$, where  $\mathbf{s}'$ is the iteration mapping defined on $\set{X_1,\ldots,X_s}$ by  $\mathbf{s}'(X)=\mathbf{s}(X)-1$ and $\mathbf{s}'(X')=\mathbf{s}(X')$ for  $X'\neq X$.

Therefore, we have that
\begin{align*}
  S                     &= \mu X.\tilde{S}(S_1,\ldots,S_{m-1},S_m(X))  \\
    \unravel{S}         &\uberEq{def}  \tilde{S}\big(S_1,\ldots,S_{m-1},S_m(\mu X.S'(X))\big) \\
  \ufold{S}{\mathbf s}  &= \tilde{S}\Big(\ufold{S_1}{\mathbf{s}},\ldots,\ufold{S_{m-1}}{\mathbf{s}},\bufold{S_m(\mu X.S'(X))}{\mathbf{s}'}. 
\end{align*}
On the one hand,  then it follows from inner  induction hypothesis that the claims (\ref{unfold:equiv:lemma:item:3:eq:1}), (\ref{unfold:equiv:lemma:item:3:eq:2}) and (\ref{unfold:equiv:lemma:item:3:eq:3})
hold for  $S_i$ with respect to $\ufold{S_i}{\mathbf{s}}$  for $i=1,\ldots,m-1$, since $\mathbf{h}(S_i) < \mathbf{h}(S)$. 
On the other hand, since $\mathbf{s}'(X)=\mathbf{s}(X)-1$ and $\mathbf{s}'(X')=X'$ for any $X'\neq X$, then it follows from the outer induction hypothesis  that there is a fixed-point free \ce
$\tilde{S}_k(Y^1,\ldots,Y^k)$ and \ces $\zeta_1,\ldots,\zeta_k$ such that

\begin{align}
  \bunravel{S_m(\mu X.S'(X))}            &= \tilde{S}_m(\zeta_1,\ldots,\zeta_k)  \label{unfold:equiv:lemma:item:3:eq:1'} \\
  \ufold{S_m(\mu X.S'(X))}{\mathbf s'}  &= \tilde{S}_m(\emptylist,\ldots,\emptylist)  \label{unfold:equiv:lemma:item:3:eq:2'}
\end{align}
such that
\begin{align}
  \min\bset{\Pi_{Y^i}\big(\tilde{S}_m(Y^1,\ldots,Y^k) \big) \gvert i=1,\ldots,k} \ge \mathbf{m}'  \label{unfold:equiv:lemma:item:3:eq:3'}
\end{align}

 where $\mathbf{m'}=\min\set{\mathbf{s}'(X_i) \gvert i=1,\ldots,s}$.
If $\mathbf{m} > \mathbf{s}(X)$,  then  we are done since in this case $\mathbf{m}'=\mathbf{m}$.
Otherwise, if $\mathbf{m} = \mathbf{s}(X)$ then $\mathbf{m}' = \mathbf{s}(X)-1$.  But since  $\mu X.S'(X)$ is monotonic then $\Pi_{X}(S'(X))\ge 1$.
That is, there is at least one jump between the root of  $S'(X)$ and $X$.
This jump is either  between the root of $\tilde{S}(S_1,\ldots,S_{m-1},X^m)$ and $X^{m}$, i.e.  $\Pi_{X^m}\big(\tilde{S}(S_1,\ldots,S_{m-1},X^m)\big)\ge 1$  and in this case we are done;
or between the root of $S_m(X)$ and $X$, i.e.  $\Pi_{X}\big(S_m(X)\big)\ge 1$  and in this case we can assume without loss of generality that $X$ is an immediate sub-\ce of $S_m(X)$,
say $\mu X.S'(X)=\zeta_k$, and thus we get the desired result since $\Pi_{X^m}\big(T(X^1,\ldots,X^m) \big) \le  1+ \Pi_{Y^k}\big(T(X^1,\ldots,S_m(Y^k)) \big)$.
\end{proof}


\setcounter{theorem}{\value{counter:composition:unif:lemma}}

\begin{lemma}[Composition Lemma] 
\label{composition:unif:lemma:appendix}
Let $S$ and $R$ be \ces.
Assume that  there are  fixed-point free \ces  $S'(X_1,\ldots, X_k)$  and  $R'(Y_1,\ldots, Y_l)$, where $k\ge 1$ and $l\ge 1$, and
 \ces  $\xi_1,\ldots,\xi_k$  where $\xi_i \in \Phi(S)$, and
 \ces  $\zeta_1,\ldots,\zeta_l$ where $\zeta_i \in \Phi(R)$, such that  $S$ and $R$ can be written as:
  \begin{align*}
    S= S'(\xi_1,\ldots,\xi_k) &&&
    R= R'(\zeta_1,\ldots,\zeta_l). 
  \end{align*}

Then, there is a fixed-point free \ce $T(Z_1,\ldots,Z_m)$  and \ces $T_1,\ldots,T_m$, where $m \ge 1$, such that 
  \begin{align*}
    S \combb R &= T(T_1,\ldots,T_m)
  \end{align*}
  where for any $i=1,\ldots,m$, there is an alternative between the two following choices.
  
  \begin{enumerate}[(a)]
  \item \label{composition:unif:lemma:item:2:a:appendix} There are $j \in \set{1,\ldots,k}$, a \ce  $R^i(Y^1,\ldots, Y^{l'})$ that is  a sub-\ce of $R'(Y_1,\ldots, Y_l)$ with $l'\le l$, and a set of \ces
    $\set{\zeta^1,\ldots,\zeta^{l'}} \subseteq \set{\zeta_1,\ldots,\zeta_{l}}$   such that 
    \begin{align}
      \label{composition:lemma:prop:1:appendix}
    T_i = \xi_j \combb R^i(\zeta^1,\ldots, \zeta^{l'})   &&\tor &&      T_i = \xi_j.
  \end{align}
\item \label{composition:unif:lemma:item:2:b:appendix} There are $j \in \set{1,\ldots,l}$, a \ce  $S^i(X^1,\ldots, X^{k'})$ that is  a sub-\ce of $S'(X_1,\ldots, X_k)$ with $k'\le l$, and a set of \ces
    $\set{\xi^1,\ldots,\xi^{k'}} \subseteq \set{\xi_1,\ldots,\xi_k}$   such that 

  \begin{align}
    \label{composition:lemma:prop:2:appendix}
    T_i = S^i(\xi^1,\ldots,\xi^{k'}) \combb \zeta_j    &&\tor &&      T_i = \zeta_j.
      \end{align}
  \end{enumerate}
\end{lemma}
\begin{proof} 
The proof is by structural induction on the  fixed-point free \ces  $S'(X^1,\ldots, X^k)$  and  $R'(Y^1,\ldots, Y^l)$.
  The base case is when $k=l=1$ and $S'(X_1)=X_1$ and $R'(Y_1)=Y_1$. In this case we have $S'(\xi_1)=\xi_1$ and $R'(\zeta_1)=\zeta_1$. The result is obvious since $S \combb R=\xi_1 \combb \zeta_1$.
  For the  induction step assume that the claim holds for some \ces $S''$ and $R''$, and we shall show it for any $S$ and $R$ such that either
  \emph{(i)} $S''$ is an immediate sub-\ce of $S$ and $R''=R$, or
  \emph{(ii)} $S=S''$ and $R''$ is an immediate sub-\ce of $R$, or
  \emph{(iii)} $S''$ (resp. $R''$) is an immediate sub-\ce of $S$ (resp. $R$).
  The proof  is not hard and involves  straightforward computations.  We only elucidate the case when $S$ is a pattern-matching and $R$ is arbitrary, and the case when both $S$ and $R$ are $\most$ \ces. The remaining cases fall into one of these two.\\
$\bullet$ If $S'(X_1,\ldots, X_k) = u;S''(X^1,\ldots, X^k)$ and $R$ is arbitrary, then in this case
  \begin{align*}
    S\combb R &= S'(\xi_1,\ldots,\xi_k) \combb  R'(\zeta_1,\ldots,\zeta_l)  \notag \\
    &= u;\big(S''(\xi_1,\ldots,\xi_k) \combb R'(\zeta_1,\ldots,\zeta_l)\big).  \notag \\
  \end{align*}
  From the induction hypothesis it follows  that there is a fixed-point free \ce $T'(Z_1,\ldots,Z_m)$ and \ces $T'_1,\ldots,T'_m$ with the right properties (\ref{composition:lemma:prop:1:appendix}) and (\ref{composition:lemma:prop:2:appendix}) such that $S''(\xi_1,\ldots, \xi_k)\combb R'(\zeta_1,\ldots, \zeta_l)=T'(T'_1,\ldots,T'_m)$.
  By  letting $T(Z_1,\ldots,Z_m)= u ; T'(Z_1,\ldots,Z_m)$ we get the desired result.
\\
$\bullet$ If $S'(X_1,\ldots, X_k)=\most(S''(X_1,\ldots, X_k))$ and $R(X_1,\ldots, X_k)=\most(R''(X_1,\ldots, X_k))$, then

 \begin{align}
   \label{most:morphism:eq:1}
   S \combb R   &= S'(\xi_1,\ldots, \xi_k) \combb R'(\zeta_1,\ldots, \zeta_l)\notag \\
                &= \most\big(S''(\xi_1,\ldots, \xi_k)\big)  \combb \most\big(R''(\zeta_1,\ldots, \zeta_l)\big) \notag \\
                &= \mathbf{If} \,\Big(\most\big(S''(\xi_1,\ldots, \xi_k)\big)\, \& \,\most{\big(R''(\zeta_1,\ldots, \zeta_l)\big)}\Big) \notag  \\
                & \;\;\;\;\;\mathbf{Then}\;            \most\Big(\big(S''(\xi_1,\ldots, \xi_k)\combb R''(\zeta_1,\ldots, \zeta_l)\big)  \oplus S''(\xi_1,\ldots, \xi_k) \oplus R''(\zeta_1,\ldots, \zeta_l) \Big).     \tag{Rule~\ref{most:ext:1}} 
 \end{align}
 On the one hand, it follows from the induction hypothesis that there is a fixed-point free \ce $T'(Z'_1,\ldots,Z'_m)$ and \ces $T'_1,\ldots,T'_m$ with the right properties (\ref{composition:lemma:prop:1:appendix}) and (\ref{composition:lemma:prop:2:appendix}) such that $S''(\xi_1,\ldots, \xi_k)\combb R''(\zeta_1,\ldots, \zeta_l)=T'(T'_1,\ldots,T'_m)$.
 One the other hand, let $T(Z_1^1,\ldots,Z_k^1,Z^2_1,\ldots,Z^2_l,Z'_1,\ldots,Z'_m,Z_1^3,\ldots,Z_k^3,Z^4_1,\ldots,Z^4_l)$  the fixed-point free \ce with free variables
 $Z_1^1,\ldots,Z_k^1,Z^2_1,\ldots,Z^2_l,Z'_1,\ldots,Z'_m,Z_1^3,\ldots,Z_k^3,Z^4_1,\ldots,Z^4_l$ defined as follows:
   \begin{multline*}
  T(Z_1^1,\ldots,Z_k^1,
   Z^2_1,\ldots,Z^2_l,
   Z'_1,\ldots,Z'_m,
   Z_1^3,\ldots,Z_k^3, 
   Z^4_1,\ldots,Z^4_l)= \\
      \begin{split}
      &\mathbf{If} \big(\most(S''(Z_1^1,\ldots,Z_k^1)) \,\&\, \most(R''(Z_1^2,\ldots,Z^2_l))\big) \\
      &\mathbf{Then}\; \most\big(T'(Z'_1,\ldots,Z'_m)  \oplus S''( Z_1^3,\ldots,Z_k^3) \oplus R''(Z^4_1,\ldots,Z^4_l)\big)
\end{split}
   \end{multline*}
   Let $T_i^1,T_i^3,T_j^2, T_j^4$ be the \ces defined by  
   \begin{align*}
     T_i^1&=  T_i^3 = \xi_i, & \textrm{ for } i=1,\ldots, k  \\
     T_j^2&=  T_j^4 = \zeta_j, & \textrm{ for } j=1,\ldots, l
   \end{align*}
   Therefore, $S \combb R$ can be written as
   \begin{align*}
     S \combb R =   T(T_1^1,\ldots,T_k^1,
     T^2_1,\ldots,T^2_l,
     T'_1,\ldots,T'_m,
     T_1^3,\ldots,T_k^3, 
     T^4_1,\ldots,T^4_l)
   \end{align*}
   which satisfies the properties (\ref{composition:lemma:prop:1:appendix}) and (\ref{composition:lemma:prop:2:appendix}). 
\end{proof}   

\subsection{Proofs for Section~\ref{correction:unif:general:setting:unfold:sec}}

\setcounter{theorem}{\value{counter:comparing:unif:unfolding:lemma}}

\begin{lemma}
\label{comparing:unif:unfolding:lemma:appendix}
There exist fixed-point free \ces $T_1,\ldots,T_m, T(Z_1,\ldots,Z_m)$, where each $Z_i$ is a free fixed-point variable and $m\ge 1$, such that
$\ufold{S}{\mathbf{s}_1}\combb \ufold{R}{\mathbf{r}_1}$
and $ \ufold{S}{\mathbf{s}_2}\combb \ufold{R}{\mathbf{r}_2}$ can be written as  
\begin{align*}
  \ufold{S}{\mathbf{s}_1}\combb \ufold{R}{\mathbf{r}_1} &= T(T_1,\ldots,T_m)              \\
 \ufold{S}{\mathbf{s}_2}\combb \ufold{R}{\mathbf{r}_2} &=  T(\emptylist,\ldots,\emptylist).
\end{align*}
\end{lemma}
\begin{proof} 
  The proof is by induction on $\Gamma(S,R,\mathbf{s}_1,\mathbf{s}_2,\mathbf{r}_1,\mathbf{r}_2) \uberEq{def}\big(\delta(\ufold{S}{\mathbf{s}_1}),\delta(\ufold{S}{\mathbf{s}_2}), \delta(\ufold{R}{\mathbf{r}_1}),\delta(\ufold{R}{\mathbf{r}_2}) \big)$.
  The base case is when $\Gamma(S,R,\mathbf{s}_1,\mathbf{s}_2,\mathbf{r}_1,\mathbf{r}_2)=(0,0,0,0)$, i.e. $S$ and $R$ are either $\emptylist$ or $@\varepsilon.\tau$. This case is trivial.
  For the induction step, assume that the claim holds for \ces $\tilde{S},\tilde{R}$ and iteration mappings $\mathbf{\tilde{s}}_1, \mathbf{\tilde{s}}_2, \mathbf{\tilde{r}}_1,  \mathbf{\tilde{r}}_2$,
  and we shall prove it for any  \ces ${S},{R}$ and iteration mappings $\mathbf{{s}}_1, \mathbf{{s}}_2, \mathbf{{r}}_1,  \mathbf{{r}}_2$ where
  $\Gamma(S,R,\mathbf{s}_1,\mathbf{s}_2,\mathbf{r}_1,\mathbf{r}_2)=\Gamma(\tilde{S},\tilde{R},\mathbf{\tilde{s}}_1,\mathbf{\tilde{s}}_2,\mathbf{\tilde{r}}_1,\mathbf{\tilde{r}}_2) +(1,1,0,0)$,
  or   $\Gamma(S,R,\mathbf{s}_1,\mathbf{s}_2,\mathbf{r}_1,\mathbf{r}_2)=\Gamma(\tilde{S},\tilde{R},\mathbf{\tilde{s}}_1,\mathbf{\tilde{s}}_2,\mathbf{\tilde{r}}_1,\mathbf{\tilde{r}}_2) +(0,0,1,1)$, or
    $\Gamma(S,R,\mathbf{s}_1,\mathbf{s}_2,\mathbf{r}_1,\mathbf{r}_2)=\Gamma(\tilde{S},\tilde{R},\mathbf{\tilde{s}}_1,\mathbf{\tilde{s}}_2,\mathbf{\tilde{r}}_1,\mathbf{\tilde{r}}_2) +(1,1,1,1)$.
  We only discuss the cases when $\delta(\ufold{S}{\mathbf{s}_2}) \ge 1$ (and hence  $\delta(\ufold{S}{\mathbf{s}_1}) \ge 1$ since $\mathbf{s}_1 \ge \mathbf{s}_2$)
  and $\delta(\ufold{R}{\mathbf{r}_2}) \ge 1$ (and hence $\delta(\ufold{R}{\mathbf{r}_1}) \ge 1$ since  $\mathbf{r}_1 \ge \mathbf{r}_2$),
  because the cases when $\delta(\ufold{S}{\mathbf{s}_2}) = 0$ or $\delta(\ufold{S}{\mathbf{s}_2}) = 0$ (but not both)  are just a particular case of the   general case that follows, and  can be handled similarly
  using the composition Lemma ~\ref{composition:unif:lemma}. 
  We distinguish two cases depending on $S$ and $R$.\\
  \underline{First case.} If neither $S$ nor $R$ is a fixed-point \ce, then there exist  fixed-point free \ces $S'(X^1,\ldots,X^k)$ and  $R'(Y^1,\ldots,Y^l)$ and \ces $S_1,\ldots,S_k$ and   $R_1,\ldots,R_l$, where each $S_i$ (resp. $R_i$) is an immediate sub-\ce of $S$ (resp. $R$), i.e. $\delta(S'(X^1,\ldots,X^k))=1$  (resp. $\delta(R'(Y^1,\ldots,Y^l))=1$),  such that $S$ and $R$ can be written as:
  \begin{align*}
    S &= S'(S_1,\ldots,S_k) \\
    R &= R'(R_1,\ldots,R_l).
  \end{align*}
  Hence,
  \begin{align}
    \label{eq:14:local}
    \begin{cases}
    \ufold{S}{\mathbf{s}_1} &= S'(\ufold{S_1}{\mathbf{s}_1},\ldots,\ufold{S_k}{\mathbf{s}_1}) \\
    \ufold{R}{\mathbf{r}_1} &= R'(\ufold{R_1}{\mathbf{r}_1},\ldots,\ufold{R_l}{\mathbf{r}_1}),
    \end{cases}
    &&\tand&&
    \begin{cases}
    \ufold{S}{\mathbf{s}_2} &= S'(\ufold{S_1}{\mathbf{s}_2},\ldots,\ufold{S_k}{\mathbf{s}_2}) \\
    \ufold{R}{\mathbf{r}_2} &= R'(\ufold{R_1}{\mathbf{r}_2},\ldots,\ufold{R_l}{\mathbf{r}_2}).
    \end{cases}
    \end{align}
  It follows from the composition Lemma~\ref{composition:unif:lemma} that there exist  a fixed-point free \ce $T(Z_1,\ldots,Z_m)$, and \ces  $T^1_1,\ldots,T^1_m$ and $T^2_1,\ldots,T^2_m$  such that
  \begin{align*}
    \ufold{S}{\mathbf{s}_1} \combb  \ufold{R}{\mathbf{r}_1} &= T(T^1_1,\ldots,T^1_m) \\
    \ufold{S}{\mathbf{s}_2} \combb  \ufold{S}{\mathbf{s}_2} &= T(T^2_1,\ldots,T^2_m),
  \end{align*}
    where the Item (\ref{composition:unif:lemma:item:2:a:appendix})  or  (\ref{composition:unif:lemma:item:2:b:appendix}) holds. We only discuss the first possibility (since the second is symmetric)
  according to which, for any $i=1,\ldots,m$    there is $j \in \set{1,\ldots,k}$, and a \ce  $R^i(Y_1,\ldots, Y_{l'})$ that is  a sub-\ce of $R'(Y_1,\ldots, Y_l)$ with $l'\le l$, and a set of \ces
    $\set{R^1_1,\ldots,R^1_{l'}} \subseteq \set{R_1,\ldots,R_{l}}$   such that 
    \begin{align}
      T^1_i = \ufold{S_j}{\mathbf{s}_1} \combb R^i(R^1_1,\ldots,R^1_{l'})   &&\tor &&      T^1_i = \ufold{S_j}{\mathbf{s}_1} \\
      T^2_i = \ufold{S_j}{\mathbf{s}_1} \combb R^i(R^2_1,\ldots,R^2_{l'})   &&\tor &&      T^2_i = \ufold{S_j}{\mathbf{s}_2}.
    \end{align}
    If  $T^1_i = \ufold{S_j}{\mathbf{s}_1}$ and hence $T^2_i = \ufold{S_j}{\mathbf{s}_2}$, then the claim follows from Remark~\ref{double:ufold:rq}. 
    Otherwise, the claim follows from the induction hypothesis. 
\\
\underline{Second case.} If $S$ if a fixed-point \ce, say $\mu X.\tilde{S}(X)$, then we distinguish two cases.
If $\mathbf{s}_2(X)=0$ then the claim holds trivially since in this case $\ufold{S}{\mathbf{s}_2}=\emptylist$. If $\mathbf{s}_2(X) >0$ and therefore $\mathbf{s}_1(X) >0$ since $\mathbf{s}_1 \ge \mathbf{s}_2$, then it follows from Eq.(\ref{unfold:equiv:lemma:eq:2:appendix}) of Lemma~\ref{unfold:equiv:lemma}, that there exists a fixed-point free \ce $S'(X^1,\ldots,X^k)$ and \ces $S_1,\ldots,S_{k-1},S_k(X)$, where $\delta(S'(X^1,\ldots,X^k))=0$, such that $\ufold{\mu X.\tilde{S}(X)}{\mathbf{s}_1}$ and
 $\ufold{\mu X.\tilde{S}(X)}{\mathbf{s}_2}$ can  be written as
  \begin{align*}
    \ufold{\mu X.\tilde{S}(X)}{\mathbf{s}_1} &=S'\Big(\ufold{S_1}{\mathbf{s}_1},\ldots,\ufold{S_{k-1}}{\mathbf{s}_1},\bufold{S_k(\mu X.\tilde{S}(X))}{\mathbf{s}'_1}\Big)\\
    \ufold{\mu X.\tilde{S}(X)}{\mathbf{s}_2} &=S'\Big(\ufold{S_1}{\mathbf{s}_2},\ldots,\ufold{S_{k-1}}{\mathbf{s}_2},\bufold{S_k(\mu X.\tilde{S}(X))}{\mathbf{s}'_2}\Big) 
  \end{align*}
  where $\mathbf{s}'_v(X)=\mathbf{s}'_v(X)-1$ and $\mathbf{s}'_v(X')=\mathbf{s}_v(X')$ for $X' \neq X$, for $v=1,2$.
  Thus the  reasoning is very similar to the one made in the first case, more precisely  it is done by taking Eq.(\ref{eq:14:local}) in which  we replace  $\ufold{S_k}{\mathbf{s}_v}$ 
  by  $\bufold{S_k(\mu X.\tilde{S}(X))}{\mathbf{s}'_v}$,  for  $v=1,2$.  Thus the induction hypothesis can be applied as well.
\end{proof}


\setcounter{theorem}{\value{counter:main:lemma:mophism:quasi}}

\begin{lemma}
  \label{main:lemma:mophism:quasi:appendix}
  Let $S$ and $R$ be  \ces with  bound  fixed-point variables $\boundv{S}=\set{ X_1,\ldots,X_s}$ and $\boundv{R}=\set{Y_1,\ldots,Y_r}$.
  Let $\Eu{M} \in \mathfrak{M}(S,R)$ be a memory  with respect  to $S$ and $R$.
  Let  $\mathbf{s}:\{X_1,\ldots,X_s\} \to  \mathbb{N}$ and $\mathbf{r}:\{X_1,\ldots,X_r\} \to  \mathbb{N}$ be iteration mappings.
There is a $(\ceSet,\ceSetFree)$-quasi-simulation $\morphyy$  between   $\NF(\tuple{S,R,\Eu{M}})$ and  $\NF(\tuple{\ufold{S}{\mathbf{s}},\ufold{R}{\mathbf{r}},\emptyset})$.
In particular, the following diagram commutes.
\[\begin{tikzcd}
\mycal{C} \times \mycal{C} \arrow{r}{\combb} \arrow[swap]{d}{\ufold{\cdot}{\mathbf{s}}  \times \ufold{\cdot}{\mathbf{r}}} &  \mycal{C} \arrow[dash]{d}{\morphyy} \\
\ceSetFree \times \ceSetFree  \arrow{r}{\combb} & \ceSetFree
\end{tikzcd}
\]
\end{lemma}
\begin{proof}
We  make of use of Lemma~\ref{composition:unif:lemma}.
The proof is by structural induction on $\ufold{S}{\mathbf s}$ and $\ufold{\mathbf r}{\mathbf r}$,  according to which the $(\ceSet,\ceSetFree)$-simulation  $\morphyy$ will be inductively constructed.
The base case holds trivially. For the induction step we assume that the claim holds for \ces   $\ufold{S''}{\mathbf s''}$ and $\ufold{R''}{\mathbf r''}$ and we shall prove  for any \ces $\ufold{S}{\mathbf s}$ and $\ufold{R}{\mathbf r}$
such that either
\emph{(i)}  $\ufold{S''}{\mathbf s''}$ is an immediate sub-\ce  of  $\ufold{S}{\mathbf s}$ and $R''=R$ and $\mathbf{r}''=\mathbf{\mathbf r}$, or
\emph{(ii)} $\ufold{R''}{\mathbf r''}$ is an immediate sub-\ce  of  $\ufold{S}{\mathbf r}$ and $S''=S$ and $\mathbf{s}''=\mathbf{\mathbf s}$, or
\emph{(iii)}  $\ufold{S''}{\mathbf s''}$ (resp. $\ufold{R''}{\mathbf r''}$) is an immediate sub-\ce  of  $\ufold{S}{\mathbf s}$ (resp. $\ufold{S}{\mathbf r}$).
We   distinguish three cases depending on $S$ and $R$:
\begin{enumerate}
\item If $S$ and $R$ are fixed-point free, then this case is trivial since $\ufold{S}{\mathbf{s}}=S$ and $\ufold{R}{\mathbf{r}}=R$.
\item \label{case:2:lemma:unif:morphism:appendix} If $S$ and $R$ are of the form $S=S'(S_1,\ldots,S_k)$ and $R=R'(R_1,\ldots,R_l)$ for fixed-point free \ces $S'(X_1,\ldots,X_k)$ and $R'(Y_1,\ldots,Y_l)$,
  i.e. $S=u;S'$ or  $S=S'\oplus S''$  or  $S=\most(S')$ or  $S=\tifthen{S'}{S''}$ or   $S=\bigand_{i=1,k}@p_i.S_i$ and similarly for $R$,  then the result follows immediately from Lemma~\ref{composition:unif:lemma} since in these cases
  $\ufold{S'(S_1,\ldots,S_k)}{\mathbf{s}} = S'\big(\ufold{S_1}{\mathbf{s}},\ldots,\ufold{S_k}{\mathbf{s}}\big)$ and   $\ufold{R'(R_1,\ldots,R_l)}{\mathbf{r}} = R'\big(\ufold{R_1}{\mathbf{r}},\ldots,\ufold{R_l}{\mathbf{r}}\big)$,
  since the induction hypothesis can be applied on each $\ufold{S_i}{\mathbf{r}}$ and $\ufold{R_j}{\mathbf{r}}$, for $i\in\set{1,\ldots,k}$ and $j\in\set{1,\ldots,l}$.
   
\item \label{case:3:lemma:unif:morphism:appendix} If $S$ is fixed-point $S=\mu X.S'(X)$, with $X \in \set{X_1,\ldots,X_s}$, then  $S$ is replaced by $S'(S)$ in the unification, and thus  we reduce this case to the case~\ref{case:2:lemma:unif:morphism:appendix} above as follows: 
\begin{align*}
      \tuple{S,R,\Eu{M}} =\tuple{\mu X.S'(X),R,\Eu{M}}  \reduces  
\begin{cases} 
  \mu Z. \tuple{S'(S),R,\Eu{M}'} & \tif (S,R,\cdot) \notin \Eu{M}\\
   Z                                & \tif (S,R,Z) \in \Eu{M}
\end{cases}
\end{align*}
  where $Z = \fresh{S,R}$  and $\Eu{M}' = \Eu{M} \cup \set{(S,R,Z)}$,  and 
\begin{align*}
    \tuple{\ufold{S}{\mathbf{s}},\ufold{R}{\mathbf{r}},\emptyset}  & = \tuple{\ufold{\mu X. S'(X)}{\mathbf{s}},\ufold{R}{\mathbf{r}},\emptyset}  \\
& = 
\begin{cases} 
 \tuple{\emptylist,\ufold{R}{\mathbf{r}},\emptyset}  & \tif \mathbf{s}(X)=0\\
 \tuple{\ufold{\mu X.S'(X)}{\mathbf{s}},\ufold{R}{\mathbf{r}},\emptyset}   & \tif   \mathbf{s}(X)>0
\end{cases} \\
& = 
\begin{cases} 
\emptylist  & \tif \mathbf{s}(X)=0\\
 \tuple{\ufold{\mu X.S'(X)}{\mathbf{s}},\ufold{R}{\mathbf{r}},\emptyset}   & \tif   \mathbf{s}(X)>0.
\end{cases}
\end{align*}
If $\mathbf{s}(X)=0$ then this case is trivial since there is by definition a \qsim between any fixed-point \ce and $\emptylist$, as well as between any fixed-point variable $Z$ and $\emptylist$.
If $\mathbf{s}(X)>0$ and $(S,R,Z) \in \Eu{M}$ then there is by definition a \qsim  between $Z$ and $\NF\big(\tuple{\ufold{\mu X.S'(X)}{\mathbf{s}}\big),\ufold{R}{\mathbf{r}},\emptyset}\big)$.
If $\mathbf{s}(X)>0$ and $(S,R,\cdot) \notin \Eu{M}$ then it follows from Eq. (\ref{unfold:equiv:lemma:eq:2:appendix}) of Lemma~\ref{unfold:equiv:lemma} that there exist a  fixed-point free \ce $\tilde{S}(X^1,\ldots,X^m)$ and \ces $S_1,\ldots,S_{m-1},S_m(X)$,  with  $m \ge 1$, such that 
$S'(X)$ can be written as $S'(X)=\tilde{S}(S_1,\ldots,S_{m-1},S_m(X))$. On the one hand, $S'(S)=\tilde{S}(S_1,\ldots,S_{m-1},S_m(S))$. On the other hand,
\begin{align*}
  \ufold{\mu X.S'(X)}{\mathbf{s}}=\tilde{S}\Big(\ufold{S_1}{\mathbf{s}},\ldots,\ufold{S_{m-1}}{\mathbf{s}},\bufold{S_m(\mu X.S'(X))}{\mathbf{s}'}\Big),
\end{align*}
where $\mathbf{s}'(X)=\mathbf{s}(X)-1$ and $\mathbf{s}'(X')=\mathbf{s}(X')$ for  $X'\neq X$.
This brings us back to the case~\ref{case:2:lemma:unif:morphism:appendix}  above in which the induction hypothesis can be applied on each $\ufold{S_i}{\mathbf{s}}$, for  $i=1,\ldots,m-1$ and for $\bufold{S_m(\mu X.S'(X))}{\mathbf{s}'}$.
\end{enumerate}
\end{proof}


\setcounter{theorem}{\value{counter:properties:morphism:lemma}}

\begin{lemma}
\label{properties:morphism:lemma:appendix}
Let $S$ and $R$ be  \ces with  bound  fixed-point variables $\boundv{S}=\set{ X_1,\ldots,X_s}$ and $\boundv{R}=\set{Y_1,\ldots,Y_r}$.
Let  $\mathbf{s}:\{X_1,\ldots,X_s\} \to  \mathbb{N}$ and $\mathbf{r}:\{X_1,\ldots,X_r\} \to  \mathbb{N}$ be iteration mappings.
Let $\phi_\mu$  be the mapping induced by the  $(\ceSet,\ceSetFree)$-quasi-simulation $\morphyy$ between $S \combb R$ and   $\ufold{S}{\mathbf{s}} \combb \ufold{R}{\mathbf{r}}$
constructed in the proof of Lemma~\ref{main:lemma:mophism:quasi}.
The mapping  $\phi_\mu$ enjoys the following properties.
\begin{enumerate}
\item \label{properties:morphism:lemma:item:1:appendix} 
  For any  fixed-point \ce $\mathbf{T}$  in   $S \combb R$, there exist  \ces $\mu X'.S'(X')$ and $R'$, mappings $\mathbf{s}':\{X_1,\ldots,X_s\} \to  \mathbb{N}$ and $\mathbf{r}':\{Y_1,\ldots,Y_r\} \to  \mathbb{N}$, and a memory $\Eu{M}'$ such that one of the four following cases holds.
  \begin{enumerate}
  \item \label{properties:morphism:lemma:item:1:1:appendix}   $\mathbf{T}=\NF\big(\tuple{\mu X'.S'(X),R',\Eu{M}'}\big)$ and 
                                                    $\phi_{\mu}(\mathbf{T})= \big(\ufold{\mu X'.S'(X')}{\mathbf{s}'} \combb \ufold{R'}{\mathbf{r}'}\big)$. 

  \item \label{properties:morphism:lemma:item:1:2:appendix}   $\mathbf{T}=\NF\big(\tuple{R',\mu X'.S'(X'),\Eu{M}'}\big)$ and
    $\phi_{\mu}(\mathbf{T})=\big(\ufold{R'}{\mathbf{r}'} \combb \ufold{\mu X'.S'(X')}{\mathbf{s}'}\big)$. 

   \item \label{properties:morphism:lemma:item:1:3:appendix}   $\mathbf{T}=\mu X'.S'(X')$, with $X' \in \set{X_1,\ldots,X_s}$  and $\mu X'.S'(X') \in \Phi_{\mu}(S)$, and 
        $\phi_{\mu}(\mathbf{T})= \ufold{\mu X'.S'(X')}{\mathbf{s}'}$.   

   \item \label{properties:morphism:lemma:item:1:4:appendix}   $\mathbf{T}=\mu X'.S'(X')$, with $X' \in \set{Y_1,\ldots,Y_s}$  and $\mu X'.S'(X') \in \Phi_{\mu}(R)$, and 
        $\phi_{\mu}(\mathbf{T})= \ufold{\mu X'.S'(X')}{\mathbf{r}'}$.   
  \end{enumerate}

\item \label{properties:morphism:lemma:item:2:appendix} For any   fixed-point sequence  
  \begin{align*}
     \mathbf{T}_1 \subcer \cdots \subcer  \mathbf{T}_m 
   \end{align*}
  in     $\Eu{T}(S \combb R)$ with $m \ge 1$ and   for any $i=1,\ldots,m$, there are  iteration  mappings   $\mathbf{s}_i:\{X_1,\ldots,X_s\} \to  \mathbb{N}$ and $\mathbf{r}_i:\{Y_1,\ldots,Y_r\} \to  \mathbb{N}$, such that
  one of the following two cases holds:
  \begin{enumerate}
  \item \label{case:a:appendix} There is a  \ce $S_i(X^i) \in \Phi(S)$ with $X^i \in \set{ X_1,\ldots,X_s}$, and a \ce $R_i \in \Phi(R)$ such that 
  \begin{align*}
    \phi_{\mu}(\mathbf{T}_i)= \ufold{\mu X^i.S_i(X^i)}{\mathbf{s}_i} \combb \ufold{R_i}{\mathbf{r}_i},
  \end{align*}
   and  for $i=1,\ldots,m-1$ and for any  $X \in \set{ X_1,\ldots,X_s}$ and any $Y \in \set{ Y_1,\ldots,Y_r}$,  we have that
 \begin{align}
   \label{iteration:X:eq:appendix}
   \mathbf{s}_{i+1}(X) = 
    \begin{cases}
      \mathbf{s}_{i}(X), & \tif X\neq X^i\\
      \mathbf{s}_{i}(X^i)-1, & \tif X = X^i
    \end{cases} && \tand &&     \mathbf{r}_{i+1}(Y) =  \mathbf{r}_{i}(Y) 
 \end{align}
  \item \label{case:b:appendix} There is a \ce  $S_i \in \Phi(S)$, and  a  \ce $R_i(Y^i) \in \Phi(R)$ with $Y^i \in \set{ Y_1,\ldots,Y_r}$, such that 
  \begin{align*}
    \phi_{\mu}(\mathbf{T}_i)=  \ufold{S_i}{\mathbf{s}_i} \combb \ufold{\mu Y^i.R_i(Y^i)}{\mathbf{r}_i}, 
      \end{align*}
\end{enumerate}
   and  for $i=1,\ldots,m-1$ and for any  $X \in \set{ X_1,\ldots,X_s}$ and any $Y \in \set{ Y_1,\ldots,Y_r}$,  we have that
\begin{align}
  \label{iteration:Y:eq:appendix}
   \mathbf{s}_{i+1}(X) =  \mathbf{s}_{i}(X) 
   && \tand &&  
   \mathbf{r}_{i+1}(Y) = 
    \begin{cases}
      \mathbf{r}_{i}(Y), & \tif Y\neq Y^i\\
      \mathbf{r}_{i}(Y^i)-1, & \tif Y = Y^i
    \end{cases}
  \end{align}
 
\end{enumerate}
\end{lemma} 
\begin{proof} 
  \begin{enumerate}
  \item For Item~\ref{properties:morphism:lemma:item:1:appendix}, we just replaced the  $(\ceSet,\ceSetFree)$-quasi-simulation relation  $\morphyy$ of Lemma~\ref{properties:quasi:simulation:lemma} by its induced mapping $\phi_\mu$.

\item  For  Item~\ref{properties:morphism:lemma:item:2:appendix}, we consider a fixed-point \ce $\mathbf{T}_{i}$ in $\Eu{T}(S\combb R)$ with $i\in\set{1,\ldots,m-1}$,
   and we shall see that  Eq.(\ref{iteration:X:eq:appendix}) and Eq.(\ref{iteration:Y:eq:appendix}) hold for
   any \ce $\mathbf{T}$ in $\Eu{T}(S\combb R)$ such that 
  \begin{align*}
         \mathbf{T}_{i} \subcer  \mathbf{T}.
  \end{align*}

  For this purpose, assume that 
  \begin{align}
    \label{assumption:X}
    \phi_{\mu}(\mathbf{T}_i)  &= \ufold{\mu X^i.S_i(X^i)}{\mathbf{s}_i} \combb \ufold{R_i}{\mathbf{r}_i}  
    \end{align}
    since the case when
    \begin{align}
      \label{assumption:Y}
    \phi_{\mu}(\mathbf{T}_i)  &=  \ufold{S_i}{\mathbf{s}_i}  \combb  \ufold{\mu Y^i.R_i(Y^i)}{\mathbf{r}_i} 
      \end{align}
      can be handled similarly. 
  It follows from Eq.(\ref{unfold:equiv:lemma:eq:2:appendix}) of Lemma~\ref{unfold:equiv:lemma} that there is a \ce $\tilde{S}$  such that   $\ufold{\mu X^i.S_i(X^i)}{\mathbf{s}_i}$  can be written as
  \begin{align}
    \label{iteration:X:eq:appendix:1}
    \ufold{\mu X^i.S_i(X^i)}{\mathbf{s}_i} =  \ufold{\tilde{S}}{\mathbf{s}'_i}  &&\textrm{ where } &&
    \mathbf{s}'_i(X)=\begin{cases}
                           \mathbf{s}_i(X)  & \tif X\neq X^i \\
                           \mathbf{s}_i(X)-1  & \tif X=X^i.
                         \end{cases}
  \end{align}
  On the other hand, notice that there is a fixed-point free \ce $\hat{S}^i(X_1,\ldots,X_k)$ (resp. $\hat{R}^i(Y_1,\ldots,Y_l)$) with $k\ge 1$ (resp. $l\ge 1$),
  and fixed-point \ces $\xi_1,\ldots,\xi_k$ (resp. $\zeta_1,\ldots,\zeta_k$) each one is in $\Phi_{\mu}(S)$ (resp.  $\Phi_{\mu}(R)$), such that
  $\ufold{\tilde{S}}{\mathbf{s}'_i}$ (resp. $\ufold{R_{i+1}}{\mathbf{r}_i}$) can be written as 
  \begin{align*}
    \ufold{\tilde{S}}{\mathbf{s}'_i}  &= \hat{S}^i(\ufold{\xi_1}{\mathbf{s}'_i},\ldots,\ufold{\xi_k}{\mathbf{s}'_i}) \\
        \ufold{R_{i+1}}{\mathbf{r}_i} & = \hat{R}^i(\ufold{\zeta_1}{\mathbf{r}_i},\ldots,\ufold{\zeta_l}{\mathbf{r}_i}).
    \end{align*}

  It follows from the composition Lemma~\ref{composition:unif:lemma} that there is a fixed-point free \ce  $T_i(Z_1,\ldots,Z_m)$ such that
  \begin{align*}
    \phi_{\mu}(\mathbf{T}_i)  &= \ufold{\mu X^i.S_i(X^i)}{\mathbf{s}_i} \combb \ufold{R_i}{\mathbf{r}_i} \\
    &= \hat{S}^i(\ufold{\xi_1}{\mathbf{s}'_i},\ldots,\ufold{\xi_k}{\mathbf{s}'_i}) \combb \hat{R}^i(\ufold{\zeta_1}{\mathbf{r}_i},\ldots,\ufold{\zeta_l}{\mathbf{r}_i})  \\
    &=T(\alpha_1,\ldots,\alpha_m)     
  \end{align*}
  such that  for any $v=1,\ldots,m$, one the following two cases holds.
  \begin{enumerate}
  \item
There is $w \in \set{1,\ldots,k}$, and a \ce $\tilde{R}^v(Y^1,\ldots,Y^{l'})$ that is a sub-\ce of  $\hat{R}^i(Y_1,\ldots,Y_l)$ with $l'\le l$,
     and a set $\set{\zeta^1,\ldots,\zeta^{l'}} \subseteq \set{\zeta_1,\ldots,\zeta_l}$ such that
    \begin{align*}
      \alpha_v= \ufold{\xi_w}{\mathbf{s}'_i}  \combb \tilde{R}^v(\ufold{\zeta^1}{\mathbf{r}_i},\ldots,\ufold{\zeta^{l'}}{\mathbf{r}_i}) &&\tor&&    \alpha_v= \ufold{\xi_w}{\mathbf{s}'_i}.
    \end{align*}
    But since $\tilde{R}^v(\ufold{\zeta^1}{\mathbf{r}_i},\ldots,\ufold{\zeta^{l'}}{\mathbf{r}_i})=\ufold{\tilde{R}^v(\zeta^1,\ldots,\zeta^{l'})}{\mathbf{r}_i}$,
    $\mathbf{T}_i \subcer \phi_{\mu}^{-1}(\alpha_v)$, and 
    the iteration mappings $s_i'$ and $s_i$ satisfy Eq.(\ref{iteration:X:eq:appendix:1}), then we get Eq.(\ref{iteration:X:eq:appendix}).

  \item
        There is $w \in \set{1,\ldots,k}$, and a \ce $\tilde{S}^v(X^1,\ldots,X^{k'})$ that is a sub-\ce of  $\hat{S}^i(X_1,\ldots,X_k)$ with $k'\le k$,
        and a set $\set{\xi^1,\ldots,\xi^{k'}} \subseteq \set{\xi_1,\ldots,\xi_k}$ such that
    \begin{align*}
      \alpha_v=   \tilde{S}^v(\ufold{\xi^1}{\mathbf{s}'_i},\ldots,\ufold{\xi^{k'}}{\mathbf{s}'_i}) \combb \ufold{\zeta_w}{\mathbf{r}_i} &&\tor&&    \alpha_v= \ufold{\zeta_w}{\mathbf{r}_i}.
    \end{align*}
    But since $\tilde{R}^v(\ufold{\zeta^1}{\mathbf{r}_i},\ldots,\ufold{\zeta^{l'}}{\mathbf{r}_i})=\ufold{\tilde{R}^v(\zeta^1,\ldots,\zeta^{l'})}{\mathbf{r}_i}$ and
    the we have $\mathbf{T}_i \subcer \phi_{\mu}^{-1}(\alpha_v)$, and 
    the iteration mappings $s_i'$ and $s_i$ satisfy Eq.(\ref{iteration:X:eq:appendix:1}), then we get Eq.(\ref{iteration:X:eq:appendix}).
  \end{enumerate}
  In summary, we  assumed  that $\phi_{\mu}(\mathbf{T}_i)$  satisfies   Eq.(\ref{assumption:X}) and we get  Eq.(\ref{iteration:X:eq:appendix}).
  However  if we  assume   that $\phi_{\mu}(\mathbf{T}_i)$  satisfies   Eq.(\ref{assumption:Y}) then  we get  Eq.(\ref{iteration:Y:eq:appendix}) by similar arguments.
 \end{enumerate}
 \end{proof} 


\subsection{Proofs for Section~\ref{equiv:unif:with:unif:unfolding:sec}}

Before proving  Lemma~\ref{two:bound:of:Omega}, we want to get a certain  fixed-point \ce from each  $\mathbf{T}_i$ of the sequence $\Eu{S}$.
More precisely, notice that for any $i\in \set{1,\ldots,m}$, one of the following situations holds.
\begin{enumerate}[i.)]
\item If $\mathbf{T}_i=\NF\tuple{S_i,R_i,\Eu{M}_i}$, then either $S_i$ is a fixed-point \ce regardless of $R_i$ that could be a fixed-point \ce as well,
  or   $R_i$ is a fixed-point \ce and $S_i$ is not.
  In the first case we  want to get $S_i$, and in the second we want to get $R_i$.
\item   Otherwise, if  $\mathbf{T}_i$ is a fixed-point sub-\ce of $S$ or $R$, then we want to get $\mathbf{T}_i$.
\end{enumerate}
The formal definition follows.
\begin{definition}
For any $i\in \set{1,\ldots,m}$, we define
  \begin{align*}
    \nbrr{(\mathbf{T}_i)} = \begin{cases}
      S_i & \tif \mathbf{T}_i =\NF\tuple{S_i,R_i,\Eu{M}_i} \tand S_i \in \Phi_{\mu}(S) \\
      R_i & \tif \mathbf{T}_i =\NF\tuple{S_i,R_i,\Eu{M}_i} \tand R_i \in \Phi_{\mu}(R) \\
      \mathbf{T}_i & \tif \mathbf{T}_i \in   \widetilde{\Phi}_{\mu}(S) \cup \widetilde{\Phi}_{\mu}(R). 
    \end{cases}
  \end{align*}
\end{definition}

We need the following simple Fact.
\begin{fact}
  \label{set:max:fact}
  For any finite sets $A,A',B,B' \subset \N$,
  \begin{enumerate}[(i)]
  \item if  $\max(A) \le \max(A')$ and if $\max(B) \le \max(B')$ then $\max(A \cup B) \le \max( A' \cup B')$.
  \item \label{set:max:fact:item:2} Therefore, to show that $\max(A \cup B) \le \max( A' \cup B')$, it suffices to show that    $\max(A) \le \max(A')$ and   $\max(B) \le \max(B')$.
  \end{enumerate}
\end{fact}

\setcounter{theorem}{\value{counter:two:bound:of:Omega}}

\begin{lemma}
  \label{two:bound:of:Omega:appendix}
For any left-maximal sequence 
  \begin{align*}
   \mathbf{T}_1  \subcer \cdots \subcer  \mathbf{T}_m  
  \end{align*}
 in $\Eu{T}$ with $m\ge 2$, and  for any $p$ and $q$ where $1 \le p <q \le m$,
\begin{enumerate}
\item If for $i=1,\ldots,q$, there are \ces $S_i \in \widetilde{\Phi}(S)$ and $R_i \in \widetilde{\Phi}(R)$, and iteration mappings \\$\mathbf{s}_i:\set{X_1,\ldots,X_s} \to \N$ and $\mathbf{r}_i:\set{Y_1,\ldots,Y_s} \to \N$ such that 
  \begin{align*}
  \phi_{\mu}(\mathbf{T}_i)& = \ufold{S_i}{\mathbf{s}_i} \combb  \ufold{R_i}{\mathbf{r}_i}
\end{align*}
then
\begin{align}
       \omega(\mathbf{T}_q) &\in \set{D^{\star}\big((\mathbf{s}_1,\mathbf{r}_1),(\mathbf{s}_q,\mathbf{r}_q)\big), D^{\star}\big((\mathbf{s}_1,\mathbf{r}_1),(\mathbf{s}_q,\mathbf{r}_q)\big)-1}.       \label{two:bound:of:Omega:eq:1:appendix}
  \end{align}

\item If  there is a \ce $\xi_m\in \widetilde{\Phi}(S) \cup \widetilde{\Phi}(R)$ and an iteration mapping $\mathbf{s}_m$ such that
\begin{align*}
  \phi_{\mu}(\mathbf{T}_m)& = \ufold{\xi_m}{\mathbf{s}_m} 
\end{align*}
then
\begin{align}
    \label{two:bound:of:Omega:eq:3:appendix}
    \min\set{\mathbf{s}_m(X) \gvert X \in \dom(\mathbf{s}_m)} \ge \D(\mathbf{T}_m) .
  \end{align}
\end{enumerate}
\end{lemma}    
\begin{proof}  
\begin{enumerate}
\item Since $\omega(\mathbf{T}_q)=n-\Omega^{\#}(\mathbf{T}_1,\mathbf{T}_q)$, showing Eq.(\ref{two:bound:of:Omega:eq:1:appendix}) amounts to show
\begin{align}
  \label{two:bound:of:Omega:eq:1-1}
   n-\Omega^{\#}(\mathbf{T}_1,\mathbf{T}_q) &\in \set{D^{\star}\big((\mathbf{s}_1,\mathbf{r}_1),(\mathbf{s}_q,\mathbf{r}_q)\big),  D^{\star}\big((\mathbf{s}_1,\mathbf{r}_1),(\mathbf{s}_q,\mathbf{r}_q)\big)-1}.
\end{align}  
Since   $\mathbf{s}_1 \ge \mathbf{s}_q$ and  $\mathbf{r}_1 \ge \mathbf{r}_q$,  then  for  $v=1,\ldots,s$  there exist positive numbers  $\alpha_v^q$ where $n-\alpha_v^q \ge 0$,
and  for $w=1,\ldots,r$,  there exist positive numbers $\beta_w^q$ where $n -\beta_w^q \ge 0$,
such that the iteration mappings  $\mathbf{s}_q, \mathbf{r}_q$ can be  written as  
\begin{align*}
  \begin{cases}
    \mathbf{s}_q(X_v) &= n- \alpha_v^q \\
  \mathbf{r}_q(Y_w) &= n-  \beta_w^q.
  \end{cases}
  \end{align*}
On the one hand,  from the Definition~\ref{distance:iteration:mapping:def} of $d^\star$ and $D^\star$, we get
\begin{align*}
  d^{\star}(\mathbf{s}_1,\mathbf{s}_q) &=
   \begin{cases} \min\set{\mathbf{s}_q(X_v) \gvert \mathbf{s}_q(X_v) \neq \mathbf{s}_1(X_v) \textrm{ for } v=1,\ldots,s}  &\tif \mathbf{s}_1 > \mathbf{s}_q \\
    \infty & \tif \mathbf{s}_1=\mathbf{s}_q
   \end{cases}\\
   &=\begin{cases} \min\set{n-\alpha_v^q \gvert n-\alpha_v^q \neq  n  \textrm{ for } v=1,\ldots,s}  &\tif \mathbf{s}_1 > \mathbf{s}_q \\
    \infty & \tif \mathbf{s}_1=\mathbf{s}_q
   \end{cases} \\
      &=\begin{cases} n-\max\set{\alpha_v^q  \gvert   \textrm{ for } v=1,\ldots,s}  &\tif \mathbf{s}_1 > \mathbf{s}_q \\
    \infty & \tif \mathbf{s}_1=\mathbf{s}_q.
    \end{cases}
\end{align*}
Similarly
\begin{align*}
  d^{\star}(\mathbf{r}_1,\mathbf{r}_q) &=
  \begin{cases} n-\max\set{\beta_w^q \gvert \textrm{ for } w=1,\ldots,r}  &\tif \mathbf{r}_1 > \mathbf{r}_q \\
    \infty & \tif \mathbf{r}_1=\mathbf{r}_q.
    \end{cases}
\end{align*}

Let 
\begin{align*}
  \mathbf{m}_{S} &= \max\set{\alpha_v^q \gvert   \textrm{ for } v=1,\ldots,s}\\
  \mathbf{m}_{R}&=  \max\set{\beta_w^q \gvert    \textrm{ for } w=1,\ldots,r}. 
\end{align*}
Hence
\begin{align*}
  D^{\star}\big((\mathbf{s}_1,\mathbf{s}_q), d^{\star}(\mathbf{r}_1,\mathbf{r}_q)\big)
         &= \min( d^{\star}(\mathbf{s}_1,\mathbf{s}_q), d^{\star}(\mathbf{r}_1,\mathbf{r}_q) ) \\
          & = \begin{cases}
               \min(n-\mathbf{m}_S,n-\mathbf{m}_R)   &\tif  \mathbf{s}_q>\mathbf{s}_1 \tand \mathbf{r}_q>\mathbf{r}_1\\
               n-\mathbf{m}_S &     \tif  \mathbf{s}_q>\mathbf{s}_1 \tand \mathbf{r}_q=\mathbf{r}_1   \\
               n-\mathbf{m}_R &     \tif  \mathbf{s}_q=\mathbf{s}_1 \tand \mathbf{r}_q>\mathbf{r}_1   
  \end{cases}\\
            & = \begin{cases}
               n-\max(\mathbf{m}_S,\mathbf{m}_R)   &\tif  \mathbf{s}_q>\mathbf{s}_1 \tand \mathbf{r}_q>\mathbf{r}_1\\
               n-\mathbf{m}_S &     \tif  \mathbf{s}_q>\mathbf{s}_1 \tand \mathbf{r}_q=\mathbf{r}_1   \\
               n-\mathbf{m}_R &     \tif  \mathbf{s}_q=\mathbf{s}_1 \tand \mathbf{r}_q>\mathbf{r}_1   
             \end{cases}\\
\end{align*}

On the other hand,  since $\mathbf{T}_i=\NF(\tuple{S_i,R_i,\Eu{M}_i})$, for $i=1,\ldots, q$, then  consider the sequence of tuples
\begin{align*}
  \Eu{S}_q=\tuple{S_1,R_1,\Eu{M}_1},\ldots,\tuple{S_q,R_q,\Eu{M}_q}
\end{align*}
and recall  the  definition  of $\Omega^{\#}(\mathbf{T}_1,\mathbf{T}_q)$ from Eq.(\ref{definition:omega:Omega:def:Omega:1})  of Definition~\ref{definition:omega:Omega:def}:
\begin{align*}
\Omega^{\#}(\mathbf{T}_1,\mathbf{T}_q)=\max\set{\nbr_{\Eu{S}_q}(S_i), \nbr_{\Eu{S}_q}(R_i) \gvert S_i \in \widetilde{\Phi}_{\mu}(S), R_i \in \widetilde{\Phi}_{\mu}(R),  i=1,\ldots,q}.
\end{align*}

We distinguish three cases depending on the iteration mappings $\mathbf{s}_1, \mathbf{s}_q, \mathbf{r}_1,\mathbf{r}_q$.
\begin{itemize}
\item If $\mathbf{s}_q=\mathbf{s}_1$ then $d^{\star}(\mathbf{s}_1,\mathbf{s}_q)=\infty$, and $\mathbf{r}_q>\mathbf{r}_1$ and hence $d^{\star}(\mathbf{r}_1,\mathbf{r}_q) =n-\mathbf{m}_R$.
  In this case \\ $\max\set{\nbr_{\Eu{S}_q}(S_i)\gvert  S_i \in \widetilde{\Phi}_{\mu}(S),  i=1,\ldots,q} \in \set{0,1}$ and $\max\set{\nbr_{\Eu{S}_q}(R_i)\gvert  R_i \in \widetilde{\Phi}_{\mu}(R),  i=1,\ldots,q} \ge 1$, and hence $\Omega^{\#}(\mathbf{T}_1,\mathbf{T}_q)=\max\set{\nbr_{\Eu{S}_q}(R_i)\gvert  R_i \in \widetilde{\Phi}_{\mu}(R),  i=1,\ldots,q}$. Therefore, in this case showing Eq.(\ref{two:bound:of:Omega:eq:1-1})
  amounts  to show that
  \begin{align*}
      n-\max\set{\nbr_{\Eu{S}_q}(R_i) \gvert  R_i \in \widetilde{\Phi}_{\mu}(R),  i=1,\ldots,q} &\in \set{ n-\mathbf{m}_{R}, n-\mathbf{m}_{R}-1}, i.e. \\
      \max\set{\nbr_{\Eu{S}_q}(R_i) \gvert  R_i \in \widetilde{\Phi}_{\mu}(R),  i=1,\ldots,q} &\in \set{  \mathbf{m}_{R},  \mathbf{m}_{R}+1}.
  \end{align*}

\item If $\mathbf{r}_q=\mathbf{r}_1$ then $d^{\star}(\mathbf{r}_1,\mathbf{r}_q)=\infty$, and $\mathbf{s}_q>\mathbf{s}_1$ and hence $d^{\star}(\mathbf{s}_1,\mathbf{s}_q) = n-\mathbf{m}_S$.
  With similar reasoning, in this case we need to show that
  \begin{align*} \max\set{\nbr_{\Eu{S}_q}(S_i)\gvert  S_i \in \widetilde{\Phi}_{\mu}(S),  i=1,\ldots,q} \in \set{\mathbf{m}_{S}, \mathbf{m}_{S}+1}. \end{align*}

\item If $\mathbf{s}_q>\mathbf{s}_1$ and  $\mathbf{r}_q>\mathbf{r}_1$ then $d^{\star}(\mathbf{s}_1,\mathbf{s}_q) =n-\mathbf{m}_S$ and $d^{\star}(\mathbf{r}_1,\mathbf{r}_q) =n-\mathbf{m}_R$.
  In this case showing Eq.(\ref{two:bound:of:Omega:eq:1-1}) amounts to show
  \begin{align*}
   n-  \max\big\{\nbr_{\Eu{S}_q}(S_i), \nbr_{\Eu{S}_q}(R_i) \gvert S_i \in \widetilde{\Phi}_{\mu}(S),& R_i \in \widetilde{\Phi}_{\mu}(R),  i=1,\ldots,q\big\}  \in \\ & \set{n-\max(\mathbf{m}_S,\mathbf{m}_R),n-\max(\mathbf{m}_S,\mathbf{m}_R)-1}.
  \end{align*}
  That is,
    \begin{align*}
   \max\big\{\nbr_{\Eu{S}_q}(S_i), \nbr_{\Eu{S}_q}(R_i) \gvert S_i \in \widetilde{\Phi}_{\mu}(S),& R_i \in \widetilde{\Phi}_{\mu}(R),  i\in[1,q]\big\}  \in  \set{\max(\mathbf{m}_S,\mathbf{m}_R),\max(\mathbf{m}_S,\mathbf{m}_R)+1}.
  \end{align*}
  It follows from Item (\ref{set:max:fact:item:2}) of Fact~\ref{set:max:fact} that to show  Eq.(\ref{two:bound:of:Omega:eq:1-1}) it suffices to show that
  \begin{align*}
    \begin{cases}
     \max\set{\nbr_{\Eu{S}_q}(S_i)\gvert S_i \in \widetilde{\Phi}_{\mu}(S), i=1,\ldots,q} &\in  \set{\mathbf{m}_{S},\mathbf{m}_{S}+1} \\
     \max\set{\nbr_{\Eu{S}_q}(R_j)\gvert S_i \in \widetilde{\Phi}_{\mu}(R), i=1,\ldots,q} &\in  \set{\mathbf{m}_{R},\mathbf{m}_{R}+1}.
    \end{cases}
    \end{align*}

\end{itemize}

Summing up these three cases, to show  Eq.(\ref{two:bound:of:Omega:eq:1-1}) it suffices to assume that   $\mathbf{s}_q>\mathbf{s}_1$ and to show 
\begin{align}
 \label{two:bound:of:Omega:eq:final}
  \max\set{\nbr_{\Eu{S}_q}(S_i)\gvert S_i \in \widetilde{\Phi}_{\mu}(S), i=1,\ldots,q} &\in  \set{\mathbf{m}_{S},\mathbf{m}_{S}+1} 
\end{align}


Let $\xi \in \widetilde{\Phi}_{\mu}(S) \cap \set{S_1,\ldots,S_q}$ be a fixed-point  \ce. Indeed,  $\xi$    appears  $\nbr_{\Eu{S}_q}(\xi)$ times in $\Eu{S}_{q}$ and let $\tilde{q}$ be the greatest $i \in \set{1,\ldots,q}$
such that $\xi =S_i$. Since $\xi$ is by definition  a fixed-point \ce, then it can be written as $\xi = \mu X_{\tilde{v}}.\tilde{S}(X_{\tilde{v}})$, for some $\tilde{v} \in \set{1,\ldots,s}$ and for some \ce
$\tilde{S}(X_{\tilde{v}}) \in \widetilde{\Phi}(S)$.
To show Eq.(\ref{two:bound:of:Omega:eq:final}), is suffices to show that either
\begin{enumerate}[(i)]
\item \label{prof:link:Omega:d:item:1} $\tilde{q}=q$ and in this case $\nbr_{\Eu{S}_q}(\xi)=\alpha^q_{\tilde{v}}+1$, or
\item  \label{prof:link:Omega:d:item:2} $\tilde{q}\neq q$ and in this case  $\nbr_{\Eu{S}_q}(\xi)=\alpha^q_{\tilde{v}}$.
\end{enumerate}
The proof is by induction on $q$.
For the base case  $q=1$, we claim that $S_1$ is a fixed-point \ce  because otherwise $\mathbf{s}_1=\mathbf{s}_{2}$ which contradicts the assumption  $\mathbf{s}_{2} > \mathbf{s}_1$.
Hence let $S_1=\xi$. Recall  that $\mathbf{s}_1(X_v)=n$  for $v=1,\ldots,s$. In this case it  follows from  Eq.(\ref{iteration:X:eq:appendix})  of Item~\ref{properties:morphism:lemma:item:2:appendix} of Lemma~\ref{properties:morphism:lemma}  that
$\mathbf{s}_{2}(X_{\tilde{v}})=\mathbf{s}_{1}(X_{\tilde{v}})-1=n-1$ and that  $\mathbf{s}_{2}(X_{v})=\mathbf{s}_{1}(X_{v})$ for any $v\in \set{1,\ldots,s} \setminus \set{\tilde{v}}$.
That is, $\alpha^{2}_{\tilde{v}} =1$, and $\alpha^{2}_{v} = 0$ for any $v \neq \tilde{v}$.
\begin{enumerate}[(i)]
\item \label{prof:link:Omega:d:item:1'} If $S_{2}=\xi$, i.e. $\tilde{q}=q=2$,  then $\nbr_{\Eu{S}_q}(\xi)=2=\alpha_{\tilde{v}}^2 +1$.
\item  \label{prof:link:Omega:d:item:2'} If $S_{2}\neq\xi$, i.e.  $\tilde{q}=1 \neq q=2$, then  in this case  $\nbr_{\Eu{S}_q}(\xi)=1=\alpha^2_{\tilde{v}}$.
\end{enumerate}
For the induction step assume that the claim holds for  $q$ and let us prove it for $q+1$. 

\begin{enumerate}[(1)]
\item \label{prof:link:Omega:d:item:1''} If $S_{q}=\xi$,   then by the induction hypothesis $\nbr_{\Eu{S}_q}(\xi)=\alpha_{\tilde{v}}^q +1$. Besides, from Eq.(\ref{iteration:X:eq:appendix})  of Item~\ref{properties:morphism:lemma:item:2:appendix} of Lemma~\ref{properties:morphism:lemma} we have that $\alpha_{\tilde{v}}^{q+1}=\alpha_{\tilde{v}}^q +1$.  
\begin{enumerate}[(i)]
\item If $S_{q+1}=\xi$, i.e. $\tilde{q}=q+1$, then in this case we have $\nbr_{\Eu{S}_{q+1}}(\xi)=\nbr_{\Eu{S}_{q}}(\xi)+1=(\alpha_{\tilde{v}}^q +1)+1=\alpha_{\tilde{v}}^{q+1} +1$.     
\item If $S_{q+1}\neq \xi$, i.e. $\tilde{q}=q$, then in this case we have $\nbr_{\Eu{S}_{q+1}}(\xi)=\nbr_{\Eu{S}_{q}}(\xi)=\alpha_{\tilde{v}}^q +1=\alpha_{\tilde{v}}^{q+1}$.     
\end{enumerate}
\item  \label{prof:link:Omega:d:item:2''} If $S_{q}\neq \xi$,   then by the induction hypothesis $\nbr_{\Eu{S}_q}(\xi)=\alpha_{\tilde{v}}^q$. Besides, from Eq.(\ref{iteration:X:eq:appendix})  of Item~\ref{properties:morphism:lemma:item:2:appendix} of Lemma~\ref{properties:morphism:lemma} we have that $\alpha_{\tilde{v}}^{q+1}=\alpha_{\tilde{v}}^q$.  
\begin{enumerate}[(i)]
\item If $S_{q+1}=\xi$, i.e. $\tilde{q}=q+1$, then in this case we have $\nbr_{\Eu{S}_{q+1}}(\xi)=\nbr_{\Eu{S}_{q}}(\xi)+1=\alpha_{\tilde{v}}^q +1=\alpha_{\tilde{v}}^{q+1} +1$.     
\item If $S_{q+1}\neq \xi$, i.e. $\tilde{q}\neq q$ and $\tilde{q}\neq q+1$, then in this case we have $\nbr_{\Eu{S}_{q+1}}(\xi)=\nbr_{\Eu{S}_{q}}(\xi)=\alpha_{\tilde{v}}^q =\alpha_{\tilde{v}}^{q+1}$.     
\end{enumerate}
\end{enumerate}

\item  To show Eq.(\ref{two:bound:of:Omega:eq:3:appendix}), assume that $\xi_m \in \widetilde{\Phi}(S)$, the case where $\xi_m \in \widetilde{\Phi}(R)$ is similar. Let
  \begin{align*} \mathbf{m}=\min\set{\mathbf{s}_m(X) \gvert X \in \set{X_1,\ldots,X_s}}. \end{align*}
  Recall from Eq.(\ref{definition:omega:Omega:def:DS:2}) of  Definition~\ref{definition:omega:Omega:def}  of $\D$ that 
  \begin{align*}
         \D(\mathbf{T}_m)   &= \begin{cases}
           n & \tif m=1  \\
           \min\big\{\D(\mathbf{T}_{m-1}), d^{\star}(\mathbf{s}_1,\mathbf{s}_m)\big\} & \tif m>1.
         \end{cases}
  \end{align*}
  We want to show that $\mathbf{m} \ge \D(\mathbf{T}_{m-1})$.
  If $m=1$ then $\mathbf{m}=n=\D(\mathbf{T}_m)$, hence   the claim trivially holds.
  If $m>1$ then $\D(\mathbf{T}_m)= \min\big\{\D(\mathbf{T}_{m-1}), d^{\star}(\mathbf{s}_1,\mathbf{s}_m)\big\}$ and in this case we distinguish two cases depending on $\mathbf{s}_m$.
  If $\mathbf{s}_m=\mathbf{s}_1$ then $\mathbf{m}=n$,  $ \D(\mathbf{T}_{m-1})<n$ and  $d^{\star}(\mathbf{s}_1,\mathbf{s}_m)=\infty$, 
  thus  $ \D(\mathbf{T}_m)  = \D(\mathbf{T}_{m-1}) < n =\mathbf{m}$. Therefore,  $\mathbf{m} \ge \D(\mathbf{T}_{m})$.
  If $\mathbf{s}_m < \mathbf{s}_1$ then  in this case $\mathbf{m}= d^{\star}(\mathbf{s}_1,\mathbf{s}_m)$, which is obviously greater or equal to  $\min\big\{\D(\mathbf{T}_{m-1}), d^{\star}(\mathbf{s}_1,\mathbf{s}_m)\big\}$.

\end{enumerate}
\end{proof}  


\setcounter{theorem}{\value{counter:distance:fixed-point-in-tree}}

\begin{lemma} 
  \label{distance:fixed-point-in-tree:appendix}
Let 
  \begin{align*}
    \mu \bb{Z}_1.\bb{T}_1(\bb Z_1)  \subcer   \T
  \end{align*}
be a  sequence in $\partial \Eu{T}$.
 Define $\bb T^{\star}_{1}(Z)$ to  be the (unique) \ce satisfying
\begin{align*}
  \bb T^{\star}_1(\T)= \bb T_{1}(\bb Z_1).
\end{align*}
We have that
\begin{align}
  \label{distance:fixed-point-in-tree:eq:appendix}
   1  \le \Pi_{Z}(\bb T^{\star}_1(Z)).   
\end{align}
\end{lemma}
\begin{proof}
  The idea of the proof is to show that either   \emph{(i)} both $\mu \bb{Z}_1.\bb{T}_1(\bb Z_1)$ and  $\T$  result from the unification of the same \ce with another \ce,
  i.e. $\nbrr{(\mu \bb{Z}_1.\bb{T}_1(\bb Z_1))}=\nbrr{(\T)}$ and in this case we know from Lemma~\ref{position:cross:in:generated:ces:lemma} that there is at least one jump between the root of  $\mu \bb{Z}_1.\bb{T}_1(\bb Z_1)$ and $\T$.
  Or,
  \emph{(ii)} $\mu \bb{Z}_1.\bb{T}_1(\bb Z_1)$ and  $\T$  result from the unification of different \ces, i.e.  $\nbrr{(\mu \bb{Z}_1.\bb{T}_1(\bb Z_1))}=\nbrr{(\T)}$. In this case
  there must be  a \ce ${{T}}$  that lies in $\Eu{T}$ between $\mu \bb{Z}_1.\bb{T}_1(\bb Z_1)$ and  $\T$ such that    $\nbrr{({{T}})}=\nbrr{(\T)}$.
  Hence there is at least one jump between the root of ${{T}}$ and $\T$, and therefore  there is at least one jump between the root of  $\mu \bb{Z}_1.\bb{T}_1(\bb Z_1)$ and $\T$.
   We need the two following claims, where Claim~\ref{claim:1} is used to prove Claim~\ref{claim:2} which will be used to prove  this Lemma.
  \begin{claim}
    \label{claim:1} Consider a sequence  $\Eu{S}_m$: $T_1 \subcer \ldots \subcer T_{m}$ in $\Eu{T}$  where $T_1$ is the root of $\Eu{T}$, with $m \ge 1$. 
    For any $q=1,\ldots,m$, if there are two \ces $M,M'$ such that $\Omega^{\#}(T_1,T_q)=\nbr_{\Eu{S}_q}{(M)}=\nbr_{\Eu{S}_q}{(M')}$ then  $T_q$ is not in $\partial \Eu{T}$.
  \end{claim}
  \begin{proof}
    Assume that there are only two \ces $M$ and $M'$ such that $\Omega^{\#}_{\Eu{s}}{(T_1,T_q)}=\nbr_{\Eu{S}_q}{(M)}=\nbr_{\Eu{S}_q}{(M')}$. The case where there are more than two can be handled similarly.
    Indeed, there is a \ce $T_p$ in $\Eu{S}_q$ on which the number of occurences   of $M$ (resp. or $M'$) has reached the maximum while that of $M'$ (resp. $M$) did  not.
    More precisely, there is $p<q$ such that either 
	$\Omega^{\#}(T_1,T_q)=\Omega^{\#}(T_1,T_p) = \nbr_{\Eu{S}_q}{(M)}=\nbr_{\Eu{S}_q}{(M')}+1$ or $\Omega^{\#}(T_1,T_q)=\Omega^{\#}(T_1,T_p) = \nbr_{\Eu{S}_q}{(M')}=\nbr_{\Eu{S}_q}{(M)}+1$.
	Assume that
    the first case holds since the second case can be handled similarly.
    Recall that $\omega(T_q)=\omega(T_p)$ since $\Omega^{\#}(T_1,T_q)=\Omega^{\#}(T_1,T_p)$.
    Towards a contradiction: assume that  $T_q$ is in  $\partial \Eu{T}$.
    If $T_p \in \partial\Eu{T}$  then  by Item (\ref{link:tree:derivative:rq:item:3}) of Remark~\ref{link:tree:derivative:rq}  we have  $\omega(T_q)=\omega(T_p)+1$, which is a contradiction.
    If $T_p$ is not in  $\partial\Eu{T}$ then by Item (\ref{link:tree:derivative:rq:item:2}) of Remark~\ref{link:tree:derivative:rq}  we have  $\omega(T_q)=\omega(T_p)+1$, which is a contradiction.
    This ends the proof of Claim~\ref{claim:1}.
  \end{proof}
  \begin{claim} \label{claim:2}
    Let $\TOne$ and  $\T$  be two \ces in $\partial\Eu{T}$ where $\TOne  \subcer \T$. If $\nbrr{(\TOne)}  \neq \nbrr{(\T)}$ then the sequence in $\Eu{T}$ that lies between $\TOne$ and $\T$ is not empty,
    and   there exists  a \ce ${{T}}$ in this sequence  such that  $\nbrr{({{T}})}=\nbrr{(\T)}$.
  \end{claim}
  \begin{proof}
    Assume that $\nbrr{(\TOne)}=M_1 $ and $\nbrr{(\T)}=M_2$. Let $\Eu{S}_1$ (resp. $\Eu{S}_2$) be the sequence in $\Eu{T}$ from the root of $\Eu{T}$ to $\TOne$ (resp. to $\T$).
    By  Item (\ref{link:tree:derivative:rq:item:3}) of Remark~\ref{link:tree:derivative:rq}
    we have $\nbr_{\Eu{S}_2}{(M_2)}=\nbr_{\Eu{S}_1}{(M_1)}+1$.   However, either there is at least one \ce, say ${{T}}$, in $\Eu{T}$  in the sequence between $\TOne$ and $\T$ such that $\nbrr{({{T}})}=\nbrr{(\T)}$, or
    $\nbr_{\Eu{S}_1}{(M_2)}=\nbr_{\Eu{S}_1}{(M_1)}$. But this second possibility is not possible, since otherwise, by  Claim~\ref{claim:1}  we would have had that $\TOne$  is not in $\partial{(\Eu{T})}$ which contradicts
    the assumption of the current claim.     This ends the proof of Claim~\ref{claim:2}.  
  \end{proof}

  To prove  Lemma~\ref{distance:fixed-point-in-tree}  we distinguish two cases depending whether the sequence  in $\Eu{T}$ that lies between $\mu \bb{Z}_1.\bb{T}_1(\bb Z_1)$ and $\T$ is empty or not.
  Let   $\Eu{S}'$ be such sequence.
  If  $\Eu{S}'$  is empty, then if follows from Claim~\ref{claim:2} that  $\nbrr{(\mu \bb{Z}_1.\bb{T}_1(\bb Z_1))}=\nbrr{(\T)}$.
  This means that during the unification process, the same \ce which is a sub-\ce of $S$ or $R$  appeared twice, which implies that  there is a position jump between $T^{\star}_1(Z)$ and $Z$.
  Or more formally, it  follows from Lemma~\ref{position:cross:in:generated:ces:lemma} that  $1 \le \Pi_{Z}(\bb T^{\star}_1(Z))$.
  Otherwise,  if the  sequence  $\Eu{S}'$ is not empty then from Claim ~\ref{claim:2} it follows that there
  is a \ce $\zeta$ in $\Eu{S}'$ such that $\nbrr{(\zeta)}=\nbrr{(\mf{\bb T}_2)}$. Since $\mf{\bb T}_2$ is a sub-\ce of $\zeta$, then there is a unique \ce $\zeta^{\star}(Z)$ such that $\zeta^{\star}(\mf{\bb T}_2)=\zeta$.
  Thus by using the same Lemma~\ref{position:cross:in:generated:ces:lemma} we deduce that $1 \le \Pi_{Z}(\zeta^{\star})$. But since $\zeta$ is a sub-\ce of $\bb T^{\star}_1(\bb T_2)$, then $\zeta^\star(Z)$ is a sub-\ce of $\bb T^{\star}_1(Z)$ and hence   $1 \le \Pi_{Z}(\bb T^{\star}_1(Z))$ as well.
  This ends the proof of Lemma~\ref{distance:fixed-point-in-tree}.
\end{proof}
